\newif\ifllncs
 \llncsfalse

\ifllncs
\documentclass[envcountsame]{llncs}

\else
\documentclass[letterpaper,11pt]{article}
\usepackage{amsthm}
\usepackage{fullpage}
\usepackage{authblk}
\fi

\usepackage{amsmath,amsfonts,amssymb}
\usepackage{breakcites}
\usepackage[pdfstartview=FitH,colorlinks,linkcolor=blue,citecolor=blue]{hyperref} 
\usepackage[margin=1cm]{caption}
\DeclareMathAlphabet{\mathbbold}{U}{bbold}{m}{n}

\usepackage{graphicx} 

\title{On One-Shot Signatures, Quantum vs Classical Binding, and Obfuscating Permutations\ifllncs{\thanks{The current document contains a proceedings version. The full version of this work can be found at \cite{shmueli2025one}.}}\fi}

\ifllncs
\author{Omri Shmueli \and Mark Zhandry} 
\date{} 
\institute{NTT Research} 

\else
\author[1]{Omri Shmueli}
\author[1,2]{Mark Zhandry}
\affil[1]{NTT Research}
\affil[2]{Stanford University}
\date{}
\fi

\ifllncs
\else
\newtheorem{theorem}{Theorem}
\newtheorem{lemma}[theorem]{Lemma}
\newtheorem{remark}[theorem]{Remark}
\newtheorem{corollary}[theorem]{Corollary}
\newtheorem{definition}[theorem]{Definition}
\newtheorem{proposition}[theorem]{Proposition}
\newtheorem{question}[theorem]{Question}
\fi
\newtheorem{construction}[theorem]{Construction}

\newcommand{\Z}{\mathbb{Z}}

\newcommand{\bv}{\mathbf{b}}

\newcommand{\rv}{\mathbf{r}}

\newcommand{\vv}{\mathbf{v}}

\newcommand{\xv}{\mathbf{x}}

\newcommand{\zv}{\mathbf{z}}

\newcommand{\Am}{\mathbf{A}}

\newcommand{\Id}{\mathbf{I}}

\newcommand{\As}{\mathcal{A}}
\newcommand{\Bs}{\mathcal{B}}
\newcommand{\Ds}{\mathcal{D}}

\newcommand{\Ms}{\mathcal{M}}
\newcommand{\Os}{\mathcal{O}}
\newcommand{\Ps}{\mathcal{P}}
\newcommand{\Qs}{\mathcal{Q}}

\newcommand{\Ss}{\mathcal{S}}

\newcommand{\Xs}{\mathcal{X}}
\newcommand{\Ys}{\mathcal{Y}}

\newcommand{\negl}{\mathsf{negl}}
\newcommand{\poly}{\mathsf{poly}}
\newcommand{\polylog}{\mathsf{polylog}}

\newcommand{\colspan}{\mathsf{ColSpan}}

\newcommand{\prp}{{\Pi}}
\newcommand{\merge}{{\sf M}}
\newcommand{\punc}{{\sf Punc}}
\newcommand{\eval}{{\sf Eval}}
\newcommand{\prf}{{\sf F}}
\newcommand{\prg}{{\sf PRG}}

\newcommand{\permute}{{\sf Permute}}
\newcommand{\iO}{{\sf iO}}
\newcommand{\gen}{{\sf Gen}}
\newcommand{\hash}{{\sf Hash}}
\newcommand{\sign}{{\sf Sign}}
\newcommand{\ver}{{\sf Ver}}
\newcommand{\setup}{{\sf Setup}}

\newcommand{\aux}{\mathsf{aux}}
\newcommand{\crs}{{\sf CRS}}
\newcommand{\pk}{{\sf pk}}
\newcommand{\sk}{{\sf sk}}

\newcommand{\transposition}[2]{{(#1\;\;#2)}}
\newcommand{\neighborswap}[1]{\transposition{#1}{#1+1}}

\newif\ifnotes
\notestrue

\ifnotes
\newcommand{\omri}[1]{$\ll$\textsf{\color{blue} Omri: { #1}}$\gg$}

\else
\newcommand{\omri}[1]{}
\fi

\newcommand{\ceil}[1]{\lceil {#1} \rceil}

\newcommand{\Hyb}{\mathsf{Hyb}}
\newcommand{\Oracle}{\mathcal{O}}
\newcommand{\Nat}{\mathbb{N}}
\newcommand{\bbZ}{\mathbb{Z}}
\newcommand{\Img}{\text{Img}}

\newcommand{\PRP}{\mathsf{PRP}}
\newcommand{\PRF}{\mathsf{PRF}}
\newcommand{\Adv}{\mathcal{A}}
\newcommand{\AdvB}{\mathcal{B}}
\newcommand{\secp}{\lambda}
\newcommand{\Obf}{\mathsf{O}}
\newcommand{\td}{\mathsf{td}}

\newcommand{\LF}{\mathsf{LF}}
\newcommand{\LFGen}{\LF\mathsf{.KeyGen}}
\newcommand{\LFF}{\LF\mathsf{.F}}

\newcommand{\hashL}{L} 
\newcommand{\hashQ}{Q} 

\newcommand{\linspan}{\mathsf{Span}}

\newcommand{\matA}{\mathbf{A}}
\newcommand{\matB}{\mathbf{B}}
\newcommand{\matC}{\mathbf{C}}

\newcommand{\vecA}{\mathbf{a}}
\newcommand{\vecB}{\mathbf{b}}
\newcommand{\vecC}{\mathbf{c}}
\newcommand{\vecD}{\mathbf{d}}
\newcommand{\vecE}{\mathbf{e}}
\newcommand{\vecF}{\mathbf{f}}

\newcommand{\vecR}{\mathbf{r}}
\newcommand{\vecS}{\mathbf{s}}
\newcommand{\vecT}{\mathbf{t}}
\newcommand{\vecU}{\mathbf{u}}
\newcommand{\vecV}{\mathbf{v}}
\newcommand{\vecW}{\mathbf{w}}
\newcommand{\vecX}{\mathbf{x}}
\newcommand{\vecZ}{\mathbf{z}}

\newcommand{\ket}[1]{|{#1}\rangle}

\begin{document}
\maketitle

\begin{abstract}
\noindent
One-shot signatures (OSS) were defined by Amos, Georgiou, Kiayias, and Zhandry (STOC'20). These allow for signing exactly one message, after which the signing key self-destructs, preventing a second message from ever being signed. While such an object is impossible classically, Amos et al observe that OSS may be possible using quantum signing keys by leveraging the no-cloning principle. OSS has since become an important conceptual tool with many applications in decentralized settings and for quantum cryptography with classical communication. OSS are also closely related to separations between classical-binding and collapse-binding for post-quantum hashing and commitments. Unfortunately, the only known OSS construction due to Amos et al. was only justified in a classical oracle model, and moreover their justification was ultimately found to contain a fatal bug. Thus, the existence of OSS, even in a classical idealized model, has remained open. 

\medskip 

We give the first standard-model OSS, with provable security assuming (sub-exponential) indistinguishability obfuscation (iO) and LWE. This also gives the first standard-model separation between classical and collapse-binding post-quantum commitments/hashing, solving a decade-old open problem. Along the way, we also give the first construction with unconditional security relative to a classical oracle. To achieve our standard-model construction, we develop a notion of permutable pseudorandom permutations (permutable PRPs), and show how they are useful for translating oracle proofs involving random permutations into obfuscation-based proofs. In particular, obfuscating permutable PRPs gives a trapdoor one-way permutation that is \textit{full-domain}, solving another decade-old-problem of constructing this object from (sub-exponential) iO and one-way functions.
\end{abstract}

\ifllncs
\pagestyle{plain}
\else

\thispagestyle{empty}
\newpage
\tableofcontents
\newpage
\thispagestyle{empty}

\fi

\section{Introduction}\label{sec:intro}

One-shot signatures (OSS) were originally proposed by Amos, Georgiou, Kiayias, and Zhandry \cite{STOC:AGKZ20}. Here, we have a signer who generates a verification key/signing key pair, publishes the verification key, but keeps the signing key secret. Later, using the secret signing key, the signer can sign any message of their choosing. However, doing so \emph{provably} destroys their signing key, rendering it impossible to sign even a second message relative to the same verification key. Note that the signer has full control over generating the verification key, signing key, message, and signature, and even with all this control they can never produce a second signature.

Such a protocol is clearly impossible classically. Using quantum protocols, however, it is conjectured by~\cite{STOC:AGKZ20} that an OSS is possible. Namely, the setup procedure now produces a \emph{quantum} signing key, and the (quantum) signing algorithm is such that it requires measuring the key, which destroys it and prevents subsequent signatures. Verification keys, messages, signatures, and even the verification algorithm itself remain classical. Such an object is not trivially impossible, since the measurement principle of quantum mechanics means that computing the signature may irreversibly destroy the signing key.

OSS yields many applications that are not otherwise known: smart contracts without any blockchain \cite{Sattath22}, overcoming lower-bounds in consensus protocols, a solution to the blockchain scalability problem \cite{coladangelo2020quantum}, quantum money with classical communication, and more. Because of these advantages, OSS has gained significant interest within the blockchain community, were a major focus of recent workshops~\cite{QSigWorkshop,QuantumMoneyWorkshop}, and it has been claimed that practical OSS would ``completely change the endgame of blockchains''~\cite{Drake23}. See Section~\ref{sec:motivation} for more discussion.

Unfortunately, the status of OSS has been unclear. OSS lives atop a hierarchy of notions relating to quantum money, with even milder versions being notoriously difficult to construct.~\cite{STOC:AGKZ20} observe that prior work of~\cite{FOCS:AmbRosUnr14,EC:Unruh16} can be used to give OSS relative to a \emph{quantum} oracle\footnote{That is, an oracle that performs general unitary transformations.}, but there was no known way to actually instantiate this oracle, even heuristically.~\cite{STOC:AGKZ20} then propose a construction relative to a \emph{classical} oracle\footnote{That is, an oracle implementing a classical function but which can be queried in superposition.}, which can then be heuristically instantiated using (post-quantum) indistinguishability obfuscation (iO). To date, this was the only plausible candidate construction.

To justify the plausibility of their construction,~\cite{STOC:AGKZ20} prove the security of their scheme in the classical oracle model. Unfortunately, their proof turned out to have a bug discovered by~\cite{Bartusek23}, which appears fatal and the proof in~\cite{STOC:AGKZ20} has been retracted. See Section~\ref{sec:bug} for discussion. This bug does not indicate an actual attack, but has left the existence of OSS uncertain.

\paragraph{Post-quantum hashing and commitments.} Even before quantum computers will enable powerful quantum protocols, they will pose a threat to current classical cryptography. Sometimes, even if the underlying building blocks are replaced with ``post-quantum'' equivalents, security vulnerabilities may remain. One notable example is the security of commitments: as observed by Unruh~\cite{EC:Unruh16}, the classical notion of binding for commitments is insufficient against quantum attacks. That is, even if the security game for commitments is ``upgraded'' to allow the adversary to run a quantum computer but is otherwise unchanged, this is not enough to guarantee meaningful security in a quantum world. Instead a stronger, inherently quantum, notion of \emph{collapse}-binding is needed. When instantiating commitments from hashing, (post-quantum) collision-resistance is likewise not enough, and instead a stronger notion of \emph{collapsing} is needed.

One may wonder if classical-binding / collision-resistance actually \emph{implies} collapse-binding / collapsing. Unruh argues that this is likely \emph{not} the case, by noting that a prior work of~\cite{FOCS:AmbRosUnr14} gives counterexample relative to a \emph{quantum} oracle. However, there is no known standard-model separation, or even a classical-oracle separation. It is entirely consistent with existing results that classical binding implies collapse-binding in a non-relativizing way. As evidence for this, known \emph{positive} results~\cite{AC:Unruh16,C:Zhandry22c} have given collapse-binding commitments/hashes from essentially all of the same post-quantum assumptions known to imply classical binding. Nevertheless, our understanding of the relationship between classical-binding and quantum collapsing-binding is far from complete.

Interestingly,~\cite{STOC:AGKZ20,TQC:DalSpo23} shows that a separation between classical and collapse-binding actually \emph{implies} a one-shot signature, showing that these questions are closely linked.\footnote{The first connection between such a separation and unclonable cryptography was due to~\cite{EC:Zhandry19b}, who shows that a separation implies the weaker object called quantum lightning. These works improved the implication to OSS.} This is in fact how the construction in~\cite{STOC:AGKZ20} works.

\paragraph{One-way permutations from obfuscation?} The concept of cryptographically useful \emph{program obfuscation} dates to the pioneering work of Diffie-Hellman~\cite{DifHel76}. Obfuscation allows for embedding cryptographic secrets into public software. The proposal of~\cite{DifHel76}, though not phrased in this language, is the following: Obfuscate a pseudorandom permutation (PRP) to get a \emph{trapdoor} one-way permutation. More generally, obfuscation heuristically translates security using oracles into standard-model security, by actually giving out the obfuscated code of the oracles.

It ultimately took over 30 years for cryptographically useful general-purpose obfuscation to emerge \cite{FOCS:GGHRSW13}. Unfortunately, by then it was shown that obfuscation cannot in general translate oracles into the standard model~\cite{C:BGIRSV01}. For example, obfuscating an arbitrary PRP provably cannot guarantee any meaningful security. Instead, the community has settled on a much weaker but precise notion of \emph{indistinguishability obfuscation} (iO), which says that the obfuscations of programs with the same functionality are computationally indistinguishable~\cite{C:BGIRSV01}. 

Though much weaker than ideal, numerous techniques have been developed to use iO. Most revolve around the use of pseudorandom \emph{functions} (PRFs). While PRFs, just like PRPs, in general cannot be obfuscated, a strengthening known as a \emph{puncturable} PRF~\cite{CCS:KPTZ13,AC:BonWat13,PKC:BoyGolIva14} actually does give provable guarantees when using iO~\cite{STOC:SahWat14}, leading to numerous positive results.

Despite many successes, some major open questions remain. Notably, Diffie and Hellman's original proposal of obtaining a trapdoor permutation by obfuscating a PRP remains open. In fact, to the best of our knowledge, there are no known positive iO results that obfuscate PRPs to achieve \emph{any} goal. One key challenge is that there has been no construction of a puncturable \emph{PRP} or analogous object, and some evidence suggests that they may be \emph{impossible}~\cite{TCC:BonKimWu17}. 

Note that~\cite{TCC:BitPanWic16,EC:GPSZ17} construct trapdoor permutations using iO, though these permutations are not \emph{full-domain}, i.e., where if $n$ is the number of bits needed to represent any element in the domain (formally, the permutation domain $\chi$ is such that $\chi \subseteq \{ 0, 1 \}^{n}$), then we have the equality $\chi = \{0,1\}^n$. These works use a very different structure and do not simply obfuscate a PRP. Moreover, these solutions have a major drawback that the domain $\chi$ of the permutation is a sparse set in $\{ 0, 1 \}^{n}$ that cannot be directly sampled from nor efficiently recognized. This complicates their use in applications (see for example~\cite{TCC:CanLic18}) and it is unknown how to use them in quantum settings \cite{ITCS:MorYam23}, where one further wants to efficiently generate quantum superpositions over subsets of the domain. An open question is if full-domain (trapdoor) OWPs are possible from iO, whether by obfuscating a PRP or by other means. There is some indication that such a OWP may in fact be impossible~\cite{TCC:AshSeg16}.

\subsection{Results and Paper Organization}
There are two parts to our work.

\medskip
\noindent
\paragraph{Part 1: Classical Oracle OSS.} The first part proves Theorems \ref{thm:main1} and \ref{thm:main2} given below. It includes an oracle construction of a non-collapsing hash function and an unconditional proof of collision resistance. A proof overview is given in Section~\ref{sec:overview} and the full proof is in \ifllncs{the full version \cite{shmueli2025one}}\else{Section \ref{sec:oss_oracle}}\fi.
\begin{theorem} \label{thm:main1}
Relative to a classical oracle, secure OSS exists.
\end{theorem}
Similarly to the blueprint suggested in~\cite{STOC:AGKZ20}, we prove Theorem~\ref{thm:main1} through separating quantum hashing notions:
\begin{theorem} \label{thm:main2}
Relative to a classical oracle, there exist post-quantum collision-resistant hash functions that are non-collapsing, and there exist post-quantum classically-binding commitments that are not collapse-binding.
\end{theorem}
We do not actually know how to prove security for the oracle in~\cite{STOC:AGKZ20}. Instead, we prove Theorems~\ref{thm:main1} and~\ref{thm:main2} by a new construction.
Our construction is inspired by the previous work of \cite{STOC:AGKZ20} and a proposal made by \cite{Bartusek23} at the NTT Research Quantum Money Workshop~\cite{QuantumMoneyWorkshop}.
We prove security through a new technique, showing the random self-reducibility of our collision problem, and a sequence of reductions to simpler problems. Ultimately we show collision finding in our hash is no easier than collision finding in plain 2-to-1 random functions, which is known to be hard. This gives the first classical-oracle construction of OSS\footnote{By ``classical-oracle construction,'' we mean the first construction \emph{provably} and \emph{unconditionally} secure relative to a classical oracle.}, and also the first such separation of classical- and collapse-binding. The construction is even the first classical-oracle construction of \emph{quantum lightning}, a weaker notion than OSS, proposed in~\cite{EC:Zhandry19b}. 

\medskip
\noindent
\paragraph{Part 2: Standard-model OSS.} The second part of this paper proves Theorems \ref{thm:main3}, \ref{thm:main4}, \ref{thm:main5} and \ref{thm:main6} below. It includes the development of a new cryptographic notion we call permutable pseudorandom permutations (permutable PRPs). We show how permutable PRPs can be used to prove the security of our construction in the standard model, and also for solving a number of long-standing open problems in cryptography, as elaborated below.

\begin{theorem}\label{thm:main3}
There exists secure OSS assuming each of the following: (1) sub-exponentially-secure indistinguishability obfuscation, (2) sub-exponentially-secure one-way functions, and (3) (polynomially-secure) LWE with a sub-exponential noise-modulus ratio.
\end{theorem}
\begin{theorem}\label{thm:main4}
Under the same assumptions as Theorem~\ref{thm:main3}, there exist post-quantum collision-resistant hash functions that are non-collapsing, and there exist post-quantum classically-binding commitments that are not collapse-binding.
\end{theorem}
As before, Theorem~\ref{thm:main3} follows from Theorem~\ref{thm:main4}, which is proved in \ifllncs{the full version \cite{shmueli2025one}}\else{Section \ref{sec:standard}}\fi. This gives the first standard-model OSS\footnote{By ``standard-model,'' we mean with provable security under widely-used computational assumptions, as opposed to merely conjecturing that a construction is secure.}, and also the first standard-model separation between classical- and collapse-binding, solving this decade-old question\footnote{\cite{EC:Unruh16} was first made public in early 2015.}. Even for the weaker notion of quantum lightning, for which there exist a handful of candidates, this is the first construction with provable security under widely-studied assumptions.

The basic idea is to obfuscate the functions in our oracle construction rather than putting it in an oracle. The main technical gap between proving security in the oracle model and the standard model stems from the fact that the construction uses random permutations. Our oracle construction uses a (truly) random permutation and our standard model construction uses a pseudorandom permutation (PRP). Obfuscating PRPs and getting any formal security guarantees, however, is a known challenge in cryptography, independently of quantum computation.

We develop a new notion of PRPs, called \emph{permutable} PRPs, that can be obfuscated using iO with provable security. A permutable PRP $\prp$ very roughly allows the following: Given a key $k$ and a (known) permutation $\Gamma$ in the form of having its circuit, one can produce a ``permuted'' key $k^\Gamma$, which allows for computing $\Gamma\left( \prp\left( k,\cdot \right) \right)$ (as well as the inverse $\prp^{-1}\left( k, \Gamma^{-1}\left( \cdot \right) \right)$). Moreover, the key $k^\Gamma$ hides the fact that the outputs were permuted by $\Gamma$. See Section~\ref{sec:overview} for a more detailed explanation, and \ifllncs{the full version \cite{shmueli2025one} }\else{Section \ref{sec:prps} }\fi for a formal definition, as well as a formal statement and proof of the following:
\begin{theorem}[Informal]\label{thm:main5}There exist permutable PRPs for a large class of $\Gamma$ assuming (1) sub-exponentially-secure iO and (2) sub-exponentially-secure one-way functions.\end{theorem}
Such PRPs can be seen as the PRP analogue of a puncturable pseudorandom function (PRF), which is one of the main techniques used to prove security in the iO literature. To the best of our knowledge, this is the first example which provably obfuscates a PRP. We also show that our techniques are quite general, and in \ifllncs{the full version \cite{shmueli2025one} }\else{Section \ref{sec:prps} }\fi we prove the following:
\begin{theorem}\label{thm:main6}
There exist \emph{full-domain} trapdoor one-way permutations (OWPs), assuming (1) sub-exponentially-secure iO and (2) sub-exponentially-secure one-way functions.
\end{theorem}
Here, we remind that our notion of ``full-domain'' means that the permutation domain is just $\{0,1\}^n$.
We thus solve the decade-old problem of constructing full-domain (trapdoor) one-way permutations from iO\footnote{\cite{TCC:BitPanWic16} was first made public in early 2015.}, and give an answer to the decades-old problem of obfuscating PRPs to obtain trapdoor permutations. In light of the impossibility in~\cite{TCC:AshSeg16}, Theorem~\ref{thm:main6} may seem surprising. However, we explain in Section~\ref{sec:owpimposs} that it actually does \emph{not} contradict their impossibility. 

Our trapdoor OWP simply obfuscates a permutable PRP. The proof is straightforward given a permutable PRP; the bulk of the technical effort is then in constructing the permutable PRP in Theorem~\ref{thm:main5}. This demonstrates the utility of our new PRP notion.

For another application, by plugging into the elegant proof of quantumness of~\cite{ITCS:MorYam23}, we immediately obtain:
\begin{corollary}\label{cor:main}Assuming sub-exponentially (classically) secure iO and sub-exponentially (classically) secure one-way functions, there exists a proof of quantumness protocol.\end{corollary}
This is the first proof of quantumness using iO, as the non-full-domain trapdoor permutations of~\cite{TCC:BitPanWic16,EC:GPSZ17} cannot be used in this construction.

\paragraph{Paper Organization.}
In the remainder of the introduction, we discuss additional motivation and related work (Section~\ref{sec:motivation}), why the bug in~\cite{STOC:AGKZ20} is unfixable (Section~\ref{sec:bug}), and why the impossibility of~\cite{TCC:AshSeg16} does not apply to our construction (Section~\ref{sec:owpimposs}). In Section~\ref{sec:overview}, we provide an overview of our techniques. A reader only interested in our results on obfuscating PRPs, including our application to trapdoor permutations, can find an overview in Section~\ref{sec:overviewprps}. Section~\ref{sec:overviewprps} is entirely classical and can be read independently of the rest of the paper without any background in quantum computing. A reader interested in our OSS construction should start with Section~\ref{sec:overvieworacle} which gives our oracle construction and an overview of the oracle proof. Then after developing our techniques for obfuscating PRPs, we explain how to translate our oracle construction into a standard-model construction in Section~\ref{sec:overviewstandard}. 

\subsection{Motivation and Other Related Work}\label{sec:motivation}

\paragraph{Quantum money and variants.} Quantum money uses unclonable quantum states as currency to prevent counterfeiting. One-shot signatures lives at the top of a hierarchy of related concepts:
\begin{itemize}
    \item {\bf Secret key quantum money.} This was originally proposed by Wiesner~\cite{Wiesner83}. A major drawback, however, is that only the mint is able to verify, leading to a number of limitations.
    \item {\bf Public key quantum money.} This was proposed by Aaronson~\cite{CCC:Aaronson09} to remedy the various issues with Wiesner's scheme. Here, anyone can verify banknotes but only the mint can create new notes. It has been a major challenge to construct public key quantum money. Several candidate constructions exist~\cite{ITCS:FGHLS12,KSS22,EC:LiuMonZha23,ITCS:Zhandry24a}. It was also shown to exist in a classical oracle model~\cite{STOC:AarChr12}, which was later improved to a standard-model construction using iO by~\cite{EC:Zhandry19b}.
    \item {\bf Quantum lightning.}
    This concept allows anyone to mint banknotes together with classical serial numbers such that anyone can verify pairs of a serial number and a quantum banknote, but ensures that no user can create two valid banknotes with the same serial number. This last property of unclonability (even for the state generator) allows to use quantum lightning in decentralized settings, in ways that standard (public-key) quantum money cannot generally be used.
    Quantum lightning was first suggested by~\cite{LAFGHKS09} (under the name ``collision-free quantum money'') and rigorously formalized by~\cite{EC:Zhandry19b}. It implies in particular a public-key quantum money scheme, where a banknote is a quantum lightning state/serial number pair $|\$\rangle,\sigma$, together with a signature on $\sigma$, signed using the mint's (standard classical post-quantum) signing key.
    Some of the quantum money candidates are also quantum lightning~\cite{ITCS:FGHLS12,KSS22,EC:LiuMonZha23,ITCS:Zhandry24a}. But others, including the classical oracle and iO results mentioned above, are not quantum lightning. Prior to our work, all quantum lightning schemes required novel computational assumptions that had not been studied by the wider cryptography community. Moreover, no prior scheme has provable security in a classical oracle model. In fact, our OSS gives in particular the first quantum lightning scheme with provable security under widely-studied assumptions, and the first unconditional proof in a classical oracle model.
    \item {\bf One-shot signatures (OSS).} Further strengthening quantum lightning, OSS treats the lightning state as a quantum signing key, which can sign a single message and then provably self-destructs. Among other things, the upgrade from quantum lightning to OSS adds the ability to use only classical communication and local quantum computation, decrease the needed coherence times for quantum lightning states, and more. No known provable classical-oracle constructions nor standard-model instantiations (under \emph{any} reasonable-sounding assumption) were known prior to this work.
\end{itemize}

\paragraph{Quantum cryptography with classical communication.} An interesting application of OSS is to send quantum money using classical communication. This may sound impossible at first, but~\cite{STOC:AGKZ20} observe that it is nevertheless possible to send quantum money using a simple classical \emph{interactive} protocol: for the mint to send a money state to Alice, Alice will create a lightning state/OSS signing key and serial number \emph{for herself} $|\$_{\sf Alice}\rangle,\sigma_{\sf Alice}$, and send $\sigma$ to the mint. The mint then signs the serial number, providing the signature $\sigma_{{\sf mint}\rightarrow{\sf Alice}}$, which the mint sends back to Alice. Now Alice's money state is $|\$_{\sf Alice}\rangle,\sigma_{\sf Alice},\sigma_{{\sf mint}\rightarrow{\sf Alice}}$, which was obtained by just sending classical messages! Then, if Alice wants to send money to Bob, Bob simply creates a new lightning state/OSS signing key and serial number \emph{for himself} $|\$_{\sf Bob}\rangle,\sigma_{\sf Bob}$, sends $\sigma_{\sf Bob}$ to Alice, who then signs $\sigma_{\sf Bob}$ using her quantum lightning/OSS signing key, obtaining signature $\sigma_{{\sf Alice}\rightarrow {\sf Bob}}$, which she sends to Bob. Bob now has the money state $|\$_{\sf Bob}\rangle,\sigma_{\sf Bob},\sigma_{{\sf Alice}\rightarrow{\sf Bob}},\sigma_{{\sf mint}\rightarrow{\sf Alice}}$. By the OSS guarantee, Alice's money state has now self-destructed, meaning she no longer has the money but Bob does. Bob can then send the money to Charlie, etc.\footnote{The Mint-to-Alice step actually only requires quantum lightning, but Alice-to-Bob, etc seem to require the stronger OSS.}

The Mint-to-Alice step using classical communication was previously solved by~\cite{STOC:Shmueli22} using iO and other tools, but the money could not be subsequently sent to Bob without quantum communication\footnote{Using the results from \cite{C:Shmueli22} it is possible to use only classical communication and send money to Bob. However, this still necessitates Alice to communicate with the Bank repeatedly for every new transaction.}. Our work allows for Alice to send money to any party in the system classically and directly. These works on quantum money with classical communication are part of a broader class of protocols for performing quantum cryptography using classical communication, such as tests of quantumness and certified randomness~\cite{FOCS:BCMVV18}, position verification~\cite{ITCS:LiuLiuQia22}, and certified deletion~\cite{C:BarKhuPor23,C:BarKhu23,TCC:KitNisYam23,EC:BGKMRR24}.

\paragraph{Cryptocurrencies and Blockchains.} One-shot signatures have numerous applications in the cryptocurrency and blockchain settings. For example, they can be used to give decentralized currency and even smart contracts without a blockchain at all~\cite{EC:Zhandry19b,STOC:AGKZ20,Sattath22}. OSS are known \cite{coladangelo2020quantum} to provide a solution to the blockchain scalability problem\footnote{The work of \cite{coladangelo2020quantum} uses a strengthening of quantum lightning, where there is an additional procedure to destroy quantum lightning states and produce a classical "proof of deletion". This strengthening is known to follow from OSS, but not from standard quantum lightning.}, using only classical communication. They also have various advantages for other blockchain-related tasks, such as decreasing the threshold for perfect finality, eliminating slashing and leakage risks in liquid staking, and more~\cite{Drake23}. 

\paragraph{Obfuscating pseudorandom objects.} Perhaps the main technique in the literature for using iO is the punctured programming paradigm~\cite{STOC:SahWat14}, which primarily utilizes a puncturable pseudorandom function~\cite{CCS:KPTZ13,AC:BonWat13,PKC:BoyGolIva14}. These are functions where one can give out a ``punctured'' key, which allows for evaluating the PRF on all but a single point $x$. Meanwhile, even with this ability, the value on $x$ remains pseudorandom.

Puncturable \emph{invertible} functions were considered in~\cite{TCC:BonKimWu17}, though their construction expands the input size, so it is not a permutation but rather an (efficiently invertible) injection. They also discuss puncturable PRPs, and explain that they are impossible in some settings. In particular, in the setting where the domain is polynomial-sized, if the punctured key reveals the permutation on all points but $x$, it also reveals $x$.

Our notion of permutable PRPs avoids this issue and even is valid in the small-domain setting. We can ``puncture'' a permutable PRP at a point $x$ with output $y$ by choosing a random image $y'$, letting $\Gamma$ be the transposition which swaps $y$ and $y'$, and outputting the permuted key $k^\Gamma$. This allows for computing the PRP on all points except for $x$ (and also the pre-image $x'$ of $y'$), but hides the true value of $y$. The impossibility of punctured PRPs has come up in several settings (e.g.~\cite{C:SACM21,USENIX:MZRS22,EC:LiuMonZha23,EPRINT:HPPY24}). It would be interesting to explore if our notion could be useful in these settings.

\paragraph{(Trapdoor) OWPs from iO.} Trapdoor one-way permutations (OWP) were proposed by~\cite{FOCS:Yao82a} to abstract the ideas behind public key cryptosystems based on both RSA and Rabin. The domain and range of these objects are sets with algebraic structure. However, it became much simpler to describe applications in terms of a full-domain OWP where the domain and range are simply $\{0,1\}^n$. This simplification lead to several conceptual errors (see~\cite{GolRot13} for explanation). To account for these errors, when using sparse-domain OWPs, one must often stipulate additional conditions (called ``enhanced'' or ``doubly enhanced'') that are trivially satisfied by full-domain OWPS.

Once iO emerged as a powerful cryptographic tool~\cite{FOCS:GGHRSW13}, a natural question was whether iO could give a third way of building a trapdoor permutation. A particular motivation is to achieve post-quantum security, since RSA and Rabin are both quantumly insecure due to Shor's algorithm~\cite{FOCS:Shor94}. Constructions were given~\cite{TCC:BitPanWic16,EC:GPSZ17}, but these constructions were ``messy,'' with sparse domains that could not be efficiently recognized and required cryptographic procedures to sample. An interesting question was whether a clean, full-domain trapdoor permutation was possible from iO. Our work mostly resolves this question, though we still have the limitation that the space of keys is sparse.

As a concrete application, our trapdoor OWPs can be plugged into the elegant proofs of quantumness of~\cite{ITCS:MorYam23}, to obtain proofs of quantumness from (classically-hard sub-exponential) iO and one-way functions\footnote{The OWP part of our work does not need LWE.}. This is the first proof of quantumness based on iO. The prior OWPs of~\cite{TCC:BitPanWic16,EC:GPSZ17}, despite being doubly enhanced, did not suffice for this application. More generally, permutations on simple domains like $\{0,1\}^n$ appear much more useful in the context of quantum cryptosystems, whereas permutations on highly structured domains break the delicate structure of quantum states.

\subsection{The Bug in the previous work of AGKZ} \label{sec:bug}

The bug in~\cite{STOC:AGKZ20} as found by~\cite{Bartusek23} is rather technical, and we do not discuss it here. However, we argue that the bug is likely unfixable using the techniques employed by~\cite{STOC:AGKZ20}. The issue is that in \cite{STOC:AGKZ20}, the security of the OSS boils down to showing the collision resistance of a certain hash function, and the proof technique uses what is known as the inner-product adversary method -- first developed in~\cite{STOC:AarChr12}. Unfortunately, the inner-product adversary method is unlikely to be able to prove even the basic collision problem is hard, even before adding all the structure that~\cite{STOC:AGKZ20} need to obtain their OSS scheme. The reason is that the collision problem has a small \emph{certificate complexity}, and it is known that the adversary method cannot prove good lower-bounds for problems with small certificate complexity, as shown by~\cite{CCC:Aaronson03}. This certificate complexity barrier extends to all ``positive weight'' adversary methods, including the inner-product method. Thus, it seems any fixed proof for~\cite{STOC:AGKZ20} must use additional techniques.\footnote{One caveat is that the barrier only applies to lower bounds for the number of queries to distinguish a function with collisions from a function that is injective. Such a lower-bound immediately implies a lower-bound for actually finding collisions. But the converse is not true, and the certificate complexity barrier does \emph{not} seem to rule out directly proving the hardness of finding collisions. Nevertheless, it does not seem like the techniques of~\cite{STOC:AGKZ20} circumvent this barrier.}

Note that most of our oracle result is proved using standard adversary method techniques. However, our result assumes the standard collision lower bound as previously proved in~\cite{JACM:AarShi04,QIC:Zhandry15}. These results were proved using the polynomial method that is not subject to the certificate complexity barrier.

\subsection{On the Lower Bound for OWPs from iO}\label{sec:owpimposs}
Asharov and Segev~\cite{TCC:AshSeg16} show a barrier for constructing OWPs from iO and one-way functions. Specifically, they show that any ``black-box'' construction of one-way permutations from iO and one-way functions, cannot be ``domain-invariant''. Here, black-box in the context of iO is a bit subtle, but their notion captures most known iO techniques, including ours.
Recall that in a (keyed) OWP, we sample one permutation out of exponentially many permutations in an efficiently samplable family. As for domain-invariance, the abstract of~\cite{TCC:AshSeg16} defines it as ``each permutation [in the family] may have its own domain, but these domains are independent of the underlying building blocks''. At first glance, this would seem to contradict our full-domain trapdoor permutation, since the domain of all permutations in the family of our construction is just $\{0,1\}^n$, clearly independent of any building block.

However, there are really two types of domain invariance: One for the key, and one for the actual input. The formal specification of the impossibility in~\cite{TCC:AshSeg16} reveals that their notion of domain-invariance refers to schemes that satisfy \emph{both} notions. In fact, in the model they propose, obfuscating a PRP actually does yield an \emph{input}-domain-invariant (trapdoor) OWP; also, since the key is an obfuscated program and most strings do not correspond to valid programs / keys, it is not key-domain invariant. 

Our (trapdoor) OWPs are similarly input-domain invariant but not key-domain invariant, since our key is likewise an obfuscated program, with the added benefit of provable security under iO.  Thus, our result does not contradict the impossibility of~\cite{TCC:AshSeg16}. Instead, it helps clarify the impossibility as being more about key-domain-invariance.

\subsection{Directions for Future Work}
Here are some natural follow-up directions from our work:
\begin{itemize}

    \item
    Can OSS be achieved without using iO? While there are several approaches (using new hardness assumptions) that seem to give weaker objects like quantum lightning~\cite{ITCS:FGHLS12,KSS22,EC:LiuMonZha23,ITCS:Zhandry24a}, they do not seem amenable to producing classical signatures, even heuristically. 
    
    \item
    Is it possible to remove the need for sub-exponential hardness in our constructions? Sub-exponential hardness arises in two key places. The first is in our PRP: we start with a ``base'' PRP (based just on one-way functions) which is permutable only for the class of ``neighbor swaps'' that exchange a given $z$ with $z+1$, and upgrade it to more general permutations $\Gamma$ by decomposing such permutations into a sequence of neighbor swaps. The proof incurring a security loss for each step in the decomposition. General permutations, even just transpositions (which are enough to get trapdoor OWPs) require an exponential-sized decomposition. In order to get around this limitation, it seems that the ``base'' PRP needs to support a richer class than just neighbor swaps. Is there a base PRP that support, say, general transpositions, from just one-way functions? The second place sub-exponential hardness comes in is when we are switching from the construction to a simulated distribution; this utilizes a hybrid over all possible outputs. Perhaps there is a clever way to change all possible outputs in one go.
    
    \item
    Can we achieve a ``clean iO'' approach to OSS that uses just iO and generic primitives? Note that it is unlikely that OSS can be obtained from iO and one-way functions: OSS requires collision resistance, which is believed to not be possible from iO and one-way functions alone~\cite{FOCS:AshSeg15}. But perhaps iO and collision-resistance is enough. Or at a minimum, maybe our LWE assumption can be replaced with other algebraic techniques to give a diversity of assumptions.
    \item
    We introduced a new technique for obfuscating PRPs, and it would be interesting to see if this enables any new results, or at least streamlines old results. 
\end{itemize}

Finally, we conclude with a fascinating complexity-theoretic question inspired by our techniques. Any permutation $\Gamma$ on $\{0,1\}^n$ can be decomposed into an exponentially long product of transpositions $\Gamma=\transposition{a_1}{a_2} \circ \transposition{a_3}{a_4} \circ \cdots$. While such a decomposition is clearly inefficient, the following question asks if the partial products can be made small:
\begin{question} [Efficient Permutation Decomposition Problem] \label{question:decompose}
For any permutation $\Gamma$ on $\{0,1\}^n$ such that $\Gamma,\Gamma^{-1}$ have circuits of size $s$, is it possible to decompose $\Gamma$ into a product of transpositions $\Gamma=\tau_1\circ \tau_2\circ\cdots \circ \tau_T$ such that each partial product $\Gamma_t:=\tau_1\circ\tau_2\circ\cdots\circ\tau_t$ for $t\in [T]$ and its inverse $\Gamma_t^{-1}$ have circuits of size $\poly(s,n)$?
\end{question}
As part of constructing our permutable PRPs and applying them to construct OSS, we show that very general families of permutations can be efficiently decomposed in this way (see Figure~\ref{fig:decomposable} for a non-exhaustive list of such families).  On the other hand, we do not know how to handle all efficient permutations, and even for very simple permutations like multiplying by a scalar mod $N$, we only know a general solution assuming the Extended Riemann Hypothesis.

\section{Technical Overview}\label{sec:overview}
In this section we explain the main techniques shown in this work.

\subsection{Definitions}
We first give the formal definition of one-shot signatures.

\begin{definition} [One-Shot Signature Scheme] \label{definition:OSS}
A one-shot signature (OSS) scheme is a tuple of algorithms $\left( \setup, \gen, \sign, \ver \right)$ together with message-space $\left( \Ms_\lambda \right)_\lambda$ satisfying the following:
\begin{itemize}
    \item
    $\crs \gets \setup\left( 1^\lambda \right)$: A classical probabilistic polynomial-time algorithm that given the security parameter, samples the classical common reference string (CRS).
    
    \item $\left( \pk, \ket{\sk} \right) \gets \gen(\crs)$: A quantum polynomial-time algorithm that takes the classical $\crs$ and samples a classical public key $\pk$ and quantum secret key  $\ket{\sk}$.
    
    \item
    $\sigma \gets \sign\left( \crs, \ket{\sk}, m \right)$: A quantum polynomial-time algorithm that given $\crs$ and the quantum key $\ket{\sk}$, given any message $m \in \Ms_\lambda$ produces a classical signature $\sigma$.
    
    \item
    $\ver\left( \crs, \pk, m, \sigma \right) \in \{ 0, 1 \}$: A classical deterministic polynomial-time algorithm which verifies a message $m$ and signature $\sigma$ relative to a public key $\pk$.
    
    \item {\bf Correctness:}
    There exists a negligible function $\negl$ such that, for any $\lambda$ and any $m\in\Ms_\lambda$ we have
    $$
    \Pr_{
    \substack{
    \crs\gets\setup\left( 1^\lambda \right), \\
    \left( \pk, \ket{\sk} \right) \gets \gen\left( \crs \right), \\
    \sigma_{m} \gets \sign(\crs, \ket{\sk}, m)
    }
    }\left[
    \ver\left( \crs, \pk, m, \sigma_{m} \right) = 1
    \right]
    \geq 1 - \negl\left( \lambda \right)
    \enspace .
    $$
    
    \item {\bf Security:}
    For any QPT algorithm $\Adv$, there exists a negligible function $\negl$ such that for all $\lambda$ the following probability is bounded by $\leq \negl \left( \lambda \right)$,
    $$
    \Pr_{
    \substack{ \crs\gets\setup(1^\lambda), \\ \left( \pk, m_0, m_1, \sigma_0, \sigma_1 \right) \gets \Adv\left( \crs \right) }
    }
    \left[
         \ver\left( \crs, \pk, m_0, \sigma_0 \right) = 1 \;\land\;
         \ver\left( \crs, \pk, m_1, \sigma_1 \right) = 1
    \right]
    \enspace .
    $$
\end{itemize}
\end{definition}

It is straightforward to adapt the above definition to utilize an oracle. In this case, the oracle will play the role of $\crs$, and we will omit the algorithm $\setup$.

\paragraph{OSS from Non-collapsing Hashing.} As observed informally by~\cite{STOC:AGKZ20}, a collision-resistant but non-collapsing hash function gives an OSS, and this was made formal and general by~\cite{TQC:DalSpo23}. 
Our OSS therefore will be built from such a hash function. Here, we give the notion of collision resistant and non-collapsing hash functions.

\begin{definition} [Collision-Resistant Always-Non-Collapsing Hash] \label{definition:CR_always_NC_hash}
A collision-resistant always-non-collapsing hash function is a pair of PPT algorithms $\left( \setup, H \right)$ such that
\begin{itemize}
    \item
    $\crs \gets \setup\left( 1^\lambda \right)$: A classical probabilistic polynomial-time algorithm that given the security parameter, samples the classical common reference string (CRS).
    
    \item
    $y \gets H\left( \crs, x \right)$: A classical deterministic polynomial-time algorithm that given the CRS and input $x$, outputs $y$.
    
    \item {\bf Collision-resistance:}
    For any QPT algorithm $\As$, there exists a negligible function $\negl$ such that for all $\lambda$,
    $$
    \Pr_{
    \substack{ \crs\gets\setup\left( 1^\lambda \right), \\ \left( x_0, x_1 \right) \gets \Adv\left( \crs \right) }
    }
    \left[
    H\left( \crs, x_0 \right) = H\left( \crs, x_1 \right)
    \right]
    \leq
    \negl\left( \lambda \right)
    \enspace .
    $$
    
    \item {\bf Always non-collapsing:}
    There exists a pair of QPT algorithms $\left( \Ss, \Ds \right)$ and a negligible function $\negl\left( \lambda \right)$ such that for all $\lambda$,
    \begin{align*}
    &\left|
    \Pr
    \left[
    \Ds\left( x, \aux \right) = 1
    \;
    :
    \;
    \substack{ \crs\gets\setup\left( 1^\lambda \right) , \\
    \left( \ket{\psi}, \aux \right) \gets \Ss\left( \crs \right), \\
    x \gets {\sf Measure}\left( \ket{\psi} \right) }
    \right] - \right.\\
    &\;\;\;\;
    \left.
    \Pr
    \left[
    \Ds\left( \ket{ \psi_{y} },\aux \right) = 1
    \;
    :
    \;
    \substack{ \crs\gets\setup\left( 1^\lambda \right) , \\
    \left( \ket{\psi}, \aux \right) \gets \Ss\left( \crs \right), \\
    \ket{ \psi_{y} } \gets {\sf PartialMeasure }_{H\left( \crs,\cdot \right)}\left( \ket{ \psi } \right) }
    \right]
    \right|
    \geq
    1 - \negl(\lambda)
    \enspace .
    \end{align*}
    Here, $x \gets {\sf Measure}\left( \ket{ \psi } \right)$ means to measure $\ket{ \psi }$ in the computational basis arriving at measurement $x$. $\ket{ \psi_{y} } \gets {\sf PartialMeasure }_{H\left( \crs,\cdot \right)}$ means to compute $H\left( \crs,\cdot \right)$ (in superposition, with output register on the side) with input register $\ket{\psi}$, and measuring the result, resulting in outcome $y$. Then the state $\ket{ \psi }$ collapses to $\ket{ \psi_{y} }$, which contains some superposition of preimages of $y$.
\end{itemize}
\end{definition}
We will not formally give the proof that such a hash function implies OSS, but give the sketch for the interested reader. The idea is that $\gen$ runs $\Ss$ to get $|\psi\rangle,\aux$, and then applies $H(\crs,\cdot)$ to  $|\psi\rangle$ and measures, obtaining $y$. We set $\pk=y$, and $|\sk\rangle=|\psi_y\rangle,\aux$ where $|\psi_y\rangle$ is the post-measurement state. Observe that the support of $|\psi_y\rangle$ are strings $x$ that hash to $y$.

For a bit $b$, let $|\psi_{y,b}\rangle$ be the state post-selecting on $x$ whose first bit is $b$. One can create $|\psi_{y,b}\rangle$ for a random choice of $b$ my simply measuring the first qubit. Then very roughly it is shown in~\cite{STOC:AGKZ20,TQC:DalSpo23} how to use the distinguisher $\Ds$ to move back and forth between $|\psi_{y,0}\rangle$ and $|\psi_{y,1}\rangle$.

To sign a bit $b$, one obtains $|\psi_{y,b}\rangle$ and measures, obtaining a string $x$ beginning with the bit $b$ that hashes to $y$. This is a signature on $b$, which is easy to verify. Moreover, the collision-resistance of $H$ implies that it is computationally infeasible to find two signatures. Thus, we have an OSS for a single bit. Then by parallel repetition, we obtain an OSS for arbitrary messages.

\subsection{OSS Relative to a Classical Oracle \ifllncs\else(Section ~\ref{sec:oss_oracle})\fi} \label{sec:overvieworacle}

Since OSS follows from non-collapsing collision-resistant hash functions, we will focus on constructing the latter. 

\paragraph{The Construction From~\cite{STOC:AGKZ20}.} The OSS of~\cite{STOC:AGKZ20} utilizes what we will call a \emph{coset partition function} in this work. This is a function $Q$ that is many-to-1, and where the pre-image set of any image is a coset of a linear subspace: that is, the pre-image sets have the form $Q^{-1}(y) := \{ \matA_y \cdot \rv + \vecB_y \: : \: \rv \in \bbZ_2^k \}$, where $\matA_y,\rv_y$ are a matrix/vector pair that depends on the image $y$. The function $Q$ will be the hash function, provided as an oracle.

In order to be non-collapsing,~\cite{STOC:AGKZ20} employ the hidden subspaces approach of~\cite{STOC:AarChr12}, and additionally provide a separate oracle $D$, which provides membership testing for the linear space $\matA^\perp := \{ \zv \: : \: \matA_y \cdot \zv = 0 \}$. This allows for testing if a state is in the uniform superposition $\ket{ Q^{-1}(y) } := \sum_{x\in Q^{-1}(y)} \ket{ x}$ : first use $Q$ to test that the state is in the support of $Q^{-1}(y)$, then apply the quantum Fourier transform (QFT), and use $D$ to test that the resulting state has support on $\matA^\perp$. The only state that passes both verifications is the state $\ket{ Q^{-1}(y) }$.

In order to obtain a non-collapsing hash/OSS scheme, then one needs to prove that $Q$ is collision-resistant, given the oracles for $Q$ and $D$. As mentioned, the proof provided in~\cite{STOC:AGKZ20} contained a fatal bug, and while the scheme is plausibly collision-resistant, it remains unclear how to prove this.

\paragraph{Relaxing the structure.} A key challenge with the approach in~\cite{STOC:AGKZ20} is that it is non-trivial come up with a coset partition function that is not trivially insecure. The cosets need to have very different orientations ($\matA_y$ values), lest the function $Q$ be periodic and therefore subject to quantum period-finding algorithms~\cite{FOCS:Shor94}. But how do we perfectly partition the domain into cosets that are all of different orientations, but no overlaps? The coset partition function used in~\cite{STOC:AGKZ20} is derived in a very specific way, and the various pre-image sets are highly correlated, making reasoning about them quite challenging. This is further complicated by the $D$ oracle, which contains even more information about the cosets. An interesting proposal was made by~\cite{Bartusek23}: Simply have the hash function be a random function $H$, but provide an additional oracle which maps each pre-image set $H^{-1}(y)$ to a random coset $S_y$ embedded in a much larger space. The point is that, the cosets $S_y$ for different $y$'s no longer need to partition the space they live in, as they are independent of each other.

Our following construction is inspired by the above general principles. We sample a random secret permutation $\Pi$ on $\{0,1\}^n$, and let $H(x),J(x)$ denote the first $r$ bits and last $n-r$ bits of $\Pi(x)$, respectively. $H:\{0,1\}^n\rightarrow\{0,1\}^r$ will be our hash function. For each $y\in\{0,1\}^r$, we also choose a random coset $S_y\subseteq\{0,1\}^k$ of dimension $n-r$, described by a matrix/vector pair $\matA_y \in \bbZ_{2}^{ k \times (n - r) },\vecB_y \in \bbZ_{2}^{k}$ as $S_y = \{ \matA_y\cdot\rv+\vecB_y:\rv\in\{0,1\}^{n-r}\}$. We then provide an oracle $\Ps(x)$ which outputs $\left( H(x), \vecU = \matA_{H(x)} \cdot J(x) + \vecB_{H(x)} \right)$. In other words, it computes $y:=H(x)$, then uses this information to determine a coset, and then uses the remaining permutation information $J(x)$ to determine the point $\vecU$ in that coset that $x$ maps to. We likewise provide the oracle $\Ds(y,\vecV)$ which checks if $\vecV\in S_y^\perp$, where $S_y^\perp$ is the kernel of $\matA_y$.

We have to provide one more oracle, which is denoted $\Ps^{-1}\left( y, \vecU \right)$, which inverts the oracle $\Ps$: It outputs $x$ if $\Ps(x) = \left( y, \vecU \right)$ and otherwise outputs a special symbol $\bot$, indicating that $\left( y, \vecU \right)$ is not in the range of $\Ps$. This oracle is necessary in order to actually preserve the non-collapsing property: Given the uniform superposition of pre-images $\ket{H^{-1}\left( y \right)}$, applying the oracle $\Ps$ gives $\sum_{x\in H^{-1}(y)}|x,\Ps(x)\rangle=\sum_{x\in H^{-1}(y)}|x,y,\matA_y\cdot J(x)+\vecB_y\rangle$. We would like to apply the QFT to the register containing $\matA_y\cdot J(x)+\vecB_y$, but this will not work: since the register $x$ remains around and is entangled with the register containing $\matA_y\cdot J(x)+\vecB_y$, performing the QFT will actually yield random junk. The oracle $\Ps^{-1}$ allows us to un-compute $x$, thereby allowing verification to work as desired.

We now need to prove the collision-resistance of $H$. Unfortunately, quantum query lower-bounds for oracles with inverses is a notorious challenging problem. This problem, together with the presence of the oracle $\Ds$, were the primary barriers to proving the security of this construction.

In what follows, we describe our proof in the oracle setting. At a very high-level, our proof will gradually eliminate parts of the oracles $\Ps,\Ps^{-1},\Ds$, until all that remains is a simple hash function oracle, and we can then invoke the known query lower-bounds for collision-finding to conclude security.

\paragraph{Warm-up: A Random Self-Reduction.} As a warm-up that will lead to our proof, we introduce a random self-reduction for the oracles $\Ps,\Ps^{-1},\Ds$. Given an instance of the oracles $\overline{\Ps},\overline{\Ps^{-1}},\overline{\Ds}$ (with underlying arbitrary permutation $\overline{\Pi}$ on $\{ 0, 1 \}^{n}$ and arbitrary matrix/vector pairs $\overline{\matA}_y \in \bbZ_{2}^{k \times (n - r)}$,$\overline{\vecB}_y \in \bbZ_{2}^{k}$ for all $y \in \{ 0, 1 \}^{r}$, where all we know is that $\overline{\matA}_y$ is full-rank), we can construct another instance as follows. Choose a random permutation $\Gamma$ on $\{ 0, 1 \}^{n}$ and for every $y \in \{ 0, 1 \}^{r}$, choose a random full-rank $\matC_y\in\{0,1\}^{k\times k}$ and random $\vecD_y\in\{0,1\}^k$. Note that $\matC_y,\vecD_y$ define a random affine permutation $L_y(\vecU)=\matC_y\cdot\vecU+\vecD_y$ on $\{0,1\}^k$. Then define:
\begin{itemize}
    \item
    $\Ps(x)$: Compute $\overline{x} \gets \Gamma(x)$, $(y,\overline{\vecU}) \gets \overline{\Ps}( \overline{x} )$, $\vecU \gets L_y(\overline{\vecU})$ and output $(y,\vecU)$.
    
    \item
    $\Ps^{-1}\left( y,\vecU \right)$: Compute $\overline{\vecU} \gets L_{y}^{-1}\left( \vecU \right)$, $\overline{x} \gets \overline{\Ps}^{-1}\left( y, \overline{\vecU} \right)$, $x \gets \Gamma^{-1}\left( \overline{x} \right)$ and output $x$.
    
    \item
    $\Ds\left( y, \vecV \right)$: Output $\overline{\Ds}\left( y, \vecV^{T} \cdot \matC_{y} \right)$.
\end{itemize}
The resulting oracle implicitly sets $\Pi := \overline{\Pi}\circ \Gamma$, and for every $y \in \{ 
0, 1 \}^{r}$ sets $\matA_y := \matC_y\cdot\overline{\matA}_y$, $\vecB_y := \matC_y\cdot \overline{\vecB}_y+\vecD_y$, which all distribute independently of the underlying $\overline{\Pi}$, $\overline{\matA}_y$, $\overline{\vecB}_y$.
Thus this random self-reduction turns any instance $\overline{\Ps},\overline{\Ps^{-1}},\overline{\Ds}$ into a fresh independent instance $\Ps,\Ps^{-1},\Ds$. Moreover, if we let $H$ and $\overline{H}$ be the functions outputting the first $r$ bits of $\Ps$ and $\overline{\Ps}$ respectively, we can turn a collision for $H$ into a collision for $\overline{H}$ by applying $\Gamma$. Thus, we can derive the collision-resistance for random instances based on the collision resistance of any fixed instance. We note that there does not appear to be an analogous random self-reduction for the oracles of~\cite{STOC:AGKZ20}. We will use variants of this self-reduction to gradually remove components from our oracle.

\paragraph{Step 1: Bloating the Dual.} The first step of our actual proof, inspired by techniques of~\cite{EC:Zhandry19b}, is to ``bloat the dual''. This means, for each $y$, we choose a random super-space $T_y^\perp$ of $S_y^\perp := \colspan\left( \matA_{y} \right)^{\bot}$, and replace $\Ds$ with the oracle $\Ds'$ which checks for membership in $T_y^\perp$. The super-space $T_y^\perp$ is chosen at random such that (1) it is a sparse subset of $\{ 0, 1 \}^{k}$ and (2) $S_y^\perp$ is a sparse subset of $T_{y}^{\bot}$. We can prove that $\Ds'$ is indistinguishable form $\Ds$, since the points in $T_y^\perp$ but not in $S_y^\perp$ are random sparse points hidden from the adversary. Thus an adversary that finds a collision given the more informative $\Ds$, will also find a collision given $\Ds'$.

In~\cite{EC:Zhandry19b}, this technique was employed to reduce the security of the quantum money scheme of~\cite{STOC:AarChr12} to an information-theoretic statement. Crucially in that proof, \emph{both} the primal $S$ and dual $S^\perp$ were bloated, which then implies that the original space $S$ is information-theoretically hidden. In our case, however, the original subspaces $S_y$ are still part of the oracles $\Ps,\Ps^{-1}$, and it is unclear how to bloat them. This key difference is that in~\cite{EC:Zhandry19b}, there was only a single sparse space $S$, and it could be bloated by consuming some of the abundant ``free'' ambient space. However, in our case, the entire domain is partitioned into sets $H^{-1}(y)$, and it seems that we would need to bloat each of them. But this is impossible, since there is no ``free'' space left as each point is already part of some subspace. Nevertheless, in the following, we will see that the dual-bloating is still useful.

\paragraph{Step 2: Simulating the Dual.} We will now completely eliminate the dual oracle $\Ds'$. To do so, we use a version of the random self-reduction described above, but starting from a smaller instance $(\overline{\Ps},\overline{\Ps}^{-1},\overline{\Ds})$. We will show that embedding the smaller instance lets us (without querying the smaller instance) know some information about the subspaces $S_y^\perp$, namely a random super-space $T_y^\perp$ of $S_y^\perp$. This allows us to simulate $\Ds'$ \emph{without querying $\overline{\Ds}$ at all}, and in turn to eliminate the dual oracle $\overline{\Ds}$ entirely.

In more detail, let $s=\log_2 (|T_y^\perp|/|S_y^\perp|)$, which is how much larger (in terms of dimension) $T_y^\perp$ is from $S_y^\perp$. Let $\overline{\Ps},\overline{\Ps}^{-1}$ be an instance of our oracle, but with input space $\{0,1\}^{n - (n - r - s)} = \{0,1\}^{r + s}$ and output space $\{0,1\}^r \otimes\{0,1\}^{k - (n - r - s)}$. In other words, we shrink the input $x$ and the vector part of the output $\vecU$ each by $n - r - s$ bits, but we keep the $y$ part of the output (the actual hash function output) the same length. We will not have access to the oracle $\overline{\Ds}$. 

We simulate $\Ps,\Ps^{-1},\Ds'$ using $\overline{\Ps},\overline{\Ps}^{-1}$. However, since $\overline{\Ps},\overline{\Ps}^{-1}$ is a smaller instance, we first expand it into a full-sized instance (i.e., without considering how our mapping distributes). Essentially, we just pass the first $r + s$ bits of $x$ into $\overline{\Ps}$, and the last $n - r - s$ bits we output in the clear as part of $\vecU$. In more detail: 
\begin{itemize}
    \item
    $\Ps(x)$:
    Break the input $x\in\{0,1\}^n$ into two parts: $x := \left( \overline{x} \in \bbZ_{2}^{r + s}, \widetilde{x} \in \bbZ_{2}^{n - r - s} \right)$ and we moreover interpret $\widetilde{x}$ as a vector of dimension $n - r - s$. Now query $\left( y, \overline{\vecU} \right) \gets \overline{\Ps}\left( \overline{x} \right)$, and set $\vecU = \left( \overline{\vecU}, \widetilde{x} \right)$. Output $\left( y, \vecU \right)$.
    
    \item
    $\Ps^{-1}\left( y, \vecU \right)$: 
    Write $\vecU = \left( \overline{\vecU} \in \bbZ_{2}^{k - (n - r - s)}, \widetilde{x} \in \bbZ_{2}^{n - r - s} \right)$, compute $\overline{x} \gets \overline{\Ps}^{-1}\left( y, \overline{\vecU} \right)$ and output $x = \left( \overline{x}, \widetilde{x} \right)$.
    
    \item
    $\Ds'\left( y, \vecV \right)$:
    Write $\vecV = \left( \overline{\vecV} \in \bbZ_{2}^{k - (n - r - s)}, \widetilde{x} \in \bbZ_{2}^{n - r - s} \right)$. Output 1 if and only if $\widetilde{x} = 0^{n - r - s}$.
\end{itemize}
Now this clearly does not simulate the correct distribution of oracles $\Ps,\Ps^{-1},\Ds'$ since for example the last $n - r - s$ bits of input are clearly visible in the output, and $\Ds'$ checks a fixed subspace. However, it is a valid instance of the oracles $\Ps,\Ps^{-1},\Ds'$, in the sense that there is \emph{some} choice permutation $\Pi$, cosets $S_y$ and spaces $T_y^\perp$ that gives these oracles, or in other words, we output an oracle that's inside the support of $\Ps,\Ps^{-1},\Ds'$.
Note that our output oracle satisfies that for every $y \in \{ 0, 1 \}^{r}$ we have that $T_y^\bot$ is just the space of vectors whose last $n - r - s$ entries are 0 -- we'll use this fact later.
We can then apply our random self-reduction to generate a correctly distributed instance.\footnote{Note that in our warm-up, the random self-reduction re-randomized an instance of the actual construction $\left( \Ps,\Ps^{-1},\Ds \right)$, but we are now re-randomizing an instance $\left( \Ps, \Ps^{-1}, \Ds' \right)$, when the dual is bloated. The same random self-reduction works just as well in this setting.}

To re-randomize $\left( \Ps,\Ps^{-1},\Ds' \right)$, let $\overline{H}$ the hash function in the oracles $\overline{\Ps},\overline{\Ps}^{-1}$, let $H$ the hash function in the oracles $\left( \Ps,\Ps^{-1},\Ds' \right)$ and let $\widetilde{H}$ the hash function in the oracles $\left( \widetilde{ \Ps },\widetilde{ \Ps^{-1} }, \widetilde{ \Ds } \right)$, which are generated by applying the random self reducibility on $\left( \Ps,\Ps^{-1},\Ds' \right)$. By the random self-reducibility, we know that any algorithm which finds collisions for $\widetilde{H}$ given access to $\left( \widetilde{ \Ps },\widetilde{ \Ps^{-1} }, \widetilde{ \Ds } \right)$ will find collisions for the simulated $H$ relative to $\left( \Ps,\Ps^{-1},\Ds' \right)$. To complete our elimination of the dual oracle, we need to show that such collisions in $H$ actually yield collisions in $\overline{H}$. For simplicity, we will work with the non-random-self reduced oracles described above $\left( \Ps,\Ps^{-1},\Ds' \right)$. First, observe that if we write $x = \left( \overline{x} \in \bbZ_{2}^{r + s}, \widetilde{x} \in \bbZ_{2}^{n - r - s} \right)$, then $H(x) = \overline{H}(\overline{x})$. There are two ways a collision $x,x'$ in $H$ can occur:
\begin{itemize}
    \item
    $\overline{x} \neq \overline{x}'$. In this case $\overline{x}$ and $\overline{x}'$ form a collision in $\overline{H}$, as desired.
    
    \item
    $\overline{x} = \overline{x}'$. We will call these ``bad'' collisions, which we will now handle. Observe that in these cases, since $x \neq x'$, we must have $\widetilde{x} \neq \widetilde{x}'$. Also, since $\overline{x} = \overline{x}'$, we have that $\overline{\vecU} = \overline{\vecU}'$. But letting $\vecU = \left( \overline{\vecU} \in \bbZ_{2}^{k - (n - r - s)}, \widetilde{x} \in \bbZ_{2}^{n - r - s} \right)$ and $\vecU' = \left( \overline{\vecU} \in \bbZ_{2}^{k - (n - r - s)}, \widetilde{x}' \in \bbZ_{2}^{n - r - s} \right)$, this means that $\vecU - \vecU'$ is a non-zero vector whose first $k - (n - r - s)$ entries are $0$. 

    In that case, let $T_y$ be the linear space that is dual to $T_y^\perp$. Recall that in the non-re-randomized case, $T_y^\perp$ is the space of vectors whose last $n - r - s$ entries are 0; this means that $T_y$ is the space of vectors whose \emph{first} $k-(n - r - s)$ entries are 0. Thus, ``bad'' collisions give $\left( \vecU - \vecU' \right) \in T_y$. This property is moreover preserved by applying the random self-reduction. In contrast, a general collision will have $\vecU-\vecU'$ in a much larger space, namely $S_y$, the analogous linear space dual to $S_y^\perp$.\footnote{$S_y$ contains, and is $s$ dimensions larger than $T_y$, since $S_y^\perp$ is contained, and is $s$ dimensions smaller than $T_y^\perp$.} 
    
    Thus, we see that bad collisions cause $\vecU-\vecU'$ to concentrate in a much smaller space than general collisions. Following a technique introduced in \cite{C:Shmueli22} (which we necessarily make optimal, in \ifllncs{the full version \cite{shmueli2025one}}\else{Lemmas \ref{lemma:unconditional_dual_subspace_anti_concentration} and \ref{lemma:dual_subspace_concentration}}\fi), we will use this fact to show that an algorithm which produces bad collisions actually distinguishes the bloated from un-bloated cases, contradicting our earlier proof of the indistinguishability of these two cases.
    
    In slightly more detail, this is because we can run an adversary several times (using a separate random self-reduction each time) to collect several independent vectors $\vecU-\vecU'$. If the vectors are always concentrated in $T_y$, the span of them will be contained in $T_y$ and therefore have dimension $n-r-s$. On the other hand, if we do the same for the original un-bloated oracles, we argue (again using self-reducibility) that the vectors $\vecU-\vecU'$ will be random in $S_y$ and hence with enough of them we span the whole larger space of dimension $n-r$. Looking at the dimension of the spanned space therefore distinguishes the original and bloated duals, which we already showed was impossible. Thus, the original algorithm must at least \emph{occasionally} output collisions that are not bad.
\end{itemize}

\paragraph{Step 3: Reducing to a Coset-Partition Function.} We have now reduced our problem to proving the collision resistance of $H$ when only given the oracles $\Ps,\Ps^{-1}$ but where the adversary does not have access to the oracle $\Ds$. We now prove such collision resistance, assuming the collision-resistance of some coset-partition function $Q$.  This may seem counter-productive; after all, we deliberately moved away from the coset-partition structure of~\cite{STOC:AGKZ20}. Looking ahead, we will see that, since we have now stripped away the dual oracle, we actually \emph{can} analyze the coset-partition structure.

We explain that given a coset-partition function $Q : \{ 0, 1 \}^{n} \rightarrow \{ 0, 1 \}^{r}$, we can readily construct an instance of $\Ps,\Ps^{-1}$: $\Ps(x)$ sets $y=Q(x)$, and $\vecU = \left( x, 0^{k-n} \right)$. This also makes $\Ps^{-1}\left( y, \left( x, 0^{k-n} \right) \right)$ trivial as $x$ is already present in the clear; we just need to verify that $Q(x)=y$ before outputting $x$ (and outputting $\bot$ otherwise). This gives a valid instance of the oracles $\Ps,\Ps^{-1}$ since the pre-image sets of $Q$ are already cosets, so outputting $x$ padded with 0's satisfies the structure of $\Ps$. The oracles $\Ps,\Ps^{-1}$ clearly do not have the correct distribution, but we can once more apply the random self-reduction to simulate a proper random instance of $\Ps,\Ps^{-1}$.

Importantly, observe that the reduction which simulate $\Ps,\Ps^{-1}$ only require access to $Q$ \emph{in the forward direction}. This means we have reduced our problem to the collision resistance of $Q$, given only forward queries to $Q$.

\paragraph{Step 4: Constructing Hard Coset-Partition Functions.}
It may appear that we are still stuck: while we have eliminated the dual and inverse oracles, we have re-introduced the extra structure of a coset partition function. Thankfully, our reduction from the previous Step 3 actually shows that \emph{any} coset partition function $Q$ can be used to simulate the oracles $\Ps,\Ps^{-1}$, as long as $Q$ maps from $n$ to $r$ bits and each preimage set, which is a coset, has size $2^{n - r}$. This means that finding any distribution over such $Q$ that is provably collision resistant, will constitute the collision resistance of our construction.

We now explain a simple construction of such a function $Q$. We start by observing that a 2-to-1 function is trivially a coset partition function. This is because any pre-image set, which is 2 points $x,x'$, is automatically a coset of the $1$-dimensional subspace $S := \{ 0^{n}, x \oplus x' \} \subseteq \bbZ_{2}^{n}$ with a constant shift of $x$ (or equivalently, a constant shift of $x'$). Also, the known quantum collision lower-bounds~\cite{JACM:AarShi04,QIC:Zhandry15} imply that a random 2-to-1 function is provably collision resistant. Now, just any random 2-to-1 function is not enough for us, but the following can be done.

We observe that a straightforward combination of several existing results proves that a random $2$-to-$1$ function, with the added property that it is shrinking by 1 bit, is also collision resistant. Given such shrinking $\ell$ i.i.d. random 2-to-1 functions $H_{1} : \{ 0, 1 \}^{n'} \rightarrow \{ 0, 1 \}^{n' - 1}$, $\cdots$, $H_{\ell} : \{ 0, 1 \}^{n'} \rightarrow \{ 0, 1 \}^{n' - 1}$, take $Q : \{ 0, 1 \}^{ n' \cdot \ell } \rightarrow \{ 0, 1 \}^{ n' \cdot \ell - \ell }$ to be the $\ell$-wise parallel application of them. The pre-image sets are then the direct sums of $\ell$ pre-image sets of the underlying 2-to-1 functions, and direct sums of cosets are cosets. Parallel application also preserves collision resistance. Putting everything together we set $n := n' \cdot \ell$, $r := n' \cdot \ell - \ell$, which proves the oracle-security of our construction, thereby proving Theorems~\ref{thm:main1} and~\ref{thm:main2}.

\subsection{Obfuscating PRPs \ifllncs\else(Section ~\ref{sec:prps})\fi}\label{sec:overviewprps}

Our next goal is to turn our oracle proof into a standard-model proof. The natural approach is, instead of providing oracles for the various functions $\Ps,\Ps^{-1},\Ds$, to provide obfuscations of these functions. To make the functions efficient, we will replace the random permutation $\Pi$ with a pseudorandom permutation, and the random choices of $\matA_y,\vecB_y$ with values generated by a pseudorandom function.

But we immediately run into a problem: There are simply no known techniques for proving the security of obfuscated pseudorandom permutations when using the standard notion of indistinguishability obfuscation (iO). This is because iO only provides a seemingly very weak guarantee: Functionally-equivalent programs are indistinguishable when obfuscated. In order to use iO, some of the program transformations actually need happen outside of the iO, which in turn requires other cryptographic techniques. Most of the iO literature follows the punctured programming approach~\cite{STOC:SahWat14}, which uses the notion of a ``punctured PRF'' \cite{CCS:KPTZ13,AC:BonWat13,PKC:BoyGolIva14}. These, very roughly, allow for giving out a program that either computes the PRF correctly, \emph{or} the program computes the PRF everywhere but a single (known) point, and for that one point the output is uniformly random. Security says that these two programs are indistinguishable.

\paragraph{Permutable PRPs.} Unfortunately, pseudorandom permutations provably cannot be punctured in the same sense as PRFs, as explained by \cite{TCC:BonKimWu17}. Roughly, the reason is that replacing the output at a single point with a uniform random value means the function is no longer a permutation.

We instead define the notion of a \emph{permutable} PRP. Here, given $k \in \{ 0, 1 \}^{\secp}$ a PRP secret key and some fixed permutation $\Gamma$ on $\{ 0, 1 \}^{n}$, it is possible to produce a circuit which computes either (1) $\prp(k,\cdot)$ and its inverse correctly, or (2) $\Gamma(\prp(k,\cdot))$ and its inverse. Security requires that (1) and (2) are indistinguishable, even if $\Gamma$ is known. This avoids the impossibility of~\cite{TCC:BonKimWu17} since we always maintain that the program being computed is a permutation. When used in iO proofs, permutable PRPs readily give that $\iO(\prp(k,\cdot))$ is computationally indistinguishable from $\iO(\Gamma(\prp(k,\cdot)))$. This even holds true if given obfuscated programs for the inverses, and even for more complicated programs that may query the permutations several times. Observe that a similar statement for puncturable pseudorandom functions readily follows from~\cite{TCC:CLTV15}, but that case inherently has no inverse.

\paragraph{Constructing Permutable PRPs for Neighbor Swaps.} For now we will focus on an extremely simple setting, where we restrict $\Gamma$ to be a swap  between $z \in \{ 0, 1 \}^{n}$ and $(z + 1) \in \{ 0, 1 \}^{n}$ for some value $z$, but leave all other points unaffected. We will call such $\Gamma$ a \emph{neighbor swap}.

Even for the simple case of a neighbor swap, a permutable PRP is non-trivial. An ``obvious'' choice is to take a PRP built from a pseudorandom \emph{function} (PRF), and instantiate the PRF with a puncturable PRF. But it is not at all clear a priori how puncturing the underlying PRF allows for swapping outputs. 

In fact, this strategy \emph{cannot} work in general. For example, we argue that it cannot work for Luby-Rackoff PRPs. To see this, we observe that in the setting where the domain size is polynomial, a neighbor-swap PRP actually is also a standard PRP. This is because the neighbor-swap property ensures that you can computationally undetectably swap $z$ and $z+1$ in the output truth table. Since any permutation on a polynomial-domain can be decomposed into a polynomial-length sequence of neighbor swaps, this says that you can apply a \emph{random} permutation un-detectably. But composing with a random permutation actually gives a truly random permutation, thus showing that the truth table of the PRP is computationally indistinguishable from the truth table of a random permutation. We can also scale this argument up, and show that if the neighbor-swap PRP is sub-exponentially secure, then it must be secure against a ``truth-table'' adversary that is provided the entire (exponential-sized) truth table. We then recall that standard PRP constructions such as Luby-Rackoff \emph{cannot} achieve security in the small-domain/ truth-table setting~\cite{AC:Patarin01}.

\begin{figure}
\centering\includegraphics[width=8cm]{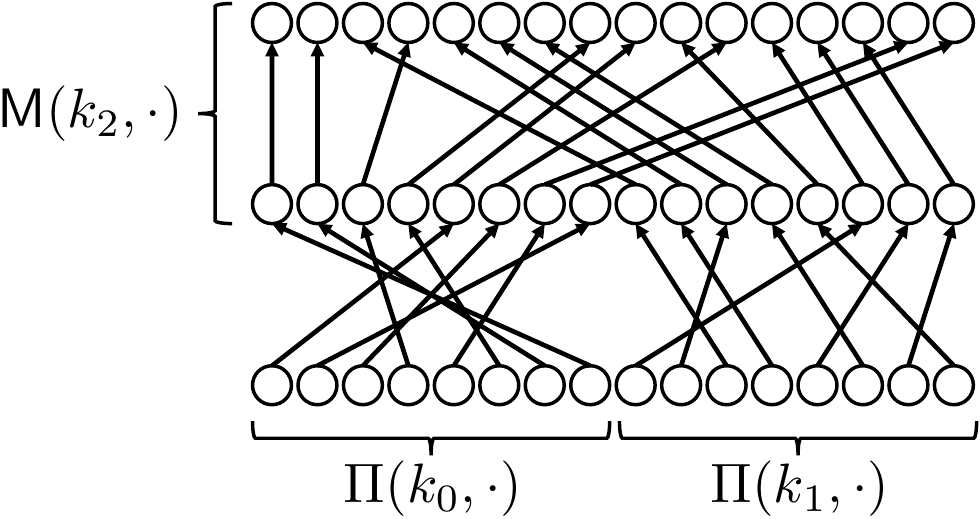}\hspace{2cm}
\caption{\label{fig:prp-constr} The PRP construction $\Pi(k,\cdot)$ of~\cite{FSE:GraPor07}. There are $2^n$ input nodes on the bottom, one per input, and $2^n$ output nodes. The left and right halves correspond to inputs starting with 0 and 1 respectively. Here, $k_0,k_1,k_2$ are keys pseudorandomly derived from $k$. Then $\merge\left( k_2, \cdot \right)$ is a pseudorandom object we call a Merge which preserves the order of each half but otherwise pseudorandomly scrambles them. $\Pi\left( k_0, \cdot \right)$ and $\Pi\left( k_1, \cdot \right)$ are then recursive calls to the construction on inputs of length $n-1$.}
\end{figure}

Guided by the above discussion, we look to the literature on small-domain permutations, specifically the work of~\cite{FSE:GraPor07}. While \cite{FSE:GraPor07} present their construction in a very procedural way involving ``permutators'', ``splitters'' and ``repartitors'', we will present the idea in a more conceptual way. We will then explain how our permuting algorithm executes neighbor swaps efficiently for any (possibly exponential) domain size. The first observation is that one can construct a random permutation as follows: Divide into two equal sized piles (those strings starting with 0 and 1, respectively), and then recursively, randomly permute each pile independently. Then randomly merge the two piles together; this step preserves the order in each pile (which were already shuffled, so this is fine), but randomly determines how the two piles are interleaved. It is not hard to see that every permutation uniquely corresponds to a triple containing a merge (determining the interleaving) and two recursive permutations. Thus this process perfectly simulates a random permutation. See Figure~\ref{fig:prp-constr}.

It turns out that permutations structured in this way can be evaluated efficiently, even for an exponential-size domain. To evaluate on a point, we only need to efficiently evaluate \emph{one} of the recursive permutations for the pile it belongs to, and then the merge operation (which we simply call a Merge). As there will be $n$ levels of the recursion and each level only makes a single recursive call, this means that we can evaluate the overall permutation using $n$ Merge evaluations. We show how to compute the Merge operation efficiently using a random data-structure we call a tally tree. A tally tree is a full binary tree on $2^n$ nodes, where leaf $z$ corresponds to an \emph{output} of the Merge, and is given the value of the first bit of the pre-image of $z$. Then the internal nodes count the total number of 1's in the leaves of the sub-tree rooted at that node. An example tally tree $T$ for a marge $\merge$ is given in Figure~\ref{figure:tally}.
\begin{figure}
\centering\includegraphics[width=8cm]{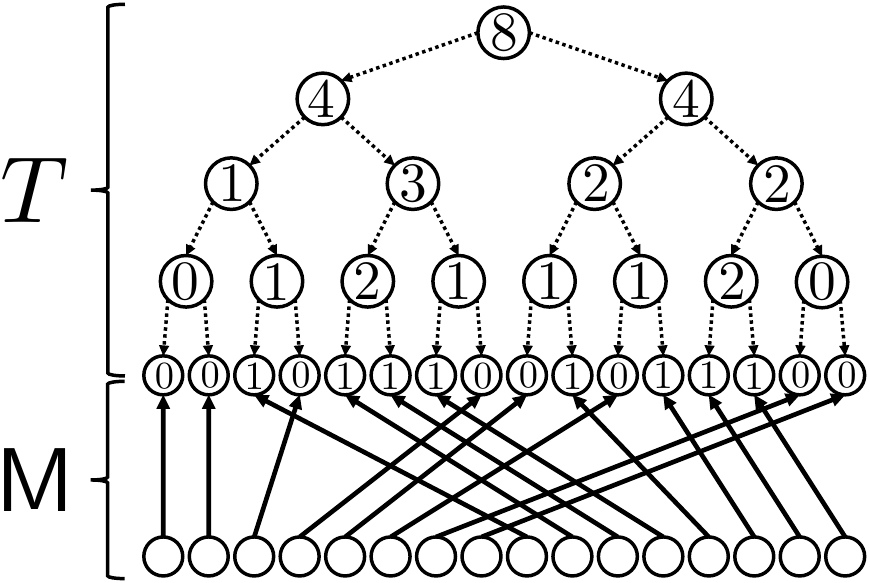}
\caption{\label{figure:tally} A merge $\merge$ and the associated tally tree $T$, for the case $N_0=N_1=8$, $N=16$.}
\end{figure}

Using a tally tree we can efficiently evaluate the Merge. We start by describing the inverse: since the Merge preserves the order within each half, we can compute the inverse of $z$ just knowing two pieces of information: (1) which half the pre-image of $z$ belongs to, and (2) how many nodes to the left of $z$ go to the same half. This can be computed easily using the tally tree. Then the forward direction can be computed using a binary search algorithm which makes queries to the inverse and again exploits the order-preserving property.

The next observation is, rather than choosing a random merge and computing the tree in a bottom-up manner, we can actually sample the tree in a top-down manner. In this view, the value in each node follows an appropriate hypergeometric distribution based on the value of its parent. Then the random choices can be simulated using a pseudorandom function (PRF). The values in the tree are left implicitly determined by the PRF key, to be computed on the fly as needed during evaluation (which only visits a polynomial-sized portion of the tree since evaluation is efficient). The result is that keys can be small (independent of the domain size), and evaluation takes time poly-logarithmic in the domain size.

We start with this construction, and first show that we can permute $z$ and $z+1$ by either permuting the Merge or by permuting one of the recursive calls to $\prp$. In particular, if the pre-images of $z$ and $z+1$ lie in different halves of the domain, then we can permute within the Merge. If the pre-images lie in the same half, then we cannot permute the Merge since this would violate the ordering of elements in that half. In this case, however, the order-preserving property of the Merge guarantees that the pre-images of $z,z+1$ are adjacent, and so we can instead permute within the recursive $\prp$ call for that half. See Figure~\ref{fig:prp-permute}.
\begin{figure}
\centering\begin{tabular}{r}\includegraphics[width=7.7cm]{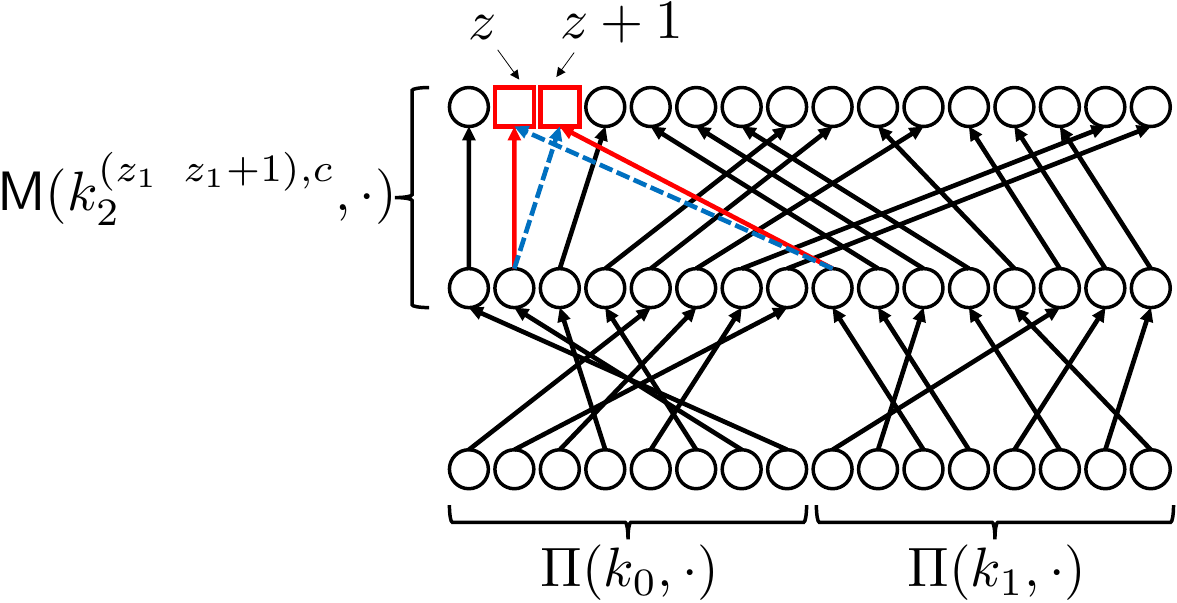}\ifllncs\\\else\hspace{1cm}\fi\includegraphics[width=6.3cm]{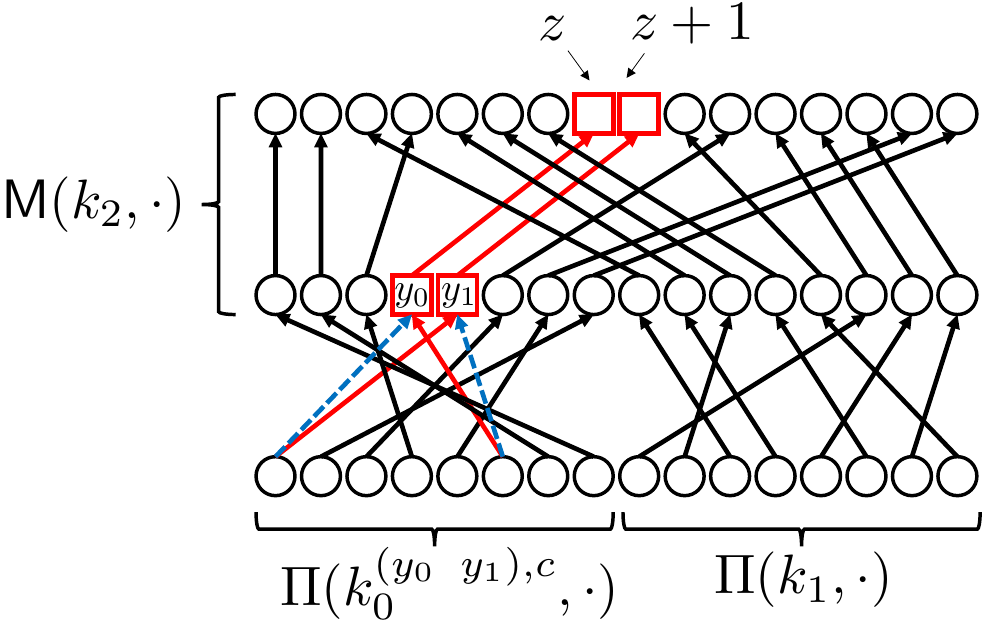}\end{tabular}\hspace{2.5cm}
\caption{\label{fig:prp-permute} Our permuting algorithm. Here, $k^{\neighborswap{z},c}$ means the permuted key, where $c=0$ means no swapping occurs, and $c=1$ means $z$ and $z+1$ are swapped. Red squares indicate the points $z,z+1$, red solid arrows indicate the original permutation ($c=0$), and blue dashed lines indicate the permuted permutation ($c=1$). Let $y_0,y_1$ be the pre-images of $z,z+1$ in the Merge.\ifllncs Top: \else Left: \fi The case where $y_0, y_1$ lie in different halves, and we permute $z,z+1$ by permuting the merge $\merge$.\ifllncs Bottom: \else Right: \fi The case where $y_0,y_1$ lie in the same half, and we permute $z,z+1$ by permuting the recursive application of $\prp$, as indicated by the red squares.}
\end{figure}

Thus, we have reduced the task of permuting $\Pi$ to permuting the Merge, and we only need to concern ourselves with the case where the pre-images of $z,z+1$ lie in different halves. If we look at the leaf nodes $z,z+1$ in the tally tree, since they have pre-images in different halves, this means one of the nodes has a value of 0 and the other a value of 1. Moreover, permuting $z, z+1$ corresponds to exchanging which is a 0 and which is a 1. 

What we show is that such permuting can be accomplished by puncturing the underling pseudorandom function (PRF) that is used to generate the tally tree. We use a puncturable PRF that can be punctured at several points, resulting in those points being replaced with random values. Concretely, we puncture at the nodes $z,z+1$, as well as all nodes on the paths from these nodes to the root, and also the siblings of those nodes. By analyzing the induced hypergeometric distributions, we conclude after puncturing that the values at $z,z+1$ are actually statistically equally likely to be 0 or 1. Thus, we can swap their values without detection. See Figure~\ref{fig:puncturing} for an example.
\begin{figure}
\centering\includegraphics[width=7.5cm]{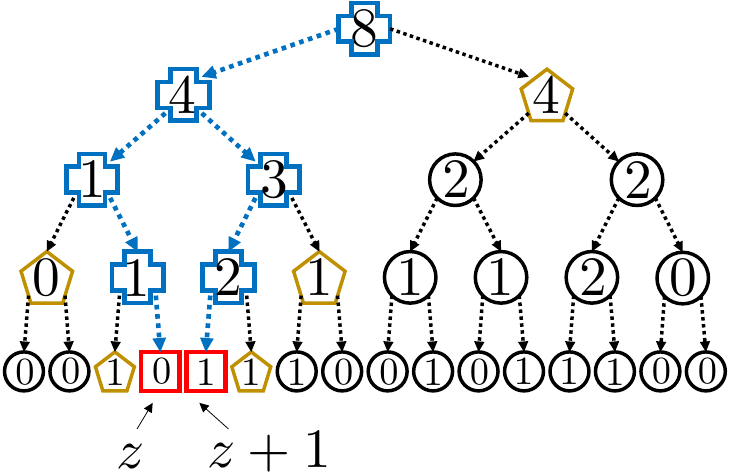}\hspace{1cm}\includegraphics[width=7.5cm]{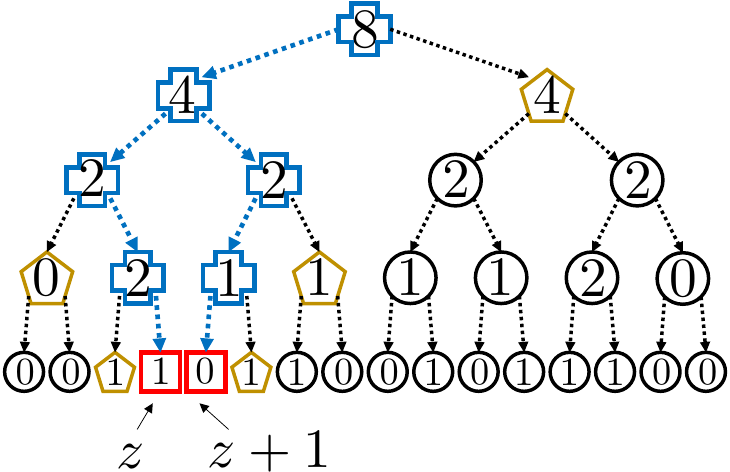}
\caption{\label{fig:puncturing}\ifllncs Top: \else Left: \fi The tally tree for a permuted merge key $k^{\neighborswap{z},0}$, where no swapping happens. The red squares are the positions to be swapped $z,z+1$. The blue crosses are the paths from root to these nodes and the gold pentagons are the siblings of these paths; these nodes have all been punctured.\ifllncs Bottom: \else Right: \fi The tally tree for the permuted merge key $k^{\neighborswap{z},1}$. This is identical to the\ifllncs top \else left \fi tree, except that we swapped the bits at $z$ and $z+1$, and accordingly updated the tallies in the paths to $z,z+1$ (the blue crosses). The gold pentagons and all their descendants remain unchanged.}
\end{figure}

\begin{remark}\cite{FSE:GraPor07} is not particularly efficient from a practical perspective, due to having to sample from hypergeometric distributions. Other works have given more efficient constructions of small-domain PRPs~\cite{STOC:Morris05,C:HoaMorRog12,C:RisYil13,EC:MorRog14}, but it is unclear if these also give permutable PRPs when instantiated with a puncturable PRF.
\end{remark}

\paragraph{Extending to More General Permutations.} Neighbor swaps are not sufficient for most applications. We therefore explain how to extend the construction to handle much more general permutations $\Gamma$.

The idea is simple, and builds upon the intuition that neighbor swaps are enough to generate all permutations. To get the circuits computing $\prp(k,\cdot)$ or $\Gamma(\prp(k,\cdot))$, we simply obfuscate those programs using iO. To show that the two cases are indistinguishable, we decompose $\Gamma$ into a sequence of (exponentially-many) neighbor swaps, and perform a sequence of hybrids where we apply one neighbor-swap at a time.

There are two caveats to this approach. One is that we will need exponentially-many hybrids, and therefore we must rely on sub-exponentially-secure iO. The underlying puncturable PRF must also be sub-exponentially secure, which is implied by sub-exponential one-way functions.

The more challenging caveat is the following. If we decompose $\Gamma$ as $\Gamma = \tau_N \circ \tau_{N-1} \circ \cdots \circ \tau_1$ for neighbor-swaps $\tau_i$, we actually require that each partial permutation $\Gamma_t := \tau_N \circ \tau_{N-1} \circ \cdots \circ \tau_t$ is computable by a small circuit. These need to be small because, at the appropriate hybrid, they will be hard-coded into the program that gets obfuscated. To maintain security given the obfuscated inverse circuits, we also need each $\Gamma_t^{-1}$ to have small circuits. But in general, even if $\Gamma, \Gamma^{-1}$ have small circuits, it is not obvious that $\Gamma$ can be decomposed in a way which guarantees that each $\Gamma_t, \Gamma_t^{-1}$ have smalls circuit. This leads to our Efficient Permutation Decomposition (EPD) Problem (Question~\ref{question:decompose}), which asks whether such a guarantee is possible in general.\footnote{Note that Question~\ref{question:decompose} was about decomposing into transpositions rather than neighbor swaps. But as\ifllncs shown in the full version \cite{shmueli2025one} \else we will see\fi, transpositions can be decomposed into neighbor swaps, and thus the two versions are equivalent.}

We do not know how to resolve the EPD problem in general. Instead, we define a notion of ``decomposable'' circuits, which are $\Gamma$ that can be decomposed into a (potentially exponentially-long) sequence in such a way that the $\Gamma_t,\Gamma_t^{-1}$ all have small circuits.\footnote{Note that there may be many possible decompositions into neighbor-swaps, and some decompositions may give small circuits while others may not.} Fortunately, decomposable circuits capture many natural types of permutations, such as general (non-neighbor) transpositions, linear cycles $j\mapsto j+1\mapsto \cdots \mapsto \ell-1\mapsto\ell\mapsto j$, involutions (i.e., such that $\Gamma\circ\Gamma$ is the identity), and more. See Figure~\ref{fig:decomposable} for a longer list of examples of decomposable permutations, which are shown to be decomposable in\ifllncs the full version of this work\else Section \ref{sec:prps}\fi. Our proof shows that an obfuscated neighbor-swap PRP is a permutable PRP for all decomposable circuits. This gives the informal Theorem~\ref{thm:main5}.

\begin{figure}
\centering \begin{tabular}{l}
Transpositions $i\leftrightarrow j$.\\
Linear cycles $j\rightarrow j+1\rightarrow j+2\rightarrow\cdots\rightarrow\ell\rightarrow j$.\\
Affine mod-2 permutations $\xv\rightarrow \matA\cdot\xv+\vecV\bmod 2$ if $\det(A)\bmod 2=1$.\\
Products $(x,y)\rightarrow(\Gamma_0(x),\Gamma_1(x))$ for decomposable $\Gamma_0,\Gamma_1$.\\
Compositions $\Gamma=\Gamma_0\circ\Gamma_1$ for decomposable $\Gamma_0,\Gamma_1$.\\
Controlled permutations $(x,y)\rightarrow (x,\Gamma_x(y))$ for decomposable $\Gamma_x$.\\
Efficient involutions ($\Gamma\circ\Gamma=\mathbf{I}$).\\
Conjugations $\Lambda^{-1}\circ\Gamma\circ\Lambda$ for decomposable $\Gamma$ and efficient $\Lambda,\Lambda^{-1}$.\\
Permutations with ancillas $(x,0^{n})\rightarrow (\Gamma(x),0^{n})$ for efficient $\Gamma,\Gamma^{-1}$.
\end{tabular}
\caption{\label{fig:decomposable} A non-exhaustive list of decomposable permutations. Note that for conjugations, $\Lambda$ need not be decomposable. For ancillas, $\Gamma:\{0,1\}^n\rightarrow\{0,1\}^n$ need not be decomposable, but the permutation $(x,0^{n})\rightarrow (\Gamma(x),0^{n})$ is only a partial function that is not specified on the rest of the domain. The permutable PRP domain in this case is the entire space including ancillas.}
\end{figure}

\begin{remark}Interestingly, despite showing that a wide range of permutations are decomposable, even for very simple permutations, decompositions can be tricky. For example, $x\mapsto 2x\bmod N$ for odd $N$ is decomposable, though our proof uses the fact that the map is a merge with an efficiently computable tally tree, and shows that such tally trees allow for decomposition. This generalizes to $x\mapsto ax\bmod N$ for \emph{polynomial} $a$. But the only way we know how to handle general linear maps requires decomposing $a$ into the product ($\bmod \: N$) of polynomially-bounded scalars. This in turn requires that such polynomial scalars generate $\Z_N^*$. This is implied by the Extended Riemann Hypothesis~\cite{Bach90}. However, we leave a general unconditional decomposition, even for scalar multiplication, as an intriguing open question.
\end{remark}

\paragraph{Applications.} We show that permutable PRPs can be very useful in obfuscation applications.\footnote{As we describe it, our applications will obfuscate a permutable PRP, where our permutable PRPs are built using obfuscation. We can equivalently simply obfuscate the underlying neighbor swap PRP, and all the hybrids in the proof of security for the permutable PRP can instead be carried out at the level of the application. However, for conceptual simplicity, we believe it is beneficial to describe our results using the permutable PRP abstraction.} Concretely, we show that obfuscating a permutable PRP for general transpositions (which in particular are decomposable) immediately gives a \emph{full-domain} trapdoor one-way permutation, yielding Theorem~\ref{thm:main6}. The proof goes through several hybrids:
\begin{itemize}
    \item Hybrid 0: The adversary is given an obfuscation of $\prp(k,\cdot)$ and a challenge point $y^*$, and has to find $x^*$ such that $y^* = \prp(k, x^*)$.
    
    \item Hybrid 1: We choose a random output $y'$, and ``puncture'' the program so if $\Pi(k,x)=y'$, the program outputs $y^*$ instead of $y'$. Otherwise the program is unchanged. This means there are now two points $x^*,x'$ that map to $y^*$. This change follows standard iO techniques.
    
    \item Hybrid 2: We use permuted key $k^{\transposition{x^*}{x'},0}\gets\permute\left( k, \transposition{x^*}{x'}, 0 \right)$ but where the transposition $\transposition{x^*}{x'}$ is still turned off. Recall per Figure~\ref{fig:decomposable} that transpositions are decomposable. We use the permuted key instead of the real key $k$ for all purposes of this hybrid. Since these are functionally equivalent, by the security of the iO, the transition to this hybrid is computationally undetectable. 
    
    \item Hybrid 3: Now we compose the PRP with the transposition swapping $x^*,x'$. That is, we sample $k^{\transposition{x^*}{x'},1}\gets\permute\left( k, \transposition{x^*}{x'}, 1 \right)$ and use it instead of $k^{\transposition{x^*}{x'},0}$, which was used in the previous hybrid. By the security of the permutable PRP this change is computationally indistinguishable.
    Note that, importantly we transpose $x^*$ and $x'$ both within the obfuscated program and also to verify the adversary's output (where we check that the circuit's output equals $y^*$). This swap means that now the adversary has to find $x'$ instead of $x^*$. That is, after we move to this hybrid, we can make another change to how we check the success of this hybrid: Take the original PRP key $k$ and simply check that $y' = \prp\left( k, x \right)$.
    
    \item Hybrid 4: Now we move to using the original key $k$ inside the obfuscated circuit instead of the permuted key $k^{\transposition{x^*}{x'},1}$. Note that the functionality of the outside circuit (which is also obfuscated by iO) is unchanged because it is still the case that on both, $x^*$, $x'$ it outputs $y^*$. This change is computationally indistinguishable by the security of the iO.

    \paragraph{Finalizing the reduction.}
    In the setting of the last hybrid experiment, we sample a random PRP key $k$ and a uniformly random $y'$. We then obfuscate the circuit $P_1$, which computes the PRP with key $k$, checks if the output was $y'$, and if so outputs $y^*$ instead of $y'$ (or in other words, does something unrelated to the original problem). The adversary then finds $x'$ which eventually leads to finding $y'$. By standard iO techniques, finding such uniformly random $y'$ is computationally intractable, even given the key $k$, the value $y^*$ and the obfuscation of the program $P_{1}$.
\end{itemize}

An elaborated formal proof is found in \ifllncs{the full version \cite{shmueli2025one}}\else{Section \ref{sec:prps} (Theorem \ref{theorem:owp_from_io})}\fi.

\subsection{OSS In the Standard Model\ifllncs\else (Section \ref{sec:standard})\fi} \label{sec:overviewstandard}

We finally turn to constructing OSS in the standard model. The idea is to simply replace the random permutation $\Pi$ with a permutable PRP $\prp(k,\cdot)$ and to have the matrix/vector pairs $\matA_y,\vecB_y$ be generated pseudorandomly from a (puncturable) pseudorandom function. Fortunately, the preceding sections have already set up most of the ideas we need for the proof, as our sequence of reductions for proving our oracle construction actually have natural cryptographic analogs. 

For our random self-reduction, we need to argue that the simulated obfuscated program is indistinguishable from the real distribution. In the oracle model, we have perfect indistinguishability. In the standard model case, we no longer have perfect indistinguishability since the simulated permutation $\Pi$ is now the composition of two permutable PRPs, and the matrix/vector pairs $\matA_y,\vecB_y$ are derived from underlying terms $\overline{\matA}_y,\overline{\vecB}_y,\matC_y,\vecD_y$, each of which are pseudorandomly generated. Fortunately, a standard hybrid over every $y$, utilizing punctured PRF security lets us move to the correct distribution over $\matA_y,\vecB_y$, and permutable PRP security lets us replace the composition of two permutable PRPs with a single PRP.

For dual-bloating, we follow the approach of \cite{EC:Zhandry19b}, which shows how to increase the size of of obfuscated subspaces using iO and one-way functions. We can likewise eliminate the dual as in the oracle case, as our argument works just as well in the standard model. Finally, we can embed a coset-partition function just as we did in the oracle case, but now using our computational random self-reduction.

There are, however, two main hurdles to getting everything to work. The first is that, in order to show that the obfuscation of the composition of two permutable PRPs is indistinguishable from a single permutable PRP, we need our permutable PRPs itself to be decomposable. In \ifllncs{the full version \cite{shmueli2025one} }\else{Theorem \ref{thm:op-prp} }\fi we construct permutable PRPs by our permuting algorithm, while carefully also maintaining their decomposability. This allows for overall composition between permutable PRPs.

Another caveat is that we need a post-quantum collision-resistant 2-to-1 function $L$, since we can no longer rely on the unconditional existence relative to an oracle. It turns out that we also need this $L$ to have some extra properties, that were not present in the oracle case. Specifically, when we showed that embedding $Q$ (the parallel repetition of $L$) results in a valid instance of the oracles $\Ps,\Ps^{-1}$, we actually argued that this means there \emph{exists} an implicit permutation $\Pi$ determined by $Q$. Very roughly, since $Q$ has collisions it therefore loses information, and $\Pi$ outputs $Q$ plus some additional bits to recover the lost information. For the oracle proof, we only needed the existence of $\Pi$ to argue that the simulated oracles had the same distribution as the real oracles. But for our standard-model argument, we will actually have a hybrid where $\Pi$ is hard-coded into the obfuscated program. That means we need $\Pi$ to have a small circuit. Unfortunately, $\Pi$ essentially requires inverting $Q$, and so in general is computationally inefficient. Note that the inefficient circuit is not needed in normal usage, but just in the proof.

As a result, we need the function $Q$ to have a number of specific properties, which mainly boil down to having $L$ a collision-resistant 2-to-1 function that has a trapdoor. Fortunately, we can construct such ``trapdoor'' 2-to-1 hash functions from lattices, following a similar approach as the claw-free trapdoor functions built from LWE~\cite{FOCS:BCMVV18}.

However, focusing on the LWE-based hash functions leads into the third and fourth caveats. For the third: The range of the LWE-based hash function is actually larger than the domain. In contrast, for our oracle construction described above, the function $Q$ needs to be surjective in order for filling in the lost information to actually yield a permutation. But this means the range of $Q$ must be smaller than the domain since it is many-to-one. This also presents a problem for invoking permutable PRP security, since permutable PRPs allow for composing with \emph{permutations}, but our LWE-based hash function is no longer a permutation.

Fortunately, we can expand the range by padding $y$ with 0's, and composing the output of the construction with \emph{another} permutable PRP. This leads to our ultimate construction, which we describe visually in Figure~\ref{figure:standard-constr}. We can also split up the evaluation of $Q$ into multiple stages, where each stage is a (decomposable) permutation which allows us to invoke permutable PRP security.

\begin{figure}
\centering\includegraphics[width=10cm]{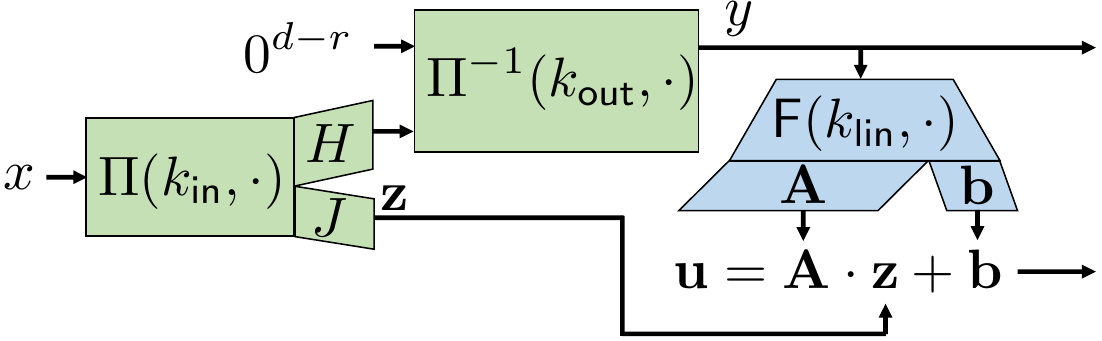}
\caption{\label{figure:standard-constr} The program $P$ that gets obfuscated to give the program $\Ps$. The programs $P^{-1},D$ are defined accordingly. Here, $\Pi,\Pi^{-1}$ is a permutable PRP, and $\prf$ is a puncturable pseudorandom function. The output of $H$ is $r$ bits, $\zv$ is $n-r$ bits, and $y$ is $d$ bits.}
\end{figure}

The fourth caveat is that the LWE-based hash function $L$ is really only 2-to-1 on an overwhelming fraction of the domain, but it is easy to devise points where the function is only 1-to-1 on those points. Intuitively, this makes the parallel repetition $Q$ of $L$ be, accordingly, only an \emph{approximate} coset partition function (i.e., it is not the case that for every output point of $Q$, the preimage set is a coset of the same dimensions). This breaks the straightforward iO proof: When we embed $Q$ into the programs, the programs will actually not have the same functionality due to the 1-to-1 points. This prevents us from naively using iO, which requires functionally equivalent programs.

To account for this, we add a trigger at all the sparse ``bad'' 1-to-1 points, and if this trigger occurs, the program behaves differently. We could try to embed this trigger in the original construction, but it is not clear if we would be able to carry out all of our re-randomization steps; even if it could work, it would significantly complicate each of the steps. Instead, we add the trigger only in the proof. Standard iO techniques allow for adding sparse triggers on \emph{random} points that are otherwise independent of the program, but not fixed triggers like we need.\footnote{Observe that the ``bad'' points are fixed once the hash function is chosen.} Here, we use our permutable PRPs again. We show that, as long as the fixed trigger is sandwiched between two permutable PRPs (which they are, for our case, because we are already composing another PRP on the output) and meets some other technical conditions, we can un-detectably add the trigger. See \ifllncs{the full version \cite{shmueli2025one} }\else{Lemma \ref{lemma:intervaltrigger} in Section \ref{sec:prps} }\fi for a precise statement, and Figure \ref{figure:sparsetrigger} for a visual representation of the kinds of program changes which we can make. 

Our proof of the sparse trigger works as follows. To trigger at a single \emph{fixed} point $y$, we actually first add a trigger at a \emph{random} point $y'$. Adding random triggers follows from standard iO techniques. Then, we use the permutable PRP property to exchange $y$ and $y'$, meaning now the trigger occurs at $y$. This step requires a permutable PRP both before and after the trigger, since after the trigger, we need to return $y$ and $y'$ to their original values. With a bit more work, we can extend this to interval triggers, as long as the interval is at most a sub-exponential fraction of the domain. 

Once we add the trigger, then embedding $Q$ actually yields an equivalent program and allows the proof to go through. By piecing all of the steps together, we obtain Theorems~\ref{thm:main3} and~\ref{thm:main4}.

\begin{figure}
\centering\includegraphics[width=10cm]{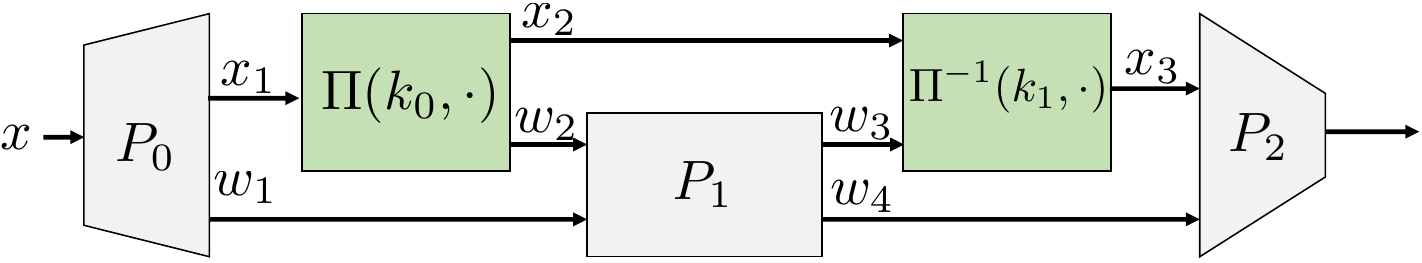}\\\vspace{0.5cm}\includegraphics[width=10cm]{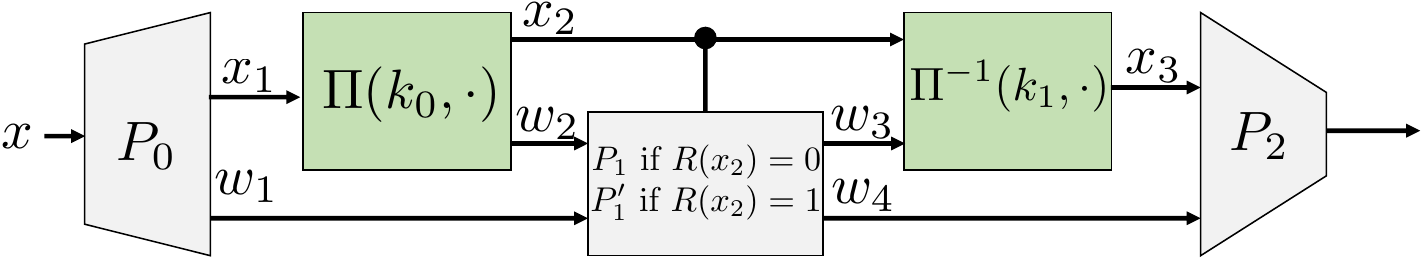}
\caption{\label{figure:sparsetrigger} The original program $P$ (top) and triggered program $P'$ (bottom). Here, $P_0,P_1,P_1',P_2$ are arbitrary programs as long as they are independent of the keys $k_0,k_1$. Our indistinguishability result shows that for the case where $R(x_2)$ tests if $x_2$ (or some component thereof) is in a fixed sparse interval, then obfuscating these two programs yields computationally indistinguishable programs. Note that $R$ is fixed and known, but due to the application of the permutable PRPs, the points that hit the trigger are computationally hidden once obfuscated.}
\end{figure}

\ifllncs
\else
\begin{remark}\label{remark:other} One may wonder if there are other instantiations of the needed 2-to-1 collision-resistant hash function. Aside from needing the structure of being 2-to-1, another requirement that limits instantiations is the apparent need for a trapdoor. This seems to prevent us from using hash functions based on LPN~\cite{EC:BLVW19,AC:YZWGL19} or super-singular isogeny graphs~\cite{JC:CLG09}. One can always make a some-what tautological assumption that obfuscating the function which discards the last bit of a (permutable) PRP gives such a function (the trapdoor being the un-obfuscated key), though we do not know how to prove this based on any standard assumption. Another possibility is to look at constructions from group actions such as~\cite{alamati2022candidate}, which can plausibly be post-quantum instantiated using \emph{ordinary} isogenies over elliptic curves. Their construction has a trapdoor and can be split up into decomposable parts analogous to the LWE-based construction, and has \emph{most} of the domain being 2-to-1. However, a problem is that a non-negligible fraction of the domain is 1-to-1. This breaks the step where we embed a trigger into the obfuscated program, since the trigger is no longer sparse. An interesting open question is whether their construction can be modified so that an exponentially-small fraction of the domain is 1-to-1. 
\end{remark}
\fi

\ifllncs
\bibliographystyle{alpha}
\bibliography{bib,abbrev0,crypto}
\appendix
\fi

\ifllncs
\else

\section{Cryptographic Tools} \label{sec:cryptodefs}
We present the general cryptographic tools and techniques which are used in this work. Some of the results in this section are new and may be of independent interest.

\paragraph{Cryptographic Preliminaries.}
We use the following known primitives and notions. Both are implicitly classical primitives with security holding against quantum algorithms.

\begin{definition} [Puncturable PRFs] \label{def:pprf}
A puncturable pseudorandom function (P-PRF) is a pair of efficient algorithms $\left( \prf, \punc, \eval \right)$ with associated output-length function $m\left( \lambda \right)$ such that:
\begin{itemize}

  \item $\prf : \{0,1\}^\lambda \times \{0,1\}^* \rightarrow \{0,1\}^{m\left( \lambda \right)}$ is a deterministic polynomial-time algorithm.
  
  \item $\punc\left( k, S \right)$ is a probabilistic polynomial-time algorithm which takes as input a key $k \in \{0,1\}^\lambda$ and a set of points $S \subseteq \{0,1\}^*$. It outputs a punctured key $k^S$.
  
  \item $\eval\left( k^S, x \right)$ is a deterministic polynomial-time algorithm.
  
  \item {\bf Correctness:}
  For any $\lambda \in \Nat$, $S \subseteq \{ 0, 1 \}^*$, $k \in \{0,1\}^\lambda$, $x \notin S$, and $k^S$ in the support of $\punc\left( k, S \right)$, we have that $\eval\left( k^S, x \right) = \prf\left( k, x \right)$.
  
  \item {\bf Security:}
  For any quantum polynomial-time algorithm $\As$, there exists a negligible function $\epsilon$ such that the following experiment with $\As$ outputs 1 with probability at most $\frac{1}{2} + \epsilon\left( \lambda \right)$:
  \begin{itemize}
    \item $\As(1^\lambda)$ produces a set $S \subseteq \{0,1\}^*$.
    
    \item The experiment chooses a random $k\gets\{0,1\}^\lambda$ and computes $k^S \gets \punc\left( k, S \right)$. For each $x\in S$, it also sets $y_x^0 := \prf\left( k, x \right)$ and samples $y_x^1 \gets \{0,1\}^{m(\lambda)}$ uniformly at random. Then it chooses a random bit $b$. It finally gives $k^S, \{\left( x, y_x^b \right) \}_{ x \in S }$ to $\As$.
    
    \item $\As$ outputs a guess $b'$ for $b$. The experiment outputs 1 if $b'=b$.
  \end{itemize}
  
\end{itemize}
\end{definition}

\paragraph{Different security levels.}
For arbitrary functions $f_{0}, f_{1} : \Nat \rightarrow \Nat$, we say that the P-PRF is $\left( f_{0}, \frac{1}{f_{1}} \right)$-secure if in the above the security part of the definition, we ask that the indistinguishability holds for every adversary of size $\leq f_{0}(\secp)$ and we swap $\epsilon(\secp)$ with $\frac{1}{ f_{1}(\secp) }$. Concretely, a sub-exponentially secure P-PRF scheme would be one such that there exists a positive real constant $c > 0$ such that the scheme is $\left( 2^{\lambda^c}, \frac{1}{2^{\lambda^c}} \right)$-secure.

\begin{definition} [Indistinguishability Obfuscation (iO)] \label{def:iO}
An indistinguishability obfuscator (iO) for Boolean circuits is a probabilistic polynomial-time algorithm $\iO\left( \cdot, \cdot, \cdot \right)$ with the following properties:
\begin{itemize}

  \item {\bf Correctness:}
  For all $\lambda, s \in \Nat$, Boolean circuits $C$ of size at most $s$, and all inputs $x$ to $C$,
  \[
  \Pr\left[
  \Obf_{C}(x) = C(x) \: : \: \Obf_{C} \leftarrow \iO\left( 1^\lambda, 1^s, C \right)
  \right] = 1
  \enspace .
  \]

  \item {\bf Security:}
  For every polynomial $\poly(\cdot)$ there exists a negligible function $\epsilon$ such that the following holds.
  Let $\secp, s \in \Nat$, and let $C_{0}$, $C_{1}$ two classical circuits of (1) the same functionality (i.e., for every possible input they have the same output) and (2) both have size $\leq s$.
  $$
  \biggl\{
  \Obf_{C_{0}} \: : \: \Obf_{ C_{0} } \leftarrow \iO\left( 1^\lambda, 1^s, C_{0} \right)
  \biggr\}
  $$
  $$
  \approx_{ \left( \poly(\secp), \epsilon(\secp) \right) }
  \biggl\{
  \Obf_{C_{1}} \: : \: \Obf_{ C_{1} } \leftarrow \iO\left( 1^\lambda, 1^s, C_{1} \right)
  \biggr\}
  \enspace .
  $$ 
\end{itemize}
\end{definition}

\paragraph{Different security levels.}
For arbitrary functions $f_{0}, f_{1} : \Nat \rightarrow \Nat$, we say that an iO scheme is $\left( f_{0}, \frac{1}{f_{1}} \right)$-secure if in the above the security part of the definition, we swap $\poly$ with a concrete function $f_{0}$ and the negligible function with the concrete $\frac{1}{f_{1}}$. Concretely, a sub-exponentially secure iO scheme would be one such that there exists a positive real constant $c > 0$ such that the scheme is $\left( 2^{\lambda^c}, \frac{1}{2^{\lambda^c}} \right)$-secure.

\begin{definition} [Lossy Functions] \label{definition:lf}
A lossy function (LF) scheme consists of classical algorithms $(\LFGen$, $\LFF)$ with the following syntax.

\begin{itemize}
    \item
    $\pk \gets \LFGen\left( 1^\secp, b, 1^{\ell} \right)$: a probabilistic polynomial-time algorithm that gets as input the security parameter $\secp \in \Nat$, a bit $b$ and a lossyness parameter $\ell \in \Nat$, $\secp \geq \ell$. The algorithm outputs a public key.

    \item
    $y \gets \LFF\left( \pk, x \right)$: a deterministic polynomial-time algorithm that gets as input the security parameter $\secp \in \Nat$, the public key $\pk$ and an input $x \in \{ 0, 1 \}^\secp$ and outputs a string $y \in \{ 0,1 \}^{m}$ for some $m \geq \secp$.
\end{itemize}

The scheme satisfies the following guarantees.

\begin{itemize}
    \item {\bf Statistical Correctness for Injective Mode:}
    There exists a negligible function $\negl(\cdot)$ such that for every $\secp, \ell \in \Nat$, 
    $$
    \Pr_{
    \pk \gets \LFGen\left( 1^\secp, 0, 1^{\ell} \right)
    }
    \left[
    \Big| \Img\left( \LFF\left( \pk, \cdot \right) \right) \Big| = 2^\secp
    \right]
    \geq
    1 - \negl(\secp) \enspace .
    $$

    \item {\bf Statistical Correctness for Lossy Mode:}
    There exists a negligible function $\negl(\cdot)$ such that for every $\secp, \ell \in \Nat$, 
    $$
    \Pr_{
    \pk \gets \LFGen\left( 1^\secp, 1, 1^{\ell} \right)
    }
    \left[
    \Big| \Img\left( \LFF\left( \pk, \cdot \right) \right) \Big| \leq 2^\ell
    \right]
    \geq
    1 - \negl(\secp) \enspace .
    $$
    
    \item {\bf Security:}
    For every polynomial $\poly(\cdot)$ there exists a negligible function $\epsilon$ such that the following holds.
    Let $\secp, \ell \in \Nat$, then (note that in the following computational indistinguishability, the security parameter is $\ell$ and not $\secp$),
    $$
    \biggl\{
    \pk_{0} \: : \: \pk_{0} \gets \LFGen\left( 1^\lambda, 0, \ell \right)
    \biggr\}
    $$
    $$
    \approx_{ \left( \poly(\ell), \epsilon(\ell) \right) }
    \biggl\{
    \pk_{1} \: : \: \pk_{1} \gets \LFGen\left( 1^\lambda, 1, \ell \right)
    \biggr\}
    \enspace .
    $$ 
\end{itemize}

\end{definition}

\paragraph{Different security levels.}
For arbitrary functions $f_{0}, f_{1} : \Nat \rightarrow \Nat$, we say that an LF scheme is $\left( f_{0}, \frac{1}{f_{1}} \right)$-secure if in the above the security part of the definition, we swap $\poly$ with a concrete function $f_{0}$ and the negligible function with the concrete $\frac{1}{f_{1}}$. Concretely, a sub-exponentially secure LF scheme would be one such that there exists a positive real constant $c > 0$ such that the scheme is $\left( 2^{\ell^c}, \frac{1}{2^{\ell^c}} \right)$-secure.

\subsection{Two iO Techniques}
Here, we recall two standard iO techniques that we will abstract as useful lemmas.

\paragraph{Sparse Random Triggers.}
Let $P$ be some program and $P'$ an arbitrary different program. Let $R$ be a function with range $[N]$ for some $N$ that is exponential in the security parameter. Let $J_y$ (for `join') be the following program:
$$
J_y(x) =
\begin{cases}
P(x) \text{ if } R(x) \neq y \\
P'(x) \text{ if } R(x) = y
\end{cases}
\enspace .
$$

\begin{lemma} \label{lem:iopuncture}
Suppose one-way functions exist. For sufficiently large polynomial $s$ and for $y$ chosen uniformly in $\{0,1\}^\lambda$, $\iO(1^\lambda,1^s,P)$ and $\iO(1^\lambda,1^s,J_y)$ are computationally indisitnguishable even given the description of $P$. Moreover, $y$ is computationally unpredictable given $P,\iO(1^\lambda,1^s,J_y)$
\end{lemma}
\begin{proof}We first prove indistinguishability through a sequence of hybrid programs:

\begin{itemize}
    \item $\Hyb_{0}$: The original obfuscation of $P$.
\end{itemize}
Here, the adversary is given $\iO(1^\lambda,1^s,P)$.

\begin{itemize}
    \item $\Hyb_{1}$: Adding a uniformly random trigger that applies only if it is in the image of a sparse PRG.
\end{itemize}
We assume an injective length-doubling pseudorandom generator $\prg:[N]\rightarrow[N]^2$. These follow from injective one-way functions, which in turn follow from plain one-way functions and iO~\cite{TCC:BitPanWic16}. Here, we choose a random $w\gets[N]^2$. The adversary is given $\iO(1^\lambda,1^s,J'_w)$ where  $J_w'(x)=\begin{cases}P(x)\text{ if }\prg(R(x))\neq w\\P'(x)\text{ if }\prg(R(x))=w\end{cases}$.
Note that since $\prg$ is length-doubling, we have that with overwhelming probability over the choice of $w$, the second line of $J'_w$ will never be triggered. Therefore, $J'_w$ is functionally equivalent to $P$. Therefore, by $\iO$ security, as long as $s$ is larger than the size of $J'_w$ (which is larger than the size of $P$), hybrids 0 and 1 are indistinguishable.

\begin{itemize}
    \item $\Hyb_{2}$: Changing the trigger to be a random element inside the image of the PRG.
\end{itemize}
Here, we switch to $w=\prg(y)$ for a random $y$. Indistinguishability from Hybrid 1 follows immediately from the pseudorandomness of $\prg$.

\begin{itemize}
    \item $\Hyb_{3}$: Dropping the use of the PRG and checking its image directly.
\end{itemize}
Now the adversary is given $\iO(1^\lambda,1^s,J_y)$ for a random $y$. Observe that since $w=\prg(y)$ and $\prg$ is injective, $J_y$ and $J'_w$ have equivalent functionalities. Therefore, by $\iO$ security, hybrids 2 and 3 are indistinguishable. This completes the proof of indistinguishability.

For the computational unpredictability of $y$, consider an adversary starting in hybrid 3 which outputs $y$ with probability $\epsilon$. Since the indistinguishability of hybrids 2 and 3 did not rely on the randomness of $y$, we can switch to hybrid 2 and still obtain an adversary that outputs $y$ with probability at least $\epsilon-\negl$. Now we observe that the view of the adversary only depends on $w=\prg(y)$, and in the end the adversary produces $y$ with non-negligible probability. Thus, by a straightforward reduction to the one-wayness of $\prg$, we conclude that $\epsilon-\negl$, and hence $\epsilon$ itself, must be negligible.
\end{proof}

\paragraph{Swapping distributions.}
We now move to the next standard technique. Let $\{ D^x_0 \}_x, \{ D^x_1 \}_x$ be two families of distributions over the same domain $\Ys$, which can also be thought of as deterministic functions $D_0(x;r)$, $D_1(x; r)$ that take as input an index $x$ and some random coins $r$. Let $P$ be a program that makes queries to an oracle $O: \Xs \rightarrow \Ys$ for some set $\Xs$. Then we have the following:

\begin{lemma} \label{lem:distswap}
Let $\left( \prf, \punc \right)$ be a $\left( f_{\prf}, \delta_\prf \right)$-secure puncturable PRF and $\iO$ be a $\left( f_{\iO}, \delta_\iO \right)$-secure iO. Let $\Xs$ a finite set and let $D_0 := \{ D_{0, x} \}_{x \in \Xs}$,$D_1 := \{ D_{1, x} \}_{x \in \Xs}$ two ensembles of distributions, such that for every $x \in \Xs$, $D_{0, x}$, $D_{1, x}$ are $\left( f_{D}, \delta_{D} \right)$-indistinguishable. Let $E_0(x) = D_0\left( x ; \prf(k,x) \right)$ and $E_1(x) = D_1\left( x ; \prf\left( \overline{k}, x \right) \right)$. Then for a sufficiently large polynomial $s$, $\iO\left( 1^\lambda, 1^s, P^{E_0} \right)$ and $\iO\left( 1^\lambda, 1^s, P^{E_1} \right)$ are $\left( \min\left( f_{\prf}, f_{\iO}, f_{D} \right), O\left( |\Xs|\cdot \left( \delta_\prf + \delta_\iO + \delta_{D} \right) \right) \right)$-computationally indistinguishable, where $k,\overline{k}\gets\{0,1\}^\lambda$ are uniformly random keys. 
\end{lemma}

\begin{proof}
We assume $\Xs = [N]$ by giving some ordering to the set. We prove security through a sequence of hybrids.

$\Hyb_{i.0}$: The adversary gets $\iO(1^\lambda,1^s,P^{E_{i,0}})$, where $E_{i,0}$ has $k,\overline{k}$ hard-coded and is defined as $E_{i,0}(x)=\begin{cases}D_0(x;\prf(k,x))&\text{ if }x\geq i\\D_1(\prf(x;\overline{k},x))&\text{ if }x< i\end{cases}$.

$\Hyb_{i.1}$: The adversary gets $\iO(1^\lambda,1^s,P^{E_{i,1}})$, where to generate $E_{i,1}$, we compute $k^{i}\gets\punc(k,i),\allowbreak \overline{k}^{i}\gets\punc(\overline{k},i)$, sample $y\gets D_0(i;\prf(k,i))$ and let $E_{i,1}(x)=\begin{cases}D_0(x;\prf(k,x))&\text{ if }x> i\\ y&\text{ if }x=i\\D_1(x;\prf(\overline{k},x))&\text{ if }x< i\end{cases}$.
Observe that by our choice of $y$, $E_{i,1}$ is identical to $E_{i,0}$, and hence the programs $P^{E_{i,0}}$ and $P^{E_{i,1}}$ are equivalent. Thus, by the security of $\iO$, Hybrids $i.0$ and $i.1$ are indistinguishable except with probability $\delta_\iO$.

$\Hyb_{i.2}$: Here, we still obfuscate $P^{E_{i,1}}$, but instead switch to $y\gets D_0(i;r)$ for fresh random coins $r$. Observe that the entire experiment except for $y$ is simulatable using just the punctured key $k^i$, and the only difference for $y$ is that we replace $\prf(k,i)$ with a random string. Thus Hybrids $i.1$ and $i.2$ are indistinguishable except with probability $\delta_\prf$.

$\Hyb_{i.3}$: Now we change to $y\gets D_1(r)$. Hybrids $i.2$ and $i.3$ are indistinguishable except with probability $\epsilon$.

$\Hyb_{i.4}$: Now we change to $y\gets D_1(\prf(\overline{k},i))$. Hybrids $i.3$ and $i.4$ are indistinguishable except with probability $\delta_\prf$.

Next, we observe that $E_{i,1}$, when using $y\gets D_1(i;\prf(\overline{k},i))$, is actually functionally equivalent to $E_{i+1,0}$. Thus, we see that the programs $P^{E_{i,1}}$ and $P^{E_{i+1,0}}$ are functionally equivalent. By $\iO$ security, we therefore have that Hybrid $i.4$ and $(i+1).0$ are indistinguishable except with probability $\delta_\iO$.

The proof then follows by observing that Hybrid $1.0$ corresponds to $\iO(1^\lambda,1^s,P^{D_0(\prf(k,\cdot))})$ and Hybrid $(N+1).0$ corresponds to $\iO(1^\lambda,1^s,P^{D_1(\prf(\overline{k},\cdot))})$

\end{proof}

\subsection{Information-Theoretical Hardness of Hidden Subspace Detection} \label{subsection:unconditional_hidden_subspace}
One of our central objects in this paper are quantumly accessible classical functions that check membership in some secret linear subspace $S \subseteq \bbZ_{2}^{k}$.

\paragraph{Information-Theoretical Subspace Hiding.}
We start with a quantum lower bound for detecting a change between two oracles: One allows access to membership check for some given (known) subspace $S$, and the other allows access to membership check in a random superspace $T$ of $S$. This is an information-theoretical version of the subspace-hiding obfuscation introduced in \cite{EC:Zhandry19b}.

\begin{lemma} \label{lemma:unconditional_subspace_hiding}
Let $k, r, s \in \Nat$ such that $r + s \leq k$ and let $S \subseteq \bbZ_{2}^{k}$ a subspace of dimension $r$. Let $\Ss_{s}$ the uniform distribution over subspaces $T$ of dimension $r + s$ such that $S \subseteq T \subseteq \bbZ_{2}^{k}$. For any subspace $S' \subseteq \bbZ_{2}^{k}$ let $\Oracle_{S'}$ the oracle that checks membership in $S'$ (outputs $1$ if and only if the input is inside $S'$).

Then, for every oracle-aided quantum algorithm $\Adv$ making at most $q$ quantum queries, we have the following indistinguishability over oracle distributions. 
$$
\{
\Oracle_{S}
\}
\approx_{O\left( \frac{ q \cdot s }{ \sqrt{ 2^{k - r - s} } } \right)}
\{
\Oracle_{T}
\;
:
\;
T \gets \Ss_{s} 
\} \enspace .
$$
\end{lemma}

\begin{proof}
    We prove the claim by a hybrid argument, increasing the dimension of the random superspace $T$ by $1$ in each step, until we made an increase of $s$ dimensions. In the first step we consider a matrix $\matB \in \bbZ_{2}^{k \times (k - r)}$, the columns of which form a basis for $S^{\bot}$. Note that the oracle $\Oracle_{S}$ can be described as accepting $\vecX \in \bbZ_{2}^{k}$ iff $\vecX^{T} \cdot \matB = 0^{k - r}$.
    
    Next we sample a uniformly random $\vecA \in \bbZ_{2}^{k - r}$ and consider the oracle $\Oracle_{S_{1}}$ that accepts $\vecX \in \bbZ_{2}^{k}$ iff either $\vecX^{T} \cdot \matB = 0^{k - r}$ or $\vecX^{T} \cdot \matB = \vecA$. Two things can be verified: (1) Due to the randomness of $\vecA$, by standard quantum lower bounds, $\{ \Oracle_{S} \} \approx_{ O\left( \frac{ q }{ \sqrt{ 2^{k - r} } } \right) } \{ \Oracle_{ S_{1} } \; : \; \vecA \gets \bbZ_{2}^{k - r} \}$, and (2) The set $S_{1}$ is a random superspace of $S$ with dimension $r + 1$.

    This proves our claim for $s = 1$. For a general $s$ we can make an $s$-step hybrid argument, where at each step we have $S_{i}$ which is a random $(r + i)$-dimensional superspace of $S$. Overall we get that for a $q$-query algorithm $\Adv$ the distinguishing advantage between $\{ \Oracle_{S} \}$ and $\{ \Oracle_{ S_{1} } \; : \; T \gets \Ss_{s} \}$ is 
    $$
    \sum_{i \in [s]} O\left( \frac{ q }{ \sqrt{ 2^{k - r - (i - 1)} } } \right)
    \leq
    O\left( \frac{ q \cdot s }{ \sqrt{ 2^{k - r - s} } } \right)
    \enspace ,
    $$
    as needed.
\end{proof}

We will also use the following corollary of Lemma \ref{lemma:unconditional_subspace_hiding}. The corollary says that it is still hard to distinguish membership check between $S$ and $T$, also when we duplicate the oracle access $\ell$ times. The corollary follows by a direct simulation reduction.

\begin{corollary} \label{corollary:unconditional_iterated_subspace_hiding}
Let $k, r, s \in \Nat$ such that $r + s \leq k$ and let $S \subseteq \bbZ_{2}^{k}$ a subspace of dimension $r$. Let $\Ss_{s}$ the uniform distribution over subspaces $T$ of dimension $r + s$ such that $S \subseteq T \subseteq \bbZ_{2}^{k}$. For any subspace $S' \subseteq \bbZ_{2}^{k}$ let $\Oracle_{S'}$ the oracle that checks membership in $S'$ (outputs $1$ if and only if the input is inside $S'$).

Then, for every oracle-aided quantum algorithm $\Adv$ making at most $q$ quantum queries, we have the following indistinguishability over oracle distributions. 
$$
\{
\Oracle^{1}_{S}, \cdots, \Oracle^{\ell}_{S}
\}
\approx_{O\left( \frac{ q \cdot \ell \cdot s }{ \sqrt{ 2^{k - r - s} } } \right)}
\{
\Oracle^{1}_{T}, \cdots, \Oracle^{\ell}_{T}
\;
:
\;
T \gets \Ss_{s} 
\} \enspace .
$$
\end{corollary}

An additional corollary which follows by combining the above, with a standard hybrid argument on the first statement \ref{lemma:unconditional_subspace_hiding} is as follows. Note that in the below statement, the first oracle distribution is where each of the $\ell$ oralces samples an i.i.d. superspace $T_{i}$, and the second oracle distribution is where we sample $T$ once, and then duplicate its oracle.

\begin{corollary} \label{corollary:unconditional_subspace_hiding_main_corollary}
Let $k, r, s \in \Nat$ such that $r + s \leq k$ and let $S \subseteq \bbZ_{2}^{k}$ a subspace of dimension $r$. Let $\Ss_{s}$ the uniform distribution over subspaces $T$ of dimension $r + s$ such that $S \subseteq T \subseteq \bbZ_{2}^{k}$. For any subspace $S' \subseteq \bbZ_{2}^{k}$ let $\Oracle_{S'}$ the oracle that checks membership in $S'$ (outputs $1$ if and only if the input is inside $S'$).

Then, for every oracle-aided quantum algorithm $\Adv$ making at most $q$ quantum queries, we have the following indistinguishability over oracle distributions. 
$$
\{
\Oracle_{T_{1}}, \cdots, \Oracle_{T_{\ell}}
\;
:
\;
\forall i \in [\ell] : T_{i} \gets \Ss_{s} 
\}
$$
$$
\approx_{O\left( \frac{ q \cdot \ell \cdot s }{ \sqrt{ 2^{k - r - s} } } \right)}
\{
\Oracle^{1}_{T}, \cdots, \Oracle^{\ell}_{T}
\;
:
\;
T \gets \Ss_{s} 
\}
\enspace .
$$
\end{corollary}

The below is an information theoretical version of our Lemma \ref{lemma:dual_subspace_concentration} and corresponding proof.

\begin{lemma} \label{lemma:unconditional_dual_subspace_anti_concentration}
Let $k, r, s \in \Nat$ such that $r + s \leq k$ and let $S \subseteq \bbZ_{2}^{k}$ a subspace of dimension $r$. Let $\Ss_{s}$ the uniform distribution over subspaces $T$ of dimension $r + s$ such that $S \subseteq T \subseteq \bbZ_{2}^{k}$. For any subspace $S' \subseteq \bbZ_{2}^{k}$ let $\Oracle_{S'}$ the oracle that checks membership in $S'$ (outputs $1$ if and only if the input is inside $S'$).

Assume there is an oracle-aided quantum algorithm $\Adv$ making at most $q$ quantum queries and outputting a vector $\vecU \in \bbZ_{2}^{k}$ at the end of its execution, such that 
$$
\Pr
\left[
\Adv^{\Oracle_{T}} \in \left( T^{\bot} \setminus \{ 0 \} \right)
\;
:
\;
T \gets \Ss_{s}
\right]
\geq
\epsilon
\enspace .
$$
Also, denote $t := k - r - s$, $\ell := \frac{ k\left( t + 1 \right) }{ \epsilon }$ and assume (1) $\frac{t \cdot \frac{1}{\epsilon}}{2^{ s - t }} \leq o(1)$ and (2) $\frac{ q \cdot \ell^2 \cdot s }{ \sqrt{ 2^{t} } } \leq o(1)$. Then, it is necessarily the case that 
$$
\Pr
\left[
\Adv^{\Oracle_{T}} \in \left( S^{\bot} \setminus T^{\bot} \right)
\;
:
\;
T \gets \Ss_{s}
\right]
\geq
\frac{\epsilon}{16 \cdot k \cdot (t + 1)}
\enspace .
$$
\end{lemma}

\begin{proof}
We start with defining the following reduction $\AdvB$, that will use the circuit $\Adv$ as part of its machinery.

\paragraph{The reduction $\AdvB$.}
The input to $\AdvB$ contains $\ell := \frac{ k \cdot \left( t + 1 \right) }{ \epsilon }$ samples of oracles $\left( \Oracle^{(1)}, \cdots, \Oracle^{(\ell)} \right)$, for $t := k - r - s$. Given the $\ell$ oracles, execute $\Adv^{ \Oracle^{(i)} }$ for every $i \in [\ell]$ and obtain $\ell$ vectors $\{ u_{1}, \cdots, u_{\ell} \}$. Then, take only the vectors $\{ v_{1}, \cdots, v_{m} \}$ that are inside $S^{\bot}$, and then compute the dimension of their span, $D := \dim\left( \linspan\left( v_{1}, \cdots, v_{m} \right) \right)$. Note that the number of queries that $\AdvB$ makes is $q \cdot \ell$.

\paragraph{Executing $\AdvB$ on the oracle distribution $\Ds_{1}$.}
Consider the following distribution $\Ds_{1}$: Sample $\ell$ i.i.d superspaces $T_{1}, \cdots, T_{\ell}$, and for each of them, give access to its membership check oracle: $\Oracle_{T_1}, \cdots, \Oracle_{T_\ell}$. Let us see what happens when we execute $\AdvB$ on a sample from the distribution $\mathcal{D}_{1}$.

Consider the $\ell$ vectors $\{ u_{1}, \cdots, u_{\ell} \}$ obtained by executing $\Adv$ on each of the input oracles. Recall that $\ell := \frac{1}{\epsilon} \cdot k \cdot \left( t + 1 \right)$ and consider a partition of the vectors into $t + 1$ consecutive sequences (or buckets), accordingly, each of length $\frac{1}{\epsilon} \cdot k$. In order to show that the probability for the reduction $\AdvB$ to have $D \geq t + 1$ is high, we show that with high probability, in each bucket $j \in [t + 1]$ there is a vector $u_{i}$ that's inside the corresponding dual $T_{i}^{\bot}$, but such that also the intersection between $T_{i}^{\bot}$ and each of the previous $j - 1$ dual subspaces that were hit by $\Adv$, is only the zero vector $0^{k}$. Note that the last condition indeed implies $D \geq t + 1$.

For every $i \in [\ell]$ we define the probability $p_{i}$. We start with defining it for the indices in the first bucket, and then proceed to define it recursively for the rest of the buckets. For indices $i \in [\frac{1}{\epsilon} \cdot k]$ in the first bucket, $p_{i}$ is the probability that given access to $\Oracle_{T_{i}}$, the output of $\Adv$ is $u_{i} \in \left( T_{i}^{\bot} \setminus \{ 0 \} \right)$, and in such case we define the $i$-th execution as successful. We denote by $T_{(1)}$ the first subspace in the first bucket where a successful execution happens (and define $T_{(1)} := \bot$ if no success happened). For any $i$ inside any bucket $j \in \left( [t + 1] \setminus \{ 1 \} \right)$ that is not the first bucket, we define $p_{i}$ as the probability that (1) $u_{i} \in \left( T_{i}^{\bot} \setminus \{ 0 \} \right)$ and also (2) the intersection between $T_{i}^{\bot}$ and each of the dual subspaces of the previous winning subspaces $T_{(1)}, \cdots, T_{(j - 1)}$, is only $\{ 0^k \}$. That is, $p_{i}$ is the probability that the output of the adversary hits the dual subspace, and also the dual does not have a non-trivial intersection with any of the previous successful duals. Similarly to the first bucket, we denote by $T_{(j)}$ the first subspace in bucket $j$ with a successful execution.

We prove that with high probability, all $t + 1$ buckets have at least one successful execution. To see this, we define the following probability $p'$ which we show lower bounds $p_{i}$, and is defined as follows. First, let $\overline{T}_{1}$, $\cdots$, $\overline{T}_{t}$ any $t$ subspaces, each of dimension $r + s$, thus the duals $\overline{T}_{1}^{\bot}$, $\cdots$, $\overline{T}_{t}^{\bot}$ are such that each has dimension $t$. $p'_{ \left( \overline{T}_{1}, \cdots, \overline{T}_{t} \right) }$ is the probability that (1) when sampling $T^{\bot}$, the intersection of $T^{\bot}$ with each of the $t$ dual subspaces $\overline{T}_{1}^{\bot}$, $\cdots$, $\overline{T}_{t}^{\bot}$ was only the zero vector, and also (2) the output of the adversary $\Adv$ was inside $T^{\bot}$. $p'$ is defined is the minimal probability taken over all possible choices of $t$ subspaces $\overline{T}_{1}$, $\cdots$, $\overline{T}_{t}$. After one verifies that indeed for every $i$ we have $p' \leq p_{i}$, it is sufficient to lower bound $p'$.

\paragraph{Lower bound for the probability $p'$.}
The probability $p'$ is for an event that's defined as the logical AND of two events, and as usual, equals the product between the probability $p'_{0}$ of the first event (the trivial intersection between the subspaces), times the conditional probability $p'_{1}$ of the second event (that $\Adv$ hits a non-zero vector in the dual $T^{\bot}$), conditioned on the first event.

First we lower bound the probability $p'_{0}$ by upper bounding the complement probability, that is, we show that the probability for a non-trivial intersection is small. Consider the random process of choosing a basis for a subspace $T$ and note that it is equivalent to choosing a basis for the dual $T^{\bot}$. The process of choosing a basis for the dual has $t$ steps, and in each step we choose a random vector in $S^{\bot}$ that's outside the span we aggregated so far. Given a dual subspace $\overline{T}^{\bot}$ of dimension $t$, what is the probability for the two subspaces to have only a trivial intersection? It is exactly the sum over $z \in [t]$ (which we think of as the steps for sampling $T^{\bot}$) of the following event: In the $t$-step process of choosing a basis for $T^{\bot}$, index $z$ was the first to cause the subspaces to have a non-zero intersection. Recall that for each $z \in [t]$, the probability that $z$ was such first index to cause an intersection, equals the probability that the $z$-th sampled basis vector for $T^{\bot}$ is a vector that's inside the unified span of $\overline{T}^{\bot}$ and the aggregated span of $T^{\bot}$ so far, after $z - 1$ samples. This amounts to the probability
$$
\sum_{z \in [t]} \frac{ |\overline{T}^{\bot}| \cdot 2^{z - 1} }{ |S^{\bot}| }
=
\sum_{z \in [t]} \frac{ 2^{t} \cdot 2^{z - 1} }{ 2^{k - r} }
=
2^{ -s } \cdot \sum_{z \in \{ 0, 1, \cdots, t - 1 \}} 2^{z}
$$
$$
=
2^{ -s } \cdot \left( 2^{t} - 1 \right)
< 
2^{t - s}
\enspace .
$$
Since the above is an upper bound on the probability for a non-trivial intersection between $T^{\bot}$ and one more single subspace, by union bound, the probability for $T^{\bot}$ to have a non-trivial intersection with at least one of the $t$ subspaces $\overline{T}_{1}^{\bot}$, $\cdots$, $\overline{T}_{t}^{\bot}$ is upper bounded by $t \cdot 2^{t - s}$. This means that $p'_{0} \geq 1 - t \cdot 2^{t - s}$.

The lower bound for the conditional probability $p'_{1}$ is now quite easy: Note that since $\Pr\left[ A | B \right] \geq \Pr\left[ A \right] - \Pr\left[ \lnot B \right]$, letting $A$ the event that $\Adv$ outputs a vector in the dual $T^{\bot}$ and $B$ the event that $T^{\bot}$ has only a trivial intersection with all other $t$ subspaces, we get $p'_{1} \geq \epsilon - t \cdot 2^{t - s}$. By our assumption that $\frac{t \cdot \frac{1}{\epsilon}}{2^{ s - t }} \leq o(1)$, we have $p'_{1} \geq \frac{\epsilon}{2}$. Overall we got $p' := p'_{0} \cdot p'_{1} \geq \left( 1 - t \cdot 2^{t - s} \right) \cdot \frac{\epsilon}{2} > \frac{\epsilon}{4}$.

Finally, to see why we get an overall high probability for $D \geq t + 1$ on a sample from $\Ds_{1}$, observe the following. In each bucket there are $\frac{k}{\epsilon}$ attempts, each succeeds with probability at least $\frac{\epsilon}{4}$ and thus the overall success probability in a bucket is $\geq 1 - e^{-\Omega(k)}$. Accordingly, the probability to succeed at least once in each of the $t + 1$ buckets (and thus to satisfy $D \geq t + 1$) is $\geq 1 - (t + 1)\cdot e^{-\Omega(k)}$, by considering the complement probability and applying union bound. Overall the probability for $D \geq t + 1$ is thus $\geq 1 - e^{ -\Omega(k) }$.

\paragraph{Executing $\AdvB$ on the distribution $\Ds_{2}$.}
Consider a different distribution $\Ds_{2}$: Sample $T$ once, then allow an $\ell$-oracle access to it, $\Oracle^{(1)}_{T}, \cdots, \Oracle^{(\ell)}_{T}$. Note that $\AdvB$ is a $q \cdot \ell$-query algorithm and thus by Corollary \ref{corollary:unconditional_subspace_hiding_main_corollary} there is the following indistinguishability of oracles with respect to $\AdvB$:
$$
\mathcal{D}_{1}
\approx_{O\left( \frac{ q \cdot \ell^2 \cdot s }{ \sqrt{ 2^{k - r - s} } } \right)}
\mathcal{D}_{2} \enspace .
$$
Since given a sample oracle from $\Ds_{1}$, the algorithm $\AdvB$ outputs $D \geq t + 1$ with probability $\geq 1 - e^{ -\Omega(k) }$, by the above indistinguishability, whenever we execute $\AdvB$ on a sample from $\Ds_{2}$, then with probability at least $\geq 1 - e^{ -\Omega(k) } - O\left( \frac{ q \cdot \ell^2 \cdot s }{ \sqrt{ 2^{k - r - s} } } \right) \geq 1 - O\left( \frac{ q \cdot \ell^2 \cdot s }{ \sqrt{ 2^{k - r - s} } } \right)$ we have $D \geq t + 1$. By our assumption in the Lemma that $O\left( \frac{ q \cdot \ell^2 \cdot s }{ \sqrt{ 2^{k - r - s} } } \right) \leq \frac{1}{2}$, with probability at least $\frac{1}{2}$ we have $D \geq t + 1$ given a sample from $\Ds_{2}$.
By an averaging argument, it follows that with probability at least $\frac{1}{2} \cdot \frac{1}{2} = \frac{1}{4}$ over sampling the superspace $T$, the probability $p_{T}$ for the event where $D \geq t + 1$, is at least $\frac{1}{2} \cdot \frac{1}{2} = \frac{1}{4}$. Let us call this set of superspaces $T$, "the good set" of samples, which by definition has fraction at least $\frac{1}{4}$. Recall two facts: (1) the dimension of $T^{\bot}$ is $t$, (2) The dimension $D$ aggregates vectors inside $S^{\bot}$. The two facts together imply that in the event $D \geq t + 1$, it is necessarily the case that there exists an execution index $i \in [\ell]$ in the reduction $\AdvB$ where $\Adv$ outputs a vector in $\left( S^{\bot} \setminus T^{\bot} \right)$, given membership check in $T$.

For every $T$ inside the good set we thus know that with probability $\frac{1}{4}$, one of the output vectors of $\Adv$ will be in $\left( S^{\bot} \setminus T^{\bot} \right)$. Since these are $\ell$ i.i.d. executions of $\Adv$, by union bound, for every $T$ inside the good set, when we prepare an oracle access to $T$ and execute $\Adv$, we will get $\Adv^{\Oracle_{T}} \in \left( S^{\bot} \setminus T^{\bot} \right)$ with probability $\frac{1}{4 \cdot \ell}$. We deduce that for a uniformly random $T$ which we then prepare oracle access to, the probability for $\Adv^{\Oracle_{T}} \in \left( S^{\bot} \setminus T^{\bot} \right)$ is at least the probability for this event and also that $T$ is inside the good set, which in turn is at least
$$
\frac{1}{4} \cdot \frac{1}{4 \cdot \ell}
=
\frac{1}{16 \cdot \ell}
:=
\frac{\epsilon}{16 \cdot k \cdot (t + 1)}
\enspace ,
$$
which finishes our proof.
\end{proof}

\subsection{Cryptographic Hardness of Hidden Subspace Detection}
In this subsection we prove the cryptographic analogues of the information theoretical lower bounds from the previous section.

We start with stating the subspace-hiding obfuscation property of indistinguishability obfuscators from \cite{EC:Zhandry19b}.

\begin{lemma} \label{lemma:subspace_hiding}
Let $k, r, s \in \Nat$ such that $r + s \leq k$ and let $S \subseteq \bbZ_{2}^{k}$ a subspace of dimension $r$. Let $\Ss_{s}$ the uniform distribution over subspaces $T$ of dimension $r + s$ such that $S \subseteq T \subseteq \bbZ_{2}^{k}$. For any subspace $S' \subseteq \bbZ_{2}^{k}$ let $C_{S'}$ some canonical classical circuit that checks membership in $S'$, say be Gaussian elimination. Let $\iO$ an indistinguishability obfuscation scheme that is $\left( f(\secp), \epsilon(\secp) \right)$-secure, and assume that $\left( f(\secp), \epsilon(\secp) \right)$-secure injective one-way functions exist. 

Then, for every security parameter $\secp$ such that $\secp \leq k - r - s$ and sufficiently large $p := p(\secp)$ polynomial in the security parameter, we have the following indistinguishability,
$$
\{
\Obf_{S}
\;
:
\;
\Obf_{S} \gets \iO\left( 1^{\secp}, 1^{p}, C_{S} \right)
\}
\approx_{\left( f(\secp) - \poly(\secp), \; s \cdot \epsilon(\secp) \right)}
$$
$$
\{
\Obf_{T}
\;
:
\;
T \gets \Ss_{s}, 
\Obf_{T} \gets \iO\left( 1^{\secp}, 1^{p}, C_{T} \right)
\} \enspace .
$$
\end{lemma}

The following generalization of Lemma \ref{lemma:subspace_hiding} is derived by a random self-reducibility argument and is formally proved by using an additional layer of indistinguishability obfuscation.

\begin{lemma} \label{lemma:iterated_subspace_hiding}
Let $k, r, s \in \Nat$ such that $r + s \leq k$ and let $S \subseteq \bbZ_{2}^{k}$ a subspace of dimension $r$. Let $\Ss_{s}$ the uniform distribution over subspaces $T$ of dimension $r + s$ such that $S \subseteq T \subseteq \bbZ_{2}^{k}$. For any subspace $S' \subseteq \bbZ_{2}^{k}$ let $C_{S'}$ some canonical classical circuit that checks membership in $S'$, say be Gaussian elimination. Let $\iO$ an indistinguishability obfuscation scheme that is $\left( f(\secp), \epsilon(\secp) \right)$-secure, and assume that $\left( f(\secp), \epsilon(\secp) \right)$-secure injective one-way functions exist.  

Then, for every security parameter $\secp$ such that $\secp \leq k - r - s$ and sufficiently large $p := p(\secp)$ polynomial in the security parameter, we have the following indistinguishability,
$$
\{
\Obf^{1}_{S}, \cdots , \Obf^{\ell}_{S}
\;
:
\;
\forall i \in [\ell]
,
\Obf^{i}_{S} \gets \iO\left( 1^{\secp}, 1^{p}, C_{S} \right)
\}
\approx_{\left( f(\secp) - \ell \cdot \poly(\secp), \; \left( 2 \cdot \ell + s \right) \cdot \epsilon(\secp) \right)}
$$
$$
\{
\Obf^{1}_{T}, \cdots , \Obf^{\ell}_{T}
\;
:
\;
T \gets \Ss_{s}, 
\forall i \in [\ell]
,
\Obf^{i}_{T} \gets \iO\left( 1^{\secp}, 1^{p}, C_{T} \right)
\} \enspace .
$$
\end{lemma}

\begin{proof}
    We first observe that as long as the circuit size parameter $1^p$ is sufficiently large, and specifically, larger than the size of an obfuscated version of the plain circuit $C$, then it is indistinguishable to tell whether said circuit is obfuscated under one or two layers of obfuscation. More precisely, due to the perfect correctness of obfuscation, the circuit $C$ and an obfuscated $\Obf_{C} \gets \iO\left( 1^{\secp}, 1^{p}, C \right)$ have the same functionality (for every sample out of the distribution of obfuscated versions of $C$), and thus, as long as $p \geq |\Obf_{C}|$, then the distributions $\Ds_{0} := \{ \Obf_{C} \gets \iO\left( 1^{\secp}, 1^{p}, C \right) \}$ and $\Ds_{1} := \{ \Obf^{2}_{C} \gets \iO\left( 1^{\secp}, 1^{p}, \Ds_{0} \right) \}$ are $\left( f(\secp), \epsilon(\secp) \right)$-indistinguishable.

    An implication of the above is that for a sufficiently large parameter $p$, the distribution
    $$
    \{
    \Obf^{1}_{S}, \cdots , \Obf^{\ell}_{S}
    \;
    :
    \;
    \forall i \in [\ell]
    ,
    \Obf^{i}_{S} \gets \iO\left( 1^{\secp}, 1^{p}, C_{S} \right)
    \}
    $$
    is $\left( f(\secp), \ell \cdot \epsilon(\secp) \right)$-indistinguishable from a distribution where $C_{S}$ is swapped with an obfuscation of it (let us denote this modified distribution with $\Ds_{S}$), and the distribution 
    $$
    \{
    \Obf^{1}_{T}, \cdots , \Obf^{\ell}_{T}
    \;
    :
    \;
    T \gets \mathcal{S}_{S}, 
    \forall i \in [\ell]
    ,
    \Obf^{i}_{T} \gets \iO\left( 1^{\secp}, 1^{p}, C_{T} \right)
    \}
    $$
    is $\left( f(\secp), \ell \cdot \epsilon(\secp) \right)$-indistinguishable from a distribution where $C_{T}$ is swapped with an obfuscation of it (let us denote this modified distribution with $\Ds_{T}$). To complete our proof we will show that $\Ds_{S}$ and $\Ds_{T}$ are appropriately indistinguishable, and will get our proof by transitivity of computational distance. 

    The indistinguishability between $\Ds_{S}$ and $\Ds_{T}$ follows almost readily from the the subspace hiding Lemma \ref{lemma:subspace_hiding}: One can consider the reduction $\AdvB$ that gets a sample $\Obf$ which is either from $\Ds_{0} := \{ \Obf_{S} \gets \iO\left( 1^{\secp}, 1^{p}, C_{S} \right) \}$ or from $\Ds_{1} := \{ \Obf_{T} \gets \iO\left( 1^{\secp}, 1^{p}, C_{T} \right) \}$ for an appropriately random superspace $T$ of $S$. $\AdvB$ then generates $\ell$ i.i.d obfuscations $\{ \Obf^{(1)}, \cdots, \Obf^{(\ell)} \}$ of the (already obfuscated) circuit $\Obf$ and executes $\Adv$ on the $\ell$ obfuscations. One can see that when the input sample $\Obf$ for $\AdvB$ came from $\Ds_{0}$ then the output sample of the reduction comes from the distribution $\Ds_{S}$, and when the input sample $\Obf$ for $\AdvB$ came from $\Ds_{1}$ then the output sample of the reduction comes from the distribution $\Ds_{T}$. Since the reduction executes in complexity $\ell \cdot \poly(\secp)$, this means that 
    $$
    \Ds_{S}
    \approx_{\left( f(\secp) - \ell \cdot \poly(\secp), \; s \cdot \epsilon(\secp) \right)}
    \Ds_{T}
    \enspace .
    $$
    To conclude, by transitivity of computational indistinguishability, we get that the two distributions in our Lemma's statement are $\left( f(\secp) - \ell \cdot \poly(\secp), \; \left( 2 \cdot \ell + s \right) \cdot \epsilon(\secp) \right)$-indistinguishable, as needed.
\end{proof}

A corollary which follows by combining the above Lemma \ref{lemma:iterated_subspace_hiding}, with a standard hybrid argument on the first lemma \ref{lemma:subspace_hiding} is as follows.

\begin{corollary} \label{corollary:subspace_hiding_main_corollary}
Let $k, r, s \in \Nat$ such that $r + s \leq k$ and let $S \subseteq \bbZ_{2}^{k}$ a subspace of dimension $r$. Let $\Ss_{s}$ the uniform distribution over subspaces $T$ of dimension $r + s$ such that $S \subseteq T \subseteq \bbZ_{2}^{k}$. For any subspace $S' \subseteq \bbZ_{2}^{k}$ let $C_{S'}$ some canonical classical circuit that checks membership in $S'$, say be Gaussian elimination. Let $\iO$ an indistinguishability obfuscation scheme that is $\left( f(\secp), \epsilon(\secp) \right)$-secure, and assume that $\left( f(\secp), \epsilon(\secp) \right)$-secure injective one-way functions exist. 

Then, for every security parameter $\secp$ such that $\secp \leq k - r - s$ and sufficiently large $p := p(\secp)$ polynomial in the security parameter, we have the following indistinguishability,
$$
\{
\Obf_{T_{1}}, \cdots , \Obf_{T_{\ell}}
\;
:
\;
T_{i} \gets \Ss_{s} \forall i \in [\ell]
,
\Obf_{T_{i}} \gets \iO\left( 1^{\secp}, 1^{p}, C_{T_{i}} \right)
\}
$$
$$
\approx_{\left( f(\secp) - \ell \cdot \poly(\secp), \; \left( 2 \cdot \ell \cdot s \right) \cdot \epsilon(\secp) \right)}
\{
\Obf^{1}_{T}, \cdots , \Obf^{\ell}_{T}
\;
:
\;
T \gets \Ss_{s}, 
\forall i \in [\ell]
,
\Obf^{i}_{T} \gets \iO\left( 1^{\secp}, 1^{p}, C_{T} \right)
\}
\enspace .
$$
\end{corollary}

\paragraph{Improved Results on Hardness of Concentration in Dual of Obfuscated Subspace.}
As part of this work we strengthen the main technical lemma (Lemma 5.1) from \cite{C:Shmueli22}. Roughly speaking, in \cite{C:Shmueli22} it is shown that an adversary $\Adv$ that gets an obfuscation $\Obf_{T}$ for a random superspace $T$ of $S$, and manages to output a (non-zero) vector in the dual $T^{\bot} \setminus \{ \mathbf{0} \}$ with probability $\epsilon$, has to sometimes output vectors in $S^{\bot} \setminus T^{\bot}$, i.e. with probability at least $\Omega\left( \epsilon^{2} \right) / \poly(k)$. Below we strengthen the probability to $\Omega\left( \epsilon \right) / \poly(k)$.

\begin{lemma} [IO Dual Subspace Anti-Concentration] \label{lemma:dual_subspace_concentration}
Let $k, r, s \in \Nat$ such that $r + s \leq k$ and let $S \subseteq \bbZ_{2}^{k}$ a subspace of dimension $r$. Let $\Ss_{s}$ the uniform distribution over subspaces $T$ of dimension $r + s$ such that $S \subseteq T \subseteq \bbZ_{2}^{k}$. For any subspace $S' \subseteq \bbZ_{2}^{k}$ let $C_{S'}$ some canonical classical circuit that checks membership in $S'$, say be Gaussian elimination. Let $\iO$ an indistinguishability obfuscation scheme that is $\left( f(\secp), \frac{1}{f(\secp)} \right)$-secure, and assume that $\left( f(\secp), \frac{1}{f(\secp)} \right)$-secure injective one-way functions exist. 

Let $\secp \in \Nat$ the security parameter such that $\secp \leq k - r - s$ and let $p := p(\secp)$ a sufficiently large polynomial in the security parameter. Denote by $\Oracle_{ \secp, p, s }$ the distribution over obfuscated circuits that samples $T \gets \Ss_{s}$ and then $\Obf_{T} \gets \iO\left( 1^{\secp}, 1^{p}, C_{T} \right)$.

Assume there is a quantum algorithm $\Adv$ of complexity $T_{\Adv}$ such that,
$$
\Pr
\left[
\Adv(\Obf_{T}) \in \left( T^{\bot} \setminus \{ 0 \} \right)
\;
:
\;
\Obf_{T} \gets \Oracle_{ \secp, p, s }
\right]
\geq \epsilon
\enspace .
$$

Also, denote $t := k - r - s$, $\ell := \frac{ k\left( t + 1 \right) }{ \epsilon }$ and assume (1) $\frac{t \cdot \frac{1}{\epsilon}}{2^{ s - t }} \leq o(1)$ and (2) $\frac{ \ell \cdot \left( k^3 + \poly(\secp) + T_{\Adv} \right) }{f(\secp)} \leq o(1)$.

Then, it is necessarily the case that 
$$
\Pr
\left[
\Adv\left( \Obf_{T} \right) \in \left( S^{\bot} \setminus T^{\bot} \right)
\;
:
\;
\Obf_{T} \gets \Oracle_{ \secp, p, s }
\right]
\geq
\frac{\epsilon}{16 \cdot k \cdot (t + 1)}
\enspace .
$$
\end{lemma}

\begin{proof}
We start with defining the following reduction $\AdvB$, that will use the circuit $\Adv$ as part of its machinery.

\paragraph{The reduction $\AdvB$.}
The input to $\AdvB$ contains $\ell := \frac{ k \cdot \left( t + 1 \right) }{ \epsilon }$ samples of obfuscations $\left( \Obf^{(1)}, \cdots, \Obf^{(\ell)} \right)$, for $t := k - r - s$. Given the $\ell$ obfuscations, execute $\Adv\left( \Obf^{(i)} \right)$ for every $i \in [\ell]$ and obtain $\ell$ vectors $\{ u_{1}, \cdots, u_{\ell} \}$. Then, take only the vectors $\{ v_{1}, \cdots, v_{m} \}$ that are inside $S^{\bot}$, and then compute the dimension of their span, $D := \dim\left( \linspan\left( v_{1}, \cdots, v_{m} \right) \right)$. 
Note that the running time of $\AdvB$ is $\ell \cdot T_{\Adv} + \ell \cdot k^{3}$, where $\ell \cdot T_{\Adv}$ is for producing the $\ell$ outputs of $\Adv$ and $\ell \cdot k^{3}$ is for (naively) executing Gaussian elimination $\ell$ times, to repeatedly check whether the new vector $v_{i}$ adds a dimension i.e., whether it is outside of the span $\linspan\left( v_{1}, \cdots, v_{i - 1} \right)$ of the previous vectors.

\paragraph{Executing $\AdvB$ on the distribution $\Ds_{1}$.}
Consider the following distribution $\Ds_{1}$: Sample $\ell$ i.i.d superspaces $T_{1}, \cdots, T_{\ell}$, and for each of them, send an obfuscation of it: $\Obf_{T_1}, \cdots, \Obf_{T_\ell}$. Let us see what happens when we execute $\AdvB$ on a sample from the distribution $\mathcal{D}_{1}$.

Consider the $\ell$ vectors $\{ u_{1}, \cdots, u_{\ell} \}$ obtained by executing $\Adv$ on each of the input obfuscations. Recall that $\ell := \frac{1}{\epsilon} \cdot k \cdot \left( t + 1 \right)$ and consider a partition of the vectors into $t + 1$ consecutive sequences (or buckets), accordingly, each of length $\frac{1}{\epsilon} \cdot k$. In order to show that the probability for the reduction $\AdvB$ to have $D \geq t + 1$ is high, we show that with high probability, in each bucket $j \in [t + 1]$ there is a vector $u_{i}$ that's inside the corresponding dual $T_{i}^{\bot}$, but such that also the intersection between $T_{i}^{\bot}$ and each of the previous $j - 1$ dual subspaces that were hit by $\Adv$, is only the zero vector $0^{k}$. Note that the last condition indeed implies $D \geq t + 1$.

For every $i \in [\ell]$ we define the probability $p_{i}$. We start with defining it for the indices in the first bucket, and then proceed to define it recursively for the rest of the buckets. For indices $i \in [\frac{1}{\epsilon} \cdot k]$ in the first bucket, $p_{i}$ is the probability that given $\Obf_{T_{i}}$, the output of $\Adv$ is $u_{i} \in \left( T_{i}^{\bot} \setminus \{ 0 \} \right)$, and in such case we define the $i$-th execution as successful. We denote by $T_{(1)}$ the first subspace in the first bucket where a successful execution happens (and define $T_{(1)} := \bot$ if no success happened). For any $i$ inside any bucket $j \in \left( [t + 1] \setminus \{ 1 \} \right)$ that is not the first bucket, we define $p_{i}$ as the probability that (1) $u_{i} \in \left( T_{i}^{\bot} \setminus \{ 0 \} \right)$ and also (2) the intersection between $T_{i}^{\bot}$ and each of the dual subspaces of the previous winning subspaces $T_{(1)}, \cdots, T_{(j - 1)}$, is only $\{ 0^k \}$. That is, $p_{i}$ is the probability that the output of the adversary hits the dual subspace, and also the dual does not have a non-trivial intersection with any of the previous successful duals. Similarly to the first bucket, we denote by $T_{(j)}$ the first subspace in bucket $j$ with a successful execution.

We prove that with high probability, all $t + 1$ buckets have at least one successful execution. To see this, we define the following probability $p'$ which we show lower bounds $p_{i}$, and is defined as follows. First, let $\overline{T}_{1}$, $\cdots$, $\overline{T}_{t}$ any $t$ subspaces, each of dimension $r + s$, thus the duals $\overline{T}_{1}^{\bot}$, $\cdots$, $\overline{T}_{t}^{\bot}$ are such that each has dimension $t$. $p'_{ \left( \overline{T}_{1}, \cdots, \overline{T}_{t} \right) }$ is the probability that (1) when sampling $T^{\bot}$, the intersection of $T^{\bot}$ with each of the $t$ dual subspaces $\overline{T}_{1}^{\bot}$, $\cdots$, $\overline{T}_{t}^{\bot}$ was only the zero vector, and also (2) the output of the adversary $\Adv$ was inside $T^{\bot}$. $p'$ is defined is the minimal probability taken over all possible choices of $t$ subspaces $\overline{T}_{1}$, $\cdots$, $\overline{T}_{t}$. After one verifies that indeed for every $i$ we have $p' \leq p_{i}$, it is sufficient to lower bound $p'$.

\paragraph{Lower bound for the probability $p'$.}
The probability $p'$ is for an event that's defined as the logical AND of two events, and as usual, equals the product between the probability $p'_{0}$ of the first event (the trivial intersection between the subspaces), times the conditional probability $p'_{1}$ of the second event (that $\Adv$ hits a non-zero vector in the dual $T^{\bot}$), conditioned on the first event.

First we lower bound the probability $p'_{0}$ by upper bounding the complement probability, that is, we show that the probability for a non-trivial intersection is small. Consider the random process of choosing a basis for a subspace $T$ and note that it is equivalent to choosing a basis for the dual $T^{\bot}$. The process of choosing a basis for the dual has $t$ steps, and in each step we choose a random vector in $S^{\bot}$ that's outside the span we aggregated so far. Given a dual subspace $\overline{T}^{\bot}$ of dimension $t$, what is the probability for the two subspaces to have only a trivial intersection? It is exactly the sum over $z \in [t]$ (which we think of as the steps for sampling $T^{\bot}$) of the following event: In the $t$-step process of choosing a basis for $T^{\bot}$, index $z$ was the first to cause the subspaces to have a non-zero intersection. Recall that for each $z \in [t]$, the probability that $z$ was such first index to cause an intersection, equals the probability that the $z$-th sampled basis vector for $T^{\bot}$ is a vector that's inside the unified span of $\overline{T}^{\bot}$ and the aggregated span of $T^{\bot}$ so far, after $z - 1$ samples. This amounts to the probability
$$
\sum_{z \in [t]} \frac{ |\overline{T}^{\bot}| \cdot 2^{z - 1} }{ |S^{\bot}| }
=
\sum_{z \in [t]} \frac{ 2^{t} \cdot 2^{z - 1} }{ 2^{k - r} }
=
2^{ -s } \cdot \sum_{z \in \{ 0, 1, \cdots, t - 1 \}} 2^{z}
$$
$$
=
2^{ -s } \cdot \left( 2^{t} - 1 \right)
< 
2^{t - s}
\enspace .
$$
Since the above is an upper bound on the probability for a non-trivial intersection between $T^{\bot}$ and one more single subspace, by union bound, the probability for $T^{\bot}$ to have a non-trivial intersection with at least one of the $t$ subspaces $\overline{T}_{1}^{\bot}$, $\cdots$, $\overline{T}_{t}^{\bot}$ is upper bounded by $t \cdot 2^{t - s}$. This means that $p'_{0} \geq 1 - t \cdot 2^{t - s}$.

The lower bound for the conditional probability $p'_{1}$ is now quite easy: Note that since $\Pr\left[ A | B \right] \geq \Pr\left[ A \right] - \Pr\left[ \lnot B \right]$, letting $A$ the event that $\Adv$ outputs a vector in the dual $T^{\bot}$ and $B$ the event that $T^{\bot}$ has only a trivial intersection with all other $t$ subspaces, we get $p'_{1} \geq \epsilon - t \cdot 2^{t - s}$. By our assumption that $\frac{t \cdot \frac{1}{\epsilon}}{2^{ s - t }} \leq o(1)$, we have $p'_{1} \geq \frac{\epsilon}{2}$. Overall we got $p' := p'_{0} \cdot p'_{1} \geq \left( 1 - t \cdot 2^{t - s} \right) \cdot \frac{\epsilon}{2} > \frac{\epsilon}{4}$.

Finally, to see why we get an overall high probability for $D \geq t + 1$ on a sample from $\Ds_{1}$, observe the following. In each bucket there are $\frac{k}{\epsilon}$ attempts, each succeeds with probability at least $\frac{\epsilon}{4}$ and thus the overall success probability in a bucket is $\geq 1 - e^{-\Omega(k)}$. Accordingly, the probability to succeed at least once in each of the $t + 1$ buckets (and thus to satisfy $D \geq t + 1$) is $\geq 1 - (t + 1)\cdot e^{-\Omega(k)}$, by considering the complement probability and applying union bound. Overall the probability for $D \geq t + 1$ is thus $\geq 1 - e^{ -\Omega(k) }$.

\paragraph{Executing $\AdvB$ on the distribution $\Ds_{2}$.}
Consider a different distribution $\Ds_{2}$: Sample $T$ once, then sample $\ell$ i.i.d. obfuscations of the same circuit $C_{T}$, denoted $\Obf^{(1)}_{T}, \cdots, \Obf^{(\ell)}_{T}$. By Corollary \ref{corollary:subspace_hiding_main_corollary}, 
$$
\mathcal{D}_{1}
\approx_{\left( f(\secp) - \ell \cdot \poly(\secp) , \; \frac{2 \cdot s \cdot \ell}{f(\secp)} \right)}
\mathcal{D}_{2} \enspace .
$$
Recall that the running time of $\AdvB$ is $\ell \cdot T_{\Adv} + \ell \cdot k^{3}$ and by our Lemma's assumptions, the complexity of $\AdvB$ is $\leq f(\secp) - \ell \cdot \poly(\secp)$. Since given a sample oracle from $\Ds_{1}$, the algorithm $\AdvB$ outputs $D \geq t + 1$ with probability $\geq 1 - e^{ -\Omega(k) }$, by the above indistinguishability, whenever we execute $\AdvB$ on a sample from $\Ds_{2}$, then with probability at least $\geq 1 - e^{ -\Omega(k) } - \frac{2 \cdot s \cdot \ell}{f(\secp)} \geq 1 - \frac{4 \cdot s \cdot \ell}{f(\secp)}$ we have $D \geq t + 1$. By our assumption in the Lemma that $O\left( \frac{ s \cdot \ell }{f(\secp)} \right) \leq \frac{1}{2}$, with probability at least $\frac{1}{2}$ we have $D \geq t + 1$ given a sample from $\Ds_{2}$.
By an averaging argument, it follows that with probability at least $\frac{1}{2} \cdot \frac{1}{2} = \frac{1}{4}$ over sampling the superspace $T$, the probability $p_{T}$ for the event where $D \geq t + 1$, is at least $\frac{1}{2} \cdot \frac{1}{2} = \frac{1}{4}$. Let us call this set of superspaces $T$, "the good set" of samples, which by definition has fraction at least $\frac{1}{4}$. Recall two facts: (1) the dimension of $T^{\bot}$ is $t$, (2) The dimension $D$ aggregates vectors inside $S^{\bot}$. The two facts together imply that in the event $D \geq t + 1$, it is necessarily the case that there exists an execution index $i \in [\ell]$ in the reduction $\AdvB$ where $\Adv$ outputs a vector in $\left( S^{\bot} \setminus T^{\bot} \right)$.

For every $T$ inside the good set we thus know that with probability $\frac{1}{4}$, one of the output vectors of $\Adv$ will be in $\left( S^{\bot} \setminus T^{\bot} \right)$. Since these are $\ell$ i.i.d. executions of $\Adv$, by union bound, for every $T$ inside the good set, when we prepare an obfuscation $\Obf_{T}$ of $T$ and execute $\Adv$, we will get a vector in $\left( S^{\bot} \setminus T^{\bot} \right)$ with probability $\frac{1}{4 \cdot \ell}$. We deduce that for a uniformly random $T$, the probability for $\Adv\left( \Obf_{T} \right) \in \left( S^{\bot} \setminus T^{\bot} \right)$ is at least the probability for $\Adv\left( \Obf_{T} \right) \in \left( S^{\bot} \setminus T^{\bot} \right)$ intersecting with the event that $T$ is inside the good set, which in turn is at least
$$
\frac{1}{4} \cdot \frac{1}{4 \cdot \ell}
=
\frac{1}{16 \cdot \ell}
:=
\frac{\epsilon}{16 \cdot k \cdot (t + 1)}
\enspace ,
$$
which finishes our proof.

\end{proof}

\section{One-Shot Signatures Relative to a Classical Oracle} \label{sec:oss_oracle}
In this section we present our construction of non-collapsing collision-resistant hash functions (which imply OSS) with respect to a classical oracle and proof of security. We first describe our scheme in \ref{constr:main}.

\begin{construction} \label{constr:main}
Let $\secp \in \Nat$ the statistical security parameter. Define $s := 16 \cdot \secp$ and let $n, r, k \in \Nat$ such that $r := s \cdot (\secp - 1)$, $n := r + \frac{3}{2} \cdot s$, $k := n$.

Let $\Pi: \{ 0, 1 \}^n \rightarrow \{ 0, 1 \}^n$ be a random permutation and let $F : \{ 0, 1 \}^{r} \rightarrow \{ 0, 1 \}^{k \cdot (n - r + 1)}$ a random function. Let $H(x)$ denote the first $r$ output bits of $\Pi(x)$, and $J(x)$ denote the remaining $n - r$ bits, which are interpreted as a vector in $\Z_{2}^{n - r}$. For each $y \in \{ 0, 1 \}^r$, let $\matA(y) \in \Z_2^{k \times (n-r)}$ be a random matrix with full column-rank, and $\vecB(y) \in \Z_2^k$ be uniformly random, both are generated by the output randomness of $F(y)$.

Then, we let $\Ps : \{ 0, 1 \}^n \rightarrow \left( \{ 0, 1 \}^r \times \Z_2^k \right)$, $\Ps^{-1} : \left( \{ 0, 1 \}^r \times \Z_2^k \right) \rightarrow  \{ 0, 1 \}^n$, $\Ds: \left( \{0,1\}^r \times \Z_2^k \right) \rightarrow \{ 0, 1 \}$ be the following oracles:
\begin{align*}
    \Ps(x) &= \left( \; y \; , \; \matA(y) \cdot J(x) + \vecB(y) \; \right) \text{ where } y = H(x)
    \\
    \Ps^{-1}\left( y, \vecU \right) & =
    \begin{cases}
    \Pi^{-1}\left( y, \vecZ \right) & \exists \: \vecZ \in \bbZ_{2}^{n - r}\text{ such that }\matA(y) \cdot \vecZ + \vecB(y) = \vecU
    \\
    \bot &\text{ else }
    \end{cases}
    \\
    \Ds\left( y, \vecV \right) &=
    \begin{cases}
    1 &\text{ if }\vecV^T \cdot \matA(y) = 0^{n-r} \\
    0 &\text{otherwise}
    \end{cases}
\end{align*}
We will denote the above distribution of oracles by $\Oracle_{n,r,k}$. We define our hash function as $H$, which can be easily computed by querying $\Ps$, considering only the first $r$ bits and discarding the second output.
\end{construction}

Note that since $\matA(y) \in \Z_{2}^{k \times \left( n - r \right)}$ is full column-rank, $\vecZ \in \bbZ_{2}^{n - r}$ is unique if it exists. Thus, $\Ps^{-1}\left( \Ps(x) \right) = x$ and $\Ps^{-1}\left( y, \vecU \right) = \bot$ if $\left( y, \vecU \right)$ is not in the range of $\Ps$. Thus, $\Ps^{-1}$ is the uniquely-defined inverse of the injective function $\Ps$.

\paragraph{Efficient Implementation.}
Given the description of Construction \ref{constr:main}, it is straightforward to implement it, using pseudorandom functions and pseudorandom permutations, the existence of both of which follow from the existence of one-way functions \cite{zhandry2021construct, zhandry2016note}. Specifically, we swap the truly random permutation $\Pi$ with a pseudorandom permutation $\PRP$, and swap the truly random function $F$ with a pseudorandom function $\PRF$. One can easily verify that under the security of $\PRP$, $\PRF$ for quantum queries, the classical efficient construction is computationally indistinguishable from the oracle distribution $\Oracle_{ n, r, k }$ described in Construction \ref{constr:main}.

\paragraph{Non-collapsing.}
We prove that our hash function is non-collapsing.

\begin{proposition} [Non-collapsing of $H$] \label{proposition:non_collapsing}
    The hash function $H$ defined in Construction \ref{constr:main} is (always) non-collapsing, as per Definition \ref{definition:CR_always_NC_hash}.
\end{proposition}
\begin{proof}
    We explain the non-collapsing property of our scheme by describing the algorithms $\left( \Ss_{H}, \Ds_{H} \right)$. The first algorithm $\Ss_{H}$ simply computes a uniform superposition $|+ \rangle^{\otimes n}$ over $n$ qubits (where $n$ is the input size for $H$), and outputs it as $|\psi \rangle$. The second algorithm $\Ds_{H}$, given an unknown $n$-qubit state $|\phi \rangle := \sum_{x \in \bbZ_{2}^{n}} \alpha_{x} \cdot |x\rangle$, acts as follows.
    \begin{enumerate}
        \item \label{oracle_construction_non_collapsing_step_1}
        Execute $\Ps$ in superposition to obtain 
        $$
        \sum_{x \in \bbZ_{2}^{n}} \alpha_{x} \cdot |x\rangle |y_{x}\rangle |\vecU_{x}\rangle
        \enspace .
        $$

        \item \label{oracle_construction_non_collapsing_step_2}
        Execute $\Ps^{-1}$ in superposition to un-compute the input register holding $x$, to obtain
        $$
        \sum_{x \in \bbZ_{2}^{n}} \alpha_{x} \cdot |y_{x}\rangle |\vecU_{x}\rangle
        \enspace .
        $$

        \item \label{oracle_construction_non_collapsing_step_3}
        Execute a $k$-qubit Quantum Fourier Transform over $\bbZ_{2}$ on the rightmost $k$-qubit register (holding the vectors $\vecU_{x}$). This boils down to the execution of parallel Hadamard gates $H^{\otimes k}$, to obtain
        $$
        \sum_{x \in \bbZ_{2}^{n}} \alpha_{x} \cdot |y_{x}\rangle \left( H^{\otimes k} \cdot |\vecU_{x}\rangle \right)
        \enspace .
        $$

        \item \label{oracle_construction_non_collapsing_step_4}
        Execute $\Ds\left( \cdot, \cdot \right)$ in superposition on the state and measure the output bit register. The output of $\Ds_{H}$ is identical to the output of $\Ds(\cdot, \cdot)$.
    \end{enumerate}

    To finish the explanation for non-collapsing we will show that the pair $\left( \Ss_{H}, \Ds_{H} \right)$ gives a distinguishing advantage of $\geq 1 - 2^{-\Omega(\secp)}$, between the cases of full measurement and partial measurement.
    \begin{itemize}
        \item
        In the first case, given $|\psi\rangle := |+\rangle^{\otimes n}$, it is measured in the computational basis and the algorithm $\Ds_{H}$ gets as input some classical $x \in \bbZ_{2}^{n}$. At the end of Step \ref{oracle_construction_non_collapsing_step_2} of $\Ds_{H}$ the state is $|y_{x}\rangle |\vecU_{x}\rangle$. The execution of QFT in the next step will generate a uniform superposition over all elements in $\bbZ_{2}^{k}$, at the end of Step \ref{oracle_construction_non_collapsing_step_3}. The set of elements that are accepted by $\Ds(y_{x}, \cdot)$ (for every classical $y_{x}$) is a coset of dimension $r$ and thus in particular a set of size $2^{r}$. The probability that the $k$-qubit uniform superposition will be accepted by $\Ds(y_{x}, \cdot)$ is $\frac{ 2^{r} }{ 2^{k} } = 2^{-(k - r)} \leq 2^{-\Omega(\secp)}$. It follows that the probability that $\Ds_{H}$ outputs $1$ in the first case is $\leq 2^{-\Omega(\secp)}$.

        \item 
        In the second case, given $|\psi\rangle := |+\rangle^{\otimes n}$ we compute $\Ps$ in superposition and measure a value $y$ on the side. The algorithm $\Ds_{H}$ thus gets as input the state $\sum_{x \in \bbZ_{2}^{n} : H(x) = y} \sqrt{2^{-(n - r)}} \cdot |x\rangle$ for some $y$. One can easily verify that at the end of Step \ref{oracle_construction_non_collapsing_step_2}, the state that $\Ds_{H}$ holds is
        $$
        | y \rangle
        \otimes 
        \left( \sum_{\vecU \in \colspan\left( \matA(y) \right)} \sqrt{ |\colspan\left( \matA(y) \right)|^{-1} } \cdot | \vecU + \vecB(y) \rangle \right)
        \enspace .
        $$
        Next, by standard known properties of QFT applied to a quantum state that is in a uniform superposition over a coset, it follows that at the end of Step \ref{oracle_construction_non_collapsing_step_3}, the state that $\Ds_{H}$ holds is 
        $$
        | y \rangle
        \otimes 
        \left( \sum_{\vecV \in \colspan\left( \matA(y) \right)^{\bot}} \sqrt{ |\colspan\left( \matA(y) \right)^{\bot}|^{-1} } \cdot (-1)^{\langle \vecB(y), \vecV \rangle} \cdot | \vecV \rangle \right)
        \enspace .
        $$
        It follows that the probability that $\Ds_{H}$ outputs $1$ in the second case is $1$.
    \end{itemize}
\end{proof}

\paragraph{Security.}
Most of this section will be devoted to proving security. There are two main steps to our security proof. The first part of the proof will show Theorem \ref{thm:dualtodualfree} (Proved in Sections \ref{subsection:bloating_dual_oracle} and \ref{subsection:simulating_dual_oracle}), which says that an adversary $\Adv$ that manages to find a collision in $H$ given access to $\left( \Ps, \Ps^{-1}, \Ds \right)$ sampled from $\Oracle_{ n, r, k }$, can be transformed to an adversary $\AdvB$ that finds a collision in $H$ in the dual-free setting, i.e., having only access to $\left( \Ps, \Ps^{-1} \right)$ and no access to the dual verification oracle $\Ds$. The second part of security will show Theorems \ref{thm:dualfreetocol} and \ref{theorem:n_to_ell_CPF_collision_resistant} (both proved in Section \ref{subsection:dual_free_to_two_to_one_oracle}), which together, show how a collision finder for the dual-free case can be turned into a collision finder for plain $2$-to-$1$ random functions. We obtain the following main security theorem.

\begin{theorem} [Collision Resistance of $H$] \label{theorem:oracle_main_security}
Let $\Oracle_{ n, r, k }$ the distribution over oracles defined in Construction \ref{constr:main}. Let $\Adv$ an oracle aided $q$-query (computationally unbounded) quantum algorithm. Then,
\[
\Pr
\left[
x_0 \neq x_1
\land 
H(x_0) = H(x_1) 
\;
:
\begin{array}
{rl}
\left( \Ps, \Ps^{-1}, \Ds \right) & \gets \Oracle_{n,r,k} \\
\left( x_0, x_1 \right) & \gets \Adv^{\Ps, \Ps^{-1}, \Ds}
\end{array}\right]
\leq
O\left( \frac{ \secp^{3} \cdot k^{3} \cdot q^{3} }{ 2^{ \secp } } \right)
\enspace .
\]
\end{theorem}

\begin{proof}
    Assume towards contradiction that there is an oracle aided quantum algorithm $\Adv$, making $q$ queries, that given a sample oracle $\left( \Ps, \Ps^{-1}, \Ds \right) \gets \Oracle_{ n, r, k }$ outputs a collision $(x_{0}, x_{1})$ in $H$ with probability $\epsilon$, such that $\epsilon \geq \omega\left( \frac{ \secp^{3} \cdot k^{3} \cdot q^{3} }{ 2^{ \secp } } \right)$.

    Note that by our parameter choices in Construction \ref{constr:main} and by our assumption towards contradiction $\epsilon \geq \omega\left( \frac{ \secp^{3} \cdot k^{3} \cdot q^{3} }{ 2^{ \secp } } \right)$, one can verify through calculation that (1) for $s - (n - r - s) := s'$ we have $\frac{ k^3 \cdot q^3 \cdot \frac{ 1 }{ \epsilon^2 } }{ 2^{s'} } \leq o(1)$ and also (2) $\frac{ k^{9} \cdot q^7 \cdot \frac{1}{\epsilon^4} }{ \sqrt{ 2^{n - r - s} } } \leq o(1)$. This means that the conditions of Theorem \ref{thm:dualtodualfree} are satisfied, and it follows there is a $q$-query algorithm $\AdvB$ that gets access only to $\left( \Ps, \Ps^{-1} \right)$, sampled from $\left( \Ps, \Ps^{-1}, \Ds \right) \gets \Oracle_{ r + s, \: r, \: k - (n - r - s) }$ that finds collisions in $H$ with probability $\geq \frac{\epsilon}{2^{6} \cdot k^{2}}$.
    
    By Theorem \ref{thm:dualfreetocol}, it follows there is a $q$-query algorithm $\AdvB'$ that given oracle access to any $Q : \{ 0, 1 \}^{r + s} \rightarrow \{ 0, 1 \}^{r}$ an $\left( r + s, \: r, \: s \right)$-coset partition function (as per Definition \ref{definition:coset_partition_function}), finds a collision in $Q$ with probability $\frac{\epsilon}{2^{6} \cdot k^{2}}$.
    
    The above is in particular true for any distribution $Q \gets \Qs$ over $\left( r + s, \: r, \: s \right)$-coset partition functions. Since our parameter choices in Construction \ref{constr:main} imply $s \: | \: (r + s)$, we can consider the distribution $\Qs$ over $\left( r + s, \: r, \: s \right)$-coset partition functions generated by Theorem \ref{theorem:n_to_ell_CPF_collision_resistant}. It follows by Theorem \ref{theorem:n_to_ell_CPF_collision_resistant} that 
    $$
    \frac{\epsilon}{2^{6} \cdot k^{2}}
    \leq
    O\left( \frac{ s^{3} \cdot q^{3} }{ 2^{ \frac{ (r + s) }{ s } } } \right)
    \enspace ,
    $$
    which in turn implies
    $$
    \epsilon
    \leq
    O\left( \frac{ k^{2} \cdot \secp^{3} \cdot q^{3} }{ 2^{ \secp } } \right)
    \enspace ,
    $$
    in contradiction to $\epsilon \geq \omega\left( \frac{ \secp^{3} \cdot k^{3} \cdot q^{3} }{ 2^{ \secp } } \right)$.
\end{proof}

\subsection{Bloating the Dual} \label{subsection:bloating_dual_oracle}
Let $\Oracle'_{n,r,k,s}$ denote the following distribution over $\Ps,\Ps^{-1},\Ds'$. The primal oracles $\Ps,\Ps^{-1}$ are defined identically to $\Oracle_{n,r,k}$. However, now, for $s \leq n - r$, we let $\matA(y)^{(0)} \in \bbZ_{2}^{s}$ denote the first $s$ columns of $\matA(y) \in \bbZ_{2}^{n - r}$ and $\matA(y)^{(1)} \in \bbZ_{2}^{n - r - s}$ denote the remaining $n-r-s$ columns.
Then, define $\Ds'$ as the oracle:
\[
\Ds'(y, \vecV)=
\begin{cases}
1 & \text{ if }\vecV^{T} \cdot \matA(y)^{(1)} = 0^{n-r-s} \\
0 & \text{otherwise}
\end{cases}
\]
Observe that if $\vv^T \cdot \matA(y) = 0^{n-r}$, then $\vv^T \cdot \matA(y)^{(1)} = 0^{n-r-s}$. Thus $\Ds'$ accepts all points that are accepted by $\Ds$, but also accepts additional points as well, namely those for which $\vv^T \cdot \matA(y)^{(0)} \neq 0^s$ but $\vv^T \cdot \matA(y)^{(1)} = 0^{n-r-s}$. We call this action, of moving from $\Ds$ to a more relaxed oracle $\Ds'$, "bloating the dual".

\begin{lemma} \label{lemma:bloating_dual_oracle}
Suppose there is an oracle aided $q$-query quantum algorithm $\Adv$ such that
\[
\Pr
\left[
\left( y_{0} = y_{1} \right)
\land
\left( x_0 \neq x_1 \right) \; :
\begin{array}{rl}
\left( \Ps, \Ps^{-1}, \Ds \right) & \gets \Oracle_{n,r,k} \\
(x_0, x_1) & \gets \Adv^{\Ps,\Ps^{-1},\Ds} \\
(y_b, \vecU_b) & \gets \Ps(x_b)
\end{array}
\right]
\geq
\epsilon \enspace .
\] 
Also, let $s \leq n - r$ such that (1) for $s - (n - r - s) := s'$ we have $\frac{ k^3 \cdot q^3 \cdot \frac{ 1 }{ \epsilon^2 } }{ 2^{s'} } \leq o(1)$ and (2) $\frac{ k^{9} \cdot q^7 \cdot \frac{1}{\epsilon^4} }{ \sqrt{ 2^{n - r - s} } } \leq o(1)$. Then,
\[
\Pr
\left[
\begin{array}
{rl}
& \left( y_0 = y_1 := y \right) \land \\
& \left( \vecU_0 - \vecU_1 \right) \notin \colspan \left( \matA(y)^{(1)} \right)
\end{array}
\;
:
\begin{array}
{rl}
\left( \Ps,\Ps^{-1},\Ds' \right) & \gets \Oracle'_{n,r,k,s} \\
\left( x_0, x_1 \right) & \gets \Adv^{\Ps,\Ps^{-1},\Ds'} \\
(y_b, \vecU_b) & \gets \Ps(x_b)
\end{array}
\right]
\geq
\frac{ \epsilon }{ 2^{6} \cdot k^{2} }
\enspace .
\]
\end{lemma}
Note that $\left( \vecU_{0} - \vecU_{1} \right) \notin \colspan\left( \matA(y)^{(1)} \right)$ means in particular that $\vecU_{0}, \vecU_{1}$, and hence $x_0, x_1$, are distinct (this follows because for every $y \in \bbZ_{2}^{r}$, the mapping between preimages $x$ and vectors $\vecU_{x} \in \left( \colspan\left( \matA(y) \right) + \vecB(y) \right)$ is bijective). Thus, the second expression means that $\Adv$ is finding collisions, but these collisions satisfy an even stronger requirement.

\begin{proof}
Assume there is an oracle-aided $q$-query quantum algorithm $\Adv$ that given oracle access to $\left( \Ps,\Ps^{-1},\Ds \right) \gets \Os_{n,r,k}$ outputs a pair $\left( x_{0}, x_{1} \right)$ of $n$-bit strings. Denote by $\epsilon$ the probability that $x_{0}$, $x_{1}$ are both distinct and collide in $H(\cdot)$ (i.e., their $y$-values are identical).
We next define a sequence of hybrid experiments, outputs and success probabilities for them, and explain why the success probability in each consecutive pair is statistically close.

\begin{itemize}
    \item $\Hyb_{0}$: The original execution of $\Adv$.
\end{itemize}
The process $\Hyb_{0}$ is the above execution of $\Adv$ on input oracles $\left( \Ps, \Ps^{-1}, \Ds \right)$. We define the output of the process as $(x_0, x_1)$ and the process execution is considered as successful if $x_{0}$, $x_{1}$ are both distinct and collide in $H(\cdot)$. By definition, the success probability of $\Hyb_{0}$ is $\epsilon$.

\begin{itemize}
    \item $\Hyb_{1}$: Simulating the oracles using only a bounded number of cosets $\left( \matA(y), \vecB(y) \right)$, by using small-range distribution.
\end{itemize}
Consider the function $F$ which samples for every $y \in \bbZ_{2}^{r}$ the i.i.d. coset description $\left( \matA(y), \vecB(y) \right)$. These cosets are then used in all three oracles $\Ps$, $\Ps^{-1}$ and $\Ds$. The difference between the current hybrid and the previous hybrid is that we swap $F$ with $F'$ which is sampled as follows: We set $R := \left( 300 \cdot q^{3} \right) \cdot \frac{ 2^{7} \cdot k^{2} }{ \epsilon }$ and for every $y \in \bbZ_{2}^{r}$ we sample a uniformly random $i_{y} \gets [R]$, then sample for every $i \in [R]$ a coset $\left( \matA_{i} \in \bbZ_{2}^{k \times (n - r)}, \vecB_{i} \in \bbZ_{2}^{k} \right)$ as usual. For $y \in \bbZ_{2}^{r}$ we define $F'(y) := \left( \matA_{i_{y}}, \vecB_{i_{y}} \right)$.

By Theorem A.6 from \cite{ananth2022pseudorandom}, it follows that for every quantum algorithm making at most $q$ queries and tries to distinguish between $F$ and $F'$, the distinguishing advantage is bounded by $\frac{300\cdot q^{3}}{R} < \frac{\epsilon}{8}$, which means in particular that the outputs of this hybrid and the previous one has statistical distance bounded by $\frac{\epsilon}{8}$. It follows in particular that the success probability of the current hybrid is $:= \epsilon_{1} \geq \epsilon - \frac{\epsilon}{8} = \frac{7 \cdot \epsilon}{8}$.

\begin{itemize}
    \item $\Hyb_{2}$: Relaxing dual verification oracle to accept a larger subspace, by information-theoretical subspace hiding.
\end{itemize}
The difference between the current hybrid and the previous hybrid is that in the current hybrid we make the dual verification oracle $\Ds$ more relaxed, and accept more strings. Specifically, recall that as part of sampling our oracles $\left( \Ps, \Ps^{-1}, \Ds \right)$ we sample the function $F'$ and in particular we sample $R$ i.i.d. cosets: $\left( \matA_{i} \in \bbZ_{2}^{k \times (n - r)}, \vecB_{i} \in \bbZ_{2}^{k} \right)_{i \in [R]}$. Recall that in the previous hybrid, given input $\left( y \in \bbZ_{2}^{r}, \vecV \in \bbZ_{2}^{k} \right)$, the oracle $\Ds$ accepts iff $\vecV \in \colspan\left( \matA_{i_y} \right)^{\bot}$. The change we make to the current hybrid is the following and applies only to the dual verification oracle $\Ds$: For each $i \in [R]$, we sample $T^{\bot}_{i}$ which is a uniformly random, $(k - n + r + s)$-dimensional superspace of $\colspan\left( \matA_{i} \right)^{\bot}$ (which has dimension $k - n + r$). When we execute $\Ds$ we check membership in $T^{\bot}_{i}$ rather than in the more restrictive, $(k - n + r)$-dimensional $\colspan\left( \matA \right)^{\bot}$.

By Lemma \ref{lemma:unconditional_subspace_hiding}, due to $k - (k - (n - r)) - s =  n - r - s$, for every $i \in [R]$, changing $\Ds$ to check for membership in $T^{\bot}_{i}$ instead of $\colspan\left( \matA_{i} \right)^{\bot}$, is $O\left( \frac{ q \cdot s }{ \sqrt{ 2^{n - r - s} } } \right)$-indistinguishable, for any $q$-query algorithm. Since we use the above indistinguishability $R$ times, we get $R \cdot O\left( \frac{ q \cdot s }{ \sqrt{ 2^{n - r - s} } } \right)$-indistinguishability. It follows that the success probability of the current hybrid is $:= \epsilon_{2} \geq \epsilon_{1} - R \cdot O\left( \frac{ q \cdot s }{ \sqrt{ 2^{n - r - s} } } \right) \geq \frac{7 \cdot \epsilon}{8} - O\left( \frac{ k^{2} \cdot q^4 \cdot s \cdot \frac{1}{\epsilon} }{ \sqrt{ 2^{n - r - s} } } \right)$, which in turn (by our Lemma's assumptions) is at least $\frac{3 \cdot \epsilon}{4}$.

\begin{itemize}
    \item $\Hyb_{3}$: Asking for the sum of collisions to be outside of $T_{i}$, by using dual-subspace anti-concentration.
\end{itemize}
In the current hybrid we change the success predicate of the experiment. Recall that as part of sampling the oracles in the previous hybrid, we sample $R$ i.i.d. cosets $\left( \matA_{i}, \vecB_{i} \right)_{i \in [R]}$ which are used in all three oracles $\left( \Ps, \Ps^{-1}, \Ds \right)$. We then sample $R$ i.i.d. $(k - n + r + s)$-dimensional superspaces $\left( T_{i}^{\bot} \right)_{i \in [R]}$ of the $R$ corresponding duals $\left( \colspan\left( \matA_{i} \right)^{\bot} \right)_{i \in [R]}$. The change we make to the success predicate is the following: At the end of the execution we get a pair $(x_{0}, x_{1})$ from $\Adv$. We define the process as successful if $y_{0} = y_{1} := y$ and also $\left( \vecU_{0} - \vecU_{1} \right) \notin T_{i_{y}}$, rather than only asking that $x_{0} \neq x_{1}$.

Let $\epsilon_{3}$ be the success probability of the current hybrid and note that $\Adv$ finds collisions with probability $\epsilon_{2}$ in the previous hybrid $\Hyb_{2}$ (and since this hybrid is no different, the same goes for the current hybrid). For every value $i \in [R]$ denote by $\epsilon_{2}^{(i)}$ the probability to find a collision in index $i$, or formally, that $y_{0} = y_{1} := y$, $x_{0} \neq x_{1}$ and also $i_{y} = i$. Observe that in such an event we also have $\left( \vecU_{0} - \vecU_{1} \right) \in S_{i_{y}}$. We deduce $\sum_{i \in [R]} \epsilon_{2}^{(i)} = \epsilon_{2}$. Let $L$ be a subset of indices $i \in [R]$ such that $\epsilon_{2}^{(i)} \geq \frac{\epsilon_{2}}{2 \cdot R}$ and note that $\sum_{i \in L} \epsilon_{2}^{(i)} \geq \frac{ \epsilon_{2} }{ 2 }$. For every value $i \in [R]$ also denote by $\epsilon_{3}^{(i)}$ the probability that $y_{0} = y_{1} := y$, $\left( \vecU_{0} - \vecU_{1} \right) \notin T_{i_{y}}$ and also $i_{y} = i$. We deduce $\sum_{i \in [R]} \epsilon_{3}^{(i)} = \epsilon_{3}$.

We would now like to use Lemma \ref{lemma:unconditional_dual_subspace_anti_concentration}, so we make sure that we satisfy its requirements. Let any $i \in L$, we know that by definition $\epsilon_{2}^{(i)} \geq \frac{\epsilon_{2}}{2 \cdot R}$ and also recall that $\epsilon_{2} \geq \frac{3 \cdot \epsilon}{4}$, $R := \left( 300 \cdot q^{3} \right) \cdot \frac{ 2^{7} \cdot k^{2} }{ \epsilon }$ and thus
$$
\epsilon_{2}^{(i)}
\geq
\frac{\epsilon_{2}}{2 \cdot R}
\geq
\frac{3 \cdot \epsilon}{8} \cdot \frac{1}{R}
\geq 
\Omega
\left(
\frac{ \epsilon^2 }{ q^3 \cdot k^2 }
\right)
\enspace .
$$
Let $s' := s - (n - r - s)$ and for any $i \in L$ let $\ell_{i} := \frac{k^2}{\epsilon_{2}^{(i)}} \leq O\left( \frac{ k^4 \cdot q^3 }{ \epsilon^2 } \right)$. Note that by our Lemma \ref{lemma:bloating_dual_oracle} statement's assumptions, we have (1) $\frac{(n - r - s) \cdot \frac{1}{\epsilon_{2}^{(i)}}}{2^{s'}} \leq o(1)$ and (2) $\frac{ q \cdot \ell_{i}^{2} \cdot s }{ \sqrt{ 2^{n - r - s} } } \leq o(1)$. Since this satisfies Lemma \ref{lemma:unconditional_dual_subspace_anti_concentration}, it follows that for every $i \in L$ we have $\epsilon_{3}^{(i)} \geq \frac{ \epsilon_{2}^{(i)} }{ 16 \cdot k^{2} }$. It follows that
$$
\epsilon_{3}
=
\sum_{i \in [R]} \epsilon_{3}^{(i)}
\geq 
\sum_{i \in L} \epsilon_{3}^{(i)}
\geq
\sum_{i \in L} \frac{ \epsilon_{2}^{(i)} }{ 16 \cdot k^{2} }
\geq
\frac{ \left( \frac{ \epsilon_{2} }{ 2 } \right) }{ 16 \cdot k^{2} }
\geq
\frac{ 3 \cdot \epsilon }{ 2^{7} \cdot k^{2} }
\enspace .
$$

\begin{itemize}
    \item $\Hyb_{4}$: For every $i \in [R]$, de-randomizing $T_{i}$ and defining it as the column span of $\matA_{i}^{(1)} \in \bbZ_{2}^{k \times \left( n - (r + s) \right)}$, the last $n - (r + s)$ columns of the matrix $\matA_{i}$, by using the random permutation $\Pi$ and random function $F$.
\end{itemize}
This hybrid is the same as the previous, with one change: For every $i \in [R]$, after sampling the coset $\left( \matA_{i}, \vecB_{i} \right)$, we will not continue to randomly sample $T^{\bot}_{i}$ (which, previously, was a uniformly random $((k - (n - r)) + s)$-dimensional superspace of $\colspan\left( \matA_{i} \right)^{\bot}$) and simply define $T_{i} := \colspan\left( \matA_{i}^{(1)} \right)$ such that $\matA_{i}^{(1)} \in \bbZ_{2}^{k \times (n - r - s)}$ is defined to be the last $n - r - s$ columns of the matrix $\matA_{i} \in \bbZ_{2}^{k \times (n - r)}$. We will define an intermediate hybrid $\Hyb_{3.1}$ and then explain why $\Hyb_{3} \equiv \Hyb_{3.1} \equiv \Hyb_{4}$.

\begin{itemize}
    \item
    \textbf{Using the random permutation $\Pi$.} 
    For every $i \in [R]$ consider the superspace $T_{i}^{\bot}$, which has $(k - n + r) + s$ dimensions. Accordingly, the dual $T_{i}$ has $k - ((k - n + r) + s)$ dimensions, which equals $n - r - s$. For $i \in [R]$, $j \in [n - r - s]$, denote by $\vecT_{i, j} \in \bbZ_{2}^{k}$ the $j$-th basis vector for the subspace $T_{i}$. Now, for every $i \in [R]$ let $M_{i} \in \bbZ_{2}^{(n - r) \times (n - r)}$ a full-rank matrix that represents the coordinates vectors of the basis vectors of $T_{i}$. Formally, $M_{i}$ is such that the $j$-th column of $M_{i}$, denoted $M_{i, j} \in \bbZ_{2}^{n - r}$, satisfies $\matA_{i} \cdot M_{i, j} = \vecT_{i, j} \in \bbZ_{2}^{k}$.

    In $\Hyb_{3.1}$, we define the permutation $\Gamma$ over $\{ 0, 1 \}^{n}$ defined as follows: For an input $s \in \{ 0, 1 \}^{n}$, it takes the left $r$ bits denoted $y \in \bbZ_{2}^{r}$, computes $i_{y} \in [R]$, then applies matrix multiplication by $M_{i_{y}}$ to the remaining right $n - r$ bits. Observe that since $M_{i}$ is full rank for all $i$, then $\Gamma$ is indeed a permutation. The change we make from $\Hyb_{3}$ to $\Hyb_{3.1}$ is that in the current hybrid we apply $\Gamma$ to the \emph{output} of $\Pi$ inside the execution of a query to $\Ps$, and apply $\Gamma^{-1}$ to the \emph{input} of $\Pi^{-1}$ inside the execution of a query to $\Ps^{-1}$. Note that for a truly random $n$-bit permutation $\Pi$, concatenating any fixed permutation $\Gamma$ like this is statistically equivalent to just computing $\Pi$ and $\Pi^{-1}$, thus the outputs (and in particular success probabilities) between $\Hyb_{3}$ and $\Hyb_{3.1}$ are identical.

    \item
    \textbf{Using the random function $F$.}
    In $\Hyb_{4}$, we stop applying the permutation $\Gamma$ to the output of $\Pi$ (and likewise stop applying $\Gamma^{-1}$ to the input of $\Pi^{-1}$), and also stop sampling $T_{i}^{\bot}$ and simply define it as $\colspan\left( \matA_{i}^{(1)} \right)^{\bot}$. Note that for every choice of $\Pi$, the following two distributions over oracles are statistically equivalent.
    (1) for every $i \in [R]$ sample $\matA_{i}$ uniformly at random, then sample the superspace $T_{i}^{\bot}$ uniformly at random, and then set $\Gamma$ accordingly and concatenate it with $\Pi$ as we did in the above $\Hyb_{3.1}$. (2) Just sample $\matA_{i}$ uniformly at random for every $i \in [R]$, and then set $T_{i} := \colspan\left( \matA_{i}^{(1)} \right)$.
    It follows that due to the fact that $F$ is a random function (or more precisely, because for every $i \in [R]$ we choose a uniformly random coset $\left( \matA_{i}, \vecB_{i} \right)$) then the outputs of $\Hyb_{3.1}$ and $\Hyb_{4}$ are statistically equivalent and in particular the success probability in both cases is the same.
\end{itemize}
It follows that the success probability $\epsilon_{4}$ in $\Hyb_{4}$ equals the success probability $\epsilon_{3}$ in $\Hyb_{3}$.

\begin{itemize}
    \item $\Hyb_{5}$: Moving back to using an exponential number of cosets, by using small-range distribution again.
\end{itemize}
We rewind the process of sampling an $R$-small range distribution version of $F$, and use $F$ as a standard random function. By the same argument for the indistinguishability between $\Hyb_{0}$ and $\Hyb_{1}$, the output of the current process has statistical distance bounded by $\frac{300\cdot q^{3}}{R} = \frac{\epsilon}{2^{7} \cdot k^{2}}$, which means in particular that the outputs of this hybrid and the previous hybrid has statistical distance bounded by $\frac{\epsilon}{2^{7} \cdot k^{2}}$. It follows that the success probability of the current hybrid is
$$
:=
\epsilon_{5}
\geq
\epsilon_{4} - \frac{\epsilon}{2^{7} \cdot k^{2}}
\geq
\frac{ 3 \cdot \epsilon }{ 2^{7} \cdot k^{2} } - \frac{\epsilon}{2^{7} \cdot k^{2}}
=
\frac{ \epsilon }{ 2^{6} \cdot k^{2} }
\enspace .
$$

To conclude, note that the process $\Hyb_{5}$ is exactly the process where $\Adv$ executes on input oracle sampled from $\left( \Ps,\Ps^{-1},\Ds' \right) \gets \Oracle'_{n,r,k,s}$. This finishes our proof.
\end{proof}

\subsection{Simulating the Dual} \label{subsection:simulating_dual_oracle}
In this section we prove the following lemma. 

\begin{lemma} \label{lemma:simulating_dual_oracle}
Suppose there is an oracle aided $q$-query quantum algorithm $\Adv$ such that
\[
\Pr
\left[
\begin{array}{rl}
     & y_0 = y_1 := y , \\
     &\left( \vecU_0 - \vecU_1 \right) \notin \colspan \left( \matA(y)^{(1)} \right)
\end{array}
\;
:
\begin{array}
{rl}
\left( \Ps, \Ps^{-1}, \Ds' \right) & \gets \Oracle'_{n, r, k, s} \\
\left( x_0, x_1 \right) & \gets \Adv^{ \Ps, \Ps^{-1}, \Ds' } \\
(y_b, \vecU_b) & \gets \Ps(x_b)
\end{array}
\right]
\geq
\epsilon \enspace .
\]
Then, there is an oracle aided $q$-query quantum algorithm $\AdvB$ such that
\[
\Pr
\left[
\left( \overline{y}_{0} = \overline{y}_{1} \right) \land \left( \overline{x}_{0} \neq \overline{x}_{1} \right) \; :
\begin{array}{rl}
\left( \overline{\Ps}, \overline{\Ps}^{-1}, \overline{\Ds} \right) & \gets \Oracle_{ r + s, \: r, \: k - (n - r - s) } \\
\left( \overline{x}_{0}, \overline{x}_{1} \right) & \gets \AdvB^{ \overline{\Ps}, \overline{\Ps}^{-1} } \\
\left( \overline{y}_{b}, \overline{\vecU}_b \right) & \gets \overline{\Ps}\left( \overline{x}_{b} \right)
\end{array}
\right]
\geq
\epsilon
\enspace .
\] 
\end{lemma}

\begin{proof}
We first describe the actions of the algorithm $\AdvB$ (which will use the code of $\Adv$ as part of its machinery) and then argue why it breaks collision resistance with the appropriate probability. Given oracle access to $\overline{\Ps}, \overline{\Ps}^{-1}$ which comes from $\left( \overline{\Ps}, \overline{\Ps}^{-1}, \overline{\Ds} \right) \gets \Oracle_{ r + s, \: r, \: k - (n - r - s) }$, the algorithm $\AdvB$ does the following:
\begin{itemize}
    \item
    Sample a random function $F_{\matC}$ that outputs some sufficient (polynomial) amount of random bits on an $r$-bit input, and sample a random $n$-bit permutation $\Gamma$. Define the following oracles.

    \item 
    $\left( \; y \in \bbZ_{2}^{r}, \; \vecU \in \bbZ_{2}^{k} \; \right) \gets \Ps\left( x \in \bbZ_{2}^{n} \right)$:
    \begin{itemize}
        \item
        $\left(
        \; \overline{x} \in \bbZ_{2}^{r + s},
        \; \widetilde{x} \in \bbZ_{2}^{n - r - s} \;
        \right) \gets \Gamma(x)$.

        \item 
        $\left(
        \; y \in \bbZ_{2}^{r},
        \; \overline{\vecU} \in \bbZ_{2}^{k - (n - r - s)} \;
        \right) \gets \overline{\Ps}(\overline{x})$.

        \item 
        $\left(
        \; \matC(y) \in \bbZ_{2}^{ k \times k },
        \; \vecD(y) \in \bbZ_{2}^{n - r - s}
        \right)
        \gets
        F_{\matC}(y)$.

        \item 
        $\vecU
        \gets
        \matC(y) \cdot \left( \begin{array}{c} \overline{\vecU} \\ \widetilde{x} + \vecD(y) \end{array} \right)$.
    \end{itemize}

    \item 
    $\left( \; x \in \bbZ_{2}^{n} \; \right)
    \gets
    \Ps^{-1}\left( \; y \in \bbZ_{2}^{r}, \; \vecU \in \bbZ_{2}^{k} \; \right)$:
    \begin{itemize}
        \item 
        $\left(
        \; \matC(y) \in \bbZ_{2}^{ k \times k },
        \; \vecD(y) \in \bbZ_{2}^{n - r - s}
        \right) \gets F_{\matC}(y)$.

        \item 
        $\left( \begin{array}{c} \overline{\vecU} \\ \widetilde{x} \end{array} \right) \gets \matC(y)^{-1} \cdot \vecU - \left( \begin{array}{c} 0^{k-(n - r - s)} \\ \vecD(y) \end{array} \right)$.

        \item 
        $\left( \; \overline{x} \in \bbZ_{2}^{r + s} \; \right)
        \gets
        \overline{\Ps}^{-1}\left( y, \overline{\vecU} \right)$.

        \item
        $x
        \gets
        \Gamma^{-1}\left( \overline{x}, \widetilde{x} \right)$.
    \end{itemize}

    \item 
    $\Ds'\left( \; y \in \bbZ_{2}^{r}, \; \vecV \in \bbZ_{2}^{k} \; \right) \in \{ 0, 1 \}$:
    \begin{itemize}
        \item 
        $\left(
        \; \matC(y) \in \bbZ_{2}^{ k \times k },
        \; \vecD(y) \in \bbZ_{2}^{n - r - s}
        \right)
        \gets
        F_{\matC}\left( y \right)$.

        \item 
        $\matA^{(1)}(y) := $ last $n - r - s$ columns of $\matC(y)$. 

        \item 
        Output $1$ iff $\vecV^{T} \cdot \matA^{(1)}(y) = \mathbf{0}^{n - r - s}$.
    \end{itemize}
\end{itemize}

The remainder of the reduction is simple: $\AdvB$ executes $\left( x_0, x_1 \right) \gets \Adv^{ 
\Ps, \Ps^{-1}, \Ds' }$ and then $\left( \overline{x}_{b}, \widetilde{x}_{b} \right) \gets \Gamma(x_{b})$ and outputs $\left( \overline{x}_{0}, \overline{x}_{1} \right)$. Assume that the output of $\Adv$ satisfies $y_{0} = y_{1} := y$ and also $\left( \vecU_0 - \vecU_1 \right) \notin \colspan \left( \matA(y)^{(1)} \right)$, and recall that $\matA(y)^{(1)} \in \bbZ_{2}^{k \times (n - r - s)}$ are the last $n - r - s$ columns of the matrix $\matA(y) \in \bbZ_{2}^{k \times (n - r)}$, which is generated by the reduction. We explain why it is necessarily the case that $\overline{x}_{0} \neq \overline{x}_{1}$.

First note that due to how we defined the reduction, $\matA(y) := \matC(y) \cdot \left( \begin{array}{cc} \overline{\matA}(y) & \\ & \Id_{n - r - s} \end{array} \right)$, where $\overline{\matA}(y) \in \bbZ_{2}^{(k - (n - r - s)) \times s}$ is the matrix arising from the oracles $\overline{\Ps}, \overline{\Ps}^{-1}$ and $\Id_{n - r - s} \in \bbZ_{2}^{(n - r - s) \times (n - r - s)}$ is the identity matrix of dimension $n - r - s$. Also note that because $\matC(y)$, $\overline{\matA}(y)$ are full rank then $\matA(y)$ is full rank. Now, since $\left( \vecU_0 - \vecU_1 \right) \notin \colspan \left( \matA(y)^{(1)} \right)$ and since $\matA(y)^{(1)}$ are the last $n - r - s$ columns of $\matA(y)$, it follows that if we consider the coordinates vector $\vecX \in \bbZ_{2}^{n - r}$ of $\left( \vecU_0 - \vecU_1 \right)$ with respect to $\matA(y)$, the first $s$ elements are not $0^{s}$. By linearity of matrix multiplication it follows that if we look at each of the two coordinates vectors $\vecX_{0}$, $\vecX_{1}$ (each has $n - r$ bits) for $\vecU_{0}$, $\vecU_{1}$, respectively, somewhere in the first $s$ bits, they differ.
Now, recall how we obtain the first $s$ bits of $\vecX_{b}$ -- this is exactly by applying $\overline{\Pi}$ (the permutation on $\{ 0, 1 \}^{r + s}$ arising from the oracles $\overline{\Ps}, \overline{\Ps}^{-1}$) to $\overline{x}_{b}$ and taking the last $s$ bits of the output. Since these bits differ in the output of the permutation, then the preimages have to differ, i.e., $\overline{x}_{0} \neq \overline{x}_{1}$.

Define $\epsilon_{\AdvB}$ as the probability that the output of $\Adv$ indeed satisfies $y_{0} = y_{1} := y$ and also $\left( \vecU_0 - \vecU_1 \right) \notin \colspan \left( \matA^{(1)}(y) \right)$, and it remains to give a lower bound for the probability $\epsilon_{\AdvB}$. We do this by a sequence of hybrids, eventually showing that the oracle which $\AdvB$ simulates to $\Adv$ is indistinguishable from an oracle sampled from $\Oracle'_{ n, r, k, s }$. More precisely, each hybrid describes a process, it has an output, and a success predicate on the output.

\begin{itemize}
    \item $\Hyb_{0}$: The above distribution $\left( \Ps, \Ps^{-1}, \Ds' \right) \gets \AdvB^{ 
\overline{\Ps}, \overline{\Ps}^{-1} }$, simulated to the algorithm $\Adv$.
\end{itemize}
The first hybrid is where $\AdvB$ executes $\Adv$ by the simulation described above. The output of the process is the output $(x_{0}, x_{1})$ of $\Adv$. The process execution is considered as successful if $y_{0} = y_{1} := y$ and $\left( \vecU_{0} - \vecU_{1} \right) \notin \colspan\left( 
\matA(y)^{(1)} \right)$.

\begin{itemize}
    \item $\Hyb_{1}$: Not applying the inner permutation $\overline{\Pi}$ (which comes from the oracles $\overline{\Ps}$, $\overline{\Ps}^{-1}$), by using the random permutation $\Gamma$.
\end{itemize}
Let $\overline{\Pi}$ the permutation on $\{ 0, 1 \}^{r + s}$ that's inside $\overline{\Ps}$. In the previous hybrid we apply the $n$-bit permutation $\Gamma$ to the input $x \in \bbZ_{2}^n$ and then proceed to apply the inner permutation $\overline{\Pi}$ to the first (i.e. leftmost) $r + s$ output bits of the first permutation $\Gamma$ (we also apply $\Gamma^{-1}$ to the output of the inverse of the inner permutation, in the inverse oracle $\Ps^{-1}$). The change we make to the current hybrid is that we simply apply only $\Gamma$ and discard the inner permutation and its inverse. Since a random permutation concatenated with any permutation distributes identically to a random permutation, the current hybrid is statistically equivalent to the previous and in particular the output of this process distributes identically to the output of the previous, and so does the success probability.

\begin{itemize}
    \item $\Hyb_{2}$: For every $y \in \bbZ_{2}^{r}$, taking $\matA(y)$ to be the direct output of $F$, by using the randomness of the random function.
\end{itemize}
In order to describe the change between the current and previous hybrid we first recall the structure of the oracles from the previous hybrid: Observe that in the previous hybrid, for every $y \in \bbZ_{2}^{r}$ we defined $\matA(y) := \matC(y) \cdot \left( \begin{array}{cc} \overline{\matA}(y) & \\ & \Id_{n - r - s} \end{array} \right)$, where $\matC(y) \in \bbZ_{2}^{k \times k}$ is the output of $F_{\matC}(y)$ and $\overline{\matA}(y) \in \bbZ_{2}^{(k - (n - r - s)) \times s}$ is the output of the inner random function $\overline{F}$ (which comes from the inside of the oracles $\left( \overline{\Ps}, \overline{\Ps}^{-1} \right)$). In the current hybrid we are going to ignore the inner random function $\overline{F}$, its generated matrix $\overline{\matA}(y)$ and also the pair $\matC(y)$, $\vecD(y)$, sample a fresh random function $F_{\matA}$ at the beginning of the process, and on query $y$ generate $\left( \matA(y), \vecB(y) \right) \gets F_{\matA}(y)$, for $\matA(y) \in \bbZ_{2}^{k \times (n - r)}$, $\vecB(y) \in \bbZ_{2}^{k}$.

To see why the two distributions are indistinguishable, note that the following two ways to sample $\matA(y)$, are statistically equivalent: (1) For every $y \in \bbZ_{2}^{r}$, the matrix $\matA(y)$ is generated by sampling a random full-rank matrix $\matC(y) \in \bbZ_{2}^{k \times k}$ and letting $\matA(y)$ be $\matC(y) \cdot \left( \begin{array}{cc} \overline{\matA}(y) & \\ & \Id_{n - r - s} \end{array} \right)$. (2) For every $y \in \bbZ_{2}^{r}$ just sample a full-rank matrix $\matA(y) \in \bbZ_{2}^{k \times (n - r)}$. Since we are using random functions and in the previous hybrid we are sampling $\matA(y)$ according to (1) and in the current hybrid we are sampling $\matA(y)$ according to (2), the outputs of the two hybrids distribute identically.

\paragraph{Finalizing the reduction.}
Observe that the distribution generated in the above $\Hyb_{2}$ is exactly an oracle sampled from $\Oracle'_{ n, r, k, s }$. From the lemma's assumptions, the success probability for $\Hyb_{2}$ is thus $\epsilon$. Since we also showed that the hybrids have identical success probabilities, it follows that $\epsilon_{\AdvB} = \epsilon$, which finishes our proof.
\end{proof}

We conclude this section by stating the following Theorem, which is obtained as a direct corollary from Lemmas \ref{lemma:bloating_dual_oracle} and \ref{lemma:simulating_dual_oracle}.

\begin{theorem} \label{thm:dualtodualfree}
Suppose there is an oracle aided $q$-query quantum algorithm $\Adv$ such that
\[
\Pr
\left[
x_0 \neq x_1
\land 
H(x_0) = H(x_1) 
\;
:
\begin{array}
{rl}
\left( \Ps, \Ps^{-1}, \Ds \right) & \gets \Oracle_{n,r,k} \\
\left( x_0, x_1 \right) & \gets \Adv^{\Ps, \Ps^{-1}, \Ds}
\end{array}\right]
\geq
\epsilon \enspace .
\]
Also, let $s \leq n - r$ such that (1) for $s - (n - r - s) := s'$ we have $\frac{ k^3 \cdot q^3 \cdot \frac{ 1 }{ \epsilon^2 } }{ 2^{s'} } \leq o(1)$ and (2) $\frac{ k^{9} \cdot q^7 \cdot \frac{1}{\epsilon^4} }{ \sqrt{ 2^{n - r - s} } } \leq o(1)$. Then, there is an oracle aided $q$-query quantum algorithm $\AdvB$ such that
\[
\Pr
\left[
x_0 \neq x_1 \land H(x_0) = H(x_1)
\;
:
\begin{array}{rl}
\left( \Ps, \Ps^{-1}, \Ds \right) \gets \Oracle_{r+s,\: r,\: k-(n - r - s)} \\
(x_0, x_1) \gets \AdvB^{\Ps,\Ps^{-1}}
\end{array}
\right]
\geq
\frac{ \epsilon }{ 2^{6} \cdot k^{2} }
\enspace .
\]
\end{theorem}

\subsection{Hardness of the Dual-free Case from 2-to-1 Collision-Resistance} \label{subsection:dual_free_to_two_to_one_oracle}
We start with defining coset partition functions, which are an object we will use in order to show that collision finding in the dual-free case is at least as hard as finding collisions in $2$-to-$1$ random functions.

\begin{definition} [Coset Partition Functions] \label{definition:coset_partition_function}
For $n, \ell \in \Nat$ such that $\ell \leq n$ we say a function $Q:\{0,1\}^n \rightarrow \{0,1\}^m$ is a $(n,m,\ell)$-\emph{coset partition function} if, for each $y$ in the image of $Q$, the pre-image set $Q^{-1}(y)$ has size $2^\ell$ and is a coset of a linear space of dimension $\ell$. We allow different pre-image sets to be cosets of different linear spaces.
\end{definition}

\paragraph{From Dual-free to Coset Partition Functions.} 
We next show that finding collisions in the dual-free case is no easier than finding collisions in \emph{worst-case} coset partition functions. 

\begin{theorem} \label{thm:dualfreetocol}
Let $k \geq n \geq r$. Suppose there is an oracle aided $q$-query quantum algorithm $\Adv$ such that
\[
\Pr
\left[
\left( y_{0} = y_{1} \right) \land \left( x_{0} \neq x_{1} \right) \; :
\begin{array}{rl}
\left( \Ps, \Ps^{-1}, \Ds \right) & \gets \Oracle_{ n, r, k } \\
\left( x_{0}, x_{1} \right) & \gets \Adv^{ \Ps, \Ps^{-1} } \\
\left( y_{b}, \vecU_{b} \right) & \gets \Ps(x_b)
\end{array}
\right]
\geq
\epsilon
\enspace .
\] 
Then there is an oracle aided $q$-query quantum algorithm $\AdvB$ that given any $( n, r, n - r )$-coset partition function $Q$, satisfies
\[
\Pr
\left[
\left( Q(w_0) = Q(w_1) \right) \land \left( w_0 \neq w_1 \right)
\;
:
\;
\left( w_0, w_1 \right) \gets \AdvB^{Q}
\right]
\geq
\epsilon
\enspace .
\]
\end{theorem}

\begin{proof}
$\AdvB$ works as follows. Given oracle access to some $(n, r, n - r)$-coset partition function $Q$, it chooses a random permutation $\Gamma : \{0,1\}^n \rightarrow \{0,1\}^n$, and for each $y$, it chooses a random full-column-rank matrix $\matC_y \in \bbZ_{2}^{k \times n}$ (which is possible since we assume $k \geq n$) and random vector $\vecD_{y} \in \bbZ_{2}^k$. It then runs $\Adv$, simulating the oracles $\Ps,\Ps^{-1}$ as follows:

\paragraph{The oracle $\Ps(x)$:}
\begin{itemize}
    \item
    $y \gets Q(\Gamma(x))$.
    
    \item
    Output $\left( y,\matC_y \cdot \Gamma(x) + \vecD_y \right)$.
\end{itemize}
We now claim that $\Ps$ is correctly distributed. To do so, we will first define an augmented function $Q' : \{0,1\}^n \rightarrow \{0,1\}^n$. On input $\vecZ$, the $n$-bit output of $Q'(\vecZ)$ consists of two parts. The first $r$ bits are set to $y = Q(\vecZ)$. The preimage set $Q^{-1}(y)$ is then a coset, which can be described as the set $\{ \overline{\matA}_y \cdot \vecR + \overline{\vecB}_y \}$ as $\vecR$ ranges over $\Z_2^{n-r}$ (where $\overline{\matA}_y$, $\overline{\vecB}_{y}$ are both unknown to the reduction algorithm $\AdvB$). Here, $\overline{\matA}_y \in \Z_2^{n\times (n-r)}$ has full column-rank and $\overline{\vecB}\in\Z_2^n$. Define the function $\overline{J}(\vecZ)$ that outputs the unique vector in $\Z_2^{n-r}$ such that $\vecZ = \overline{\matA}_y \cdot \overline{J}\left( \vecZ \right) + \overline{\vecB}_y$. Then define $Q'\left( \vecZ \right) = \left( Q\left( \vecZ \right) , \overline{J}\left( \vecZ \right) \right)$. Note that $Q'$ is not efficiently computable without knowing $\overline{\matA}_y$ ,$\overline{\vecB}_y$, but here we will not need it to be. Intuitively, the reason that we do not need $\overline{J}(\cdot)$ to be efficiently computable is because whenever we simulate the oracles $\left( \Ps, \Ps^{-1}, \Ds \right)$, the value $J(\cdot)$ is never output, it is only the connector between preimages $x$ and vectors $\vecU_{x}$ in the coset $\colspan\left( \matA(y) \right) + \vecB(y)$. Notice that $Q'$ is a function from $\Z_2^{n}$ to $\Z_2^n$, and it is moreover a permutation with $\left( Q' \right)^{-1}\left( y, \vecR \right) = \overline{\matA}_y \cdot \vecR + \overline{\vecB}_y$.

Observe that $\AdvB$'s simulation of $\Ps$ is implicitly setting the following parameters
\begin{align*}
\Pi(x) & = Q'(\Gamma(x)) \; , & H(x) & = Q(\Gamma(x)) \; , & J(x)=\overline{J}(\Gamma(x)) \; , \\
\matA_y  & =\matC_y\cdot\overline{\matA}_y \; , & \vecB_y & =\matC_y\cdot\overline{\vecB}_y+\vecD_y \enspace .
\end{align*}
Thus, we must check that these quantities have the correct distribution. Indeed, for every $Q$, which is in turn defines $\left( \overline{\matA}_y, \overline{\vecB}_y \right)_{ y \in \{ 0, 1 \}^{r} }$, the function $Q'(x)$ is a permutation: Given $y \in \{ 0, 1 \}^{r}$, $\vecR \in \bbZ_{2}^{n - r}$, one can recover $\vecZ \in \{ 0, 1 \}^{n}$ as $\vecZ = \overline{\matA}_y \cdot \vecR + \overline{\vecB}_y$. Hence $\Pi$ is a permutation since it is the composition of two permutations. Moreover, since one of the two permutations ($\Gamma$) is uniformly random, so is $\Pi$. 

Now we look at the distribution of $\matA_y,\vecB_y$. Recall that $\overline{\matA}_y \in \Z_2^{n\times (n-r)}$ is a full-column-rank matrix, and $\matC_y \in \Z_2^{k\times n}$ is a \emph{random} full-column-rank matrix. Thus, $\matA_y=\matC_y \cdot \overline{\matA}_y \in \Z_2^{k \times (n-r)}$ is also a random full-column-rank matrix.

Then we have that $\vecB_y=\matC_y\cdot\overline{\vecB}_y+\vecD_y$ where $\vecD_y$ is random, meaning $\vecB_y$ is random. Thus, $\Ps$ has an identical distribution to that arising from $\Os_{n,r,k}$. 

\paragraph{The oracle $\Ps^{-1}(y,\vecU)$:}
\begin{itemize}
    \item
    $x \gets
    \begin{cases}
    \Gamma^{-1}\left( \vecW \right)
    &\text{ $\exists \vecW \in \bbZ_{2}^{n}$ such that $\matC_y \cdot \vecW + \vecD_y = \vecU$} \\
    \bot
    &\text{ if no such $\vecW$ exists}
    \end{cases}$
    \item
    Output $\begin{cases}
    x &\text{ if $x \neq \bot$ and $Q\left( \Gamma(x) \right) = y$ } \\ \bot &\text{ if $x=\bot$ or $Q(\Gamma(x))\neq y$ }
    \end{cases}$
\end{itemize}

Observe that $\Ps^{-1}\left( \Ps(x) \right) = x$, and for all pairs $\left( y \in \{ 0, 1 \}^{r}, \vecU \in \bbZ_{2}^{k} \right)$ that are not in the image of $\Ps$, we have $\Ps^{-1}\left( y, \vecU \right) = \bot$. Thus, $\Ps^{-1}$ is the uniquely-defined inverse of $\Ps$. Thus, since the distribution of $\Ps$ simulated by $\AdvB$ exactly matches the distribution arising from $\Os_{n,r,k}$, the same is true of the pairs $\left( \Ps, \Ps^{-1} \right)$.

\paragraph{Finishing touches.}
Thus we saw that for every input $(n, r, n - r)$-coset partition function $Q$, the algorithm $\AdvB$ perfectly simulates the view of $\Adv$, which consists of the pair of oracles $\Ps,\Ps^{-1}$ that distribute according to $\Os_{n,r,k}$. Hence, with probability $\epsilon$, the algorithm $\Adv$ will produce a collision $x_0\neq x_1$ such that $H(x_0)=H(x_1)$. It remains to explain what $\AdvB$ does in order to obtain a collision in $Q$, for every collision in the simulated $H$. Given $(x_{0}, x_{1})$, the reduction $\AdvB$ will then compute and output $\left( w_0 = \Gamma(x_0) , w_1 = \Gamma(x_1) \right)$. Observe that if $x_0\neq x_1$, then $w_0\neq w_1$ since $\Gamma$ is a permutation. Moreover, if $H(x_0)=H(x_1)$, then 
$$
Q(w_0) = Q\left( \Gamma(x_0) \right) = H(x_0) = H(x_1) = Q\left( \Gamma(x_1) \right) = Q(w_1)
\enspace .
$$
Hence, with probability at least $\epsilon$, $\AdvB$ will output a collision for $Q$. Notice that for each query that $\Adv$ makes to $\Ps(\cdot)$, the reduction $\AdvB$ needs to make exactly one query to $Q$ and the same goes for the inverse $\Ps^{-1}(\cdot)$. Thus $\AdvB$ makes exactly $q$ queries to $Q$. This completes the proof.
\end{proof}

\paragraph{Collision-resistant Coset Partition Functions.}
We now show how to construct a collision resistant coset partition functions relative to an oracle. Our main theorem for this subsection is the following.

\begin{theorem} \label{theorem:n_to_ell_CPF_collision_resistant}
For any $n, \ell \in \Nat$ such that $\ell \: | \: n$, there exists a distribution over $\left( n, n-\ell, \ell \right)$-coset partition functions $H$, such that any algorithm making $q$ queries to $H$ can only find collisions in $H$ with probability at most $O\left( \frac{ \ell^{3} \cdot q^{3} }{ 2^{ \frac{ n }{ \ell } } } \right)$.
\end{theorem}

\begin{proof}
We will start with a much weaker goal of constructing distributions of collision-resistant 2-to-1 functions, that are shrinking by a single bit.

\begin{lemma} \label{lemma:2_to_1_collision_resistant}
A random 2-to-1 function $H:\{0,1\}^n\rightarrow\{0,1\}^{n-1}$ is collision resistant given quantum queries to $H$. In particular, any quantum algorithm making $q$ queries has a $O(q^3/2^n)$ probability of producing a collision.
\end{lemma}

\begin{proof}
The original work to lower-bound the quantum query complexity of collision resistance was~\cite{JACM:AarShi04}. That work proves that random 2-to-1 functions are collision-resistant, provided that the range is at least as large as the domain. They need the large range since their proof works via showing that the function is indistinguishable from a 1-to-1 function. But of course, a 1-to-1 function must have a range at least as large as domain. This is not quite good enough for us, as we insist on our $H$ ``losing'' one bit.

Instead, we first point out that $H$ can be extended into a permutation as follows. First let $J:\{0,1\}^n\rightarrow\{0,1\}$ be an arbitrary function which for any collision $(x_0,x_1)$ of $H$, assigned random but distinct values to $J(x_0)$ and $J(x_1)$. Since $H$ is 2-to-1, such a $J$ only needs a 1-bit range. Then $\Pi(x)=(H(x),J(x))$ is a permutation, and if $H$ is a random 2-to-1 function, then $\Pi$ is a random permutation.

We will argue that $H$ is collision-resistant \emph{even given queries to $\Pi$} (but not $\Pi^{-1}$). Note that this is potentially stronger than giving access to $H$, as an algorithm can always ignore $J$. Suppose there is an oracle aided algorithm $\Adv$ which given oracle access to $\Pi$, finds a collision in $H$ (that is, a collision in the first $n - 1$ output bits of $\Pi$) with probability $\epsilon$.

To argue for the collision resistance given $\Pi$, we recall that for every $q$-query quantum algorithm having access to $\Pi$ (but not its inverse), $\Pi$ being a random permutation is $O\left( \frac{q^3}{2^n} \right)$-indistinguishable from $\Pi$ being a random function~\cite{QIC:Zhandry15}. So now let $H$ be the first $n-1$ bits of a random function $\Pi : \{ 0, 1 \}^{n} \rightarrow \{ 0, 1 \}^{n}$, and $J$ be the last bit. Then we have that $\Adv^\Pi$ still finds a collision in this $H$ with probability at least $\epsilon-O\left( \frac{q^3}{2^n} \right)$.

But now we observe that $H$ and $J$ are simply just independent random functions, and $J$ can be simulated without making any queries to $H$ (this can even be made efficient, but that is irrelevant since we only care about query counts). Thus, we obtain an algorithm $\Bs^H$ which finds a collision in a random function $H:\{0,1\}^n\rightarrow\{0,1\}^{n-1}$ with probability at least $\epsilon - O\left( \frac{q^3}{2^n} \right)$. But we know that random functions are collision resistant, regardless of the relationship between domain and range. In particular, the probability of finding a collision in such random $H$ is at most $O\left( \frac{q^3}{2^{n - 1}} \right) = O\left( \frac{q^3}{2^n} \right)$ by \cite{QIC:Zhandry15}. Thus, we can bound $\epsilon \leq O\left( \frac{q^3}{2^n} \right)$.
\end{proof}

\paragraph{Extending to coset partition functions.}
Observe that a 2-to-1 function is trivially a coset partition function. Indeed, since the pre-image set always is a pair $\{ x_0, x_1 \}$, which is a coset of the $1$-dimensional linear space $\{ 0, x_0 \oplus x_1 \} \subseteq \bbZ_{2}^{n}$ over the field $\bbZ_2$. To conclude what we saw so far, a random $2$-to-$1$ function $H : \{ 0, 1 \}^{n} \rightarrow \{ 0, 1 \}^{n - 1}$ (that is shrinking by one bit) is collision resistant by Lemma \ref{lemma:2_to_1_collision_resistant}, and is also a $(n, n - 1, 1)$ coset partition function, which proves Theorem \ref{theorem:n_to_ell_CPF_collision_resistant} for the case $\ell = 1$.

We extend to (an almost) general $\ell$ as follows. We can let $H^\ell : \{ 0, 1 \}^{ n \cdot \ell } \rightarrow \{ 0, 1 \}^{\left( n - 1 \right) \cdot \ell}$ be the function such that
$$
    H^\ell\left( x_1, \cdots , x_\ell \right) := \left( H(x_1), \cdots, H(x_\ell) \right) \enspace .
$$
It is straightforward that any collision for $H^\ell$ immediately gives a collision for $H$. Moreover, we can simulate any query to $H^\ell$ given $\ell$ queries to $H$. Thus, any algorithm making $q$ queries to $H^\ell$ can only find collisions in $H^\ell$ with probability at most $O\left( \frac{ \ell^{3} \cdot q^{3} }{ 2^n } \right)$.

Lastly, parallel repetition preserves coset partitions, since the pre-image sets of $H^{\ell}$ are just the direct sums of $\ell$ of the pre-image sets of $H$. Thus, we obtain the theorem by replacing $n \cdot \ell$ with $n$.
\end{proof}

\section{Permutable PRPs and Applications}\label{sec:prps}
In this section we develop a new concept of pseudo-random permutations, that will be useful for our obfuscation-based construction and may also be of independent interest. Quite roughly, our notion of a permutation will allow us to compose the permutation with another (fixed) permutation, while maintaining indistinguishability between the cases for whether the permutation was applied or not.

\begin{definition} \label{def:pprp}
Let $G = \{ G_N \}_{N \in \Nat}$ be a collection where each $G_N$ is a set of permutations over $[N]$. An output-permutable PRP (OP-PRP) for $G$ is a tuple of algorithms $\left( \prp, \prp^{-1}, \permute, \eval, \eval^{-1} \right)$ with the following properties.
\begin{itemize}
    \item {\bf Efficient Permutations: } For any key $k\in\{0,1\}^\lambda$ and any desired ``block size'' $N$, $\prp(k,\cdot)$ is an efficiently computable permutation on $[N]$ with $\prp^{-1}(k,\cdot)$ being its efficiently computable inverse.
    
    \item {\bf Output Permuting: } $\permute\left( k, \Gamma, c \right)$ is a deterministic polynomial-time procedure which takes as input a key $k\in\{0,1\}^\lambda$, the circuit description of a permutation $\Gamma$ in $G_N$, and a bit $c$. It outputs a permuted key $k^{\Gamma,c}$. 

    \item {\bf Output Permuted Correctness: } For all $\lambda \in \Nat$, $k\in\{0,1\}^\lambda$, $c\in\{0,1\}$, $\Gamma\in G_N$ and $x,z\in[N]$, 
    \begin{align*}
    \eval\left( k^{\Gamma,c}, x \right) & =
    \begin{cases}
    \prp(k,x) &\text{ if } c=0
    \\
    \Gamma\left( \prp(k,x) \right) & \text{ if } c=1
    \end{cases}
    \\
    \eval^{-1}\left( k^{\Gamma,c}, z \right) & =
    \begin{cases}
    \prp^{-1}(k,z) & \text{ if } c=0
    \\
    \prp^{-1}(k,\Gamma^{-1}(z))) & \text{ if } c=1
    \end{cases}
    \end{align*}
    We call $k^{\Gamma,c}$ a permuted key.
    
    \item {\bf Security:} For any interactive quantum polynomial-time adversary $\As$, there exists a negligible function $\epsilon(\lambda)$ such that the following experiment with $\As$ outputs 1 with probability at most $\frac{1}{2} + \epsilon\left( \lambda \right)$:
    \begin{itemize}
        \item
        $\As\left( 1^\lambda \right)$ produces a block size $N$ (in binary) and the description of a permutation $\Gamma \in G_N$.
        
        \item
        The experiment chooses a random $k \gets \{0,1\}^\lambda$ and a random bit $c\in\{0,1\}$. It gives $k^{\Gamma,c}\gets\permute(k,\Gamma,c)$ to $\As$.
        
        \item
        $\As$ produces a guess $c'$ for $c$. The experiment outputs 1 if $c'=c$.
    \end{itemize}
\end{itemize}
\end{definition}

\paragraph{On the relations between parameters.}
Observe that $\lambda, N \in \Nat$ can be arbitrary, which means in particular that we can fix the block size $N$ and vary the security parameter $\lambda$, or alternatively fix $\lambda$ and vary $N$. This lets us construct permutable PRPs for small domains with high security, for example. We will often overload notation, and use the same symbol for both $\prp$ and $\eval$ (and corresponding symbols for $\prp^{-1}$ and $\eval^{-1}$), when clear from context whether we are using a permuted or normal key. In this case, an OP-PRP would be a triplet $\left( \prp, \prp^{-1}, \permute \right)$.

\paragraph{Sub-exponential Security.}
For arbitrary functions $f_{0}, f_{1} : \Nat \rightarrow \Nat$, we say that an OP-PRP is $\left( f_{0}, \frac{1}{f_{1}} \right)$-secure if in the above the security part of the definition, we ask that the indistinguishability holds for every adversary of size $\leq f_{0}(\secp)$ and we swap $\epsilon(\secp)$ with $\frac{1}{ f_{1}(\secp) }$. Concretely, a sub-exponentially secure OP-PRP scheme would be one such that there exists a positive real constant $\delta > 0$ such that the scheme is $\left( 2^{\lambda^{\delta}}, \frac{1}{2^{\lambda^{\delta}}} \right)$-secure.

\paragraph{From output-permutable to arbitrarily permutable.}
We say that a PRP is an \emph{input permutable} (IP-) PRP if we apply $\Gamma$ to the inputs rather than the outputs of $\Pi$ (and accordingly apply $\Gamma^{-1}$ to the output, and not the input, of $\Pi^{-1}$). Note that by exchanging the roles of $\prp$ and $\prp^{-1}$ (and likewise $\eval$ and $\eval^{-1}$), we can turn any OP-PRP into an IP-PRP and vice versa. Also, for a PRP of the form $\Pi_{ \mathsf{out} }\left( k_{ \mathsf{out} }, \: \Pi_{ \mathsf{in} }\left( k_{ \mathsf{in} }, \: \cdot \: \right) \right)$, that's the composition of an IP-PRP $\Pi_{ \mathsf{in} }$ and an OP-PRP $\Pi_{ \mathsf{out} }$, both properties are simultaneously satisfied. We simply call such a PRP a permutable PRP (P-PRP).

\paragraph{Decomposable Permutations.} In this work the class $G$ of permutations that we will be interested in, is the class of \emph{decomposable permutations}, defined next.

\begin{definition} [Decomposable Permutations]
Let $N \in \Nat$, let $\Gamma$ a permutation on $[N]$ and let $T, s : \Nat \rightarrow \Nat$. We say that $\Gamma$ is $\left( T(N), s(N) \right)$-decomposable if there exists a sequence of permutations $\Gamma_0, \Gamma_1, \cdots, \Gamma_{T(N)}$ such that:
\begin{itemize}
    \item $\Gamma_0$ is the identity.
    \item $\Gamma_{T(N)} = \Gamma$.
    \item Each $\Gamma_i,\Gamma_i^{-1}$ has circuit size at most $s(N)$.
    \item For each $i$, either (1) $\Gamma_i = \Gamma_{i-1}$ or (2) there exists a $z_i \in [N]$ such that $\Gamma_i = \Gamma_{i-1} \circ \neighborswap{z_i}$. Here $\neighborswap{z_i}$ is the neighbor-swap permutation for $z_{i}$, which swaps between $z_{i}$ and $\left( z_{i} + 1 \: \text{mod } N \right)$, and acts as the identity on all other elements in $[N]$.
\end{itemize}
In the uniform setting, we will additionally ask that there is a uniform polynomial-time (quantum) algorithm which given the description of $\Gamma$ and $i$ constructs both $z_i$ and the circuits for $\Gamma_i,\Gamma_i^{-1}$.
\end{definition}

\paragraph{Examples of Decomposable Families of Permutations.}
Whenever we ask that the circuit size parameter $s(N)$ is polynomial (in $\log(N)$), we do not know whether any efficiently computable (in both directions) permutation $\Pi$, $\Pi^{-1}$ is decomposable. Intuitively, the reason is that while any permutation can be written as a concatenation of neighbor swaps, it may very well be the case that somewhere along the (likely exponentially-long) sequence of permutations $\Gamma_{i}$, $\Gamma^{-1}_{i}$, some of them will not have an efficient circuit implementation. This leads to our Question \ref{question:decompose} regarding the ability to efficiently decompose permutations.
Despite our lack of complete understanding of the decomposability of permutations, we mention some examples of decomposable families of permutations.

\begin{itemize}
    \item
    \textbf{Naive composition of decomposable permutations.}
    Let $\Gamma$ that can be decomposed into a polynomial-length sequence $\Gamma = \Gamma^1 \circ \cdots\circ \Gamma^r$ such that for every $i \in [r]$, the permutation $\Gamma^i$ is $(T,s)$-decomposable. Then $\Gamma$ is $\left( rT, rs \right)$-decomposable.

    \item
    \textbf{Linear cycles.} 
    Linear cycles $\left( j \; \; j+1 \; \; j+2 \; \; \cdots \; \; \ell - 1 \; \; \ell \right)$ swapping $j$ and $j+1$, then $j+1$ and $j+2$ and eventually $\ell - 1$ and $\ell$. Equivalently, $j$ goes to position $\ell$ and all other elements in the range are subtracted by $1$.
    Linear cycles as well as their inverses, are $\left( N, \polylog(N) \right)$-decomposable. To see this one can consider a straightforward decomposition of the cycle into neighbor swaps, and furthermore the efficient implementation of each of the intermediate permutations is given by a circuit that simply check if $z$ is in the range of the cycle, if so, decrements by $1 \text{ (mod $N$)}$ or if the element is $j$, sends it to position $\ell$.
    
    \item
    \textbf{Transpositions.}
    Transpositions $\left( j\;\;\ell \right)$, or (non-neighboring) swaps, are $\left( O(N), \polylog(N) \right)$-decomposable using the decomposition
    $$
    \left( j \; \; \ell \right)
    =
    \left( j \; \; j+1 \; \; j+2 \; \; \cdots \; \; \ell-1 \right)
    \left( \ell-1 \; \; \ell \right)
    \left( j \; \; j+1 \; \; j+2 \; \; \cdots \; \; \ell-1 \right)^{-1}
    $$
    $$
    =
    \left( j \; \; j+1 \; \; j+2 \; \; \cdots \; \; \ell-1 \; \; \ell \right)
    \left( j \; \; j+1 \; \; j+2 \; \; \cdots \; \; \ell-1 \right)^{-1}
    \enspace ,
    $$
    To see the $\left( O(N), \polylog(N) \right)$ decomposition, we use the above rule of composition of decomposable permutations, for the two decomposable permutations $\left( j \; \; j+1 \; \; j+2 \; \; \cdots \; \; \ell-1 \; \; \ell \right)$ and $\left( j \; \; j+1 \; \; j+2 \; \; \cdots \; \; \ell-1 \right)^{-1}$.

    \item 
    \textbf{Permutations that are decomposable to transpositions, rather than neighbor swaps.}
    Let $\Gamma$ a permutation on $[N]$ that's $\left( T, s \right)$-decomposable, but to transpositions rather than neighbor swaps. Specifically, in the $T$-length sequence of permutations that $\Gamma$ decomposes to (and each of such permutation $\Gamma_{i}$, $\Gamma_{i}^{-1}$ has implementation of complexity $\leq s$), each consecutive pair is either identical or differs in one transposition. Then, $\Gamma$ is $\left( T \cdot O(N), s + \polylog(N) \right)$-decomposable into neighbor swaps.

    \item
    \textbf{Scalar addition.}
    $x \mapsto x + s \bmod N$ is $\left( N, \polylog(N) \right)$-decomposable, as adding $s$ is decomposable into a linear chain of transpositions:
    $$
    \left( 0 \;\; s \right) \circ \left( 1 \;\; 1 + s \right) \circ \cdots \left( N - 2 - s \;\; N - 2 \right) \circ \left( N - 1 - s \;\; N - 1 \right)
    \enspace ,
    $$
    where each point in the chain is efficiently computable (and invertible), because it effectively moves all elements in some (efficiently computable) domain $s$ steps forward, and the largest $s$ elements in the domain become the smallest $s$ elements in the domain, in an order-preserving manner. Thus scalar addition is $\left( N, \polylog(N) \right)$-decomposable. Similarly, vector addition $\vecX \mapsto \vecX + \vecS \bmod \bbZ_{N}^{n}$ is thus $\left( N \cdot n, \polylog(N) \cdot n \right)$-decomposable, by performing the addition per coordinate, and deducing decomposability by naive composition of decomposable permutations.

    \item
    \textbf{Scalar multiplication.}
    $x \mapsto ax\bmod N$ for any polynomial $a$ which has an inverse in $N$ is $(N,\polylog(N))$-decomposable, though this seems to require a bit of work. Under the Extended Riemann Hypothesis (ERH), all $a$ are $\left( N^2, \polylog(N) \right)$-decomposable. If discrete logarithms mod $N$ had small circuits, we could use the conjugation example above to reduce multiplication to addition, decomposing multiplication by any $a$. However, since discrete logarithms are presumably classically hard, we have to do something else. In Remark~\ref{rem:multdecompose} following the proof of Theorem~\ref{thm:contr-merge}, we explain how to decompose multiplications by small $a$, or more generally any $a\in\Z_N^*$ that is generated by small integers. Assuming ERH, all $a$ are generated by small integers~\cite{Bach90}, giving a conditional decomposability for all $a$.
    
    \item
    \textbf{Involutions.}
    Involutions are permutations where $\Gamma \circ \Gamma$ is the identity. Involutions that are computable by circuits of size $s$ are $\left( O(N^2), s + \polylog(N) \right)$-decomposable.
    We will decompose $\Gamma$ into a $\left( N, s \right)$-decomposition of transpositions, which will imply our wanted decomposition to neighbor swaps. Intuitively, we will visualize $\Gamma$ as the applications of disjoint transpositions, and then the way we are going to decompose $\Gamma$ is by adding each of the (possibly exponentially many) transpositions, only when both of its elements are smaller than some index. Formally,
    $$
    \Gamma_i(x) :=
    \begin{cases}
    \Gamma(x) & \text{ if } x \leq i \text{ and } \Gamma(x) \leq i
    \\
    x & \text{ otherwise }
    \end{cases}
    \enspace .
    $$
    Each permutation in the sequence is of complexity $\leq s + \polylog(N)$, we have $\Gamma_{0}$ is the identity and $\Gamma_{N} := \Gamma$. Also, $\Gamma_i = \Gamma_{i-1}$ if $\Gamma(i) \geq i$ and otherwise $\Gamma_i(x) = \Gamma_{i-1} \circ (i \;\; \Gamma(i))$.
    
    \item
    \textbf{Affine transformations.}
    Affine permutations $\xv \mapsto \Am \cdot \xv + \bv \bmod r$ where $\Am \in \Z_r^{n \times n}$ is invertible and $\bv\in\Z_r^n$, are $\left( O(r^n\times n^2), \poly(r,n) \right)$-decomposable. Note that the circuit size depends polynomially on $r$; we can improve this $\left( O(r^{2n} \times n^2), \poly(\log r,n) \right)$-decomposability under ERH. We leave it as an interesting open question to handle general $r$ unconditionally. We observe that we can handle addition by $\bv$ as above. To multiply by $\Am$, we decompose $\Am$ into $O(n^2)$ elementary row operations. Row Swaps are involutions and Row Sums are controlled additions, which are both decomposable by the above results. Finally, since $r$ is small, we can handle Scalar Multiplications by the above.

    \item
    \textbf{Conjugations.}
    If $\Gamma=\Lambda^{-1}\circ\Gamma'\circ\Lambda$ where $\Lambda,\Lambda^{-1}$ are permutations with circuits of size $U$ and $\Gamma'$ is $(T,S)$-decomposable, then $\Gamma$ is $(T,O(S+U))$-decomposable, by simply conjugating the decomposition of $\Gamma$ by $\Lambda$. Note that $\Lambda$ does \emph{not} need to be decomposable.

    \item
    \textbf{Applying a decomposable permutation to a subset of the bits.}
    If $N = N_0 \cdot N_1$, and $\Gamma_0$ is a $(T,s)$-decomposable permutation on domain $[N_0]$, we can extend it to a permutation $\Gamma$ with domain $[N] \cong [N_0] \times [N_1]$, that applies $\Gamma_{0}$ to $[N_{0}]$ and the identity to $[N_{1}]$. Then $\Gamma$ is also $(T,s)$-decomposable.
    
    \item
    \textbf{Applying a conditional decomposable permutation.}
    If $N = N_0 \cdot N_1$, and $\Gamma_0$ is a $(T,s)$-decomposable permutation on domain $[N_0]$, we can extend it to a permutation $\Gamma$ with domain $[N] \cong [N_0] \times [N_1]$ where $\Gamma$ applies $\Gamma_{0}$ to $[N_0]$ conditioned on some target value $v \in [N_{1}]$, and the identity otherwise. Then $\Gamma$ is $(T, s + \polylog(N))$-decomposable.

    \item
    \textbf{Applying a controlled decomposable permutation.}
    If $N = N_0 \cdot N_1$, and for every $v \in [N_{1}]$, $\Gamma_{v}$ is a $(T,s)$-decomposable permutation on domain $[N_0]$, we can extend to a permutation $\Gamma$ with domain $[N] \cong [N_0] \times [N_1]$ where $\Gamma$ applies $\Gamma_{v}$ to $[N_0]$ conditioned on the element in $[N_1]$ being $v$. Then $\Gamma$ is $(T \cdot N_1, s + \polylog(N))$-decomposable.
    
    \item
    \textbf{Injective functions (and permutations) with an ancilla.}
    We do not know how to generically decompose an arbitrary permutation $\Gamma$, though we can do it if we are willing to use ancilla bits. Slightly more generally then permutations with an ancilla, we will show how to compose with an efficiently computable and invertible injective function $\Gamma:[N] \rightarrow [M]$ that computes with an ancilla (this in particular implies for efficiently computable and invertible permutations).
    
    Suppose $\Gamma:[N]\rightarrow [M]$ is an injective function such that (1) $\Gamma$ is computable in size $s$ and (2) $\Gamma^{-1} : [M] \rightarrow [N]$ is computable in size $s$ (and on elements in $[M]$ that are not in the image of $\Gamma$ the output can be arbitrary in $[N]$). We construct a permutation $\Gamma'$ on $\left( [N] \times [M] \right) \cong [N\cdot M]$ that is $\left( O\left( \left( N\cdot M \right)^2 \right), s\times \polylog\left( N\cdot M \right) \right)$-decomposable, such that 
    $$
    \forall x \in [N] : \Gamma'\left( x, 0_{[M]} \right)
    =
    \left( \Gamma\left( x \right) \in [M] , 0_{[N]} \right)
    \enspace ,
    $$
    and for cases where the second input is not $0_{[M]}$, the permutation $\Gamma'$ may act arbitrarily.

    We show how to decompose $\Gamma'$ into $2$ controlled scalar addition permutations, and $1$ involution (which are all decomposable, as we know by now), which will make $\Gamma'$ also appropriately decomposable. Let 
    \begin{align*}
        \Gamma_1 & \left( x, y \right) := \left( x , y + \Gamma(x) \bmod M \right) \enspace , 
        \\
        \Gamma_{2} & \left( x, y \right) := \left( x - \Gamma^{-1}(y) \bmod N, y \right) \enspace , 
        \\
        \Gamma_{3} & \left( x, y \right) := \left( y, x \right) \enspace , 
    \end{align*}
    Then, let $\Gamma' = \Gamma_3 \circ \Gamma_2 \circ \Gamma_1$. Since each of $\Gamma_1, \Gamma_2, \Gamma_3$ are involutions (since our arithmetic is over $\bbZ_{2}$) then the composed $\Gamma'$ is accordingly decomposable using the involution case.
    
\end{itemize}

Our main theorem of this section is the following:
\begin{theorem} \label{thm:op-prp}
Let $T$ be any exponential function and $p$ any polynomial. Assuming the existence of sub-exponentially-secure one-way functions and sub-exponentially-secure iO, there exists an OP-PRP for the class of $(T,p)$-decomposable permutations. Moreover, the OP-PRP is itself $(T,p)$ decomposable.
\end{theorem}
This theorem will be proved in Sections~\ref{sec:onsmerges},~\ref{sec:onsprp}, and~\ref{sec:decomposableprp}. Before proving it, however, we will give some example applications.

\subsection{How to use OP-PRPs with Indisitnguishability Obfuscation}\label{sec:useopprp}

Here, we explain how OP-PRPs are useful for constructions involving indistinguishability obfuscation.

\paragraph{Composing with fixed permutations.} Consider a program $P^{O,O^{-1}}$ which makes queries to an oracle $O$. We show the following:
\begin{lemma} \label{lemma:op_prp_compose}
Let $\Gamma$ be a permutation, and let $\left( \prp, \prp^{-1}, \permute \right)$ be an OP-PRP for a class of permutations which includes $\Gamma$. Then for a sufficiently large polynomial $s$, $\iO\left( 1^\lambda, 1^s, P^{\prp\left( k, \cdot \right), \prp^{-1}\left( k, \cdot \right)} \right)$ is computationally indistinguishable from $\iO\left( 1^\lambda, 1^s, P^{\Gamma\left( \prp\left(k, \cdot \right) \right), \prp^{-1}\left( k, \Gamma^{-1}\left( \cdot \right) \right)} \right)$, where $k \gets \{ 0, 1 \}^\lambda$ is uniformly random.
\end{lemma}
In other words, we can compose $\prp(k,\cdot)$ with any fixed permutation $\Gamma$ applied to the output of $\prp$. This is analogous to the oracle case, where composing a random permutation with any fixed permutation gives a random permutation.
\begin{proof}
We prove security through a sequence of hybrids:

\vspace{2mm}
\noindent
$\Hyb_{0}$:
Here, the adversary is given $\iO(1^\lambda,1^s,P^{O,O^{-1}})$ where $O(\cdot) = \prp(k,\cdot)$ and $O^{-1}(\cdot) = \prp^{-1}(k,\cdot)$.

\vspace{2mm}
\noindent
$\Hyb_{1}$:
Now we sample $k^{\Gamma,0}\gets\permute(k,\Gamma,0)$ and switch to $O(\cdot)=\prp(k^{\Gamma,0},\cdot)$ and $O^{-1}(\cdot)=\prp^{-1}(k^{\Gamma,0},\cdot)$. By the correctness of the permuted key $k^{\Gamma,0}$, $O,O^{-1}$, and hence $P^{O,O^{-1}}$, is unchanged by this modification. Therefore, as long as $s$ is larger than the maximum size of $P^{O,O^{-1}}$ in Hybrids 0 and 1, by $\iO$ security the two hybrids are indistinguishable.

\vspace{2mm}
\noindent
$\Hyb_{2}$:
Now we switch to $k^{\Gamma,1}\gets\permute(k,\Gamma,1)$ and set $O(\cdot)=\prp(k^{\Gamma,1},\cdot)$ and $O^{-1}(\cdot)=\prp^{-1}(k^{\Gamma,1},\cdot)$. Indistinguishability from Hybrid 1 follows from OP-PRP security.

\vspace{2mm}
\noindent
$\Hyb_{3}$:
Now we move to $O(\cdot)=\Gamma(\prp(k,\cdot))$ and $O^{-1}(\cdot)=\prp^{-1}(k,\Gamma^{-1}(\cdot))$. Observe that these oracles $O,O^{-1}$ are functionally identical to those in Hybrid 2, and therefore so is the program $P$. Thus, indistinguishability from Hybrid 2 follows from $\iO$ security.

Thus, we have that Hybrids 0 and 3 are indistinguishable, proving Lemma~\ref{lemma:op_prp_compose}.
\end{proof}

\paragraph{Trapdoor permutations from iO and one-way functions.} Here, we show that obfuscating an OP-PRP gives a trapdoor permutation.

\begin{construction} \label{constr:trapdoorperm}
Let $\left( \prp, \prp^{-1} \right)$ be a PRP, and $\iO$ an indistinguishability obfuscator. Then define the trapdoor permutation $\left( \gen, F, F^{-1} \right)$ as:
\begin{itemize}
    \item
    $\gen\left( 1^\lambda \right)$: Sample $k \gets \{0,1\}^\lambda$. Let $P : \{0,1\}^\lambda \rightarrow \{0,1\}^\lambda$ defined as $P(x) := \prp(k,x)$, setting the block size $n = \lambda$. Let $s$ be a sufficiently large function of $\lambda$. Output $\pk = \hat{P} \gets \iO\left( 1^\lambda, 1^s, P \right)$ and $\sk = k$.
    
    \item
    $F\left( \pk, x \right)$: Interpret $\pk$ as a program $\hat{P}$, and output $\hat{P}\left( x \right)$.
    
    \item
    $F^{-1}\left( \sk, y \right)$: Interpret $\sk$ as a key $k$, and output $\prp^{-1}\left( k, y \right)$.
\end{itemize}
\end{construction}

\begin{theorem} [Trapdoor One-Way Permutations from iO and One-Way Functions] \label{theorem:owp_from_io}
Assume the existence of one-way functions. Assume $\left( \prp, \prp^{-1}, \permute \right)$ is an OP-PRP for some class that includes all transpositions, and $\iO$ is a secure iO. Then, if we instantiate Construction~\ref{constr:trapdoorperm} with $\left( \prp, \prp^{-1}, \permute \right)$ as the PRP, we get a secure trapdoor permutation.
\end{theorem}

\begin{proof}We need to show that there is no algorithm $\As$ which, given $\hat{P}\gets\iO(1^\lambda,1^s,P)$ and a random $y^*\gets\{0,1\}^\lambda$, outputs $x^*$ such that $P(x^*)=y^*$. Assume toward contradiction that there is such an $\As$ with success probability $\epsilon$. We will show that $\epsilon$ is negligible through a sequence of hybrid experiments.

\vspace{2mm}
\noindent
$\Hyb_{0}$:
Here, $\As$ is given $\hat{P}\gets\iO(1^\lambda,1^s,P)$ and $y^*$, where $y^*\gets\{0,1\}^\lambda$ and $P(\cdot)=\prp(k,\cdot)$ for a random key $k\gets\{0,1\}^\lambda$. $\As$ wins if it outputs $x$ such that $\prp(k,x)=y^*$, which by assumption is with non-negligible probability.

\vspace{2mm}
\noindent
$\Hyb_{1}$:
Now we additionally sample a random $y'\gets\{0,1\}^{\lambda}$, and switch to $\hat{P}\gets\iO(1^\lambda,1^s,P_1)$ where 
\[
P_1(x) =
\begin{cases}
y^*
& \text{ if } \prp\left( k, x \right) = y' \\
\prp\left( k, x \right)
& \text{ otherwise }
\end{cases}
\enspace .
\]
$\As$ still wins if it outputs $x$ such that $\prp(k,x) = y^*$. 
Indistinguishability from Hybrid 0 follows from Lemma~\ref{lem:iopuncture}. 

\vspace{2mm}
\noindent
$\Hyb_{2}$:
Now we sample $k^*\gets k^{\transposition{y^*}{y'},0}\gets\permute(k,\transposition{y^*}{y'},0)$, change $\hat{P}\gets\iO(1^\lambda,1^s,P_2)$, where
\[
P_2(x) =
\begin{cases}
y^*
& \text{ if } \prp\left( k^*, x \right) = y' \\
\prp\left( k^*, x \right)
& \text{ otherwise }
\end{cases}
\enspace .
\]
Now, we switch to $\As$ winning if it produces an $x$ such that $\prp(k^*,x)=y^*$. Since the programs $\prp(k,\cdot)$ and $\prp(k^{\transposition{y^*}{y'},0},\cdot):=\prp(k^*,\cdot)$ are functionally equivalent, the obfuscations of the programs $P_1$ and $P_2$ are indistinguishable by the security of the outer $\iO$. Also, the success condition of the adversary is functionally equivalent. Overall, indistinguishability from Hybrid 1 follows from $\iO$ security.

\vspace{2mm}
\noindent
$\Hyb_{3}$:
Now we switch to sampling $k^*=k^{\transposition{y^*}{y'},1}\gets\permute\left( k, \transposition{y^*}{y'}, 1 \right)$ and use this key in the program $P_2$ and also for the win condition. Indistinguishability from Hybrid 2 follows from OP-PRP security. Now observe that the win condition is $\prp(k^{\transposition{y^*}{y'},1},x)=y^*$, which is equivalent to $\prp(k,x)=y'$.

\vspace{2mm}
\noindent
$\Hyb_{4}$:
Now we switch back to giving $\As$ the program $\hat{P}\gets\iO(1^\lambda,1^s,P_1)$, but keeping the winning condition as $\prp(k,x))=y'$. Observe that switching from $k^*=k^{\transposition{y}{y'},0}$ to $k^*=k^{\transposition{y}{y'},1}$ in $P_2$ actually did not change the functionality at all: the change permuted the values of $y^*$ and $y'$ in the output of $\prp$, but since both values cause $P_2$ to output $y^*$, swapping them does not change the functionality. Therefore, the functionality of $P_2$ using either permuted key remains equivalent to $P_1$. Thus, $\As$ outputs $x$ such that $\prp(k,x)=y'$ with non-negligible probability.

Now, observe that if we additionally give the adversary $k$, it can compute $y'=\prp(k,x)$ for itself, with non-negligible probability. Thus, we obtain an adversary $\Bs$ which is given the description of $P$ (namely, the key $k$) and $\iO(1^\lambda,1^s,P_1)$, and guesses $y'$ with non-negligible probability. But this contradicts the computational unpredictability guarantee from Lemma~\ref{lem:iopuncture}. Hence the original advantage of $\As$ in inverting the one-way permutation must be negligible.
\end{proof}

\paragraph{Fixed sparse triggers.} Lemma~\ref{lem:iopuncture} allows for ``puncturing'' a program if a certain random trigger is hit, in which case the program may behave completely differently from the original program. Here, we show that the trigger can even be \emph{fixed}, as long as it is appropriately scrambled by OP-PRPs. Concretely, consider the program $P(x)$ with hard-coded OP-PRP keys $k_0, k_1$ that works as follows (and is also described in Figure \ref{figure:sparsetrigger}):
\begin{enumerate}
    \item
    Apply some polynomial-sized circuit $P_0\left( x \right)$, obtaining a pair $x_1, w_1$.
    
    \item
    Run $\prp\left( k_0, x_1 \right)$, and parse the output as $\left( x_2, w_2 \right)$.
    
    \item \label{step:p1}
    Apply some polynomial-sized circuit $P_1\left( w_1, w_2 \right)$, obtaining $w_3, w_4$.
    
    \item
    Then run $\prp^{-1}\left( k_1, \left( x_2, w_3 \right) \right)$, obtaining $x_3$.
    
    \item
    Finally feed $x_3, w_4$ into some polynomial-sized circuit $P_2$, and output the result.
\end{enumerate}
The key structural property of the program $P$ is that $x_2$ is the output of $\prp\left( k_0, \cdot \right)$ and is fed into $\prp^{-1}\left( k_1, \cdot \right)$ without modification and without affecting any other part of the circuit.

Now consider a different program $P'(x)$, which is identical except that we modify Step~\ref{step:p1} into Step~\ref{step:p1}' by embedding a trigger:
\begin{enumerate}
    \item[\ref{step:p1}'.] Apply some polynomial-sized circuit $R(x_2)$ with single-bit outputs. If $R\left( x_2 \right) = 0$, then let $\left( w_3, w_4 \right) \gets P_1\left( w_1, w_2 \right)$ as in program $P$. However, if $R\left( x_2 \right) = 1$, we instead let $\left( w_3, w_4 \right) \gets P_1'\left( w_1, w_2 \right)$, for some different polynomial-sized circuit $P_1'$.
\end{enumerate}

\paragraph{Extending the above template to a pair of circuits instead of a single circuit.}
In the above, the program $P$ is defined as a function of the programs $P_{0}$, $P_{1}$, $P_{2}$ and the program $P'$ is defined as a function of the programs $P_{0}$, $P_{1}$, $P_{2}$ in addition to the programs $R$, $P_1'$.

We next extend the above template in a way that the single program $P$ is extended to a pair of programs $P$, $Q$, and the single program $P'$ is extended to a pair of programs $P'$, $Q'$, in the following way:
\begin{itemize}
    \item
    Let $Q_{0}$, $Q_{1}$, $Q_{2}$, $P_{0}$, $P_{1}$, $P_{2}$ arbitrary programs. 
    The program $Q$ is defined as a function of $Q_{0}$, $Q_{1}$, $Q_{2}$ in the exact same way that $P$ is defined as a function of $P_{0}$, $P_{1}$, $P_{2}$. The programs $P, Q$ are sampled for the \emph{same} pair of keys $k_{0}$, $k_{1}$, with the only change that the roles of the keys are reversed between $P$ and $Q$.

    \item
    Let $R$, $Q'_1$, $P'_1$ some arbitrary programs, and let $Q_{0}$, $Q_{1}$, $Q_{2}$, $P_{0}$, $P_{1}$, $P_{2}$ the programs defining the above $P$, $Q$.
    The program $Q'$ is defined as a function of $Q_{0}$, $Q_{1}$, $Q_{2}$, $R$, $Q_1'$ in the same way that the program $P'$ is defined as a function of the programs $P_{0}$, $P_{1}$, $P_{2}$, $R$, $P_{1}'$. 
\end{itemize}
Thus, we can consider a sampling mechanism that either samples the pair $P$, $Q$ together (by having the programs $P_{0}$, $P_{1}$, $P_{2}$, $Q_{0}$, $Q_{1}$, $Q_{2}$ fixed and sampling the keys $k_{0}$, $k_{1}$) or samples the pair $P'$, $Q'$ together (by having the programs $P_{0}$, $P_{1}$, $P_{2}$, $Q_{0}$, $Q_{1}$, $Q_{2}$ fixed and also the programs $R$, $P_{1}'$, $Q_{1}'$ fixed, and sampling the keys $k_{0}$, $k_{1}$). We will show that obfuscating the un-primed and primed versions give computationally indistinguishable programs, for certain choices of $R$. In the following, we will interpret $x_2$ as an element in $[0,N)$ for some $N \in \Nat$.

\begin{lemma} [Permuted Fixed Sparse Triggers] \label{lemma:intervaltrigger}
Let $P_{0}$, $P_{1}$, $P_{2}$, $Q_{0}$, $Q_{1}$, $Q_{2}$, $R$, $P_{1}'$, $Q_{1}'$ arbitrary programs and let $P$, $Q$ (or $P'$, $Q'$) the derived pair of programs, as the template above. Assume that $R(x_2)$ outputs 1 if and only if $x_2 \in [a,b)$, for some integers $0 \leq a < b < N$ such that the ratio $\frac{N}{b - a}$ is exponential in $\lambda \in \Nat$, the security parameter. Assuming $\left( \prp, \prp^{-1}, \permute \right)$ the permutation used inside the programs, is an OP-PRP for any class that includes all involutions and $\iO$ is a secure indistinguishability obfuscator, for random choices of the keys $k_0, k_1$, for a sufficiently large polynomial $s$, we have that the obfuscated pair $\iO\left( 1^\lambda, 1^s, P \right) ,\iO\left( 1^\lambda, 1^s, Q \right)$ is computationally indistinguishable from the obfuscated pair $\iO\left( 1^\lambda, 1^s, P' \right), \iO\left( 1^\lambda, 1^s, Q' \right)$.
\end{lemma}

\begin{proof}
We prove indistinguishability through a sequence of hybrids. 

\vspace{2mm}
\noindent
$\Hyb_{0}$:
Here, we obfuscate $P,Q$.

\vspace{2mm}
\noindent
$\Hyb_{1}$:
Here, we choose a random $y\in[0,\ceil{\frac{N}{b - a}})$, and obfuscate the programs $P',Q'$, but where we replace the relation $R$ in both programs with the relation $R'$ where $R'(x_2)$ outputs 1 if and only if $x_2\in [y(b-a),(y+1)(b-a))$. This is the same as saying that $(x_2-[x_2\bmod (b-a)])\;/\;(b-a) = y$. Observe that the range of $(x_2-[x_2\bmod (b-a)])\;/\;(b-a)$ is contained in $\lceil N/(b-a)\rceil$. Indistinguishability of Hybrid 0 and Hybrid 1 follows from the indistinguishability guarantee of Lemma~\ref{lem:iopuncture} and the fact that $y$ is uniform in an exponential-sized domain.

\vspace{2mm}
\noindent
$\Hyb_{2}$:
Here, we switch to a random $y\in[0,\lfloor N/(b-a)\rfloor)$. Since $N/(b-a)$ is exponential, this is a negligible change in the distribution of $y$, hence Hybrid 1 and Hybrid 2 are indistinguishable.

\vspace{2mm}
\noindent
$\Hyb_{3}$:
Let $\pi$ be the involution on $x_2$, which exchanges the ranges $[a,b)$ and $[y(b-a),(y+1)(b-a))$. Notice that these intervals have the same size, and the latter interval is contained in $[0,N)$ since $y\leq N/(b-a)-1$. We can easily extend $\pi$ to be an involution mapping $(x_2,w_2)\mapsto (\pi(x_2),w_2)$ or $(x_2,w_3)\mapsto (\pi(x_2),w_3)$.

Now instead of obfuscating the program $P'$, we switch to obfuscating the program $P''$ which replaces $\prp(k_0,\cdot)$ with $\pi(\prp(k_0,\cdot))$ and $\prp^{-1}(k_1,\cdot)$ with $\prp^{-1}(k_1,\pi^{-1}(\cdot))$ (still using the relation $R'$). We likewise switch from $Q'$ to the analogous program $Q''$. Since this is just composing the PRP applications with the fixed involution $\pi$, Hybrids 1 and 2 are indistinguishable by Lemma~\ref{lemma:op_prp_compose}.

Observe that in $P''$ we now we apply $\pi$ to the output of $\prp(k_0,\cdot)$ and to the input to $\prp^{-1}(k_1,\cdot)$ (and the analogous statements for $Q''$). Thus, the two applications of $\pi$ cancel out, \emph{except} that the trigger is checked \emph{between} applications of $\pi$. Since $\pi$ exchanges the roles of $[a,b)$ and $[y(b-a),(y+1)(b-a))$, if were to test the output of $\prp(k_0,\cdot)$ itself, the trigger value would in fact be the interval $[a,b)$.

\vspace{2mm}
\noindent
$\Hyb_{4}$:
Now we switch to obfuscating $P',Q'$ without $\pi$ but with the correct relation $R$. This is functionally equivalent to our modified $P'',Q''$, since the permuted keys changes the trigger (when interpreted as an output of $\prp(k_0,\cdot)$ or $\prp(k_1,\cdot)$) to be $[a,b)$. Thus, by $\iO$ security, Hybrid 2 and Hybrid 3 are computationally indistinguishable. This completes the proof of Lemma~\ref{lemma:intervaltrigger}.
\end{proof}

\subsection{ONS-Merges}\label{sec:onsmerges}
We now gradually build up to our proof of Theorem~\ref{thm:op-prp}. Here, we start from a seemingly much weaker object called an Output Neighbor Swap Merge (ONS-Merge).

A neighbor swap is a permutation which exchanges some $j$ with $j+1$ and otherwise is the identity. In cycles notation, a neighbor swap would be written as $\neighborswap{j}$. 

A Merge is a permutation with the added correctness requirement. The domain $[N]$ is interpreted as pairs $(b,x)$ for $b\in\{0,1\}$ and $x\in N_b$, where $N_0+N_1=N$. The range $[N]$ remains $[N]$. We let $L=\{(0,x)\}\cong [N_0]$ and $R=\{(1,x)\}\cong [N_1]$. A Merge is then a permutation preserves the ordering of elements in $L$, and also preserves the ordering of elements in $R$. An Output Neighbor Swap Merge can be thought of as an OP-PRP for the simple class of neighbor swaps. However, many such swaps will actually break the strong order-preserving property of the merge, and hence will be illegal. We therefore need to modify the definition (both correctness and security) to account for this.

\begin{definition}\label{def:ais-merge} An Output Neighbor Swap (ONS-) \emph{Merge} is at tuple of five algorithms $(\merge$, $\merge^{-1}$, $\permute$, $\eval$, $\eval^{-1})$ with the following properties:
\begin{itemize}
    \item {\bf Efficient Permutations: } For any key $k\in\{0,1\}^\lambda$, any desired block-sizes $N_0,N_1$, $\merge(k,\cdot)$ is an efficiently computable permutation on $[N=N_0+N_1]$ with with $\merge^{-1}(k,\cdot)$ being its efficiently computable inverse.
    \item {\bf Order-Preserving:} For any key $k\in\{0,1\}^\lambda$, any block sizes $N_0,N_1$, and any two inputs $x_0<x_1\in [N_0]$ (resp. $x_0<x_1\in [N_1]$), then $\merge(k,(0,x_0))<\merge(k,(0,x_1))$ (resp. $\merge(k,(1,x_0))<\merge(k,(1,x_1))$). If $b_0\neq b_1$, there is no restriction on the ordering of $\merge(k,(b_0,x_0))$ and $\merge(k,(b_1,x_1))$
    \item {\bf Output Neighbor Swapping:} $\permute(k,\Gamma,c)$ is a deterministic polynomial-time procedure which takes as input a key $k\in\{0,1\}^\lambda$, a neighbor swap $\neighborswap{z}$, and a bit $c$. If $\merge^{-1}(k,z)=(b_0,x_0)$ and $\merge^{-1}(k,z+1)=(b_1,x_1)$ with $b_0\neq b_1$, it outputs a swapped key $k^{\neighborswap{z},c}$. Otherwise if $b_0=b_1$, the input is considered illegal and the output is $\bot$.
    \item {\bf Output Swapping Correctness:} For all $\lambda\in\Z$, $k\in\{0,1\}^\lambda$, all \emph{legal} neighbor swaps $\neighborswap{z}$, and all $x,z'\in [N]$, \begin{align*}\eval(k^{\neighborswap{z},c},x)&=\begin{cases}\prp(k,x)&\text{ if }c=0\\\neighborswap{z}\circ (\prp(k,x))&\text{ if }c=1\end{cases}\\
    \eval^{-1}(k^{\neighborswap{z},c},z')&=\begin{cases}\prp^{-1}(k,z')&\text{ if }c=0\\\prp^{-1}(k,\neighborswap{z}(z')))&\text{ if }c=1\end{cases}\end{align*}
    \item {\bf Security:} For any interactive QPT adversary $\As$, there exists a negligible function function $\epsilon(\lambda)$ such that the following experiment with $\As$ outputs 1 with probability at most $1/2+\epsilon(\lambda)$:
    \begin{itemize}
        \item $\As(1^\lambda)$ chooses a neighbor swap $\neighborswap{z}$ for $z\in[N-1]$.
        \item The experiment chooses a random $k\gets\{0,1\}^\lambda$ and a random bit $c\in\{0,1\}$. It computes $(b_0,x_0)\gets\merge^{-1}(k,z)$ and $(b_1,x_1)\gets\merge^{-1}(k,z+1)$. It checks that $b_0\neq b_1$; if $b_0=b_1$ the experiment immediately aborts and returns a random bit. If $b_0\neq b_1$, it returns $k^{\neighborswap{z},b}\gets\permute(k,\neighborswap{z},c)$ to $\As$. 
        \item $\As$ produces a guess $c'$ for $c$. The experiment outputs 1 if $c'=c$.
    \end{itemize}
\end{itemize}
\end{definition}
Note that the abort condition is necessary: in these cases, permuting the output by $\neighborswap{z}$ actually reverses the order of two strings belonging to the same set $L$ or $R$. But this breaks the order-preserving property of the permuted key, which allows for easy distinguishing. However, if $\merge^{-1}(k,z)$ and $\merge^{-1}(k,z+1)$ have are in different sets $L$ and $R$, then the order between them is arbitrary, and so we can hope that the permuted keys are indistinguishable.

\paragraph{Hypergeometric Distribution.} Let $D_{N,t,s}$ the following distribution: let $U$ be an arbitrary subset of $[N]$ of size $s$. Choose a random set $V$ of $[N]$ of size exactly $t$, and output the number of elements in $V \cap U$. This distribution is known as the hypergeometric distribution and can be efficiently sampled. More specifically, there is a sequence of functions $D_{N,t,s}^\kappa$ with domain $\{0,1\}^\kappa$ such that $D_{N,t,s}^\kappa(r)$ for random coins $r$ approximates a distribution that is $O(N\times 2^{-\kappa})$-close to $D_{N,t,s}$.

\paragraph{Tally Trees.} In order to describe our construction, we introduce the notion of a tally tree. Consider a merge $\merge$. Assign to each range element $z$ a bit $b$ indicating the first bit of the pre-image of $z$. Thus, we obtain a sequence $V$ of $N$ bits, which determines the images of the sets $L$ and $R$. The number of 0's is exactly $N_0$ and the number of 1's is exactly $N_1$. Observe that $V$ is in bijection with the merge $\merge$. This is because once you choose the set of images of $L$ (resp. $R$), the actual mapping from $L$ (resp. $R$) to those images is fixed by the ordering. 

Now, notice that $V$ alone does not actually allow for \emph{efficient} computation of $\merge$ (nor $\merge^{-1}$), since to determine the image of, say, $(0,x)$, you would need to find the position in $V$ of the $x$-th 0. But this presumably requires scanning exponentially-many bits in $V$ to find the right location.

We can, however, speed this up by supplying more information, which is exactly the \emph{tally tree}. A tally tree $T$ is a binary tree with $N$ leaves, one for each element of the range $N$. We will think of $T$ as having a fixed topology that depends on $N$, with the goal of making $T$ shallow. We will associate two quantities to each node $z$. The first is $s(z)$, which si the number of leaves of the subtree rooted at $z$. $s(z)$ is solely a function of the topology of $T$ and will just be used for notational convenience.

The second value is $v(z)$, which is the total of all $V_u$ values for all leaves $u$ in the subtree rooted at $z$. Equivalently, $v(z)=V_z$ for all leaves, and $v(p)$ for any internal node is equal to the sum $v(p)=v(u)+v(w)$ where $u,w$ are the left and right children of $p$. Observe that for the root $\varepsilon$, $v(\varepsilon)=N_1$. 

Observe that we can equivalently sample a tally tree in reverse, starting from the root. We start by setting $v(\varepsilon)=N_1$. Then suppose we have set $v(p)=t$ for a node $p$ with left child $u$ and right child $w$. Setting $v(p)=t$ stipulates that among the $s(p)$ leaves of the tree rooted at $p$, $t$ of them are set to $1$. But as we haven't set any of the descendents of $p$ yet, just the total of them, the distribution over the positions of those $t$ 1's in the subtree rooted at $p$ is uniformly random. We then have that $v(u)$ is exactly distributed according to $D_{s(p),t,s(u)}$. We then set $v(w)=v(p)-v(u)$.

Notice that by sampling from appropriate hypergeometric distributions, we can sample the nodes of $T$ in basically any order. For example, let $C$ be a \emph{cover} of $T$, meaning a set of nodes whose subtrees are disjoint and jointly include all of the leaves of $T$. We can sample $v(u)$ for all $u\in C$ in any order as follows. Let $N'$ denote the portion of the domain yet to be determined (initially $N'=N$) and $S$ the number of 1's remaining to allocate (initially $S=N_1$). Then in any order, we choose an element $u\in C$, sample $v(u)\gets D_{N',s(u),S}$, and then update $N'\mapsto N'-s(u),S\mapsto S-v(u)$. Notice that setting $v(u)$ for $u\in C$ determines $v(u')$ for all $u'$ that are ``above'' the cover. Observe that through this process, the distribution of $v(u)$ for any $u\in C$ depends on the total $\sum_{u'}v(u')$ of all $v(u')$ sampled so far, but is otherwise independent of the actual values $v(u')$.

Given a tally tree $T$, we evaluate $(b,x)\gets\merge^{-1}(z)$ as follows. First set $b=V_z$ by looking the value stored at the leaf labeled $z$. Now by the ordering property of a merge, $x$ is just a count of the number of $z'<z$ with $\merge^{-1}(z)$ having the first bit $b$. We cannot directly count such $z'$ in $V$ (since there will be exponentially-many), but we can instead use the internal nodes of the tree $T$. Namely, let $U_L=\{u\}$ be the set of nodes that are \emph{left} siblings of nodes on the path from root to $z$. Then $x=\sum_{u\in U_L} v(u)$. Observe that this process only visits a single path from root to leaf and its siblings, and therefore only $O(d)$ nodes where $d$ is the depth of the tree.

To evaluate $\merge(x)$ given $T$, we simply do a binary search, exploiting the ordered property and our ability to compute $\merge^{-1}(x)$.

\paragraph{Our Construction.} To give our construction, we will show how to implicitly generate a tally tree, which will then generate a merge. Rather than build the tally tree from the leafs toward the roots, our construction use the top-down generation, but pseudorandonly generate the values in the tree.

\begin{construction}\label{constr:merge}Let $\prf:\{0,1\}^\lambda\times\{0,1\}^*\rightarrow\{0,1\}^\kappa$, and suppose it has a puncturing algorithm $\punc$. Given a key $k\in\{0,1\}^\lambda$, the tally tree $T$ is implicitly defined as follows. We deterministically choose a topology that has depth $O(\log N)$, which determines $s(z)$ for all nodes $z$. We set $v(\varepsilon)=N_1$. Then we recursively internal nodes as follows. For a node $u$ that is a left child of some node $p$ (which already has $s(p)$ defined), define $v(u)=D^\kappa_{s(p),v(p),s(u)}(\prf(k,p))$. We then let $v(w)$ for the right child $w$ to be $v(p)-v(u)$. Note that these values are not explicitly computed, but rather left implicitly determined by $k$.

To permute a key $k$ according to a legal swap $\neighborswap{z}$ and bit $c$, collect into a set $H$ all nodes in the paths from root to $z$ or $z+1$, together with all the siblings of nodes on this path. Here, $H$ stands for ``hard-coded''. Let $S$ denote the nodes just on the paths, excluding $z,z+1$. Here, $S$ stands for ``punctured'' nodes. Assume without loss of generality that $v(z)=0$ and $v(z+1)=1$. Let $\overline{k}^S\gets\punc\left( k, S \right)$. Then do the following:
\begin{itemize}
    \item If $c=0$, output $k^{\neighborswap{z},0}=(\overline{k}^S,\{(u,v(u))\}_{u\in H}$.
    \item If $c=1$, output $k^{\neighborswap{z},1}=(\overline{k}^S,\{(u,v'(u))\}_{u\in H}$ where \[
    v'(u)=\begin{cases}v(u)+1&\text{ if $u$ is an ancestor of $z$ but not $z+1$}\\
    v(u)-1&\text{ if $u$ is an ancestor of $z+1$ but not $z$}\\
    v(u)&\text{ otherwise}
    \end{cases}
    \]
\end{itemize}
To evaluate $v(u)$ for some node $u$, we first look up if there is a pair $(u,v)\in H$, and if so produce the value $v$. Otherwise, we can use $\overline{k}^S$ to compute the remaining nodes.
\end{construction}

\begin{theorem}\label{thm:contr-merge} Suppose $(\prf,\punc)$ is a secure puncturable PRF. Assume $\kappa\geq \lambda+\log N$. Then the protocol given in Construction~\ref{constr:merge} is an ONS Merge. In particular, if $(\prf,\punc)$ is $\epsilon$-secure, then Construction~\ref{constr:merge} is $O(\epsilon+N^2\times 2^{-\kappa})$-secure. Moreover, the permutations $\merge(k,\cdot),\merge^{-1}(k,\cdot)$ are $(O(N^2),\polylog(N))$-decomposable.
\end{theorem}
\begin{proof}We first argue decomposability. Consider a tally tree $T$. The ``identity'' tally tree $T_0$ simply has all the leaves of value 1 as the rightmost leaves. Define $T_{r,i}$ as the tree which has the first (left-most) $r-1$ 1's in the leaves in the correct position, the right-most $N_1-r$ pushed all the way to the right, and the $r$-th 1 at position $N-(N_1-r)-i$. Let $i_r$ be such that the $N-(N_1-r)-i_r$ is the correct position of the $r$th 1. The tally trees $T_{r,i}$ can all be succinctly represented: let $C_r$ be the cover of the leaves containing the left-most $r-1$ 1's. Then $v(u)$ for $u$ descendant from $C_r$ is determined exactly as in $T$. Let $C_r'$ be a cover for the $N_1-r$ leaves that have all the right-most 1's. Then any $u$ descendant from $C_r'$ has $v(u)=s(u)$. Finally, let $C_r''$ be a cover for the remaining nodes. Since  there is at most 1 leaf descendant from $C_r''$ that contains a 1, we can easily compute $v(u)$ for any $u$ descendant from $C_r''$, by simply deciding whether that 1 is a descandant of $u$.

Then $T_0=T_{1,0}$, $T_{r,i_r}=T_{r+1,0}$, and $T_{N_1+1,0}=T$. Going from $T_{r,i}$ to $T_{r,i+1}$ corresponds to a neighbor swap. This gives a sequence of length $O(N^2)$ connecting $T_0$ to $T$, showing that $T$ is $(O(N^2),\polylog(N))$-decomposable.

We now prove security through a sequence of hybrids.

\vspace{2mm}
\noindent
$\Hyb_{0}$:
Here, the adversary sees $k^{\neighborswap{z},0}=(\overline{k}^S,\{(u,v_0(u))\}_{u\in S})$ where $v_0(u)$ is defined as $v_0(u)=v(u)$, which for the left child $u$ of a parent $p$ is equal to $D^\kappa_{s(p),v(p),s(u)}(\prf(k,p))$. Observe that we only need to consider the left children, as the right children are determined by their parent and sibling. The hybrid outputs a random bit at rejects if $v_0(z)=v_0(z+1)$.

\vspace{2mm}
\noindent
$\Hyb_{1}$:
Here, we replace $k^{\neighborswap{z},0}$ with $(\overline{k}^S,\{(u,v_1(u))\}_{u\in S})$ where 
for left-children $u$, $v_1(u)$ is defined as $D^\kappa_{s(p),v(p),s(u)}(r_u)$ for independent random coins $r_u\gets\{0,1\}^\kappa$. The hybrid outputs a random bit and rejects if $v_1(z)=v_1(z+1)$. By a straightforward reduction to punctured PRF security, Hybrid's 0 and 1 are indistinguishable except with probability $\epsilon$.

\vspace{2mm}
\noindent
$\Hyb_{2}$:
We now move to $(\overline{k}^S,\{(u,v_2(u))\}_{u\in S})$, where $v_2(u)$ is sampled as a fresh random sample from $D_{s(p),v(p),s(u)}$. The hybrid outputs a random bit and rejects if $v_2(z)=v_2(z+1)$. Each such sample is at most $O(N\times 2^{\kappa})$-close to the distribution $D^\kappa_{s(p),v(p),s(u)}$ sampled in Hybrid 1, and there are at most $O(N)$ such samples. Therefore, Hybrids 1 and 2 are at most $O(N^2\times 2^{-\kappa})$-close.

\vspace{2mm}
\noindent
$\Hyb_{3}$:
Here, we define $v_2$ as in Hybrid 2, but give the adversary $(\overline{k}^S,\{(u,v_3(u))\}_{u\in S})$, where, assuming $v_2(z)=0$ and $v_2(z+1)=1$, we define 
\[v_2'(u)=\begin{cases}v_2(u)+1&\text{ if $u$ is an ancestor of $z$ but not $z+1$}\\
    v_2(u)-1&\text{ if $u$ is an ancestor of $z+1$ but not $z$}\\
    v_2(u)&\text{ otherwise}
    \end{cases}\]
We now argue that the views $\{(u,v_2(u))\}_{u\in S}$ and $\{(u,v_2'(u))\}_{u\in S}$ are actually distributed identically, even given $\overline{k}^S$. To do so, observe that an equivalent way to sample $v_2$ is as follows. Let $C$ be the cover of $T$ defined as $H\setminus P$, obtained by taking all the nodes in $S$ that are \emph{not} on the paths from root to $z$ or $z+1$, and additionally include $z,z+1$ themselves. This is given by the red squares and yellow pentagons in Figure~\ref{fig:puncturing}. Then sample $v_2(u)$ for $u\in C$ according to the algorithm for sampling covers of tally trees described above, starting with $z,z+1$ and moving to the rest of $C$. Then compute internal nodes ``above'' the cover by adding the values of the node's children, which gives all pairs $(u,v_2(u))$ for $u\in S$. The nodes ``below'' the cover $C$ are then implicitly generated by $\prf$ as before. The hybrid then rejects if $v_2(z)=v_2(z+1)$.

Since the values of $T$ for the cover $C$ is distributed exactly as in the case of a random merge, we see that conditioned on not rejecting, $v_2(z)$ and $v_2(z+1)$ are distributed as random distinct bits.

Hybrid 3 then is identical, except that we swap the values of at $z$ and $z+1$. But since $v(z),v(z+1)$ were random distinct values anyway, this distributions are identical. Moreover, the rest of the cover $C$ is sampled only depending on the total $v(z)+v(z+1)$, which we know is 1 in both Hybrid 3 and Hybrid 2. Thus, the distribution of the entire cover $C$, and therefore the entire tree, is identical in Hybrid 2 and Hybrid 3.

\vspace{2mm}
\noindent
$\Hyb_{4}$ and $\Hyb_{5}$:
These are analogs of Hybrids 1,0 (respectively), except that we use the values $v_1',v_0'$, which are derived from $v_1,v_0$ analogous to $v_2'$. Hybrids 3 and 4 are indistinguishable by an identical argument to the indistinguishability of Hybrids 1 and 2. Hybrids 4 and 5 are likewise indistinguishable by an identical argument to that of Hybrids 0 and 1.

Then we observe that Hybrid 5 uses $v_0'=v'$, and is therefore exactly the case $k^{\neighborswap{z},1}$. This finishes the proof.\qed\end{proof}

\begin{remark}\label{rem:multdecompose}
The decomposition of merges actually allows us to perform scalar multiplication for small scalars. We observe that $x\mapsto 2x\bmod N$ for odd $N$ is actually a merge: indeed,it preserves the ordering for $x\in[0,(N-1)/2]$ as well as for $x\in[(N+1)/2,N-1]$. Moreover, we can compute the associated tally tree rather trivially. Thus we can decompose it as in Theorem~\ref{thm:contr-merge}. We can likewise handle $x\mapsto ax\bmod N$ for $N$ relatively prime to $a$, as long as $a$ is polynomial. The idea is to generalize the concept of a merge where the input domain is partitioned into $a$ buckets and the order-preserving property holds for each bucket. We can then moreover generate $x\mapsto ax\bmod N$ for any $a$ which can be decomposed as the product of ``small'' (polynomial in $\log N$) $a'$. Note that if the Extended Riemann Hypothesis is true, such small $a'$ generate al $a$~\cite{Bach90}.
\end{remark}

\subsection{From ONS-Merge to ONS-PRP}\label{sec:onsprp}

An Output Neighbor Swappable PRP is an OP-PRP for the class of neighbor swaps$ \neighborswap{z}$. We now show that an ONS-Merge can be used to construct an ONS-PRP. The construction is based on the small-domain PRP of~\cite{FSE:GraPor07}, though (1) we will use it in the large-domain setting, and (2) we will be implicitly implementing the underlying PRF with a \emph{puncturable} PRF. 

The construction can be thought of as performing a merge-sort in reverse, where we first partition the domain into a left and right set, of sizes $N_0$ and $N_1$ respectively. Then all the left elements are recursively shuffled using a smaller PRP, and all the right elements are shuffled using an independent smaller PRP. This step is the opposite of sorting the left and right parts separately. Then the two halves are randomly inserted into their final positions using a merge, which preserves the order of each half. The keys for the smaller PRPs and the merge will be derived pseudorandomly from the overall key.

\begin{construction}
Let $\left( \merge, \merge^{-1}, \permute' \right)$ be an ONS-Merge. Let $\prg$ be a length-tripling PRG. Then let $(\prp,\prp^{-1},\permute)$ be defined as follows. For $N=2$ ($x$ is a single bit), we will simply have $\Pi(k,x)=k\oplus x$. The only permutation is $(0\;\;1)$, which is the same as $x\mapsto 1\oplus x$. We can permute the key information-theoretically as $\permute(k,(0\;\;1))=k\oplus 1$. For larger $N$, we let $N_0=\lfloor N/2\rfloor$ and $N_1=N_0-N_1$. We interpret the domain as $[N]\cong \{(b,x)\}_{b\in\{0,1\},x\in[N_b]}$. We use this interpretation both for the input to $\prp(k,\cdot)$ as well as for the input to $\merge$. For the outputs of $\prp(k,\cdot)$ and $\merge$, we will interpret the range as $[N]$. Then we define the algorithms as follows:
\begin{itemize}
    \item $\prp(k,(b,x))$: Run $(k_0,k_1,k_2)\gets\prg(k)$, $y\gets \prp(k_b,x)$, and output $z\gets \merge(k_2,(b,y))$. Here, $\prp(k_b,x)$ is called recursively.
    \item $\prp^{-1}(k,z)$: Run $(k_0,k_1,k_2)\gets\prg(k)$, $(b,y)\gets \merge^{-1}(k_2,z)$ and Compute $x\gets\prp^{-1}(k_b,y)$ and output $(b,x)$. Here, $\prp^{-1}(k_b,y)$ is called recursively.
    \item $\permute(k,\neighborswap{z},c)$: Let $(k_0,k_1,k_2)\gets\prg(k)$. Compute $(b_0,x_0)\gets\prp^{-1}(k,z)$ and $(b_1,x_1)\gets\prp^{-1}(k,z+1)$. We break into two cases:
    \begin{itemize}
        \item If $b_0\neq b_1$ output the permuted key $k^{\neighborswap{z},c}=(k_0,k_1,k_2^{\neighborswap{z},c})$ where $k_2^{\neighborswap{z},c}\gets\permute'(k_2,\neighborswap{z},c)$.
        \item If $b_0=b_1=b$, output  the permuted key $k^{\neighborswap{z},c}=(\overline{k_0},\overline{k_1},k_2)$ where $\overline{k_{1-b}}=k_{1-b}$ and $\overline{k_b}=k_b^{\neighborswap{y},c}\gets\permute(k_b,\neighborswap{y},c)$ where $(b,y)\gets\merge^{-1}(k_2,z)$. Here, $\permute$ is recursively called.
    \end{itemize}
    \item $\prp(k^{\neighborswap{z},c},x)$: Run $\prp(k,x)$ exactly as above but swapping out $k_0,k_1,k_2$ with their permuted versions if necessary. Likewise define $\prp^{-1}(k^{\neighborswap{z},c},x)$.
\end{itemize}
\end{construction}

\begin{theorem}\label{thm:ons-prp-constr} If $\prg$ is a secure PRG and $(\merge,\merge^{-1},\punc')$ is a secure ONS-Merge, then the tuple $(\prp,\prp^{-1},\punc)$ in Construction~\ref{constr:merge} is a secure ONS-PRP. In particular, if $\prg$ is $\epsilon_\prg$-secure and $(\merge,\merge^{-1},\punc')$ is $\epsilon_\merge$-secure, then $(\prp,\prp^{-1},\punc)$ is $\epsilon_\prp$-secure where $\epsilon_\prp(\lambda,n)\leq O((\epsilon_\prg(\lambda)+\epsilon_\merge(\lambda,N))\times \log N)$. Moreover, the permutations $\prp(k,\cdot)$ and $\prp^{-1}(k,\cdot)$ are $(O(N^4),\polylog(N))$-decomposable.
\end{theorem}
\begin{proof}First, we argue efficiency. A call to $\prp$ on block-size $N$ makes a call to $\prg$, a call to $\merge$ on block-size $N$, and a recursive call to $\prp$ on block-size at most $\rceil N/2\rceil $. Solving the recurrence, $\prp$ makes $O(\log N)$ calls to $\prg$ and $O(\log N)$ calls to $\merge$ on block-sizes no more than $N$. Since $\prg$ nad $\merge$ are polynomial time, this is efficient. We can likewise analyze the other algorithms.

\medskip

Next we need to argue the correctness guarantees. In particular, we need to argue that $\permute$ makes called to $\permute'$ on neighbor swaps where the pre-images of the swapped points have different initial bits. Moreover, we need to verify that $\prp$ using a punctured key computes the correct function. 

For the case $n=1$, $\prp$ is just XORing with the key $k$. the only available neighbor swap is $\transposition{0}{1}$, which is the same as XORing with 1, which is equivalent to XORing the key with 1. Thus, the $n=1$ case achieves perfect security. Now we handle the $n>1$ case.

Consider running $\permute(k,\neighborswap{z},c)$. Recall that $\permute'(k_2,\neighborswap{z},c)$ is called exactly when the first bits of $\prp^{-1}(k,z)$ and $\prp^{-1}(k,z+1)$ are different. Notice that this is equivalent to the case where $\merge^{-1}(k_2,z)$ and $\merge^{-1}(k_2,z+1)$ have different first bits, which means that this is a valid call to $\permute'$. In this case, the correctness of the permuted key follows immediately from the correctness of the underlying $\merge$.

Now consider the case where $\prp^{-1}(k,z)$ and $\prp^{-1}(k,z+1)$ (or equivalently, $\merge^{-1}(k_2,z)$ and $\merge^{-1}(k_2,z+1)$) share the same first bit. This implies that $\merge^{-1}(k_2,z)$ and $\merge^{-1}(k_2,z+1)$ are adjacent (any supposed string between them would share the same first bit, and by order preserving, its image under $\merge(k_2,\cdot)$ would need to lie between $z,z+1$, which is impossible). Thus, when $\permute(k_b,\neighborswap{y},c)$ is recursively called where $(b,y)\gets\merge^{-1}(k_2,z)$, we have that $(b,y+1)=\merge^{-1}(k_2,z+1)$. Thus, the neighbor swap $\neighborswap{y}$ is actually transposing $\merge^{-1}(k_2,z)$ and $\merge^{-1}(k_2,z+1)$. From this, correctness of the permuted key following immediately from the (inductively justified) correctness of the underlying $\prp$ for block-size $~N/2$. 

\medskip

Now we argue decomposability. To decompose, we simply decompose $\merge(k_2,\cdot)$, and then recursively decompose $\prp(k_0,\cdot)$ and $\prp(k_1,\cdot)$. Solving the recurrence gives the decomposition statement in Theorem~\ref{thm:ons-prp-constr}.

\medskip

Finally, we prove security. Let $\As$ be a supposed adversary for $(\prp,\prp^{-1},\permute)$ with winning probability $1/2+\epsilon_\prp(\lambda,n)$. Let $k\in\{0,1\}^\lambda$ denote the experiments random key $k$, and let $k_0,k_1,k_2\gets\prg(k)$ denote the keys used to evaluate $\prp,\prp^{-1},\punc$, and let $z_0,z_1$ denote $\As's$ chosen pair. We consider the following sequence of hybrids:

\vspace{2mm}
\noindent
$\Hyb_{0}$:
This is the ONS-PRP security experiment. By assumption, the winning probability is at least $1/2+\epsilon_\prp(\lambda,N)$.

\vspace{2mm}
\noindent
$\Hyb_{1}$:
This is the same experiment, except we switch to $k_0,k_1,k_2$ being uniformly random keys, and using them in all algorithms. By a straightforward reduction to the PRG security of $\prg$, we have that $\As$ still wins with probability at least $1/2+\epsilon_\prp(\lambda,N)-\epsilon_\prg(\lambda)$. 

\vspace{2mm}
\noindent
$\Hyb_{2}$:
This is the same as Hybrid 1, except that we change the win condition. Namely, if the punctured key $k^{\Gamma,c}$ contains a punctured $k_2$, then the experiment aborts and outputs a random bit, independent of $\As$. By a straightforward reduction to the ONS-Merge security of $(\merge,\merge,\punc')$, we have that $\As$ still wins with probability at least $1/2+\epsilon_\prp(\lambda,N)-\epsilon_\prg(\lambda)-\epsilon_\merge(\lambda,N)$.

\vspace{2mm}
\noindent
$\Hyb_{3}$:
This is the same as Hybrid 1, except that we further change the win condition to always output a random bit. The only difference from Hybrid 2 is that when $k^{\Gamma,c}$ contains a punctured $k_b$, in Hybrid 2 the experiment outputs a bit dependent on $\As$'s output, whereas in Hybrid 3 it outputs a random bit. By a straightforward reduction to the ONS-Security of $(\prp,\prp^{-1},\punc)$ with block-size at most $\lceil N/2\rceil $, we have that $\As$ still wins with probability at least $1/2+\epsilon_\prp(\lambda,N)-\epsilon_\prg(\lambda)-\epsilon_\merge(\lambda,N)-\epsilon_\prp(\lambda,\lceil N/2\rceil)$. But observe that in Hybrid 3, the winning probability is exactly $1/2$. Thus we have that
\[\epsilon_\prp(\lambda,N)\leq\epsilon_\prg(\lambda)+\epsilon_\merge(\lambda,N)+\epsilon_\prp(\lambda,\lceil N/2 \rceil)\]
For the base case, observe that when $N=2$, $\prp$ is a random permutation and the punctured key contains no information, so we trivially have $\epsilon_\prp(\lambda,2)=0$. Solving the recurrence gives the statement of the theorem.\end{proof}

\subsection{Achieving OP-PRPs for Decomposable Permutations Using Obfuscation}\label{sec:decomposableprp}

\begin{construction}Let $(\prp,\prp^{-1},\permute')$ be an ONS-PRP (that is, an OP-PRP for the family of neighbor swaps $\neighborswap{z}$). Let $G=(G_n)_n$ be a family of efficiently computable permutations, and $s$ a parameter. Then let $(\prp,\prp^{-1},\permute)$ be a new PRP where $\permute$ does the following:
\begin{itemize}
    \item $\permute(k,\Gamma,b)$: Let\begin{align*}
        P&=\begin{cases}\iO(1^\lambda,1^s,\prp(k,\cdot))&\text{ if }b=0\\\iO(1^\lambda,1^s,\Gamma(\prp(k,\cdot)))&\text{ if }b=1\end{cases}\\
        P^{-1}&=\begin{cases}\iO(1^\lambda,1^s,\prp^{-1}(k,\cdot))&\text{ if }b=0\\\iO(1^\lambda,1^s,\prp(k,\Gamma^{-1}(\cdot)))&\text{ if }b=1\end{cases}
    \end{align*}
    Output $k^{\Gamma,b}=(P,P^{-1})$.
\end{itemize}
Then we augment $\prp,\prp^{-1}$ so that $\prp(\;(P,P^{-1})\;,x)=P(x)$ and $\prp^{-1}(\;(P,P^{-1})\;,z)=P^{-1}(z)$.
\end{construction}
Observe that this construction preserves the decomposability of $\prp,\prp^{-1}$, since the algorithms are exactly the same.

\begin{theorem}\label{thm:op-prp-constr} Suppose $\iO$ is $\epsilon$-secure and $(\prp,\prp^{-1},\permute')$ is $\delta$-secure as an ONS-PRP. Suppose $G=(G_n)_n$ is $(T,s')$-decomposable. Let $t$ be the maximum circuit size of $\prp,\prp^{-1}$ after hardcoding the key (including permuted keys from $\permute'$). Then as long as $s\geq s'+t$, $(\prp,\prp^{-1},\permute)$ is $O(T(\epsilon+\delta))$-secure as an OP-PRP for $G$. 
\end{theorem}
\begin{proof}Let $\Gamma_0,\cdots,\Gamma_T=\Gamma$ be the sequence of permutations where $\Gamma_i=\Gamma_{i-1}\circ \neighborswap{z_i}$ and $\Gamma_0$ is the identity. Our goal is to show that the pair $\iO(1^\lambda,\prp(k,\cdot)),\iO(1^\lambda,\prp^{-1}(k,\cdot))$ is indistinguishable from $\iO(1^\lambda,\Gamma(\prp(k,\cdot))),\iO(1^\lambda,\prp^{-1}(k,\Gamma^{-1}(\cdot)))$. Do do so, we will introduce a sequence of hybrids:

\vspace{2mm}
\noindent
$\Hyb_{i}$:
$\iO(1^\lambda, 1^s\Gamma_i(\prp(k,\cdot))), \iO(1^\lambda,\prp^{-1}(k,\Gamma_i^{-1}(\cdot)))$. In particular, Hybrid $0$ is the case where we have $\iO(1^\lambda,\prp(k,\cdot))$, $\iO(1^\lambda,\prp^{-1}(k,\cdot))$ and Hybrid $T$ is the case $\iO(1^\lambda,\Gamma(\prp(k,\cdot)))$, $\iO(1^\lambda,\prp^{-1}(k,\Gamma^{-1}(\cdot)))$.

\vspace{2mm}
\noindent
$\Hyb_{(i - 1).1}$:
$\iO(1^\lambda,,1^s,\Gamma_i(\prp(k^{\{\neighborswap{z_i},0\}},\cdot))),\iO(1^\lambda,,1^s,\prp^{-1}(k^{\{\neighborswap{z_i},0\}},\Gamma_i^{-1}(\cdot)))$. Here, we sample $k^{\{\neighborswap{z_i},0\}} \gets \permute'(k,\neighborswap{z_i},0)$.

\vspace{2mm}
\noindent
$\Hyb_{(i - 1).2}$:
$\iO(1^\lambda,1^s,\Gamma_i(\prp(k^{\neighborswap{z_i},1},\cdot))),\iO(1^\lambda,1^s,\prp^{-1}(k^{\neighborswap{z_i},1},\Gamma_i^{-1}(\cdot)))$. Here, we sample \\$k^{\neighborswap{z_i},1}\gets\permute'(k,\neighborswap{z_i},1)$.

Observe that $\prp(k^{\neighborswap{z_i},0},\cdot)$ is functionally equivalent to $\prp(k,\cdot))$, and this equivalence is preserved by composing both sides with $\Gamma_{i-1}$. Likewise for the $\prp^{-1}$ functions. The sizes of the programs being obfuscated are at most $s'+t$, where $t$ is the maximum size of the circuit $\prp$ after hardcoding the key (inclunding permuted keys by $\permute'$). Thus, by $\iO$ security, Hybrids $i-1$ and $(i-1).1$ are indistinguishable, except with probability at most $\epsilon$.

The only difference between Hybrids $(i-1).1$ and $(i-1).2$ is that we switch from $k^{\neighborswap{z_i},0}$ to $k^{\neighborswap{z_i},1}$. By the ONS-security of $(\prp,\prp^{-1},\permute')$, these two hybrids are indistinguishable except with probability at most $\delta$.

Finally, observe that $\prp(k^{\neighborswap{z_i},1},\cdot))$ is functionally equivalent to $\neighborswap{z_i}\circ\prp(k,\cdot)$. Therefore $\Gamma_{i-1}\circ \prp(k^{\neighborswap{z_i},1},\cdot))$ is functionally equivalent to $\Gamma_{i-1}\circ \neighborswap{z_i}\circ\prp(k,\cdot)$, which in turn is equivalent to  $\Gamma_i\circ\prp(k,\cdot)$. Thus, by $\iO$ security, Hybrids $(i-1).2$ and $i$ are indistinguishable except with probability at most $\epsilon$.

Thus we obtain a chain of hybrids $0 \rightarrow 0.1 \rightarrow 0.2 \rightarrow 1 \rightarrow 1.1 \rightarrow 1.2 \rightarrow 2 \rightarrow \cdots \rightarrow T-1 \rightarrow (T-1).1 \rightarrow (T-1).2 \rightarrow T$ where each transition is indistinguishable. The triangle inequality then gives the theorem.
\end{proof}

\section{One-Shot Signatures in the Standard Model} \label{sec:standard}
Following our construction in an oracle model from Section \ref{sec:oss_oracle}, we present our standard model construction of non-collapsing collision resistant hash functions, which imply OSS.

\begin{construction}\label{constr:standard}
Let $\secp \in \Nat$ the statistical security parameter. Define $s := 16 \cdot \secp$ and let $n, r, k \in \Nat$ such that $r := s \cdot (\secp - 1)$, $n := r + \frac{3}{2} \cdot s$, $k := n$. Let $d := \poly_{d}(\secp) \in \Nat$ the expansion parameter and $\kappa := \poly_{\kappa}(\secp) \in \Nat$ the cryptographic security parameter, for some sufficiently large polynomials in the statistical security parameter.

Let $\iO$ an iO scheme, $\left( \prf, \punc \right)$ a puncturable PRF, and $\left( \prp, \prp^{-1}, \permute \right)$ a permutable PRP for the class of all $\left( 2^{\poly\left( \kappa \right)}, \poly\left( \kappa \right) \right)$-decomposable permutations. Then we construct a hash function $\left( \gen, \hash \right)$ as follows:
\begin{itemize}
    \item
    $\gen\left( 1^\lambda \right)$:
    Sample $k_{\sf in}, k_{\sf out}, k_{\sf lin} \gets\{0,1\}^{\kappa}$. $\prp\left( k_{\sf in}, \cdot \right)$ is a permutation with domain $\{0,1\}^n$, $\prp\left( k_{\sf out}, \cdot \right)$ is a permutation with domain $\{0,1\}^d$, and $\prf\left( k_{\sf lin}, \cdot \right)$ is a PRF with inputs in $\{ 0,1 \}^d$ that outputs some polynomial number of bits. Let $H(\cdot)$ denote the first $r$ output bits of $\prp\left( k_{\sf in}, \cdot \right)$ and $J(\cdot)$ denote the remaining $n-r$ bits. For each $y \in \{ 0, 1 \}^d$, let $\matA(y) \in \Z_2^{k \times (n-r)}$ be a matrix with full column-rank and $\vecB(y) \in \Z_2^k$ a vector, both are generated pseudorandomly by the output of $\prf\left( k_{\sf lin}, y \right)$.
    
    As the common reference string output $\crs = \left( \mathcal{P}, \mathcal{P}^{-1}, \mathcal{D} \right)$ where $\Ps \gets \iO\left( 1^{\kappa}, P \right)$, $\Ps^{-1} \gets \iO\left( 1^{\kappa}, P^{-1} \right)$, $\mathcal{D} \gets \iO\left( 1^{\kappa}, D \right)$ such that
        \begin{align*}
        P(x) & =
        \big( \; y\; , \; \matA(y)\cdot J(x)+\vecB(y) \; \big)\text{ where }y \gets \prp^{-1}\left( k_{\sf out} , \: H(x) || 0^{d-r} \right) \\
        P^{-1}\left( y, \vecU \right) & =
        \begin{cases}
        \Pi^{-1}\left( w || \vecZ \right) & \exists\; w, \vecZ : \left( \prp\left( k_{\sf out}, y \right) = w || 0^{d-r} \right) \land \left( \matA(y) \cdot \vecZ + \vecB(y) = \vecU \right)
        \\
        \bot & \text{ else }
        \end{cases}
        \\
        D\left( y, \vecV \right) & =
        \begin{cases}
        1 & \text{ if } \vecV^T \cdot \matA(y) = 0^{n-r}
        \\
        0 & \text{ otherwise }
        \end{cases}
    \end{align*} 
    \item $\hash\left( \crs, x \right)$: Compute $( y,\vecU )\gets\mathcal{P}(x)$ and output $y$.
\end{itemize}
\end{construction}

\paragraph{Changes made from the oracle to standard model construction.}
The two first obvious differences between the oracle model Construction~\ref{constr:main} and the standard model Construction~\ref{constr:standard} are that (1) The standard model construction pseudorandomly generates the permutation $\Pi$ and matrices $\matA_y,\vecB_y$ rather than truly at random, and (2) the standard model construction obfuscates the programs rather than putting them in oracles. The third, less obvious change, is that the standard model construction additionally pads $H(x)$ with 0's and then applies $\prp^{-1}\left( k_{\sf out}, \cdot \right)$ in order to get $y \in \{ 0, 1 \}^{d}$, rather than simply setting $y = H(x) \in \{ 0, 1 \}^{r}$, which was done in the oracle construction.

The rationale behind expanding $y$ (the input to $\prf\left( k_{\sf lin}, \cdot \right)$) from $r$ to $d$ bits, comes from the adaptation of the last step of the oracle construction's security proof (Section \ref{subsection:dual_free_to_two_to_one_oracle} in the oracle construction), to the last step of the the standard model construction's security proof (Section \ref{subsection:standard_model_dualfree_to_lwe} in the standard model construction). Specifically, one can recall that in the oracle security proof, in the last step, we showed that collision finding in the dual-free case can be used for collision finding in 2-to-1 random functions. Here, in the standard model, our replacement for 2-to-1 collision resistant hash functions will be LWE-based and approximate, as elaborated in Section \ref{subsection:standard_model_dualfree_to_lwe}, and the expansion of $y$ is intended to deal with that.

Finally, note that our choice to use the inverse $\prp^{-1}\left( k_{\sf out}, \cdot \right)$ in $P$ (i.e., rather than computing the permutation in the forward direction $\prp\left( k_{\sf out}, \cdot \right)$) may seem arbitrary. However, our choice corresponds to us ultimately using the second permutation as an \emph{input}-permutable permutation, which is in turn equivalent to treating the inverse as the permutation, and then asking that it is output-permutable. We accomplish this exactly by using the inverse of $\prp$ as our forward-computation of the permutation. While in Section \ref{sec:prps} we show that output-permutable and input-permutable PRPs are in fact equivalent (and also bilaterally-permutable PRPs), the explicit use of the inverse $\prp^{-1}\left( k_{\sf out}, \cdot \right)$ intends to give insight to the security proof.

\paragraph{Non-collapsing.} The non-collapsing property of our standard model construction is identical to the oracle model non-collapsing procedure, shown in Proposition \ref{proposition:non_collapsing}.

\paragraph{Security in the Standard Model.}
The rest of this section is for dedicated for showing security in the standard model. The results in Sections \ref{section:bloating_dual_standard} (Lemma \ref{lemma:bloating_dual_standard_model}), \ref{section:simulating_dual_standard} (Lemma \ref{lemma:simulating_dual_standard_model}) and \ref{subsection:standard_model_dualfree_to_lwe} (Theorem \ref{thm:reduce2lwe}) together, eventually prove the below theorem. 

\begin{theorem} [Collision Resistance of $\hash$] \label{theorem:standard_main_security}
Let $\gen\left( 1^\secp \right)$ the generation algorithm from Construction \ref{constr:standard}. Then, for every quantum polynomial-time algorithm $\Adv$ there exists a negligible function $\negl$ such that,
\[
\Pr
\left[
\left( x_0 \neq x_1 \right)
\land 
\left( \hash(x_0) = \hash(x_1) \right) 
\;
:
\begin{array}
{rl}
\left( \Ps, \Ps^{-1}, \Ds \right) & \gets \gen\left( 1^\secp \right) \\
\left( x_0, x_1 \right) & \gets \Adv\left( \Ps, \Ps^{-1}, \Ds \right)
\end{array}\right]
\leq
\negl\left( \secp \right)
\enspace .
\]
\end{theorem}

\subsection{Bloating the Dual} \label{section:bloating_dual_standard}
For $n', r', k' \in \Nat$, $s \in \Nat \cup \{ 0 \}$ such that $r' + s \leq n' \leq k'$, we define the modified generator $\widetilde{\gen}\left( 1^{\lambda}, n', r', k', s \right)$, as follows. It samples a distribution over $\Ps,\Ps^{-1},\Ds'$, where the values of $n, r, k$ in Construction \ref{constr:standard} are replaced by $n', r', k'$ respectively. For $s$ we let $\matA^{(0)}(y) \in \bbZ_{2}^{k' \times s}$ denote the first $s$ columns of $\matA(y) \in \bbZ_{2}^{k' \times (n' - r')}$ and $\matA^{(1)}(y) \in \bbZ_{2}^{k' \times (n' - r' - s)}$ denote the remaining $n'-r'-s$ columns. Note that the standard generator is defined as $\widetilde{\gen}\left( 1^{\lambda}, n, r, k, 0 \right)$.
The functionality of $\Ps, \Ps^{-1}$ stays the same and we define $\Ds'\left( y, \vecV \right)$ as an indistinguishability obfuscation of the following function:
\[
D'\left( y, \vecV \right) =
\begin{cases}
1 & \text{ if } \vecV^{T} \cdot \matA^{(1)}(y) = 0^{n'-r'-s} \\
0 & \text{otherwise}
\end{cases}
\]
Observe that if $\vecV^T \cdot \matA(y) = 0^{n-r}$, then $\vecV^T \cdot \matA^{(0)}(y) = 0^{n-r-s}$. Thus $\Ds'$ accepts all points that are accepted by $\Ds$, but also accepts additional points as well, namely those for which $\vecV^T \cdot \matA^{(0)}(y) \neq 0^s$ but $\vecV^T \cdot \matA^{(1)}_y = 0^{n-r-s}$. We call this bloating the dual.

\begin{lemma} \label{lemma:bloating_dual_standard_model}
Let $\secp, n, r, k \in \Nat$, and assume there is a quantum algorithm $\Adv$ with complexity $T_{\Adv}$ such that,
\[
\Pr
\left[
\left( y_0 = y_1 \right) \land \left( x_0 \neq x_1 \right) \; :
\begin{array}{rl}
\left( \Ps, \Ps^{-1}, \Ds \right) & \gets \widetilde{\gen}\left( 1^{\lambda}, n, r, k, 0 \right) \\
\left( x_0, x_1 \right) & \gets \Adv\left( \Ps,\Ps^{-1},\Ds \right) \\
\left( y_b, \vecU_b \right) & \gets \Ps\left( x_b \right)
\end{array}
\right]
\geq
\epsilon \enspace .
\] 
Assume that the primitives used in Construction \ref{constr:standard} are $\left( f(\cdot), \frac{1}{f(\cdot)} \right)$-secure for some sub-exponential $f(\secp) := 2^{ \secp^{\delta} }$ for some constant real number $\delta > 0$. Also, for $w := \secp^{\frac{\delta}{2}}$, $s' := s - (n - r - s)$, assume all of the following:
\begin{enumerate}
    \item
    $\frac{ 2^r \cdot \frac{ k^{2} }{ \epsilon } }{ f\left( \kappa \right) } \leq o(1)$,

    \item 
    $\frac{ \frac{ k^{2} }{ \epsilon } }{ f\left( w \right) } \leq o(1)$,

    \item 
    $\frac{ \frac{ 2^w }{ \epsilon } }{f\left( n - r - s \right)} \leq o(1)$,

    \item 
    $\frac{(n - r - s) \cdot \frac{ 2^w }{ \epsilon }}{2^{s'}} \leq o(1)$, and

    \item 
    $\frac{ \frac{ 2^w }{ \epsilon } \cdot \left( k^5 + \poly\left( n - r - s \right) + T_{\Adv} \right) }{f\left( n - r - s \right)} \leq o(1)$.
\end{enumerate}
Then, it follows that,
\[
\Pr
\left[
\begin{array}{rl}
     & y_0 = y_1 := y , \\
     &\left( \vecU_0 - \vecU_1 \right) \notin \colspan \left( \matA^{(1)}(y) \right)
\end{array}
\;
:
\begin{array}
{rl}
\left( \Ps, \Ps^{-1}, \Ds' \right) & \gets \widetilde{\gen}(1^{\lambda}, n, r, k, s) \\
\left( x_0, x_1 \right) & \gets \Adv\left( \Ps, \Ps^{-1}, \Ds' \right) \\
(y_b, \vecU_b) & \gets \Ps(x_b)
\end{array}
\right]
\geq
\frac{ \epsilon }{ 512 \cdot k^{2} }
\enspace .
\]
\end{lemma}

Note that $\vecU_{0} - \vecU_{1} \notin \colspan\left( \matA^{(1)}(y) \right)$ means in particular that $\vecU_{0}, \vecU_{1}$, and hence $x_0, x_1$, are distinct. Thus, the second expression means that $\Adv$ is finding collisions, but these collisions satisfy an even stronger requirement.

\begin{proof}
Let $\secp \in \Nat$ and assume there is a $T_{\Adv}$-complexity algorithm $\Adv$ and a probability $\epsilon$ such that $\Adv$ gets $\left( \Ps,\Ps^{-1},\Ds \right) \gets \gen\left( 1^{\lambda} \right)$ and outputs a pair $\left( x_{0}, x_{1} \right)$ of $n$-bit strings such that with probability $\epsilon$ we have $x_{0} \neq x_{1}$ and $y_{0} = y_{1}$.
We next define a sequence of hybrid experiments. Each hybrid defines a computational process, an output of the process and a predicate computed on the process output. The predicate defines whether the (hybrid) process execution was successful or not.

\begin{itemize}
    \item $\Hyb_{0}$: The original execution of $\Adv$.
\end{itemize}
The process $\Hyb_{0}$ is the above execution of $\Adv$ on input a sample from the distribution $\left( \Ps, \Ps^{-1}, \Ds \right) \gets \widetilde{\gen}\left( 1^{\lambda}, n, r, k, 0 \right)$. We define the output of the process as $\left( x_0, x_1 \right)$ and the process execution is considered as successful if $x_{0}$, $x_{1}$ are both distinct and collide in their $y$ values. By definition, the success probability of $\Hyb_{0}$ is $\epsilon$.

\begin{itemize}
    \item $\Hyb_{1}$: Preparing to switch to a bounded number of cosets $\left( \matA(y), \vecB(y) \right)$, by using an obfuscated puncturable PRF and injective mode of a lossy function.
\end{itemize}
Let $\left( \LFGen, \LFF \right)$ a $\left( f(\cdot), \frac{1}{f(\cdot)} \right)$-secure lossy function scheme (as in Definition \ref{definition:lf}). Set $w := \secp^{\frac{\delta}{2}}$ where $f(\secp) := 2^{\secp^{\delta}}$. Sample $\pk_{\LF} \gets \LFGen\left( 1^{d}, 0, 1^{w} \right)$ and let $\LFF\left( \pk_{\LF}, \cdot \right) : \{ 0, 1 \}^{d} \rightarrow \{ 0, 1 \}^{m}$ the induced injective function.

We now consider two circuits in order to describe our current hybrid. $E_{0}\left( k_{\sf lin}, \cdot \right)$ is the circuit that given an input from $\{ 0, 1 \}^{d}$ applies $\prf\left( k_{\sf lin }, \cdot \right)$ to get $\matA(y)$, $\vecB(y)$. $E_{1}\left( \pk_{\LF}, k'_{ \sf lin }, \cdot \right)$ is the circuit that for a $\LF$ key $\pk_{\LF}$ and P-PRF key $k'_{ \sf lin }$ for a P-PRF with input size $m$ rather than $d$, given a $d$-bit input $y$, applies the lossy function with key $\pk_{\LF}$ and then the P-PRF to get $\matA'(y)$, $\vecB'(y)$.

Note that in the previous hybrid, the circuit $E_{0}$ is used in all three circuits $P$, $P^{-1}$, $D$ in order to generate the cosets per input $y$, and furthermore, each of these three circuits access $E_{0}$ only as a black box. The change that we make to the current hybrid is that we are going to use $E_{1}\left( \pk_{\LF}, k'_{ \sf lin }, \cdot \right)$ for a freshly sampled $\pk_{\LF}$, $k'_{ \sf lin }$, instead of $E_{0}$. Since we are sending obfuscations $\left( \Ps, \Ps^{-1}, \Ds \right)$ of the three circuits and due to the three circuits accessing the samplers $E_{0}$, $E_{1}$ only as black boxes, it follows by Lemma \ref{lem:distswap} that the output of the previous hybrid and the current hybrid are $\left( f(\kappa), \frac{|\Xs|}{f\left( \kappa \right)} \right)$-indistinguishable, where $\Xs$ is the set of al possible values of $y$, which has size $2^r$ (recall that while the $y's$ are of length $d >> r$, the are a sparse set inside $\{ 0, 1 \}^{d}$ because are padding with zeros), and we also recall that $\kappa$, the cryptographic security parameter is some polynomial in the statistical security parameter $\secp$. It follows in particular the success probability of the current process is $:= \epsilon_{1} \geq \epsilon - \frac{2^r}{f\left( \kappa \right)} \geq \epsilon - \frac{\epsilon}{32} = \frac{31\cdot \epsilon}{32}$.

\begin{itemize}
    \item $\Hyb_{2}$: Switching to a bounded number of cosets $\left( \matA(y), \vecB(y) \right)$, by using the lossy function.
\end{itemize}
The change from the previous hybrid to the current hybrid is that we are going to sample a lossy key $\pk_{\LF}^{1} \gets \LFGen\left( 1^{d}, 1, 1^{w} \right)$ and use it inside $E_{1}$ from the previous hybrid, instead of using an injective key $\pk_{\LF}^{0} \gets \LFGen\left( 1^{d}, 0, 1^{w} \right)$, which was used in the previous hybrid. Note that in this hybrid, there are at most $2^w$ cosets (that is, some different values of $y$ will have the same coset), by the correctness of the lossy function scheme. By the security of the lossy function scheme, the output of this hybrid is $\left( f\left( w \right), \frac{1}{ f\left( w \right) } \right)$-indistinguishable from the previous hybrid. It follows in particular the success probability of the current process is
$$
:= \epsilon_{2}
\geq \epsilon_{1} - \frac{1}{ f\left( w \right) }
\geq \frac{31\cdot \epsilon}{32} - \frac{1}{ f\left( w \right) }
\geq \frac{30\cdot \epsilon}{32}
= \frac{15 \cdot \epsilon}{16} \enspace .
$$

\begin{itemize}
    \item $\Hyb_{3}$: Using obfuscated (instead of plain) circuits for membership checks inside $D$, by using the security of the obfuscator that obfuscates $D$.
\end{itemize}
Recall that in the previous hybrid, the circuit $D$ executes as follows: $D\left( y, \vecV \right)$ computes $\matA(y)$ by access to $\prf\left( k_{\sf lin},\cdot \right)$, and then checks membership in the dual of $\colspan\left( \matA(y) \right) := S_{y}$. In the current hybrid we make the following change to $D$: We sample an additional P-PRF key $k_{S}$ at the beginning of the hybrid. Given $\matA(y)$, we apply $\prf\left( k_{S}, \cdot \right)$ to obtain pseudorandomness and generate an obfuscation $\Obf_{S_{y}^{\bot}} \gets \iO\left( 1^{\kappa}, S_{y}^{\bot} \right)$ of the circuit that checks membership inside $S_{y}^{\bot}$. The circuit $D$ decides on the membership check of $\vecV$ using the obfuscated circuit $\Obf_{S_{y}^{\bot}}$ instead of the plain circuit $S_{y}^{\bot}$.

Note that by the correctness of the obfuscation scheme (specifically, the inner obfuscation scheme that was used to obfuscate the circuit $S^{\bot}_{y}$, for every $y$), we did not change the functionality of $D$. It follows by the security of the iO that obfuscates $D$, that the obfuscations between the two cases are $\left( f\left( \kappa \right), \frac{1}{f(\kappa)} \right)$-indistinguishable, and thus the same can be said of the outputs of the hybrids. It follows in particular the success probability of the current process is
$$
:= \epsilon_{3}
\geq \epsilon_{2} - \frac{1}{ f\left( \kappa \right) }
\geq \frac{15 \cdot \epsilon}{16} - \frac{\epsilon}{32}
= \frac{29 \cdot \epsilon}{32} \enspace .
$$

\begin{itemize}
    \item $\Hyb_{4}$: Relaxing dual verification inside $D$ to accept a larger subspace $T^{\bot}_{y}$ for every $y$, by using an obfuscated puncturable PRF over subspace hiding.
\end{itemize}
We now consider two circuits in order to describe our current hybrid, both of which are used only inside the circuit $D$. The first circuit $E_{S}\left( k_{S}, \cdot \right)$ is the circuit that given an input $\matA(y)$ applies $\prf\left( k_{S}, \cdot \right)$ to obtain pseudorandomness and generate an obfuscation $\Obf_{S_{y}^{\bot}} \gets \iO\left( 1^{\kappa}, S_{y}^{\bot} \right)$. The second circuit $E_{T}\left( k_{T}, \cdot \right)$ is the circuit that given an input $\matA(y)$ applies $\prf\left( k_{T}, \cdot \right)$ to obtain pseudorandomness for two things: (1) to sample $T_{y}^{\bot} \subseteq \{ 0, 1 \}^{k}$ a superspace that contains $S^{\bot}_{y}$ and has $s$ more dimensions (recall that $S_{y}$ has $n - r$ dimensions, $S^{\bot}_{y}$ has $k - (n - r)$ dimensions and thus $T_{y}^{\bot}$ has $k - (n - r) + s$ dimensions), and (2) to generate an obfuscation $\Obf_{T_{y}^{\bot}} \gets \iO\left( 1^{\kappa}, T_{y}^{\bot} \right)$. 

For concreteness, the way we use the randomness from $\prf\left( k_{T}, \cdot \right)$ is by generating a pseudorandom full rank matrix $M_{y} \in \bbZ_{2}^{(n - r) \times (n - r)}$, multiplying by $\matA(y) \in \bbZ_{2}^{k \times (n - r)}$ to get $\overline{\matA}(y) := \matA(y) \cdot M_{y} \in \bbZ_{2}^{k \times (n - r)}$. We then take the last $n - r - s$ columns of $\overline{\matA}(y)$, denote this sub-matrix by $\overline{\matA}^{(1)} \in \bbZ_{2}^{k \times (n - r - s)}$ and define $\colspan\left( \overline{\matA}^{(1)} \right)^{\bot} := T_{y}^{\bot}$.

Let $\Xs$ the set of all possible cosets that arise from the scheme, which has size $\leq 2^w$ by the lossy function. We would first like to note the indistinguishability, per coset $i \in \Xs$, between $\Obf_{S_{i}^{\bot}}$ and $\Obf_{T_{i}^{\bot}}$, for a known $S_{i}$ and uniformly random appropriate superspace $T_{i}$. Specifically, we look at the indistinguishability for truly random bits for sampling $T$ and for sampling the obfuscation in either of the cases. For this, we would like to use Lemma \ref{lemma:subspace_hiding}. We would like to make sure the lemma's requirements are met so we note the dimensions with $'$. Note that in our case, $k$, $k - n + r$ and $s$, take the role of $k'$, $r'$ and $s'$ respectively, and thus we get $\left( f\left( n - r - s \right), \frac{1}{f\left( n - r - s \right)} \right)$-indistinguishability. 

Note that in the previous hybrid, the circuit $E_{S}$ is used in $D$, and furthermore, the access of $D$ to $E_{S}$ is only as a black box. The change that we make to the current hybrid is that we are going to use $E_{T}\left( k_{ T }, \cdot \right)$ for a freshly sampled P-PRF key $k_{T}$ instead of $E_{S}$ (which is also sampled for a freshly sampled P-PRF key $k_{S}$). Since we are sending an obfuscation $\Ds$ of $D$ and due to $D$ accessing the samplers $E_{S}$, $E_{T}$ only as black boxes, it follows by Lemma \ref{lem:distswap} that the output of the previous hybrid and the current hybrid are $\left( f\left( n - r - s \right), \frac{|\Xs|}{f\left( n - r - s \right)} \right)$-indistinguishable, where $\Xs$ is the set of all cosets, which has size $\leq 2^w$ by the lossy function. It follows in particular that the success probability of the current process is
$$
:= \epsilon_{4}
\geq \epsilon_{3} - \frac{2^w}{f\left( n - r - s \right)}
\geq \frac{29 \cdot \epsilon}{32} - \frac{\epsilon}{32}
= \frac{28\cdot \epsilon}{32}
\enspace .
$$

\begin{itemize}
    \item $\Hyb_{5}$: Asking for sum of collisions to be outside of $T_{y}$, by an obfuscated puncturable PRF over dual-subspace anti-concentration.
\end{itemize}
This hybrid is the same as the previous in terms of execution, but we change the definition of a successful execution, that is, we change the predicate computed on the output of the process. We still ask that $\left( y_0 = y_1 := y \right)$, but instead of only asking the second requirement to be $\left( x_0 \neq x_1 \right)$, we ask for a stronger condition: $\left( \vecU_{0} - \vecU_{1} \right) \in \left( S_{y} \setminus T_{y} \right)$. Note that we are not going to need to be able to efficiently check for the success of the condition, but we'll prove that it happens with a good probability nonetheless.

Let $\epsilon_{5}$ be the success probability of the current hybrid and note that $\Adv$ finds collisions with probability $\epsilon_{4}$ in the previous hybrid $\Hyb_{4}$ (and since this hybrid is no different, the same goes for the current hybrid). Let $\Xs \subseteq \{ 0, 1 \}^{m}$ the image of the lossy function $\LFF\left( \pk_{\LF}, \cdot \right)$ which we use to map our images $y$ to cosets $\left( \matA(y), \vecB(y) \right)$, that is, there are $|\Xs|$ cosets and by the lossyness we know that $|\Xs| \leq 2^w$. For every value $\vecX \in \Xs$ denote by $\epsilon_{4}^{(\vecX)}$ the probability to find a collision on value $\vecX$, or formally, that $y_{0} = y_{1} := y$, $x_{0} \neq x_{1}$ and $\vecX = \LFF\left( \pk_{\LF}, y \right)$. We deduce $\sum_{ \vecX \in \Xs } \epsilon_{4}^{(\vecX)} = \epsilon_{4}$. Let $L$ be a subset of $\Xs$ such that $\epsilon_{4}^{(\vecX)} \geq \frac{\epsilon_{4}}{2 \cdot |\Xs|}$ and note that $\sum_{\vecX \in L} \epsilon_{4}^{(\vecX)} \geq \frac{ \epsilon_{4} }{ 2 }$. We further define $\epsilon_{5}^{(\vecX)}$ as the probability to find a strong (as in the notion of $\epsilon_{5}$) collision on value $\vecX$, or formally, that $y_{0} = y_{1} := y$, $\left( \vecU_{0} - \vecU_{1} \right) \in \left( S_{y} \setminus T_{y} \right)$ and $\vecX = \LFF\left( \pk_{\LF}, y \right)$. Note that $S_{y}$, $T_{y}$ are really functions of $\vecX$ rather than of $y$, so $\left( \vecU_{0} - \vecU_{1} \right) \in \left( S_{\vecX} \setminus T_{\vecX} \right)$ and also observe that $\sum_{\vecX \in \Xs} \epsilon_{5}^{(\vecX)} = \epsilon_{5}$.

We would now like to use Lemma \ref{lemma:dual_subspace_concentration}, so we make sure that we satisfy its requirements. Let any $\vecX \in L$, we know that by definition $\epsilon_{4}^{(\vecX)} \geq \frac{\epsilon_{4}}{2 \cdot |\Xs|}$ and also recall that $\epsilon_{4} \geq \frac{28\cdot \epsilon}{32}$, $|\Xs| \leq 2^w$ and thus 
$$
\epsilon_{4}^{(\vecX)}
\geq
\frac{\epsilon_{4}}{2 \cdot |\Xs|}
\geq
\frac{28\cdot \epsilon}{64} \cdot \frac{1}{2^w}
\geq 
\Omega
\left(
\frac{ \epsilon }{ 2^w }
\right)
\enspace .
$$
Let $s' := s - (n - r - s)$ and for any $\vecX \in L$ let $\ell_{\vecX} := \frac{k^2}{\epsilon_{4}^{(\vecX)}} \leq O\left( \frac{ k^2 \cdot 2^w }{ \epsilon } \right)$. Note that by our Lemma \ref{lemma:bloating_dual_standard_model} statement's assumptions, we have (1) $\frac{(n - r - s) \cdot \frac{1}{\epsilon_{4}^{(\vecX)}}}{2^{s'}} \leq o(1)$ and (2) $\frac{ \ell_{\vecX} \cdot \left( k^3 + \poly\left( n - r - s \right) + T_{\Adv} \right) }{f\left( n - r - s \right)} \leq o(1)$. Since this satisfies Lemma \ref{lemma:dual_subspace_concentration}, it follows that for every $\vecX \in L$ we have $\epsilon_{5}^{(\vecX)} \geq \frac{ \epsilon_{4}^{(\vecX)} }{ 16 \cdot k^{2} } - \frac{1}{f(\kappa)}$, because for each $\vecX \in \Xs$, in order to use Lemma \ref{lemma:dual_subspace_concentration}, we need the randomness for the experiment to be genuinely random, which will necessitate us to invoke the security of the iO and puncturable PRF, which incurs the loss of $\frac{1}{f(\kappa)}$. It follows that
$$
\epsilon_{5}
=
\sum_{\vecX \in \Xs} \epsilon_{5}^{(\vecX)}
\geq 
\sum_{\vecX \in L} \epsilon_{5}^{(\vecX)}
\geq
\sum_{\vecX \in L} \frac{ \epsilon_{4}^{(\vecX)} }{ 16 \cdot k^{2} }
-
\frac{|\Xs|}{f(\kappa)}
\geq
\frac{ \left( \frac{ \epsilon_{4} }{ 2 } \right) }{ 16 \cdot k^{2} }
-
\frac{|\Xs|}{f(\kappa)}
\geq
\frac{ \left( \frac{ 28 \cdot \epsilon }{ 64 } \right) }{ 16 \cdot k^{2} }
-
\frac{|\Xs|}{f(\kappa)}
\geq
\frac{ \epsilon }{ 64 \cdot k^{2} }
\enspace .
$$

\begin{itemize}
    \item $\Hyb_{6}$: For every $y$, de-randomizing $T_{y}$ and defining it as the column span of $\matA(y)^{(1)} \in \bbZ_{2}^{k \times \left( n - r - s \right)}$, the last $n - r - s$ columns of the matrix $\matA(y)$, by using permutable PRPs, security of iO and an obfuscated puncturable PRF.
\end{itemize}
This hybrid is the same as the previous, with the following change. Recall how we compute dual membership check inside the circuit $D$: Given $y$, we compute $\vecX = \LFF\left( \pk_{\LF}, y \right)$ we take the P-PRF key $k_{\sf in}$ and compute $\matA(\vecX) = \matA(y)$, then use the additional P-PRF key $k_{T}$ to generate the random superspace $T_{y}^{\bot}$ of $S_{y}^{\bot}$ and additional randomness, and use both to generate an obfuscation $\Obf_{ T_{y}^{\bot} } \gets \iO\left( 1^{\kappa}, T_{y}^{\bot} \right)$. Also recall that the subspace $T_{y}^{\bot}$ is generated as follows: Generate the matrix $M_{y}$ using the P-PRF, then multiply it by $\matA(y)$ to get $\overline{\matA}(y)$, and then take the last $n - r - s$ columns of it as a basis for $T_{y}$, and $T_{y}^{\bot}$ is defined as the dual of that space.
In the current hybrid we will not sample $M_{y}$ and simply define $\overline{\matA}_{y} := \matA_{y}$. This means that $T_{y}$ is defined to be the column span of $\matA(y)^{(1)}$. We will also not obfuscate the membership check for $T_{y}^{\bot}$ inside the circuit $D$, and check membership by using a basis of $T_{y}^{\bot}$ in the plain. We next define hybrid experiments (which go in the direction from the previous hybrid $\Hyb_{5}$ to the current hybrid $\Hyb_{6}$) and explain why each consecutive pair are appropriately indistinguishable. 

\begin{itemize}
    \item
    \textbf{Using the security of the permutable PRP.}
    Assume we sample all components of our scheme, excluding the key $k_{\sf in}$ for the initial permutation $\Pi$ on $\{ 0, 1 \}^{n}$, and define the following permutation $\Gamma$ on $\{ 0, 1 \}^{n}$. Denote by $h$ the first $r$ bits of the input and by $j$ the last $n - r$ bits of the input to the permutation $\Gamma$. Recall that in the circuits $P, P^{-1}, D$, the value $y$, the coset $\matA(y), \vecB(y)$, the superspace $T_{y}^{\bot}$ and the associated matrix $M_{y} \in \bbZ_{2}^{ (n - r) \times (n - r) }$, are all computed as a function of $r$ bits, which in the construction take the role of $H(x)$. In fact, all of the above variables can be written as a function of $h \in \{ 0, 1 \}^{r}$ instead of as a function of $y$. The permutation $\Gamma$ takes $h$ and computes $M_{h}$ (the matrix for computing $T_{h}$ from the matrix $\matA(h)$), interprets $j$ as a vector in $\bbZ_{2}^{n - r}$, and applies $M_{h}$ to $j$. Note that this change means we are taking the puncturable PRF key $k_{T}$ which is used to sample $T_{y}^{\bot}$ (and formally, is used to sample the pseudorandomness for generating the matrix $M_{y}$ and for obfuscating membership check for $T_{y}^{\bot}$), which previously only existed inside the circuit $D$, and putting it also inside the circuits $P$, $P^{-1}$, because this key is needed in order to compute the permutation $\Gamma$.

    Recall the examples of decomposable permutations in Section \ref{sec:prps}. For every value $h \in \{ 0, 1 \}^{r}$ observe that multiplication by $M_{h}$ is a permutation (and moreover an affine permutation, which is decomposable efficiently).
    Then, $\Gamma$ is a controlled permutation as described in the examples of decomposable permutations in Section \ref{sec:prps}, which is controlled on decomposable permutations. Overall, we deduce that $\Gamma$ is a $\left( 2^{\poly(\secp)}, \poly(\secp) \right)$-decomposable permutation. We use the permutable PRP $\Pi$, and switch to a setting where we use the key $k_{\sf in}^{\Gamma}$ that applies $\Gamma$ to the output of $\Pi$ (and $\Gamma^{-1}$ to the input of $\Pi$), instead of just applying $\Pi$ and its inverse. By the security of the permutable PRP (Lemma \ref{lemma:op_prp_compose}), this change is $\left( f(\kappa), \frac{1}{f(\kappa)} \right)$-indistinguishable. 

    \item
    \textbf{Using the security of outside iO.}
    For every $h \in \bbZ_{2}^{r}$, we make the following change to $P, P^{-1}$. Instead of composing the permutation $\Gamma$ to the output of $\Pi$, it applies $\Pi$ as it is. However, when generating the matrix $\matA(y)$, it multiplies by $M_{y}$ to get $\overline{\matA}\left( y \right) := \matA\left( y \right) \cdot M_{y}$. An additional change we will make not to $P, P^{-1}$ but to the circuit $D$, is that we will not obfuscate the membership check circuit for $T_{y}^{\bot}$ and simply use its available basis. One can observe that we did not change the functionality of $P, P^{-1}, D$ in any of the above changes and thus by the security of the indistinguishability obfuscation that obfuscates $\Ps \gets \iO\left( 1^{\kappa}, P \right)$, $\Ps^{-1} \gets \iO\left( 1^{\kappa}, P^{-1} \right)$, $\Ds \gets \iO\left( 1^{\kappa}, D \right)$, this change is $\left( f(\kappa), \frac{1}{f(\kappa)} \right)$-indistinguishable. 

    \item
    \textbf{Using an obfuscated puncturable PRF argument over the choice of the matrix $\matA(y)$, for every $y$.}
    Note that after we did the last step, our process for generating the cosets, in all tree circuits $P, P^{-1}, D$ is the following: Given $y$ we compute $\vecX$ with the lossy function, and then apply two different puncturable PRFs (that is, with i.i.d keys, $k_{\sf lin}$ and $k_{T}$) to obtain $\left( \matA(\vecX), \vecB(\vecX) \right)$ and the matrix $M_{\vecX} \in \bbZ_{2}^{ (n - r) \times (n - r) }$. We then multiply to get the matrix that we are actually using, i.e., $\overline{\matA}(\vecX) := \matA(\vecX) \cdot M_{\vecX}$. The change we next make is to have one fresh puncturable PRF key $k'_{\sf lin}$ and use it to generate the coset $\left( \matA(\vecX), \vecB(\vecX) \right)$, without further generating the matrix $M_{\vecX}$.

    By a standard argument using an obfuscated puncturable PRF like we used numerous times in this proof (i.e., we use lemma \ref{lem:distswap} and the fact that when using real randomness, the two ways to sample $\matA(y)$ are statistically equivalent), we get that this change is $\left( f(\kappa), 2^w \cdot \frac{1}{f(\kappa)} \right)$-indistinguishable.
\end{itemize}

Observe that after the last change we are exactly in the setting of $\Hyb_{6}$, and we started in $\Hyb_{5}$. It follows in particular that the success probability of the current process is
$$
:= \epsilon_{6}
\geq \epsilon_{5} - \frac{2^w}{f\left( \kappa \right)}
\geq \frac{ \epsilon }{ 64 \cdot k^{2} } - \frac{ \epsilon }{ 128 \cdot k^{2} }
= \frac{ \epsilon }{ 128 \cdot k^{2} }
\enspace .
$$

\begin{itemize}
    \item $\Hyb_{7}$: Going back to using $2^r$ cosets rather than $\leq 2^w$, by moving from lossy mode to injective mode in the lossy function.
\end{itemize}
We make the exact same change we made between $\Hyb_{1}$ to $\Hyb_{2}$, but in the opposite direction. That is we sample an injective key $\pk_{\LF}^{0} \gets \LFGen\left( 1^{d}, 0, 1^{w} \right)$ and use it instead of the previous lossy key $\pk_{\LF}^{1} \gets \LFGen\left( 1^{d}, 1, 1^{w} \right)$, which is used in the previous hybrid. By the exact same argument (which relies on the security of the lossy function), the output of this hybrid is $\left( f(w), \frac{1}{f(w)} \right)$-indistinguishable. It follows in particular the success probability of the current process is
$$
:= \epsilon_{7}
\geq \epsilon_{6} - \frac{1}{f\left( w \right)}
\geq \frac{ \epsilon }{ 128 \cdot k^{2} } - \frac{ \epsilon }{ 256 \cdot k^{2} }
= \frac{ \epsilon }{ 256 \cdot k^{2} }
\enspace .
$$

\begin{itemize}
    \item $\Hyb_{8}$: Stop using the lossy function, by an obfuscated puncturable PRF.
\end{itemize}
We make the exact same change we made between $\Hyb_{0}$ to $\Hyb_{1}$, but in the opposite direction. That is, we drop the lossy function $\LFF$ altogether and apply the PRF to $y$ directly and not to $\vecX$, the output of the lossy function on input $y$. By the exact same argument (which relies on Lemma \ref{lem:distswap}), the output of this hybrid is $\left( f(\kappa), \frac{|\Xs|}{f\left( \kappa \right)} \right)$-indistinguishable, where $\Xs$ is the set of al possible values of $y$, which has size $2^r$ (recall that while the $y's$ are of length $d >> r$, the are a sparse set inside $\{ 0, 1 \}^{d}$ because are padding with zeros). It follows in particular the success probability of the current process is
$$
:= \epsilon_{8}
\geq \epsilon_{7} - \frac{2^r}{f\left( \kappa \right)}
\geq \frac{ \epsilon }{ 256 \cdot k^{2} } - \frac{ \epsilon }{ 512 \cdot k^{2} }
= \frac{ \epsilon }{ 512 \cdot k^{2} }
\enspace .
$$

To conclude, note that the generated obfuscations in the final hybrid $\Hyb_{8}$ form exactly the distribution $\left( \Ps, \Ps^{-1}, \Ds' \right)\gets \widetilde{\gen}\left( 1^{\secp}, n, r, k, s \right)$. This finishes our proof.
\end{proof}

\subsection{Simulating the Dual} \label{section:simulating_dual_standard}
Our next step is to show that an adversary which has access to the dual-free setting can simulate the CRS for an adversary in the restricted setting, where the dual verification check is bloated.

\begin{lemma} \label{lemma:simulating_dual_standard_model}
Let $\secp, n, r, k \in \Nat$ and assume there is a quantum algorithm $\Adv$ running in time $T_{\Adv}$ such that,
\[
\Pr
\left[
\begin{array}{rl}
     & y_0 = y_1 := y , \\
     &\left( \vecU_0 - \vecU_1 \right) \notin \colspan \left( \matA^{(1)}(y) \right)
\end{array}
\;
:
\begin{array}
{rl}
\left( \Ps, \Ps^{-1}, \Ds' \right) & \gets \widetilde{\gen}(1^{\lambda}, n, r, k, s) \\
\left( x_0, x_1 \right) & \gets \Adv\left( \Ps, \Ps^{-1}, \Ds' \right) \\
(y_b, \vecU_b) & \gets \Ps(x_b)
\end{array}
\right]
\geq
\epsilon \enspace .
\]
Assume that the primitives used in Construction \ref{constr:standard} are $\left( f(\cdot), \frac{1}{f(\cdot)} \right)$-secure, and assume $\frac{ \frac{2^r}{\epsilon} }{f\left( \kappa \right)} \leq o(1)$. Then, there is a quantum algorithm $\AdvB$ running in time $T_{\Adv} + \poly\left( \secp \right)$ such that
\[
\Pr
\left[
\left( \overline{y}_{0} = \overline{y}_{1} \right) \land \left( \overline{x}_{0} \neq \overline{x}_{1} \right) \; :
\begin{array}{rl}
\left( \overline{\Ps}, \overline{\Ps}^{-1}, \overline{\Ds} \right) & \gets \widetilde{\gen}\left( 1^{\lambda}, r + s, r, k - (n - r - s), 0 \right) \\
\left( \overline{x}_{0}, \overline{x}_{1} \right) & \gets \AdvB\left( \overline{\Ps}, \overline{\Ps}^{-1} \right) \\
\left( \overline{y}_{b}, \overline{\vecU}_b \right) & \gets \overline{\Ps}(x_b)
\end{array}
\right]
\geq
\frac{\epsilon}{2}
\enspace .
\] 
\end{lemma}

\begin{proof}
We first describe the actions of the algorithm $\AdvB$ (which will use the code of $\Adv$ as part of its machinery) and then argue why it breaks collision resistance with the appropriate probability. Given $\overline{\Ps}, \overline{\Ps}^{-1}$ which comes from $\left( \overline{\Ps}, \overline{\Ps}^{-1}, \overline{\Ds} \right) \gets \widetilde{\gen}\left( 1^{\lambda}, r + s, r, k - (n - r - s), 0 \right)$, the algorithm $\AdvB$ does the following:
\begin{itemize}
    \item
    Sample a P-PRF key $k_{\matC}$ that outputs some sufficient (polynomial) amount of random bits on an $d$-bit input, and sample a permutable PRP key $k_{\Gamma}$ for a PRP on domain $\{ 0, 1 \}^{n}$. Define the following circuits.

    \item 
    $\left( \; y \in \bbZ_{2}^{d}, \; \vecU \in \bbZ_{2}^{k} \; \right) \gets P\left( x \in \bbZ_{2}^{n} \right)$:
    \begin{itemize}
        \item
        $\left(
        \; \overline{x} \in \bbZ_{2}^{r + s},
        \; \widetilde{x} \in \bbZ_{2}^{n - r - s} \;
        \right) \gets \Pi\left( k_{\Gamma}, x \right)$.

        \item 
        $\left(
        \; y \in \bbZ_{2}^{d},
        \; \overline{\vecU} \in \bbZ_{2}^{k - (n - r - s)} \;
        \right) \gets \overline{\Ps}(\overline{x})$.

        \item 
        $\left(
        \; \matC(y) \in \bbZ_{2}^{ k \times k },
        \; \vecD(y) \in \bbZ_{2}^{ n - r - s }
        \right)
        \gets
        \prf\left( k_{\matC}, y \right)$.

        \item 
        $\vecU
        \gets
        \matC(y) \cdot \left( \begin{array}{c} \overline{\vecU} \\ \widetilde{x} + \vecD(y) \end{array} \right)$.
    \end{itemize}

    \item 
    $\left( \; x \in \bbZ_{2}^{n} \; \right)
    \gets
    P^{-1}\left( \; y \in \bbZ_{2}^{d}, \; \vecU \in \bbZ_{2}^{k} \; \right)$:
    \begin{itemize}
        \item 
        $\left(
        \; \matC(y) \in \bbZ_{2}^{ k \times k },
        \; \vecD(y) \in \bbZ_{2}^{n - r - s}
        \right) \gets \prf\left( k_{\matC}, y \right)$.

        \item 
        $\left( \begin{array}{c} \overline{\vecU} \\ \widetilde{x} \end{array} \right) \gets \matC(y)^{-1} \cdot \vecU - \left( \begin{array}{c} 0^{k-(n - r - s)} \\ \vecD(y) \end{array} \right)$.

        \item 
        $\left( \; \overline{x} \in \bbZ_{2}^{r + s} \; \right)
        \gets
        \overline{\Ps}^{-1}\left( y, \overline{\vecU} \right)$.

        \item
        $x
        \gets
        \Pi^{-1}\left( k_{\Gamma}, \left( \overline{x}, \widetilde{x} \right) \right)$.
    \end{itemize}

    \item 
    $D'\left( \; y \in \bbZ_{2}^{d}, \; \vecV \in \bbZ_{2}^{k} \; \right) \in \{ 0, 1 \}$:
    \begin{itemize}
        \item 
        $\left(
        \; \matC(y) \in \bbZ_{2}^{ k \times k },
        \; \vecD(y) \in \bbZ_{2}^{n - r - s}
        \right)
        \gets
        \prf\left( k_{\matC}, y \right)$.

        \item 
        $\matA^{(1)}(y) := $ last $n - r - s$ columns of $\matC(y)$. 

        \item 
        Output $1$ iff $\vecV^{T} \cdot \matA^{(1)}(y) = \mathbf{0}^{n - r - s}$.
    \end{itemize}

    \item 
    Use indistinguishability obfuscation in order to generate the input for $\Adv$: $\Ps \gets \iO\left( 1^{\kappa}, P \right)$, $\Ps^{-1} \gets \iO\left( 1^{\kappa}, P^{-1} \right)$, $\Ds' \gets \iO\left( 1^{\kappa}, D' \right)$.
\end{itemize}

The remainder of the reduction is simple: $\AdvB$ executes $\left( x_0, x_1 \right) \gets \Adv\left( \Ps, \Ps^{-1}, \Ds' \right)$ and then $\left( \overline{x}_{b}, \widetilde{x}_{b} \right) \gets \Pi\left( k_{\Gamma}, x_{b} \right)$ and outputs $\left( \overline{x}_{0}, \overline{x}_{1} \right)$. Assume that the output of $\Adv$ satisfies $y_{0} = y_{1} := y$ and also $\left( \vecU_0 - \vecU_1 \right) \notin \colspan \left( \matA(y)^{(1)} \right)$, and recall that $\matA(y)^{(1)} \in \bbZ_{2}^{k \times (n - r - s)}$ are the last $n - r - s$ columns of the matrix $\matA(y) \in \bbZ_{2}^{k \times (n - r)}$, which is generated by the reduction. We explain why it is necessarily the case that $\overline{x}_{0} \neq \overline{x}_{1}$.

First note that due to how we defined the reduction, $\matA(y) := \matC(y) \cdot \left( \begin{array}{cc} \overline{\matA}(y) & \\ & \Id_{n - r - s} \end{array} \right)$, where $\overline{\matA}(y) \in \bbZ_{2}^{(k - (n - r - s)) \times s}$ is the matrix arising from $\overline{\Ps}, \overline{\Ps}^{-1}$ and $\Id_{n - r - s} \in \bbZ_{2}^{(n - r - s) \times (n - r - s)}$ is the identity matrix of dimension $n - r - s$. Also note that because $\matC(y)$, $\overline{\matA}(y)$ are full rank then $\matA(y)$ is full rank. Now, since $\left( \vecU_0 - \vecU_1 \right) \notin \colspan \left( \matA(y)^{(1)} \right)$ and since $\matA(y)^{(1)}$ are the last $n - r - s$ columns of $\matA(y)$, it follows that if we consider the coordinates vector $\vecX \in \bbZ_{2}^{n - r}$ of $\left( \vecU_0 - \vecU_1 \right)$ with respect to $\matA(y)$, the first $s$ elements are not $0^{s}$. By linearity of matrix multiplication it follows that if we look at each of the two coordinates vectors $\vecX_{0}$, $\vecX_{1}$ (each has $n - r$ bits) for $\vecU_{0}$, $\vecU_{1}$, respectively, somewhere in the first $s$ bits, they differ.
Now, recall how we obtain the first $s$ bits of $\vecX_{b}$ -- this is exactly by applying $\overline{\Pi}$ (the permutation on $\{ 0, 1 \}^{r + s}$ arising from $\overline{\Ps}, \overline{\Ps}^{-1}$) to $\overline{x}_{b}$ and taking the last $s$ bits of the output. Since these bits differ in the output of the permutation, then the preimages have to differ, i.e., $\overline{x}_{0} \neq \overline{x}_{1}$.

Define $\epsilon_{\AdvB}$ as the probability that the output of $\Adv$ indeed satisfies $y_{0} = y_{1} := y$ and also $\left( \vecU_0 - \vecU_1 \right) \notin \colspan \left( \matA^{(1)}(y) \right)$, and it remains to give a lower bound for the probability $\epsilon_{\AdvB}$. We do this by a sequence of hybrids, eventually showing that the view which $\AdvB$ simulates to $\Adv$ is computationally indistinguishable from a sample from $\widetilde{\gen}\left( 1^{\lambda}, n, r, k, s \right)$. More precisely, each hybrid describes a process, it has an output, and a success predicate on the output.

\begin{itemize}
    \item $\Hyb_{0}$: The above distribution $\left( \Ps, \Ps^{-1}, \Ds' \right) \gets \AdvB\left( \overline{\Ps}, \overline{\Ps}^{-1} \right)$, simulated to the algorithm $\Adv$.
\end{itemize}
The first distribution is defined in the reduction above. The output of the process is the output $(x_{0}, x_{1})$ of $\Adv$. The process execution is considered as successful if $y_{0} = y_{1} := y$ and $\left( \vecU_{0} - \vecU_{1} \right) \notin \colspan\left( \matA(y)^{(1)} \right)$.

\begin{itemize}
    \item $\Hyb_{1}$: Not applying the inner permutation $\overline{\prp}_{\sf in}$ (which comes from the circuits $\overline{\Ps}$, $\overline{\Ps}^{-1}$), by using the security of an obfuscated permutable PRP.
\end{itemize}
Let $\overline{\prp}_{\sf in}$ the (first) permutable PRP that's inside $\overline{\Ps}$ (which is the obfuscation of the circuit $\overline{P}$). In the previous hybrid we apply the $n$-bit permutable PRP $\prp\left( k_{\Gamma}, \cdot \right)$ to the input $x \in \bbZ_{2}^n$ and then proceed to apply the inner permutation $\overline{\prp}_{\sf in}\left( \overline{k}_{\sf in}, \cdot \right)$ to the first (i.e. leftmost) $r + s$ output bits of the first permutation $\prp\left( k_{\Gamma}, \cdot \right)$. The change we make to the current hybrid is that we simply apply only $\prp\left( k_{\Gamma}, \cdot \right)$.

Recall two details: (1) By Theorem \ref{thm:op-prp}, the inner permutable PRP $\overline{\prp}_{\sf in }\left( \overline{k}_{\sf in}, \cdot \right)$ is in and of itself $\left( 2^{\poly(\secp)}, \poly(\secp) \right)$-decomposable, and (2) the circuits $P$, $P^{-1}$ which apply the permutations are both obfuscated by iO to be generate the obfuscations $\Ps$, $\Ps^{-1}$. We can treat it as a fixed permutation that acts on the output of the permutation $\Pi\left( k_{\Gamma}, \cdot \right)$ and thus it follows by Lemma \ref{lemma:op_prp_compose} that the current and previous hybrids are computationally indistinguishable, with indistinguishability $\frac{1}{f(\kappa)}$.

\begin{itemize}
    \item $\Hyb_{2}$: For every $y \in \bbZ_{2}^{d}$, taking $\matA(y)$ to be the direct output of the PRF $\prf$, by using an obfuscated punctured PRF.
\end{itemize}
In order to describe the change between the current and previous hybrid we first recall the structure of the circuits from the previous hybrid: In the previous hybrid, for every $y \in \bbZ_{2}^{r}$ we defined $\matA(y) := \matC(y) \cdot \left( \begin{array}{cc} \overline{\matA}(y) & \\ & \Id_{n - r - s} \end{array} \right)$, where $\matC(y) \in \bbZ_{2}^{k \times k}$ comes from the output $\prf\left( k_{\matC}, y \right)$ and $\overline{\matA}(y) \in \bbZ_{2}^{(k - (n - r - s)) \times s}$ is the output of the inner PRF $\overline{\prf}\left( \overline{k}_{\sf lin} \right)$ (which in turn comes from the inside of $\left( \overline{\Ps}, \overline{\Ps}^{-1} \right)$). In the current hybrid we are going to ignore the PRFs $\prf\left( k_{\matC}, y \right)$ and $\overline{\prf}\left( \overline{k}_{\sf lin} \right)$ and their generated values $\matC(y)$, $\vecD(y)$ and $\overline{\matA}(y)$ and instead, sample a fresh key $k_{\matA}$, and on query $y$ generate $\matA(y) \gets \prf\left( k_{\matA}, y \right)$, for $\matA(y) \in \bbZ_{2}^{k \times (n - r)}$.

First, note that the following two ways to sample $\matA(y)$, are statistically equivalent for every $y$: (1) The matrix $\matA(y)$ is generated by sampling a random full-rank matrix $\matC(y) \in \bbZ_{2}^{k \times k}$ and letting $\matA(y)$ be $\matC(y) \cdot \left( \begin{array}{cc} \overline{\matA}(y) & \\ & \Id_{n - r - s} \end{array} \right)$. (2) For every $y \in \bbZ_{2}^{r}$ just sample a full-rank matrix $\matA(y) \in \bbZ_{2}^{k \times (n - r)}$. This means that when truly random bits are used for generating $\matA(y)$ in the two cases, the distributions are statistically equivalent.

To see why the two distributions are computationally indistinguishable, a different description of the previous hybrid can be given as follows: We can consider a sampler $E_{0}$ that for every $y \in \{ 0, 1 \}^d$ samples $\matA(y)$  according to the first algorithm, and another sampler $E_{1}$ that samples $\matA(y)$ according to the second algorithm, and we know that for every $y$ (and recall there are $2^r$ actual values of $y$ which can appear as the output, and not $2^d$) the outputs of $E_{0}$ and $E_{1}$ are statistically indistinguishable. 

Since there are $2^{r}$ valid values for $y$, by Lemma \ref{lem:distswap}, the current hybrid is computationally indistinguishable from the previous, with indistinguishability $\frac{2^{r}}{f(\kappa)}$.

\begin{itemize}
    \item $\Hyb_{3}$: Discarding the inner obfuscations $\left( \overline{\Ps}, \overline{\Ps}^{-1} \right)$ completely, by using the security of the outer obfuscator.
\end{itemize}
The change between the current hybrid and the previous is that in the current hybrid we generate the circuits $P, P^{-1}, D'$ without using $\left( \overline{\Ps}, \overline{\Ps}^{-1} \right)$ at all. Note that this is possible, since in the previous hybrid, we moved to a circuit that did not use access to the circuits $\left( \overline{\Ps}, \overline{\Ps}^{-1} \right)$ any longer during the execution of any of the three circuits $P, P^{-1}, D'$, except from using the second permutation $\overline{\prp}_{\sf out}$, which acts on $\{ 0, 1 \}^d$ and does not need to act from inside the inner obfuscations $\left( \overline{\Ps}, \overline{\Ps}^{-1} \right)$ any more. This means that we can technically move the application of the inner permutation $\overline{\prp}_{\sf out}$ "outside of the inner circuits $\left( \overline{\Ps}, \overline{\Ps}^{-1} \right)$" and the functionality of the circuits $P, P^{-1}, D'$ did not change between the current and the previous hybrids, and thus, by the security of the indistinguishability obfuscator that obfuscates the three circuits, the current hybrid is computationally indistinguishable from the previous one, with indistinguishability of $\frac{1}{f(\kappa)}$.

\paragraph{Finalizing the reduction.}
Finally, observe that the distribution generated in the above $\Hyb_{3}$ is exactly a sample from $\widetilde{\gen}\left( 1^{\lambda}, n, r, k, s \right)$. Also observe that the outputs of $\Hyb_{0}$ and $\Hyb_{3}$ are $O\left( \frac{2^{r}}{f(\kappa)} \right)$-computationally indistinguishable. Recall that by the lemma's assumptions, with probability $\epsilon$, on a sample from $\widetilde{\gen}\left( 1^{\lambda}, n, r, k, s \right)$, the algorithm $\Adv$ outputs a pair $(x_{0}, x_{1})$ of $n$-bit strings such that $y_{0} = y_{1} := y$ and also $\left( \vecU_0 - \vecU_1 \right) \notin \colspan \left( \matA^{(1)}(y) \right)$. It follows that the probability for the same event when the input to $\Adv$ is generated by $\Hyb_{0}$, is at least $\epsilon - O\left( \frac{2^{r}}{f(\kappa)} \right) \geq \frac{\epsilon}{2}$, which finishes our proof.
\end{proof}

\subsection{Hardness of the Dual-free Case from LWE and iO} \label{subsection:standard_model_dualfree_to_lwe}
Here, we explain how to prove the collision resistance of the dual-free case (where the adversary sees $\Ps,\Ps^{-1}$, but not $\Ds$). This will follow the blueprint used in Section~\ref{subsection:dual_free_to_two_to_one_oracle} to reduce the dual-free case to a standard collision-resistance problem (which we recommend going over, before reading the following standard model security reduction). However, a few technical challenges arise. Some of these are due to needing certain steps to be efficient, where they are naively inefficient in Section~\ref{subsection:dual_free_to_two_to_one_oracle}. Another issue is that the underlying 2-to-1 function we use based on LWE~\cite{FOCS:BCMVV18} is not uniformly 2-to-1, but rather has some points that have no collisions and are only $1$-to-$1$. By careful arguments using the novel permutable PRPs as a tool, we are nevertheless able to resolve these issues and prove security. In the following section all arithmetic is implicitly modulo $2$, unless explicitly noted otherwise (and then the arithmetic modulo or lack thereof will be specified).

\paragraph{The LWE-based hash function $\hashL$.}
Here, we recall an \emph{approximate} 2-to-1 function $\hashL$ based on LWE which is a simplified version of the \emph{noisy claw-free trapdoor function} developed in~\cite{FOCS:BCMVV18}. Let $u,v,\sigma,B,\overline{B},q$ be parameters with the relationships described in Equation~\ref{eq:lweparams}.

\begin{align}\label{eq:lweparams}
\sigma&=u^{\Omega(1)} , &
\overline{B}&=\sigma\times u^{\Omega(1)} , \\
B&\geq \overline{B}\times u^{\omega(1)} , &
q&\geq B\times u^{\Omega(1)}\nonumber , \\
v&\geq \Omega(u\log q) , &
\exists&\delta \in (0, 1) : \frac{q}{\sigma} \leq 2^{u^\delta} \enspace .
\nonumber
\end{align}

The keys for the hash function have the form $\pk = \left( \matB, \vecC \right)$, where $\matB \gets \Z_q^{v \times u}$ and $\vecC \gets \matB \cdot \vecS + \vecE \bmod q$ where $\vecS \gets \Z_q^u$ and the entries of $\vecE \in \Z_q^v$ are i.i.d. sampled from discrete Gaussians of width $\sigma$, which in turn are guaranteed (w.h.p) to have entries in $( -\overline{B}, \overline{B} ]$.

We define the function $\hashL\left( \pk, \cdot \right) : \bbZ_q^u \times (-B,B]^v \times \{0,1\} \rightarrow \bbZ_{q}^{v}$ as follows.
$$
\hashL\left( \left( \matB \in \bbZ_{q}^{v \times u}, \vecC \in \bbZ_{q}^{v} \right), \left( \vecT \in \bbZ_{q}^{u}, \vecF \in (-B,B]^v \right), b \in \{ 0, 1 \} \right) =  \matB \cdot \vecT + \vecF + b\cdot \vecC \bmod q
\enspace .
$$
By choosing $B,q$ to be powers of 2, we can map the domain and range to bit-strings.

\paragraph{$\hashL$ is collision resistant.}
Observe that $\hashL\left( \pk, \cdot \right)$ is collision resistant, under the LWE assumption. This is because (with probability exponentially close to $1$ in the security parameter, over sampling $\pk$) a collision is always of the form $\left( \vecT'_0, 0 \right)$, $\left( \vecT'_1, 1 \right)$, and subtraction reveals $\vecS' = \vecT_0 - \vecT_1$, which breaks the assumption that search LWE with sub-exponential modulus/noise ratio is hard for quantum polynomial-time algorithms.

\paragraph{$\hashL$ is approximately $2$-to-$1$.}
Observe that with overwhelming probability over sampling $\pk$, the function $\hashL\left( \pk, \cdot \right)$ is \emph{almost} 2-to-1. Formally, with overwhelming probability over sampling $\vecE$, $\matB$ we have that for any two different lattice points $\matB \cdot \vecU$, $\matB \cdot \vecV$, the sets $\matB \cdot \vecU + (-B, B]^{v} \bmod q$, $\matB \cdot \vecV + (-B, B]^{v} \bmod q$ are non-intersecting (this follows because $v \geq \Omega\left( u\log q \right)$ is known to imply that the lattice is sparse with high probability). This in turn implies that for any point in the set $S := \{ \matB \cdot \vecU + ( -(B - \| \vecE \|), (B - \| \vecE \|) )^{v} \: | \: \vecU \in \bbZ_{q}^{u} \}$, the only two points in the domain of $\hashL$ that reach it are $\left( \vecT', 0 \right)$, $\left( \vecT' - \vecS', 1 \right)$. Also, due to our parameter choices, (1) From $\overline{B} \geq \sigma\times u^{\Omega(1)}$ it follows that with high probability over sampling $\vecE$, the norm $\| \vecE \|$ is bounded by $\overline{B}$ and (2) From $B \geq \overline{B}\times u^{\omega(1)}$ it follows that $\overline{B}$ (and thus also $\| \vecE \|$) is negligible compared to $B$, which makes the set $S$ be an overwhelming fraction of the entire range $\{ \matB \cdot \vecU + ( -B, B ]^{v} \: | \: \vecU \in \bbZ_{q}^{u} \}$.

\paragraph{$\hashL$ has a trapdoor to find preimage sets.}
Similarly to the noisy trapdoor claw-free functions from \cite{FOCS:BCMVV18}, the function $\hashL$ has a trapdoor that can invert it. Specifically, it is known how to efficiently sample a pair $\left( \pk, \td \right)$ where $\pk$ is statistically indistinguishable from an honestly sampled public key for $\hashL$, and given $y$ an element in the image of $\hashL\left( \pk, \cdot \right)$ and $\td$, it is possible to efficiently compute the preimage set of $y$ (which, as we know by now, is of size either $1$ or $2$).

\begin{theorem} \label{thm:reduce2lwe}
Let $n, r, k \in \Nat$ such that $r < n \leq k$ and also $\frac{n}{n - r}$ is an integer (which implies that $\frac{r}{n - r}$ is also an integer). Suppose the hardness of LWE against quantum polynomial-time algorithms, for the parameters as set in Equation~\ref{eq:lweparams}. Note that this in particular implies the collision resistance of the LWE-based function $\hashL$ defined above.

Then, for every quantum polynomial-time algorithm $\Adv$ there exists a negligible function $\negl$, such that the probability to get $\left( \Ps,\Ps^{-1} \right)$ sampled from $\left( \Ps,\Ps^{-1}, \Ds \right) \gets \widetilde{\gen}\left( 1^\secp, n, r, k, 0 \right)$, and then output a collision in $\hash$ (the function derived from $\Ps$), is $\negl\left( \frac{n}{n - r} \right)$.
\end{theorem}

\begin{proof}
Suppose there exists an adversary $\As$ which, given $\left( \Ps,\Ps^{-1} \right)$ as sampled from $\left( \Ps,\Ps^{-1}, \Ds \right) \gets \widetilde{\gen}\left( 1^\secp, n, r, k, 0 \right)$ (recall that $\Adv$ gets the sampled circuits without the dual $\Ds$), finds a collision in the associated function $\hash(x) := \prp^{-1}\left( k_{\sf out}, \; H(x) || 0^{d-r} \right)$ with non-negligible probability $\epsilon$. We will describe an adversary $\AdvB$ that violates the quantum hardness of LWE, using the adversary $\Adv$.

\paragraph{The collision resistant, approximate coset partition, trapdoor hash $\hashQ$.}
We will use the above function $\hashL$ in order to construct the function $\hashQ$.
We start with defining an instance of $\hashL$: We choose $u, q, B, v$ in a way which satisfies Equation~\ref{eq:lweparams} and also such that $2 \cdot q^u \cdot \left( 2B \right)^v = 2^{n/(n-r)}$. In other words, this makes the domain of the hash function $\hashL$ be exactly $\frac{n}{n - r}$ bits, and its output size be $v \cdot \log_{2}(q)$ bits (recall we assume that $\frac{n}{n - r}$ is an integer and we set $q$ to be a power of $2$). There is one more requirement for our parameters: We take $d$ (recall that $d$ is the number of bits that the second permutation $\Pi\left( k_{\sf out}, \cdot \right)$ acts on) to be at least the sum of the sizes (in bits) of the input and output of $\hashL$, multiplied by $(n - r)$, that is,
$$
\left( n - r \right) \cdot
\left(
\frac{n}{n - r}
+
v \cdot \log(q)
\right)
=
n
+
\left( n - r \right) \cdot v \cdot \log(q)
\leq
d
\enspace .
$$

We set $\hashQ$ to be a $(n - r)$-parallel repetition of $\hashL$, namely we sample $n - r$ i.i.d. public keys $\pk_{1}, \cdots, \pk_{n - r}$ for $\hashL$, which form the key $\pk$ for $\hashQ$, and compute the $n - r$ outputs of $\hashL$ given $n - r$ inputs.\footnote{The construction of $\hashQ$ from $\hashL$ is almost identical to the oracle construction of a coset partition function from Theorem \ref{theorem:n_to_ell_CPF_collision_resistant}, where we construct a collision-resistant coset partition function from a random $2$-to-$1$ function.}  It follows that $\hashQ$ maps $n$ input bits to $\left( n - r \right)\cdot v \cdot \log_{2}\left( q \right)$ output bits, and furthermore by our setting of parameters, $d$ is at least the sum of sizes (in bits) of the input and output sizes of $\hashQ$ -- this will come in handy later in the proof.

As a further specification of the format for computing $\hashQ$:, the input $\vecW \in \{ 0, 1 \}^{n}$ to $\hashQ$ is given in two parts: The first, vector part, consists of $r$ bits and is written $\left( \vecW_{1}, \cdots, \vecW_{n - r} \right)$ (note $\frac{r}{n - r} = \frac{n}{n - r} - 1$) where for each $i \in [n - r]$ we have $\vecW_{i} \in \{ 0, 1 \}^{\frac{r}{n - r}}$, and the second, coordinates part, consists of $n - r$ bits $\left( b_{1}, \cdots, b_{n - r} \right)$. For computing $\hashQ$, for every $i \in [n - r]$ we take $\vecW_{i}$ and $b_{i}$ and apply $\hashL\left( \pk_{i}, \cdot \right)$. We then take the output $\vecA_{i} \in \{ 0, 1 \}^{v \cdot \log(q)}$ and set it as the $i$-th packet in the output $\vecA \in \{ 0, 1 \}^{(n - r) \cdot v \cdot \log(q)}$ of $\hashQ$.

Note the following details about $\hashQ$:
\begin{itemize}
    \item
    \textbf{Collision resistance.} A collision in $\hashQ$ with input size $n$ constitutes a collision in $\hashL$ with input size $\frac{n}{n - r}$. It follows that based on the quantum hardness of LWE, for every quantum polynomial time adversary there is a negligible function $\negl$ such that the probability to find a collision in $\hashQ$ is bounded by $\negl\left( \frac{n}{n - r} \right)$.

    \item 
    \textbf{Approximate coset partition.} Due to $\hashL$ being approximately 2-to-1 and $\hashQ$ being a parallel repetition of it, it follows that $\hashQ$ is an approximate coset partition function (Definition \ref{definition:coset_partition_function}). Specifically, we know that for an overwhelming fraction of the outputs of $\hashL$, the preimage set is a coset of dimension $1$, thus it follows that for an overwhelming fraction of the outputs of $\hashQ$, the preimage set is a coset of dimension $n - r$.

    \item 
    \textbf{Existence of a (secret) trapdoor that allows efficient computation of coset description.}
    The concatenation of the trapdoors $\left( \td_{1}, \cdots, \td_{n - r} \right)$ of the $n - r$ instances of $\hashL$ gives us a trapdoor $\td$ for $\hashQ$ in the following sense: Given an element $\vecA$ in the image of $\hashQ$ and the trapdoor, one can compute a coset description of the preimage set of $\vecA$.
    
    To see this, recall that the output of $\hashQ$ is given by the concatenation of the $n - r$ outputs of $\hashL$, and for each of them, the number of preimages is either $1$ or $2$, and furthermore, these can be computed by the independent trapdoors for each of the instances of $\hashL$. Eventually, the computed coset description of $\vecA$ is given as follows: For each $i \in [n - r]$ and $b \in \{ 0, 1 \}$, let $\vecW_{i, b} \in \{ 0, 1 \}^{\frac{n}{n - r}}$ the preimage of the $i$-th part $\vecA_{i} \in \{ 0, 1 \}^{v \cdot \log\left( q \right)}$ (of $\vecA \in \{ 0, 1 \}^{(n - r) \cdot v \cdot \log_{2}\left( q \right)}$), in the function $\hashL\left( \pk_{i}, \cdot \right)$, such that it ends with $b$. If there is no such preimage that ends with $b$ we set $\vecW_{i, b} = 0^{\frac{n}{n - r}}$. Note that while it isn't necessarily the case that there are two images (one for each $b \in \{ 0, 1 \}$), it is guaranteed that at least one of them always exists.
    \begin{itemize}
        \item The coset shift is set to
        $$
        \overline{\vecB}_{\vecA}
        :=
        \sum_{i \in [n - r]} \vecW'_{i, 0}
        \enspace ,
        $$
        where the vector $\vecW'_{i, 0} \in \bbZ_{2}^{n}$ is derived from $\vecW_{i, 0} \in \bbZ_{2}^{\frac{n}{n - r}}$ by using the previously described input format to $\hashQ$. Namely, we split $\vecW'_{i, 0} \in \bbZ_{2}^{n}$ into two parts, first part of $r$ bits and a second part of $n - r$ bits. We put the first $\frac{n}{n - r} - 1 = \frac{r}{n - r}$ bits of $\vecW_{i, 0}$ as the $i$-th packet in the first, $r$-bit part of $\vecW'_{i, 0}$, and the remaining bit of $\vecW_{i, 0}$ we set as the $i$-th bit in the second, $n - r$-bit part of $\vecW'_{i, 0}$, and the rest of the bits of $\vecW'_{i, 0}$ are set to $0$.

        \item 
        As for the subspace associated with the coset, denote by $\overline{\matA}_{\vecA} \in \bbZ_{2}^{n \times n - r}$ the subspace basis, and for $i \in [n - r]$ we set the $i$-th column of the basis to be $\vecW'_{i, 0} + \vecW'_{i, 1}$, where for $b \in \{ 0, 1 \}$, the vector $\vecW'_{i, b} \in \{ 0, 1 \}^{n}$ is derived from $\vecW_{i, b} \in \{ 0, 1 \}^{ \frac{n}{n - r} }$ in the same way described above when computing the coset shift.
    \end{itemize}

    \item 
    \textbf{Public efficient computation of coordinate vector for preimage element.}
    $\hashQ$ allows anyone to efficiently compute coordinate vectors of preimages with respect to its corresponding cosets. Given any input $\vecW \in \{ 0, 1 \}^{n}$ to $\hashQ$, one can take the last $n - r$ bits of $\vecW$ to obtain a coordinates vector $\vecZ \in \bbZ_{2}^{n - r}$. Then, if one further gets access to the coset description $\left( \overline{\matA}_{\vecA}, \overline{\vecB}_{\vecA} \right)$, we have the equality 
    $$
    \vecW
    =
    \overline{\matA}_{\vecA} \cdot \vecZ + \overline{\vecB}_{\vecA} \enspace .
    $$
    In that sense, the vector $\vecZ$ which is efficiently computable from the preimage $\vecW$ is the coordinates vector of $\vecW$ with respect to the coset $\left( \overline{\matA}_{\vecA}, \overline{\vecB}_{\vecA} \right)$.
\end{itemize}

Now, given our definition of the hash function $\hashQ$ and its above special properties, we are ready to describe our reduction.

\paragraph{The reduction $\AdvB$ from collision finding in $\hashQ$ to collision finding in $\hash$.}
The following reduction is an adaptation of the oracle reduction from Theorem \ref{thm:dualfreetocol} to the standard model, using LWE, obfuscation and permutable PRPs.
Given $\pk = \left( \pk_{1}, \cdots, \pk_{n - r} \right)$ a public key for $\hashQ$ (where $\pk_{i}$ is the i.i.d. sampled public key for $\hashL$), the reduction $\AdvB$ samples permutable PRP keys $k_{\sf in}$, $k_{\sf out}$ and a puncturable PRF key $k'_{\sf lin}$. Denote by $\ell := \left( n - r \right) \cdot v \cdot \log_{2}(q)$ the output size of $\hashQ$, and we next define the functions $P$, $P^{-1}$ which we then obfuscate to get the circuits $\Ps$, $\Ps^{-1}$, which we will feed to $\As$.

\begin{itemize}
    \item $P\left( x \in \{ 0, 1 \}^n \right)$:
    \begin{enumerate}
        \item 
        $\vecW \gets \Pi\left( k_{\sf in}, x \right)$,
    
        \item
        $\vecA \gets \hashQ\left( \pk, \vecW \right)$,

        \item 
        $y \gets \Pi^{-1}\left( k_{\sf out}, \left( \vecA || 0^{d - \ell} \right) \right)$,

        \item 
        $\left( \matC_y \in \bbZ_{2}^{k \times n}, \vecD_y \in \bbZ_{2}^{k} \right) \gets \prf\left( k'_{\sf lin}, y \right)$,

        \item 
        $\vecU \gets \matC_{y} \cdot \vecW + \vecD_{y}$,

        \item 
        Output $\left( y \in \{ 0, 1 \}^{d}, \vecU \in \bbZ_{2}^{k} \right)$.
    \end{enumerate}

    \item $P^{-1}\left( y \in \{ 0, 1 \}^{d}, \vecU \in \bbZ_{2}^{k} \right)$:
    \begin{enumerate}
        \item
        $x \gets
        \begin{cases}
        \Pi^{-1}\left( k_{\sf in}, \vecW \right)
        &\text{ $\exists \vecW \in \bbZ_{2}^{n}$ such that $\matC_y \cdot \vecW + \vecD_y = \vecU$} \\
        \bot
        &\text{ if no such $\vecW$ exists}
        \end{cases}$
        
        \item
        Output
        $\begin{cases}
        x &\text{ if $x \neq \bot$ and $y = \Pi^{-1}\left( k_{\sf out}, \left( \vecA || 0^{d - \ell} \right) \right)$, for  $\vecA \gets \hashQ\left( \pk, \vecW \right)$ } \\
        \bot &\text{ otherwise }
        \end{cases}$
    \end{enumerate}

\end{itemize}

$\AdvB$ obfuscates the two circuits to get $\Ps$, $\Ps^{-1}$ and executes $\left( x_{0}, x_{1} \right) \gets \Adv\left( \Ps, \Ps^{-1} \right)$. The output of $\AdvB$ is $\left( \vecW_{0} := \Pi\left( k_{\sf in}, x_{0} \right), \vecW_{1} := \Pi\left( k_{\sf in}, x_{1} \right) \right)$ as a collision in $\hashQ$. Let $\epsilon_{\AdvB}$ the probability that $\AdvB$ outputs a collision in $\hashQ$ and we will show that if $\epsilon$ is non-negligible then so is $\epsilon_{\AdvB}$. We next define a sequence of hybrid experiments. Each hybrid defines a computational process, an output of the process and a predicate computed on the process output. The predicate defines whether the (hybrid) process execution was successful or not.

\begin{itemize}
    \item
    $\Hyb_{0}$: The original execution of the reduction $\AdvB$. 
\end{itemize}
Here, $\pk$ is sampled as a public key for the function $\hashQ$ and we execute $\left( \vecW_{0}, \vecW_{1} \right) \gets \AdvB\left( \pk \right)$. The output of this hybrid is $\left( \vecW_{0}, \vecW_{1} \right)$ and the process is defined as successful if the pair constitutes a collision in $\hashQ\left( \pk, \cdot \right)$. By definition, the success probability of this hybrid is $\epsilon_{0} := \epsilon_{\AdvB}$. The program $P$ in this hybrid is described in Figure \ref{figure:hyb_0}.

\begin{figure}
\centering\includegraphics[width=12cm]{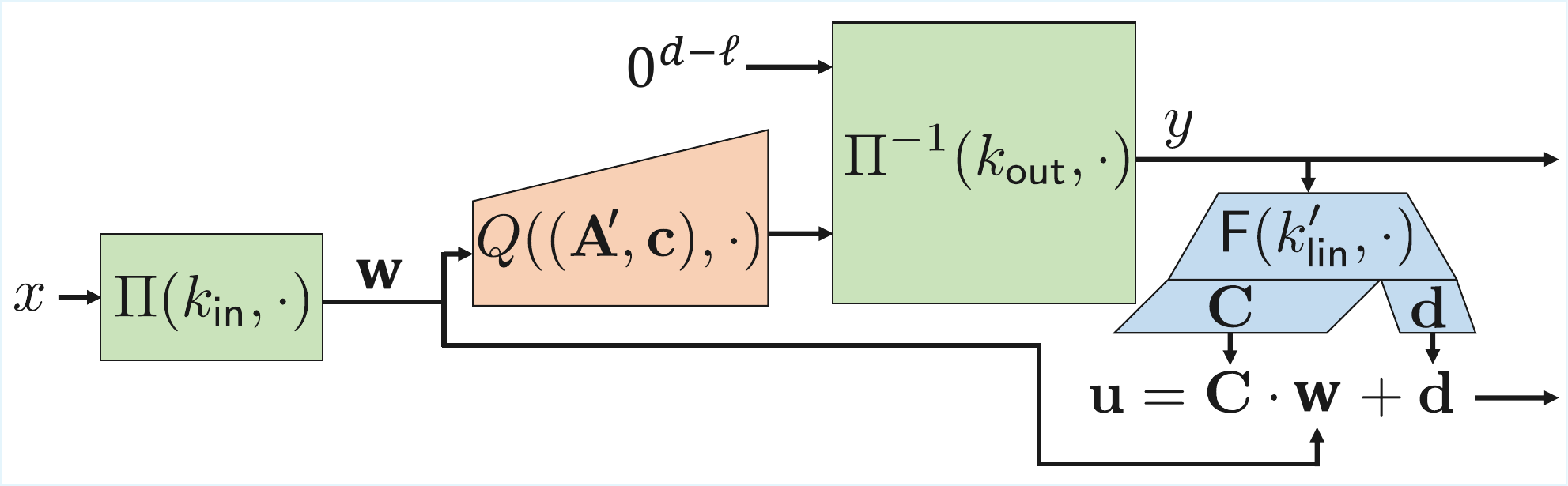}
\caption{\label{figure:hyb_0}
The obfuscated program $P$ in Hybrid 0. We simply apply $\hashQ$ to the output $\vecW$ of $\prp\left( k_{\sf in}, \cdot \right)$. The output $\vecU$ is set to a pseudorandomly-generated affine function of $\vecW$.} 
\end{figure}

\begin{itemize}
    \item
    $\Hyb_{1}$: Generating $\vecU$ from a coordinates vector $\vecZ$ and the coset $\left( \overline{\matA}_{\vecA}, \overline{\vecB}_{\vecA} \right)$ instead of the preimage $\vecW$, by using the trapdoor of $\hashQ$ and the security of the iO. Also, changing back to asking for collisions in the original $\Ps$.
\end{itemize}
Here, the public key $\pk$ for $\hashQ$ is sampled together with a trapdoor $\td$. The change we make to the following hybrid is this: In the previous hybrid, we took $\vecA$ and computed $y \in \{ 0, 1 \}^{d}$, then computed $\left( \matC_y, \vecD_y \right) \gets \prf\left( k'_{\sf lin}, y \right)$ and set $\vecU \gets \matC_{y} \cdot \vecW + \vecD_{y}$. Here, we compute $\left( \matC_y, \vecD_y \right)$ all the same, but instead of using $\vecW$ to compute $\vecU$, we use the coordinates vector $\vecZ$ of $\vecW$, the output $y$ and the trapdoor $\td$. In other words, we will maintain the same functionality but move to a circuit structure that more resembles the original distribution $P$, $P^{-1}$, by (1) having a circuit that given $y \in \{ 0, 1 \}^{d}$ computes the pseudorandom coset description $\left( \matA_{y}, \vecB_{y} \right)$, and the vector $\vecU$ is generated by $\matA_{y} \cdot \vecZ + \vecB_{y}$ where $\vecZ$ will be the last $n - r$ output bits of the first permutation $\Pi\left( k_{\sf in}, \cdot \right)$.

Specifically, recall that given an image $\vecA \in \{ 0, 1 \}^{(n - r) \cdot v \cdot \log(q)}$ of $\hashQ\left( \pk, \cdot \right)$ and the trapdoor $\td$, we can efficiently compute the coset $\left( \overline{\matA}_{\vecA}, \overline{\vecB}_{\vecA} \right)$ and also, given an input $\vecW$ we can compute $\vecZ$ the coordinates vector of $\vecW$ with respect to the coset of the output $\vecA$. 
Given $y$ we use the key $k'_{\sf lin}$ to compute $\vecA$ and then $\td$ to compute $\left( \overline{\matA}_{\vecA}, \overline{\vecB}_{\vecA} \right)$. Now, we take $\vecZ$ the last $n - r$ bits of $\vecW$ and set $\vecU \gets \matA_{y} \cdot \vecZ + \vecD_{y}$ for $\matA_{y} := \matC_{y} \cdot \overline{\matA}_{\vecA}$, $\vecD_{y} := \vecD_{y} + \matC_{y} \cdot \overline{\vecB}_{\vecA}$.

The sampling of $\pk$ with a trapdoor $\td$ is statistically indistinguishable from sampling $\pk$ without one, and furthermore by the correctness of the trapdoor, the functionality of the circuits did not change. Thus, the circuits given to $\Adv$ in $\Hyb_{0}$ are computationally indistinguishable by the security of the iO that obfuscates the circuits $P$, $P^{-1}$. It follows in particular that the success probability of the current process is $:= \epsilon_{1}$ such that $\epsilon_{\AdvB} \geq \epsilon_{1} - \negl\left( \secp \right)$.

Another change that we make to this hybrid the definition of successful execution: Instead of asking for collisions in $\hashQ$, we ask for collisions in the original $\Ps$. Since we move back and forth between collisions in $\hashQ$ and $\Ps$ using a permutation, any collision in $\hashQ$ can be translated into a collision in $\Ps$. Formally the reduction $\AdvB$ changes minimally: After executing $\Adv$ and obtaining $\left( x_{0}, x_{1} \right)$ we do not apply $\Pi\left( k_{\sf in}, \cdot \right)$ in the end to obtain $\vecW_{0}$, $\vecW_{1}$. The success probability still satisfies $\epsilon_{\AdvB} \geq \epsilon_{1} - \negl\left( \secp \right)$.

\begin{itemize}
    \item
    $\Hyb_{2}$: Discarding the trapdoor $\td$ and computing the coset $\left( \matA_{y}, \vecB_{y} \right)$ as a function of $y$ alone, using an obfuscated puncturable PRF.
\end{itemize}
The change we will make to the following hybrid is to the circuit that samples the coset $\left( \matA_{y}, \vecB_{y} \right)$. Specifically, instead of sampling the coset through the process from previous hybrid (setting $\matA_{y} := \matC_{y} \cdot \overline{\matA}_{\vecA}$, $\vecD_{y} := \vecD_{y} + \matC_{y} \cdot \overline{\vecB}_{\vecA}$ for the pair $\left( \overline{\matA}_{\vecA}, \overline{\vecB}_{\vecB} \right)$ arising from the trapdoor $\td$ and the pseudorandomly generated $\left( \matC_y, \vecD_y \right) \gets \prf\left( k'_{\sf lin}, y \right)$), we just sample a fresh puncturable PRF key $k_{\sf lin}$ and sample the coset description from scratch $\left( \matA_y \in \bbZ_{2}^{k \times (n - r)}, \vecB_y \in \bbZ_{2}^{k} \right) \gets \prf\left( k_{\sf lin}, y \right)$.

Note that in the first distribution (arising from the sampling process of $\Hyb_{1}$), the matrix $\overline{\matA}_{\vecA}$ is full rank and thus for a truly random $\left( \matC_y, \vecD_y \right)$, the pair $\left( \matA_y, \vecB_{y} \right)$ is a truly random coset. To see why the two distributions are computationally indistinguishable, a different description of the previous hybrid can be given as follows: We can consider a sampler $E_{0}$ that for every $y \in \{ 0, 1 \}^d$ samples $\matA(y)$  according to the first algorithm, and another sampler $E_{1}$ that samples $\matA(y)$ according to the second algorithm, and we know that for every $y$ (and recall there are at most $2^\ell$ actual values of $y$ which can appear as the output, and not $2^d$) the outputs of $E_{0}$ and $E_{1}$ are statistically equivalent. 

Since there are $\leq 2^{\ell}$ valid values for $y$, by Lemma \ref{lem:distswap}, the current hybrid is computationally indistinguishable from the previous, with indistinguishability $\frac{2^{\ell}}{f(\kappa)}$. It follows in particular that the success probability of the current process is $:= \epsilon_{2}$ such that $\epsilon_{\AdvB} \geq \epsilon_{2} - \negl\left( \secp \right)$. Note that once we moved to this hybrid, the program $P$ in the current hybrid applies the original circuit $P$ from the Construction \ref{constr:standard}, only applying $\hashQ$ in the middle, between the two permutations $\Pi\left( k_{\sf in}, \cdot \right)$ and $\Pi\left( k_{\sf out}, \cdot \right)$, this is described in Figure \ref{figure:hyb_2}. The rest of our hybrids are intended to get rid of $\hashQ$ and gradually embed it into the obfuscated circuits.

\begin{figure}
\centering\includegraphics[width=12cm]{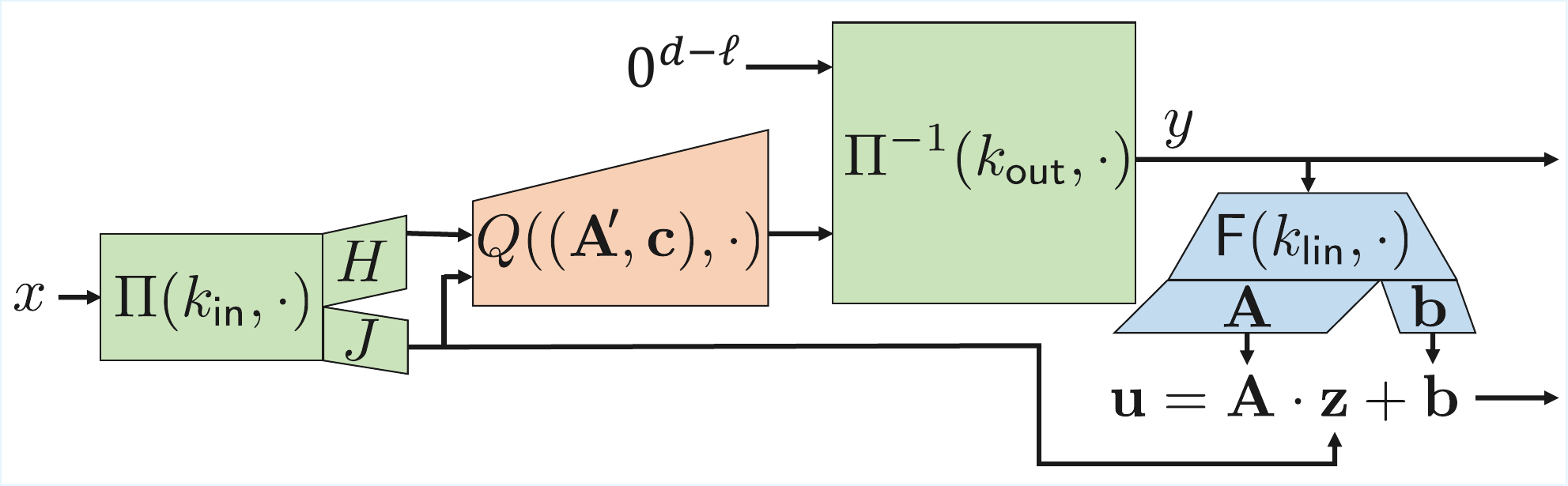}
\caption{\label{figure:hyb_2}
The obfuscated program $P$ in Hybrid 2. The output $\vecW$ of $\prp\left( k_{\sf in}, \cdot \right)$ is still the input to $\hashQ$. However, instead of generating the output vector $\vecU$ as a function of $\vecW$, it is a pseudorandomly-generated affine function of $\vecZ$, the last $n - r$ bits of $\vecW$.}
\end{figure}

\begin{itemize}
    \item
    $\Hyb_{3}$: Describing a circuit $C$ with equivalent functionality to $\hashQ$ and switching to using it instead, by using the security of iO.
\end{itemize}
We next describe a circuit $C : \{ 0, 1 \}^{n} \rightarrow \{ 0, 1 \}^{\ell}$ with identical functionality (but different circuit) than $\hashQ$. The goal of using $C$ is that (1) we can switch to it securely by the security of the iO that obfuscates the circuits $P$, $P^{-1}$, and (2) the structure of $C$ will be described in steps, such that each step is securely removable. Specifically, by the security of an obfuscated permutable PRP (in some instances using the security of the first PRP $\Pi\left( k_{\sf in}, \cdot \right)$, and sometimes the second one $\Pi\left( k_{\sf out}, \cdot \right)$) we will be able to provably remove it from the insides of the programs $P$, $P^{-1}$, using these steps. Recall that the only difference between the current hybrid and the original distribution over $P, P^{-1}$ (from Construction \ref{constr:standard}) is the circuit $C$ in between the permutations $\Pi\left( k_{\sf in}, \cdot \right)$ and $\Pi\left( k_{\sf out}, \cdot \right)$, and thus, once we remove all of its parts indistinguishably, our proof will be complete.

Recall that the original $\hashQ$ is a $(n - r)$-parallel repetition of $\hashL : \{ 0, 1 \}^{\frac{n}{n - r}} \rightarrow \{ 0, 1 \}^{v \cdot \log(q)}$, where for each $\hashL_{i}\left( \pk_{i}, \cdot \right) := L_{i}\left( \left( \matB_{i}, \vecC_{i} \right), \cdot \right)$ we have $\vecC_{i} = \matB_{i} \cdot \vecS_{i} + \vecE_{i}$ for the secret vectors $\vecS_{i} \in \bbZ_{q}^{u}$, $\vecE_{i} \in ( -B, B ]^{v}$. As we no longer rely on the security of $\hashL$ (nor $\hashQ$) for the remainder of the hybrids, we can use this knowledge in the following circuit $C$, which has the following 4 steps.
\begin{enumerate}
    \item
    \textbf{Permutation $\pi$, correctly acts on almost all inputs.}
    We apply a permutation $\pi$ which acts on $\{ 0, 1 \}^{n}$, and for each of the $n - r$ slices (each of which takes $\frac{n}{n - r}$ bits) acts as follows:
    \begin{align*}
    \left( \vecW_i \in \{ 0, 1 \}^{\frac{r}{n - r}}, b_i \in \{ 0, 1 \} \right)
    &:=
    \left( \vecT_i \in \bbZ_{q}^{u}, \vecF_i \in ( -B, B ]^{v}, b_i \in \{ 0, 1 \} \right)
    \\
    &\mapsto
    \left(
    \vecT_i' = \vecT_i + b_i \cdot \vecS_{i} \bmod q ,
    \; \vecF_i' = \vecF_{i} + b_i \cdot \vecE_{i} \bmod (-B,B],
    \; b_{i}
    \right)
    \\
    &:=
    \left(
    \vecW'_i \in \{ 0, 1 \}^{\frac{r}{n - r}},
    \; b_{i}
    \right)
    \enspace ,
    \end{align*}
    where $h \bmod (-B,B]$ means the unique integer $h' \in ( -B, B ]$ such that $h-h'$ is a multiple of $2B$. Thanks to the fact that we sum $\bmod \; (-B,B]$ in $\vecF'_{i}$ we can say that $\pi$ is indeed a permutation, but observe that for an exponentially small fraction of the inputs (specifically, inputs such that $\vecF_i' = \vecF_{i} + b_i \cdot \vecE_{i} \bmod \; (-B,B]$ does not equal $\vecF_{i} + b_i \cdot \vecE_{i} \bmod \; q$), we are losing information in $\vecF'_{i}$ with respect to the larger range $q$. Our next two steps of $C$ are intended to correct this loss of information.

    \item 
    \textbf{Computing correction signal for faulty inputs.}
    This step maps $n$ bits to $n$ bits in the following way, which sometimes erases information.
    Given $n$ input bits $\vecW \in \{ 0, 1 \}^n$ first we allocate $n - r$ new bits $w \in \{ 0, 1 \}^{n - r}$, initially set to value $0^{n - r}$. Intuitively, for $i \in [n - r]$ the bit $w_{i}$ records whether or not the adding of $\vecE_{i}$ to $\vecF_{i}$ modulo $(-B,B]$ from the previous step caused an out-of-bounds vector. Formally, we set $w_{i}$ to be $1$ iff both $b_{i} = 1$ and $\left( \vecF_{i}' - \vecE_{i} \bmod q \right) \notin ( -B, B ]^{v}$ (we call such instances of $\vecF_{i}$ that satisfy this condition, "bad"). The output of $C$ is $\left( \left( w_{1}, \cdots, w_{n - r} \right), \left( \vecW'_{1}, \cdots, \vecW'_{n - r} \right) \right)$. From hereon the circuit $C$ disregards $\vecZ := \left( b_{1}, \cdots, b_{n - r} \right)$ the last $n - r$ bits of $\vecW$.

    \item
    \textbf{Injective function $\tau$, correcting the faulty inputs from $\pi$.} 
    Let $\widetilde{B} := 3\cdot \left( B + \overline{B} \right)$ and recall $n := (n - r) + (n - r) \cdot \left( u \cdot \log(q) + v \cdot \log\left( 2\cdot B \right) \right)$.
    This part of the computation injectively maps $n$ input bits into $(n - r) \cdot \left( u \cdot \log(q) + v \cdot \log\left( 2\cdot \widetilde{B} \right) \right) > n$ output bits.
    $\tau$ computes in slices and for every $i \in [n - r]$ the input is $\left( w_{i}, \vecW'_{i} \right)$ where $\vecW'_{i} := \left( \vecT'_{i} \in \bbZ_{q}^{u}, \vecF'_{i} \in [-B, B) \right)$ and the output is $\left( \vecT''_{i} \in \bbZ_{q}^{u}, \vecF''_{i} \in [-\widetilde{B}, \widetilde{B}) \right)$ (discarding the bit $w_{i}$, but enlarging the range of $\vecF'_{i}$). Ideally, to compute the correction, all we need to do is
    $$
    \left( w_{i}, \vecF'_i \right)
    \mapsto
    \vecF_i'' = \left( \vecF'_{i} - w_{i}\cdot \vecE_{i} \bmod (-B,B] \right) + w_{i}\cdot \vecE_{i} \bmod ( -(B + \overline{B}), (B + \overline{B}) ]
    \enspace ,
    $$
    discarding the input bit $w_{i}$. However one can see that this will sometimes lose information and not be injective. Specifically, observe that only when $w_{1} = 1$ and it is \emph{not} the case that $\left( \vecF_{i}' - \vecE_{i} \bmod (-(B + \overline{B}), (B + \overline{B})] \right) \notin ( -B, B ]^{v}$, we lose information. So, we simply treat these inputs differently, by slightly artificially enlarging the output range of the function.
    
    Formally, for every $i \in [n - r]$ the function $\tau$ computes
    \begin{itemize}
        \item If $w_{i} = 0$ or $\left( \vecF_{i}' - \vecE_{i} \bmod (-(B + \overline{B}), (B + \overline{B})] \right) \notin ( -B, B ]^{v}$ then
        \begin{align*}
        &\left( w_{i} \in \{ 0, 1 \}, \vecT'_i \in \bbZ_{q}^{u}, \vecF'_i \in ( -B, B ]^{v} \right)
        \\
        &\mapsto
        \left(
        \vecT_i'' = \vecT'_i ,
        \; \vecF_i'' = \left( \vecF'_{i} - w_{i}\cdot \vecE_{i} \bmod (-B,B] \right) + w_{i}\cdot \vecE_{i} \bmod ( -(B + \overline{B}), (B + \overline{B}) ]
        \right)
        \enspace .
        \end{align*}

        \item Else, let $\mathbf{shift} \in ( -\widetilde{B}, \widetilde{B} ]^{v}$ the vector such that each of its coordinates has the same value $2\cdot \left( B + \overline{B} \right)$. Then, compute
        \begin{align*}
        &\left( w_{i} \in \{ 0, 1 \}, \vecT'_i \in \bbZ_{q}^{u}, \vecF'_i \in ( -B, B ]^{v} \right)
        \\
        &\mapsto
        \left(
        \vecT_i'' = \vecT'_i ,
        \; \vecF_i'' = \vecF'_{i} + \mathbf{shift} \bmod
        (-\widetilde{B}, \widetilde{B}]
        \right)
        \enspace .
        \end{align*}
    \end{itemize}
    
    One can verify that $\tau$ is indeed injective and efficiently invertible: If $\vecF''_{i} \in ( -B, B ]^{v}$ we know $w_{i} = 0$ and $\vecF'_{i} = \vecF''_{i}$. Otherwise if $\vecF''_{i} \in ( -(B + \overline{B}), (B + \overline{B}) ]^{v}$, we know that $w_{i} = 1$ and $\vecF'_{i} = \vecF''_{i} - \vecE_{i} \bmod ( -B, B ]^{v}$. Otherwise, if $\vecF''_{i} \in ( -\widetilde{B}, \widetilde{B} ]^{v}$, we know that $w_{i} = 1$ and $\vecF'_{i} = \vecF''_{i} - \mathbf{shift} \bmod ( -B, B ]^{v}$.

    One can verify that in terms of correctness, at the end of this step in $C$, we have 
    $$
    \vecT''_{i} = \vecT_{i} + b_{i} \cdot \vecS_{i} \bmod q
    , \enspace 
    \vecF''_{i} = \vecF_{i} + b_{i} \cdot \vecE_{i} \bmod q
    \enspace .
    $$

    \item 
    \textbf{Injective function $\gamma$ computing the final output.}
    The last step of $C$ injectively maps $(n - r) \cdot \left( u \cdot \log(q) + v \cdot \log\left( 2\cdot \widetilde{B} \right) \right)$ bits to $(n - r) \cdot v \cdot \log\left( q \right)$. The input is $\left( \vecW''_{1}, \cdots, \vecW''_{n - r} \right)$ and for each $i \in [n - r]$ we have $\vecW''_{i} := \left( \vecT''_{i} \in \bbZ_{q}^{u}, \vecF''_{i} \in ( -\widetilde{B}, \widetilde{B} ]^{v} \right)$ and the output is
    \begin{align*}
    \matB_{i} \cdot \vecT''_{i} + & \vecF''_{i} \bmod q
    \\&=
    \matB_{i} \cdot \vecT_{i} + \vecF_{i} + b_{i} \cdot \left( \matB_{i} \cdot \vecS_{i} + \vecE_{i} \right) \bmod q
    \\&:=
    \matB_{i} \cdot \vecT_{i} + \vecF_{i} + b_{i} \cdot \vecC_{i} \bmod q
    \\&:=
    \hashL\left( \left( \matB_{i}, \vecC_{i} \right), \left( \vecT_{i}, \vecF_{i}, b_{i} \right) \right)
    \enspace ,
    \end{align*}
    which implies the correct computation of the original circuit $\hashQ$.
\end{enumerate}

The circuit $C$ replaces $\hashQ$ in a specific way: We connect the output $\vecW$ of the permutation $\Pi\left( k_{\sf in}, \cdot \right)$ into $C$ (all $n$ bits, including $\vecZ$), then the output $\vecA \in \bbZ_{q}^{v}$ is padded with zeros and goes into $\Pi^{-1}\left( k_{\sf out}, \cdot \right)$ to get the value $y \in \{ 0, 1 \}^{d}$ as before, and the value $\vecZ$, which was not changed during the computation $C$, goes as usual to be a coordinate vector for the coset $\left( \matA_{y}, \vecB_{y} \right)$. The circuit $P$ in the current hybrid is described in Figure \ref{figure:circ_C}. One can verify that when $C$ is inserted to replace $\hashQ$ as described, the functionality of $P$, $P^{-1}$ does not change. It follows that the change of the obfuscated circuits $\Ps$, $\Ps^{-1}$ is indistinguishable by the security of the iO. It follows in particular that the success probability of the current process is $:= \epsilon_{3}$ such that $\epsilon_{\AdvB} \geq \epsilon_{3} - \negl\left( \secp \right)$.

\begin{figure}
\centering\includegraphics[width=12cm]{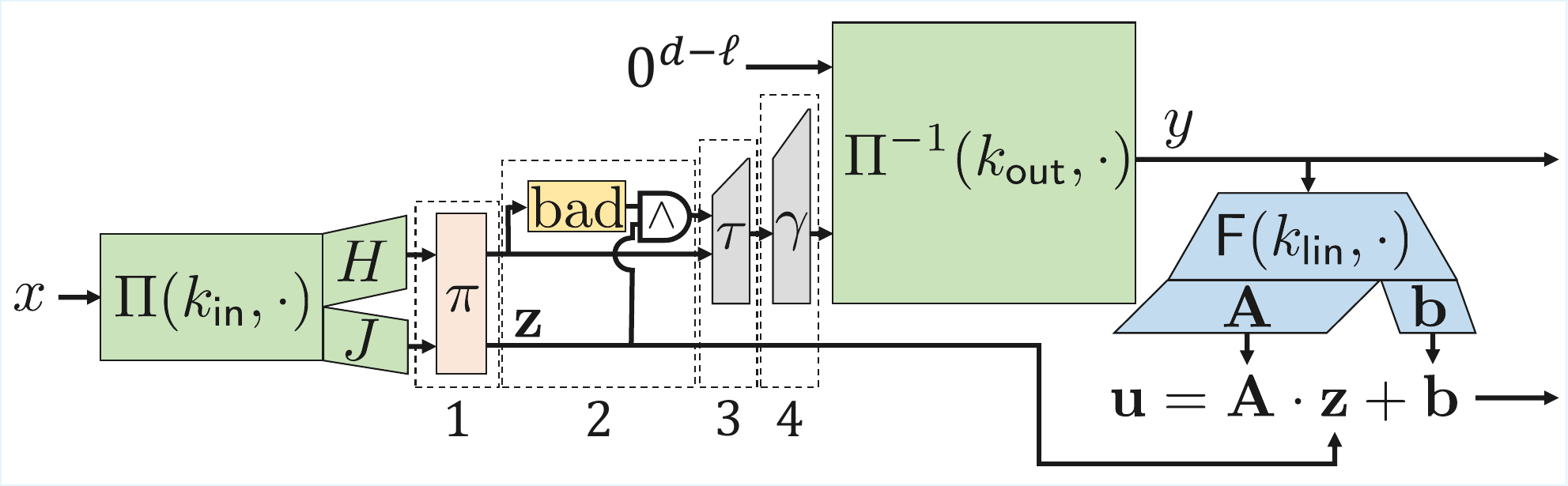}
\caption{\label{figure:circ_C}
The obfuscated program $P$ in Hybrid $3$. We insert the program $C$ instead of the original circuit for $\hashQ$, which was used in Hybrid $2$. There are $4$ steps to the computation of $C$: $1$ is the permutation $\pi$, $2$ is the correction signal, $3$ and $4$ are the injective functions $\tau$ and $\gamma$, respectively. The $4$ steps of the computation of $C$, presented in the picture, are formally described in Hybrid $3$.}
\end{figure}

\begin{itemize}
    \item
    $\Hyb_{4}$: Eliminating layers $1, 3$ and $4$ from the circuit $C$ (and keeping only layer $2$), by using permutable PRPs and the security of iO.
\end{itemize}
In this hybrid we will execute only layer $2$ of $C$ instead of the entire circuit $C$ (recall that layer $2$ computes the correction bits). To indistinguishably remove layers $1$, $3$, $4$ from $C$ we argue the following:
\begin{itemize}
    \item
    Both permutations $\Pi\left( k_{\sf in}, \cdot \right)$, $\Pi\left( k_{\sf out}, \cdot \right)$ are output-permutable, which means that $\Pi^{-1}\left( k_{\sf out}, \cdot \right)$ is input-permutable. One can verify that it thus only remains to explain why the permutation $\pi$ (from layer $1$ of $C$) is composable with $\Pi\left( k_{\sf in}, \cdot \right)$ and why the injective functions $\tau$, $\gamma$ (from layers $3$, $4$ in $C$, respectively) are composable with $\Pi^{-1}\left( k_{\sf out}, \cdot \right)$. This will invoke Lemma~\ref{lemma:op_prp_compose} which uses both the fact that permutations (or permutations that are extensions of injective functions that use ancillas of zeros, as explained in Section \ref{sec:prps}) are decomposable and the fact that the entire circuit is obfuscated under iO.

    \item 
    $\pi$ can be described as a controlled permutation, that under the controlled value (i.e., the value of $\vecZ \in \bbZ_{2}^{n - r}$) applies a (modular) addition of a vector. Since a modular addition of a vector is decomposable and $\pi$ is a controlled version of a decomposable permutation it is overall decomposable (all of these arguments are explained at the beginning of Section \ref{sec:prps}). Both functions $\tau$ and $\gamma$ are injective as explained as part of the description of $C$. Also, the input to $\Pi^{-1}\left( k_{\sf out}, \cdot \right)$ has size $d$ and is padded with zeros. Due to our parameter choices for $d$, one can verify that $d$ is at least as large as the sum of sizes for the input and output for both functions $\tau$ and $\gamma$. This means again that these functions are composable with the permutation.
\end{itemize}

The above means that we can indeed move to obfuscating only the computation of layer $2$ instead of the entirety of $C$. It follows in particular that the success probability of the current process is $:= \epsilon_{4}$ such that $\epsilon_{\AdvB} \geq \epsilon_{4} - \negl\left( \secp \right)$.

\begin{itemize}
    \item
    $\Hyb_{5}$: Eliminating the remaining layer $2$ from the circuit $C$, by using fixed permuted sparse triggers and the security of iO.
\end{itemize}
In this hybrid we will eliminate the final layer $2$ of the circuit $C$. Note that once we do this, the obfuscated programs that are given to $\Adv$ are identical to the ones given to it in the dual-free setting. To indistinguishably remove the final part of $C$ let us observe why it matches the setting of Lemma \ref{lemma:intervaltrigger}, which is also described in Figure \ref{figure:sparsetrigger}. Referring to the notations in Lemma \ref{lemma:intervaltrigger}, in our setting we do not have the circuits $P_0$ and $P_2$ and also there is no wire $w_1$. Recall that in previous hybrids , in our circuit we denoted by $\left( \vecW_{1}, \cdots, \vecW_{n - r} \right)$ the first $r$ output bits of $\Pi\left( k_{\sf in}, \cdot \right)$ and by $\vecZ \in \bbZ_{2}^{n - r}$ the remaining bits. Going back to matching to the setting of Lemma \ref{lemma:intervaltrigger}: The role of $x_2$ is taken by $\left( \vecW_{1}, \cdots, \vecW_{n - r} \right)$, and the role of $w_{2}$ is taken by $\vecZ$. Now we will consider two cases for the middle program $P_1$ (which will accordingly represent the two hybrids $\Hyb_{4}$ and $\Hyb_{5}$).

In one instance, the program in the middle is $P'_{1}$ and given $w_{2} := \vecZ$ and $x_{2} := \left( \vecW_{1}, \cdots, \vecW_{n - r} \right)$ uses the trigger program $R$ on $x_{2}$, which acts as follows: For every $i \in [n - r]$, if $\vecW_{i}$ contains a "bad" $\vecF_{i}$, then $R$ is triggered, outputs $1$ and then bit $i$ of $w_{3}$ is set to be the $i$-th bit of $w_{2}$. $w_{4}$ is always set to be $w_{2} := \vecZ$ regardless. Note that this exactly describes our process in $\Hyb_{5}$, as $w_{4}$ continues to be a coordinates vector for the coset $\left( \matA_{y}, \vecB_{y} \right)$ and $w_{3}$ continues as the signal for correction (which, in the original circuit $C$, was used in the function $\tau$).

In a different instance, the program in the middle is $P_{1}$ and given $w_{2} := \vecZ$ always outputs $w_{3} = 0^{n - r}$ (and similarly to the previous process, $w_{4}$ is always set to be $w_{2} := \vecZ$). Note that this exactly describes our $\Hyb_{4}$. The program $P^{-1}$ is changed analogously, just like the program $Q$ in Lemma \ref{lemma:intervaltrigger} is changed to $Q'$.

The above means that we can indeed move to obfuscating the original circuits $P, P^{-1}$ from the dual-free setting of Construction \ref{constr:standard}, and indistinguishably remove the remaining layer $2$ of $C$. It follows in particular that the success probability of the current process is $:= \epsilon_{5}$ such that $\epsilon_{\AdvB} \geq \epsilon_{5} - \negl\left( \secp \right)$. To conclude the proof, note that we assume that $\Adv$ succeeds in finding a collision with a non-negligible probability, i.e., $\epsilon_{5}$ is non-negligible, which makes $\epsilon_{\AdvB}$ non-negligible, which in turn implies that the quantum polynomial time algorithm $\AdvB$ breaks the collision resistance of $\hashQ$ (and thus search LWE), in contradiction.

\end{proof}

\fi

\ifllncs\else
\bibliographystyle{alpha}
\bibliography{bib,abbrev0,crypto}
\fi

\end{document}